\documentclass[%
reprint,
amsmath,amssymb,
aps,
prx,
superscriptaddress,
nofootinbib
]{revtex4-2}

\usepackage{comment}
\usepackage{graphicx}
\usepackage{dcolumn}
\usepackage{enumitem}
\usepackage{bm}
\usepackage{braket}
\usepackage{amsmath}
\usepackage{dsfont}
\usepackage[dvipsnames]{xcolor}
\usepackage{mathrsfs}
\usepackage[colorlinks=true,linkcolor=black,citecolor=NavyBlue,urlcolor=NavyBlue]{hyperref}
\usepackage{pifont}
\usepackage{hyperref}
\usepackage{natbib}
\usepackage{tikz}
\usepackage[normalem]{ulem}
\usetikzlibrary{quantikz2}
\usetikzlibrary{trees}
\usepackage{nicematrix}
\usepackage{blindtext}
\usepackage[colorinlistoftodos,prependcaption,textsize=tiny]{todonotes}
\usepackage{regexpatch}
\usepackage{physics}
\usepackage{placeins}
\makeatletter
\xpatchcmd{\@todo}{\setkeys{todonotes}{#1}}{\setkeys{todonotes}{inline,#1}}{}{}
\makeatother

\newcommand{\add}[1]{{\color{blue}#1}}

\newcommand{\nocontentsline}[3]{}
\newcommand{\tocless}[2]{\bgroup\let\addcontentsline=\nocontentsline#1{#2}\egroup}

\DeclareMathAlphabet{\mathdutchcal}{U}{dutchcal}{m}{n}

\usepackage{amsthm}
\theoremstyle{definition}

\theoremstyle{plain}
\newtheorem{theorem}{Theorem}
\newtheorem{lemma}{Lemma}

\newtheorem{proposition}{Proposition}

\definecolor{axiomatic}{HTML}{FCF3E5}
\definecolor{microscopic}{HTML}{e9eef1}
\definecolor{operational}{HTML}{fceeee}
\definecolor{theorem}{HTML}{E1F1EC}

\usepackage{tcolorbox}
\tcbuselibrary{breakable}
\newtcolorbox[]{defbox}[1][]{breakable,colback=RoyalBlue!10!white,colframe=
white,title=#1}

\newtcolorbox[]{standardbox}[1][]{breakable,colback=RoyalBlue!5!white,colframe=white,#1}
\newtcolorbox[]{axiomaticbox}[1][]{breakable,colback=axiomatic,colframe=white, #1}
\newtcolorbox[]{microscopicbox}[1][]{colback=microscopic,colframe=white,#1}
\newtcolorbox[]{operationalbox}[1][]{breakable,colback=operational,colframe=white,#1}
\newtcolorbox[]{theorembox}[1][]{breakable,colback=theorem,colframe=white,#1}

\def\BraVert{\egroup\,\mid\,\bgroup}

\def\ketbra#1#2{\ket{#1\vphantom{#2}}\!\bra{#2\vphantom{#1}}}

\def\bra#1{\mathinner{\langle{#1}|}}

\def\ket#1{\mathinner{|{#1}\rangle}}

\def\braket#1{\mathinner{\langle{#1}\rangle}}

\newcommand{\nI}{{\rm i}}
\newcommand{\nE}{{\rm e}}
\usepackage{calc}
\newcommand{\vardbtilde}[1]{\tilde{\raisebox{0pt}[0.85\height]{$\tilde{#1}$}}}

\begin{document}
\title{The Trinity of Markovian Quantum Thermodynamics:\\ Unifying the Axiomatic, Microscopic, and Operational Paradigms}
\author{Yutong Luo}
\email{yuluo@tcd.ie}
\affiliation{School of Physics, Trinity College Dublin, College Green, Dublin 2, D02 K8N4, Ireland}
\affiliation{Trinity Quantum Alliance, Unit 16, Trinity Technology and Enterprise Centre, Pearse Street, Dublin 2, D02 YN67, Ireland}
\author{Jakub Czartowski}
\email{jakub.czartowski@tcd.ie}
\affiliation{School of Physics, Trinity College Dublin, College Green, Dublin 2, D02 K8N4, Ireland}
\affiliation{Trinity Quantum Alliance, Unit 16, Trinity Technology and Enterprise Centre, Pearse Street, Dublin 2, D02 YN67, Ireland}
\author{Felix Hubmann}
\affiliation{Faculty of Physics, University of Vienna, Boltzmanngasse 5, A-1090 Vienna,
Austria}
\author{Simon Milz}
\email{s.milz@hw.ac.uk}
\affiliation{Institute of Photonics and Quantum Sciences, School of Engineering and Physical Sciences,
Heriot-Watt University, Edinburgh EH14 4AS, United Kingdom}
\author{Felix C. Binder}
\email{felix.binder@tcd.ie}
\affiliation{School of Physics, Trinity College Dublin, College Green, Dublin 2, D02 K8N4, Ireland}
\affiliation{Trinity Quantum Alliance, Unit 16, Trinity Technology and Enterprise Centre, Pearse Street, Dublin 2, D02 YN67, Ireland}

\date{\today}

\begin{abstract}
Thermodynamics imposes fundamental constraints on the evolution of quantum systems. These constraints and their dynamical consequences have been formulated within distinct paradigms, including axiomatic approaches based on quantum master equations, microscopic descriptions of open-system dynamics, and operational formulations rooted in resource theories. While each perspective has yielded important insights into thermodynamically consistent quantum dynamics, their precise relationship has remained unresolved. 
Here we establish the exact equivalence of these three paradigms in the Markovian regime. We prove that thermal Lindbladians satisfying Markovianity, time-translation symmetry, and quantum detailed balance are precisely those admitting a microscopic realisation as an energy-conserving thermal collision model and, equivalently, those generating Markovian thermal operations. This unifies the existing approaches to Markovian quantum thermodynamics and identifies its dynamical underpinnings.
We further provide an explicit microscopic protocol for simulating thermal Markovian processes with controlled finite-time simulation errors. We illustrate its applicability by providing faithful thermal collision-model implementations of a qubit thermalising in a bosonic environment and of a three-level autonomous thermal machine. In the latter case, the protocol gives rise to a finite-stroke thermal engine that not only reproduces the continuous-time dynamics but also its steady-state thermodynamic performance.
As a whole, these results establish a unified foundation for Markovian quantum thermodynamics, showing that its axiomatic, microscopic, and operational formulations are exactly equivalent and providing a universal protocol for implementing thermal processes and machines.
\end{abstract}
\maketitle

\section{Introduction}
Naturally occurring dynamical processes -- both classical and quantum -- obey fundamental thermodynamic principles. These principles determine not only the equilibrium state towards which a system evolves, but also \textit{how} it gets there. For example, a glass of cold water in warmer surroundings will steadily warm until it reaches its ambient temperature. It does not first freeze and then rapidly reheat, even though both trajectories end in the same thermal state. Likewise, an excited atom coupled to a bosonic field relaxes towards its steady state while simultaneously losing quantum coherence. Thermodynamics therefore constrains the entire evolution, not merely its endpoint. Characterising such thermodynamically consistent dynamics is a cornerstone for the field of (quantum) thermodynamics and provides the essential baseline for the rigorous treatment of non-equilibrium dynamics.

Typically, thermalisation arises through interactions between a system and its surrounding environment, which induce irreversible dynamics and relaxation. In the quantum regime, this process is most commonly studied in the limit of weak system--environment coupling and rapidly decaying environmental correlations, where the environment acts as a large memoryless reservoir and the system dynamics are effectively described by a Markovian master equation~\cite{davies1974markovian,davies1976markovian,breuer2002theory,Kosloff2013QuantumThermo}. Markovian dynamics therefore provide the standard framework for describing quantum thermalisation~\cite{Kosloff2013QuantumThermo}. Identifying which Markovian evolutions are compatible with thermodynamic principles, and the constraints these principles impose, is thus a central problem in quantum thermodynamics.

A priori, this question admits several distinct answers. In the literature, the resulting framework of \textit{Markovian quantum thermodynamics} has been formulated in three complementary ways:
\begin{itemize}
    \item[(i)] {\it axiomatically}, based on thermodynamically consistent structural constraints on quantum master equations~\cite{alicki2007quantum,fagnola2015entropy,fagnola2015entropy_DB,Lostaglio2017coherence,Alhambra2017Dynamicalmaps,dann_open_2021,dann2021quantum};
    \item[(ii)] {\it microscopically}, based on explicit open-system thermalisation models that give rise to thermodynamic master equations~\cite{davies1974markovian,davies1976markovian,Potts_2021,ciccarello2022quantum,cusumano2022quantum,oDonovan2025quantum,scandi_thermalization_2026};
    \item[(iii)] {\it operationally}, based on thermodynamically achievable state transformations without consuming thermodynamic resources~\cite{janzing2000thermodynamic,brandao2013resource,horodecki2013fundamental,lostaglio2019introductory,Lostaglio2022Continuous,korzekwa2022optimizing,son2024hierarchy}.
\end{itemize}
Despite their common physical motivation, these three paradigms are logically independent and individually incomplete. The axiomatic formulation, while structurally transparent, is abstract: it characterises the generators of Markovian master equations through thermodynamic constraints but does not generally provide an explicit microscopic mechanism -- in terms of an underlying system--environment dynamics -- that realises them. Conversely, the microscopic approach is constructive: It starts from a physical system--environment interaction yet typically proceeds by model-specific constructions and does not a priori clarify the structure of the resulting system dynamics. The operational framework, meanwhile, defines admissible processes through thermodynamic convertibility without the investment of resources, but is formulated at the level of input-output maps rather than generators of a Markovian master equation. Consequently, it does not directly reveal the structural form of the underlying dynamics.
It therefore remains unclear whether these perspectives merely intersect for particular examples or, more fundamentally, select the same class of Markovian thermal dynamics. While previous work has deduced the structural properties that energy conservation imposes on master equations~\cite{dann_open_2021}, a comprehensive unification of all three aspects of Markovian thermodynamics has thus far been lacking.

Unifying structural and microscopic descriptions of physical processes is a recurring theme in physics. For quantum channels, the Stinespring dilation theorem~\cite{Stinespring_1955} establishes the equivalence between completely positive trace-preserving evolutions and unitary dynamics on an enlarged Hilbert space. Likewise, the Gorini–Kossakowski–Sudarshan–Lindblad (GKSL) equation, which is central to this work, was originally derived from axiomatic semigroup considerations~\cite{lindblad1976generators, gorini1976completely}, but can also be obtained from a microscopic system--bath Hamiltonian under weak-coupling assumptions~\cite{davies1974markovian,davies1976markovian,breuer2002theory} or via collision models~\cite{Cattaneo2021CMcansimulate, ciccarello2022quantum, burgarth_control_2023}. In both cases, connecting complementary descriptions grounds abstract formulations in concrete physical models and reveals the structure of the underlying dynamics.

In the same vein, for a consistent framework of Markovian quantum thermodynamics, clarifying the relationship between axiomatic, microscopic, and operational descriptions is important both conceptually and practically. Markovian thermalisation processes are the elementary building blocks of quantum thermodynamics, underpinning thermal machines such as quantum heat engines and refrigerators~\cite{kosloff2014quantum,hofer2017markovian}, while also serving as standard phenomenological models of thermal noise in quantum devices~\cite{breuer2002theory,Mehl2013Noise,cleri2024quantum,naeij2025open}. They should therefore admit a unique physical interpretation independent of the theoretical framework used to describe them. Otherwise, the same process could be considered thermodynamically consistent from one perspective but not another, leading to conflicting conclusions about its physical implementation, resource requirements, and thermodynamic behaviour. For instance, a phenomenological master equation satisfying all expected thermodynamic constraints might nevertheless require athermal resources for its implementation, while a seemingly thermodynamically consistent microscopic model could generate a master equation lacking the structural features expected of genuine thermal relaxation.

In this work, we resolve these questions by establishing a precise equivalence between the axiomatic, microscopic, and operational descriptions of Markovian quantum thermodynamics, thereby providing a rigorous and systematic foundation for analysing, implementing, and simulating Markovian thermalisation in quantum systems.

\section{Summary of results}\label{sec:summary_of_results}
Our main results establish a unified framework for Markovian quantum thermodynamics by connecting its axiomatic, microscopic, and operational formulations. Under minimal structural assumptions on the system Hamiltonian, we prove that these seemingly distinct frameworks characterise exactly the same class of thermodynamically consistent dynamics. 
Consequently, \textit{any} phenomenological master equation satisfying the corresponding axioms is guaranteed to admit an underlying thermodynamically consistent realisation. Conversely, the dynamics stemming from such microscopic models are fully characterised by these axioms. And finally, both of these approaches coincide with the operational formulation of Markovian quantum thermodynamics based on memoryless thermal operations. Building on this equivalence, we develop an explicit protocol for implementing the dynamics identified by our equivalence theorem. Specifically, we construct energy-conserving thermal collision models that reproduce arbitrary thermodynamically consistent GKSL dynamics in the continuous-time limit. We further derive finite-collision-time error bounds for this construction and demonstrate its applicability through concrete single- and multi-bath examples and simulations.

In this paper, we take $\hbar = k_{\rm B} = 1$ and consider a finite-dimensional quantum system $S$ with Hamiltonian $H_S$, whose evolution from $t=0$ to $t=\tau$ is described by a dynamical map $\Lambda_\tau \equiv \overleftarrow{\rm T}\exp(\int_0^\tau \mathcal{L}_t\,{\rm d} t)$ with time-dependent generator $\mathcal{L}_t$ and $\overleftarrow{\rm T}$ denoting the chronological time-ordering. Throughout, the dynamics of the system is assumed to be \textit{memoryless} -- mirroring the fact that typical thermalisation happens via weakly coupling to a large heat reservoir. Consequently, $\mathcal{L}_t$ is the generator of a GKSL master equation modelling Markovian dynamics, with the evolution of a system state $\rho(t)$ given by $\dot \rho = \mathcal{L}_t(\rho(t))$. Determining the structural requirements that thermodynamic consistency imposes on this generator $\mathcal{L}_t$ in axiomatic, microscopic, and operational terms is one of the main aims of this paper.

We consider the following three essential formulations of Markovian quantum thermodynamics:
\begin{figure}
    \centering
    \includegraphics[width=\linewidth]{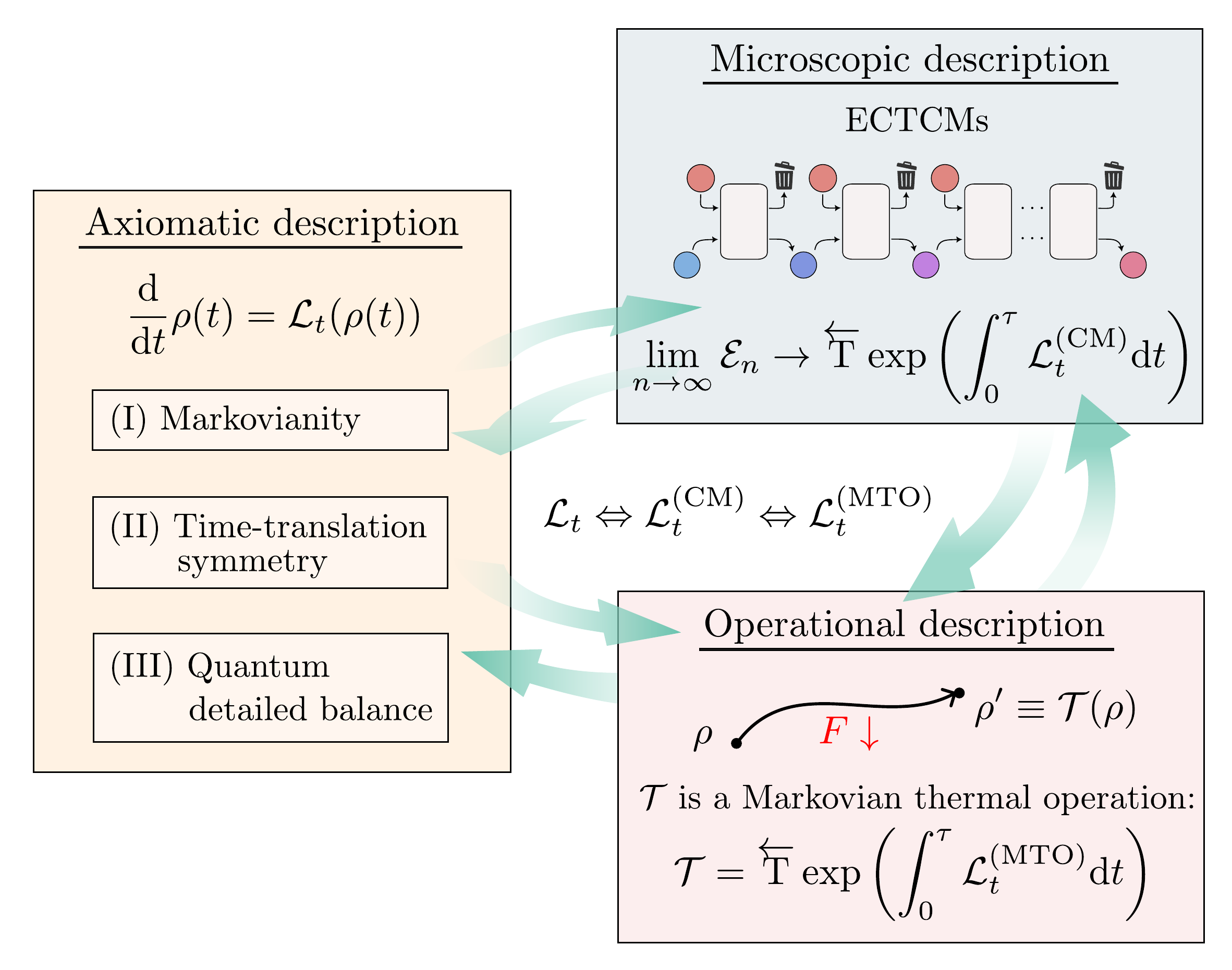}
    \caption{Equivalence between three paradigms of Markovian quantum thermodynamics. The axiomatic description (left) characterises the dynamics through a GKSL master equation subject to three thermodynamic axioms: Markovianity, time-translation symmetry, and quantum detailed balance. The microscopic description (top right) realises such dynamics via energy-conserving thermal collision models (ECTCMs), where repeated interactions with thermal ancillae converge in the continuous-time limit to a time-ordered exponential generated by $\mathcal{L}_t^{\rm (CM)}$. The operational description (bottom right) arises from the resource theory of athermality, where admissible dynamics are Markovian thermal operations (MTOs) generated by $\mathcal{L}_t^{\rm (MTO)}$ ensuring free energy monotonicity (indicated as $F\downarrow$). The arrows indicate the relationships established in this work: the sets of axiomatic generators, collision-model generators, and MTO generators coincide, such that they generate the same class of thermodynamically consistent dynamics. The three independent formulations of Markovian quantum thermodynamics are thereby unified.}
    \label{fig:equivalence}
\end{figure}

\noindent {\it Axiomatic description:} In the axiomatic formulation, the generators $\mathcal{L}_t$ are constrained by three thermodynamic axioms: Markovianity, time-translation symmetry, and quantum detailed balance. These conditions guarantee memoryless thermalisation and enforce invariance of the Gibbs state without reference to any underlying microscopic mechanism. A paradigmatic example satisfying these axioms are Davies generators, which have originally been derived in the weak-coupling limit of a system interacting with a thermal reservoir~\cite{davies1974markovian,davies1976markovian,alicki2007quantum}. The resulting dynamics decouples the evolution of energetic populations and coherence, and leaves the thermal state invariant.
Davies maps have become the canonical model for thermodynamically consistent Markovian dynamics~\cite{alicki2009thermalization,ROGA2010311,Alhambra2017Dynamicalmaps,Bardet2023RapidThermal,moroder_thermodynamics_2024}. Structurally similar master equations have been derived and analysed in Ref.~\cite{dann_open_2021} based on energy conservation assumptions of the underlying dynamics. However, the axiomatic framework encompasses a substantially broader class of physical assumptions and generators, whose microscopic and operational interpretations have remained much less understood.

\smallskip
\noindent {\it Microscopic description:} To study thermodynamic dynamics arising from a microscopic system--environment interaction satisfying thermodynamic laws, we introduce energy-conserving thermal collision models (ECTCMs), consisting of repeated interactions of system and thermal ancillae through strictly energy-conserving interaction Hamiltonians. Besides their conceptual transparency, ECTCMs provide a microscopic route to deriving Markovian master equations from strictly energy-conserving interactions, where standard weak-coupling derivations fail (Sec.~\ref{sec:EC}). Taking an appropriate continuous-time limit, the repeated interactions reproduce the desired Markovian thermodynamic dynamics. Although such thermal collision models have recently been studied in Ref.~\cite{hubmann_open_2025}, the full class of generators obtainable within this framework remains to be characterised.

\smallskip
\noindent {\it Operational description:} In the resource-theoretic formulation of quantum thermodynamics, thermodynamically free transformations are described by thermal operations, namely quantum channels implementable without access to any additional thermodynamic resources~\cite{lostaglio2019introductory}. Restricting to thermal operations generated by continuous Markovian dynamics gives rise to the class of Markovian thermal operations (MTOs)~\cite{spaventa2022capacity,vomEnde2023ExploringLimits,son2024hierarchy}. While MTOs provide an operational conceptualisation of thermodynamic Markovian processes, a complete characterisation of their generators, as well as their microscopic realisation, has remained an open problem.
\smallskip

In this work, we connect these three distinct formulations of Markovian quantum thermodynamics (Fig.~\ref{fig:equivalence}), yielding our first main result:
\begin{theorembox}
\begin{theorem}[Axiomatic--Microscopic--Operational Equivalence]\label{thm:equivalences}
For any system with a non-trivial Hamiltonian, the sets of generators defined by the axiomatic, microscopic, and operational formulations of Markovian quantum thermodynamics exactly coincide.
\end{theorem}
\end{theorembox}
The proof of Thm.~\ref{thm:equivalences} proceeds by establishing pairwise equivalences between the three formulations. Props.~\ref{prop:L_CM_=>_L} (Sec.~\ref{sec:ME_from_CM}),  and \ref{prop:ME_CM_equivalence} (Sec.~\ref{sec:constructing_ECTCM}) establish the equivalence between axiomatic and microscopic formulations. Prop.~\ref{prop:MTO_=_CM} (Sec.~\ref{sec:ECTCM_=_MTO}) establishes the equivalence between the microscopic and operational formulations. The equivalence between axiomatic and operational formulations then follows immediately by transitivity. 

As a result, our findings show that all three approaches to Markovian thermodynamics yield the same dynamical laws. Consequently, dynamics arising from these generators are guaranteed to (i) satisfy the thermodynamic axioms, (ii) admit an underlying thermodynamically consistent model, and (iii) be implementable without consuming athermal resources.

Our theorem substantially generalises previous connections between Davies maps, collision models, and Markovian thermal operations established for two-level population dynamics~\cite{lostaglio_elementary_2018}. In contrast, Thm.~\ref{thm:equivalences} establishes the exact equivalence between the axiomatic, microscopic and operational formulations for arbitrary finite-dimensional systems, including coherent dynamics.

Beyond demonstrating the equivalence between these formulations, our framework also provides practical constructive tools for simulating Markovian quantum thermodynamics ``digitally", i.e., using finite collision models. In particular, we present a simulation protocol and show that the simulation error, quantified by the deviation from the target dynamics, scales as $O(\tau^{3/2}/n^{1/2})$ in terms of the evolution time $\tau$ and the number of collisions~$n$. 

Using this protocol, we simulate paradigmatic thermodynamically consistent dynamics in concrete physical systems. First, we construct an ECTCM that reproduces the Markovian dynamics of a two-level system coupled to a bosonic thermal environment. Although the corresponding master equation is not naturally derived from an energy-conserving interaction Hamiltonian, its Lindbladian is thermal and our construction provides a fully energy-conserving microscopic implementation. Second, going beyond the single-bath setting, we transform a three-level autonomous thermal machine into a three-stroke engine, with each stroke realised by an ECTCM. We numerically verify that, in the continuous-time limit, the resulting engine reproduces the state dynamics and heat currents of the original autonomous machine. Furthermore, we prove that the two engines have identical efficiencies for arbitrary collision times chosen for the ECTCMs. Together, these examples demonstrate the practical applicability of our framework and establish ECTCMs as a versatile platform for simulating, engineering, and analysing thermodynamically consistent open-system dynamics and quantum thermal machines.

The paper is structured as follows: In Sec.~\ref{sec:axiomatic_description}, we introduce the thermodynamic axioms and derive the general form of the ensuing thermal generators. In Sec.~\ref{sec:microscopic_description}, we define ECTCMs and derive the corresponding generators in the continuous-time limit. By showing that these generators satisfy all the thermodynamic axioms and presenting an explicit ECTCM construction protocol for any given axiomatic generator, we establish the equivalence between the axiomatic and microscopic formulations of Markovian thermodynamics. Furthermore, we connect these results to the operational framework in Sec.~\ref{sec:operational_description}, proving the triple equivalence stated in Thm.~\ref{thm:equivalences}. In Sec.~\ref{sec:ECTCM_simulation}, we cast the ECTCM construction employed in the proof of Thm.~\ref{thm:equivalences} into a simulation protocol for Markovian quantum thermodynamics, together with an analysis of the finite-collision-time simulation error. Applications of the ECTCM simulation protocol to representative physical models are presented in Sec.~\ref{sec:applications}, where we demonstrate the simulation accuracy and show that the collision model simulation precisely recovers the efficiency of the thermal machine despite the simulation's discrete nature. Finally, we conclude and discuss future directions in Sec.~\ref{sec:discussion}.

\section{Axiomatic description}\label{sec:axiomatic_description}
The axiomatic approach to Markovian quantum thermodynamics imposes thermodynamic constraints on the master equation governing the evolution of a system. In particular, the dynamics must be autonomous and satisfy quantum detailed balance (see below). While yielding structurally transparent constraints on the corresponding generators, this approach does \textit{not} assume an underlying microscopic model. The physical implementation of the resulting master equations, and the resources it requires, therefore remain a priori opaque.

We begin this section by introducing the three thermodynamic axioms imposed on the generators $\mathcal{L}_t$ in Sec.~\ref{sec:thermal_axioms} and discussing their physical significance in Sec.~\ref{sec:physical_significance_axioms}. The resulting structure of the corresponding axiomatic generators is then derived in Sec.~\ref{sec:structure_thermal_Lindbladian}. 

\subsection{Thermodynamic axioms}\label{sec:thermal_axioms}
Axiomatic approaches to Markovian quantum thermodynamics have found widespread application in the literature~\cite{fagnola2015entropy,fagnola2015entropy_DB,Lostaglio2017coherence,Alhambra2017Dynamicalmaps,dann_open_2021,dann2021quantum,hewgill2021quantum,chen2025thermodynamically}. 
In this paradigm, we consider a $d_S$-dimensional system $S$ with a Hamiltonian $H_S$, which defines the energy eigenbasis $\{\ket{j}_S\}_{j=0}^{d_S-1}$ and the energy spectrum $\{\epsilon_j\}_{j=1}^{d_S-1}$ of the system. Thermalisation of $S$ over the time interval $[0,\tau]$ is induced by interactions with the environment and described by the dynamical map $\Lambda_\tau \equiv \overleftarrow{\rm T}\exp(\int_0^\tau \mathcal{L}_t\,{\rm d} t)$ with the time-dependent generator $\mathcal{L}_t$ and $\overleftarrow{\rm T}$ denoting chronological time-ordering. The axiomatic approach imposes thermodynamic constraints directly on $\mathcal{L}_t$, thereby identifying the class of thermodynamically consistent dynamics: those that are memoryless, autonomous, and consistent with thermal equilibrium. These physical requirements are encoded in the following thermodynamic axioms (see Sec.~\ref{sec:physical_significance_axioms} for a more detailed discussion of their physical significance):
\begin{axiomaticbox}[boxsep=0pt,left=1em,right=1em]{}
\begin{center}
    {\fontsize{10pt}{12pt}\selectfont {\bf Thermal Lindbladians}}
\end{center}

\begin{enumerate}[label=(\Roman*),leftmargin=2em]
    \item\label{axiom:I} {\it Markovianity}: The master equation for a system state $\rho_S(t)$ is Markovian:
    \begin{align}
        \frac{{\rm d} \rho_S(t)}{{\rm d} t} = \mathcal{L}_t(\rho_S(t)),
    \end{align}
    with the general time-dependent Lindbladian $\mathcal{L}_t$ of GKSL form~\cite{gorini1976completely,lindblad1976generators}:
    \begin{align}
        \mathcal{L}_t(\cdot) = -\nI[H_S'(t), \cdot] + \mathcal{A}_t(\cdot) - \frac{1}{2}[\mathcal{A}^\dagger_t(\mathds{1}), \cdot]_+, \hspace{-1.2em} \label{eq:L_gen}
    \end{align}
    where $[\cdot\,,\cdot]_+$ is the anti-commutator, $H_S'(t)$ is a time-dependent Hermitian operator which can be identified as the effective Hamiltonian of system $S$ after including the Lamb shift induced by the system--environment interaction, $\mathds{1}$ denotes the identity operator, $\mathcal{A}_t$ is a time-dependent completely positive (CP) map and $\mathcal{A}^\dagger_t$ is the adjoint of $\mathcal{A}_t$ with respect to the Hilbert-Schmidt inner product. We write $\mathcal{L}_t = -\nI \mathcal{L}_{H_S'(t)} + \mathcal{D}_t$ with the unitary term defined as $\mathcal{L}_{X}(\cdot) := [X, \cdot]$ and the dissipator
    \begin{align}
        \mathcal{D}_t(\cdot) := \mathcal{A}_t(\cdot) - \frac{1}{2}[\mathcal{A}^\dagger_t(\mathds{1}), \cdot]_+.
        \label{eq:D_t_gen}
    \end{align}
    \item\label{axiom:II} {\it Time-translation symmetry}: $[\mathcal{L}_t, \mathcal{L}_{H_S}]=0, \, \forall\, t$, where $H_S$ is the time-independent bare system Hamiltonian. If $\mathcal{L}_t$ satisfies the time-translation symmetry, $H_S'(t)$ and $\mathcal{A}_t$ appearing in its definition [Eq.~(\ref{eq:L_gen})] can always be chosen such that~\cite{HOLEVO1993211}
    \begin{subequations}
        \begin{align}
            [H_S'(t), H_S]&=0, \quad \forall\, t, \label{eq:H_S'_commute_H_S}\\
            [\mathcal{A}_t, \mathcal{L}_{H_S}]&=0, \quad \forall\, t. \label{eq:A_commute_L_H_S}
        \end{align}
    \end{subequations}
    \item\label{axiom:III} {\it Quantum detailed balance condition}: For all $(x,x',y,y')$, the dissipator $\mathcal{D}_t$ satisfies that
    \begin{align}
        &\quad \nE^{-\beta \epsilon_y}\bra{x'}_S\mathcal{D}_t(\ketbra{x}{y}_S)\ket{y'}_S \nonumber\\
        &= \nE^{-\beta \epsilon_{y'}}\bra{x}_S\mathcal{D}_t(\ketbra{x'}{y'}_S)\ket{y}_S^*, \, \forall\, t,
        \label{eq:QDB_L_t}
    \end{align}
    where $\beta$ is the inverse temperature, $\{\ket{j}_S\}_{j=0}^{d_S-1}$ is the orthonormal eigenbasis of $H_S$ with the corresponding eigenvalues $\{\epsilon_j\}_{j=0}^{d_S-1}$ and $d_S$ is the dimension of $H_S$. This condition is known as Gelfand--Naimark--Segal (GNS) detailed balance~\cite{agarwal1973open,alicki1976detailed}.
    In App.~\ref{app:QDB}, we show that under Axioms~\ref{axiom:I} and \ref{axiom:II}, 
    the various definitions of quantum detailed balance~\cite{AMORIM2021389,fagnola2009two}, e.g. Kubo--Martin--Schwinger (KMS) detailed balance~\cite{temme2010chi,Alhambra2017Dynamicalmaps,chen2025efficient,scandi_thermalization_2026}, coincide with Eq.~(\ref{eq:QDB_L_t}) and are equivalent to the same condition imposed on the CP map $\mathcal{A}_t$ itself, namely:
    \begin{align}
        &\quad \nE^{-\beta \epsilon_y}\bra{x'}_S\mathcal{A}_t(\ketbra{x}{y}_S)\ket{y'}_S \nonumber\\
        &= \nE^{-\beta \epsilon_{y'}}\bra{x}_S\mathcal{A}_t(\ketbra{x'}{y'}_S)\ket{y}_S^*, \, \forall\, t,
        \label{eq:QDB_A_t}
    \end{align}
    for all $(x,x',y,y')$. 
\end{enumerate}
\end{axiomaticbox}
Throughout this paper, a Lindbladian $\mathcal{L}_t$ satisfying the three axioms will be referred to as a {\it thermal Lindbladian}.
We remark that when $H_S'(t) = H_S$ and therefore $[\mathcal{L}_{H_S'(t)}, \mathcal{D}_t]=0$, the Lindbladians $\mathcal{L}_t$ defined by Axioms~\ref{axiom:I}--\ref{axiom:III} reduce to Davies generators~\cite{davies1974markovian,davies1976markovian,alicki2007quantum,ROGA2010311}, the canonical model for (rapid) Markovian thermalisation in both simple and many-body quantum systems~\cite{alicki2009thermalization,ROGA2010311,Bardet2023RapidThermal,moroder_thermodynamics_2024}. 
Previous works~\cite{Lostaglio2017coherence,dann_open_2021} characterised thermodynamically consistent Markovian dynamics by combining Markovianity~\ref{axiom:I} and time-translation symmetry~\ref{axiom:II} with the weaker requirement that the Gibbs state be stationary. In this work, we follow similar lines in the derivation of structural results, but replace Gibbs-state stationarity by the stronger condition of quantum detailed balance~\ref{axiom:III}. As we shall show, this strengthening is precisely what guarantees the existence of an energy-conserving microscopic realisation and underlies the equivalence between the axiomatic, microscopic, and operational descriptions.

\subsection{Physical significance of the three axioms}\label{sec:physical_significance_axioms}
The axioms of Markovianity~\ref{axiom:I}, time-translation symmetry~\ref{axiom:II}, and quantum detailed balance~\ref{axiom:III} represent fundamental, physical constraints expected for a thermodynamically consistent quantum process. Respectively, they encode memoryless irreversibility, autonomous processes, and thermal equilibrium. 

Markovianity~\ref{axiom:I} encapsulates the assumption that the future evolution of the system depends only on its current state. Physically, this corresponds to a situation in which the system interacts with an environment that does not retain accessible memory of the system’s past states, as is typically expected under weak system--environment coupling and rapidly decaying bath correlations~\cite{breuer2002theory}. The resulting memoryless evolution is therefore intrinsically irreversible. In the thermodynamic context, this irreversibility manifests as the monotonic accumulation of entropy production~\cite{popovic2018entropy, strasberg2019non}, in accordance with the second law.

Time-translation symmetry~\ref{axiom:II} renders the dynamics autonomous relative to the system Hamiltonian -- namely, the dynamics should be insensitive to the choice of time origin of the system's free evolution. Since energetic coherence $\ketbra{j}{j'}_S$ oscillates with Bohr frequency $\epsilon_j - \epsilon_{j'}$, this is equivalent to requiring that the dynamics cannot convert coherences between different oscillation frequencies. Consequently, different coherence modes evolve independently, which is captured by the commutation relation $[\mathcal{L}_t,\mathcal{L}_{H_S}]=0$. Beyond this physical motivation, Axiom~\ref{axiom:II} also emerges naturally in standard physical approximations and leads to profound structural consequences (see also Refs.~\citep[][Sec.~III B]{Marvian2016QuantifyCoherence} and \cite{Lostaglio2017coherence}):
\begin{itemize}
    \item[(i)] In the general theory of open quantum systems, Axiom~\ref{axiom:II} arises as a consequence of the {\it secular approximation} (also known in quantum optics as the rotating-wave approximation), which neglects terms in the Lindbladian that would violate $[\mathcal{L}_t,\mathcal{L}_{H_S}]=0$~\cite{breuer2002theory}.
    \item[(ii)] Under Axiom~\ref{axiom:II}, the channel generated by $\mathcal{L}_t$, $\Lambda_\tau \equiv \overleftarrow{\rm T}\exp(\int_0^\tau \mathcal{L}_t\,{\rm d} t)$, commutes with the free unitary evolution $\mathcal{U}_{S,t'}(\cdot) \equiv e^{-i H_S t'}(\cdot)e^{i H_S t'}$ for all $t'$. This time covariance implies that, for each fixed time $t$, $\Lambda_t$ admits a unitary dilation with an environment $E$ and a joint unitary operator $U$ such that $\Lambda_t(\cdot) = {\rm Tr}_{E}\{U(\cdot\otimes \ketbra{0}{0}_E)U^\dagger\}$ where the global dynamics preserves the total energy $[U, H_S + H_E] = 0$~\citep[][Thm.~25]{marvian2012symmetry}. Therefore, time-translation symmetry intimately connects to an energy-conservation law. 
\end{itemize}

Quantum detailed balance~\ref{axiom:III} identifies thermal equilibrium and generalises the notion of detailed balance in classical thermodynamics, which is known to imply vanishing entropy production~\cite{esposito_second_2011} and time reversibility at thermal equilibrium~\citep[][Thm.~1.2]{kelly1979reversibility}. Axiom~\ref{axiom:III} likewise leads to these fundamental equilibrium properties:
\begin{itemize}
    \item[(i)] To see that the thermal state $\gamma_S\equiv \nE^{-\beta H_S}/{\rm Tr}\{\nE^{-\beta H_S}\}$ is stationary under the axiomatic dynamics, i.e., $\mathcal{L}_t(\gamma_S) = 0$ consider two general system operators $A_S\equiv \sum_{x,y}^{d_S-1} A_S^{x,y}\ketbra{x}{y}_S$ and $B_S\equiv \sum_{x',y'}^{d_S-1} B_S^{y'x'}\ketbra{y'}{x'}_S$. Eq.~(\ref{eq:QDB_L_t}) is equivalent to the relation
    \begin{align}
        {\rm Tr}\{\mathcal{D}_t(A_S\gamma_S)B_S\} = {\rm Tr}\{A_S\mathcal{D}_t(\gamma_S B_S)\},
        \label{eq:QDB_D_with_AB}
    \end{align}
    holding for all $A_S$ and $B_S$. Choosing $A_S = \mathds{1}_S$ and using the trace-preservation condition $\mathcal{D}_t^\dagger(\mathds{1}_S) = 0$ [see Eq.~\eqref{eq:D_t_gen}], we obtain ${\rm Tr}\{\mathcal{D}_t(\gamma_S)B_S\} = 0,\,\forall\, B_S$, and thus, $\mathcal{D}_t(\gamma_S) = 0$. Furthermore, by Axiom~\ref{axiom:II} [Eq.~(\ref{eq:H_S'_commute_H_S})], the unitary part of $\mathcal{L}_t$ satisfies $\mathcal{L}_{H_S'(t)}(\gamma_S) = 0$. Combining these relations yields $\mathcal{L}_t(\gamma_S) = 0$ -- i.e., the corresponding master equation leaves the thermal state invariant. 
    \item[(ii)] The entropy production rate in Markovian quantum processes is $\dot{\Sigma}:=-D(\mathcal{L}_t(\rho_S(t))||\gamma_S)$ where $D(\sigma||\tau):={\rm Tr}\{\sigma(\ln\sigma - \ln\tau)\}$ is the quantum relative entropy~\cite{spohn1978entropy, horowitz2014equivalent}. The stationarity condition $\mathcal{L}_t(\gamma_S)=0$ therefore ensures vanishing entropy production at thermal equilibrium.
    \item[(iii)] Beyond implying vanishing entropy production, Axiom~\ref{axiom:III} also encodes a notion of time reversibility at thermal equilibrium. To see this, we omit the explicit time dependence of $\mathcal{L}_t$ and its unitary contribution $\mathcal{L}_{H_S’}$, and consider the dynamics $\Lambda_t \equiv {\rm e}^{\mathcal{D} t}$ generated by a dissipator $\mathcal{D}$ satisfying Eq.~(\ref{eq:QDB_L_t}). As discussed above, Eq.~(\ref{eq:QDB_L_t}) is equivalent to Eq.~(\ref{eq:QDB_D_with_AB}), which implies  ${\rm Tr}\{A_S\gamma_S\mathcal{D}^\dagger(B_S)\} = {\rm Tr}\{\mathcal{D}^\dagger (A_S)\gamma_S B_S\}$ for all $A_S, B_S$. It immediately follows by repeated application of this identity that ${\rm Tr}\{A_S\gamma_S(\mathcal{D}^n)^\dagger(B_S)\} = {\rm Tr}\{(\mathcal{D}^n)^\dagger (A_S)\gamma_S B_S\}$ for any integer $n\ge 1$. Therefore
    \begin{align}
        {\rm Tr}\{B_SA_S(t)\gamma_S\} = {\rm Tr}\{B_S(t)A_S\gamma_S\},
    \end{align}
    where $X(t)\equiv \Lambda^\dagger_t(X)$ is the time evolution of an operator $X$ in the Heisenberg picture. This relation shows that two-time correlation functions evaluated on the thermal state are invariant under time reversal. The dynamics is thus reversible at thermal equilibrium. Furthermore, in App.~\ref{app:QDB}, we show that Eq.~(\ref{eq:QDB_L_t}) is equivalent to the KMS condition (an a priori different notion of detailed balance) for the dissipator $\mathcal{D}$. This then implies that $\Lambda_t$ coincides with its corresponding Petz recovery map~\cite{petz1986sufficient,petz1988sufficiency} (see Thm.~8 in Ref.~\cite{Alhambra2017Dynamicalmaps}). The relationship between Petz maps and Bayesian time-reversal of quantum channels~\cite{wilde2015recoverability,parzygnat2023axioms,parzygnat2023time} suggests the interpretation of the detailed balance condition as a notion of time reversibility at thermal equilibrium.
\end{itemize}

\subsection{Structure of thermal Lindbladians}\label{sec:structure_thermal_Lindbladian}
We now examine the structure of the thermal Lindbladian $\mathcal{L}_t$, i.e., the constraints that follow from Axioms~\ref{axiom:I}–\ref{axiom:III}. Its unitary part, $\mathcal{L}_{H_S'(t)}$, is straightforward to characterise: by Axiom~\ref{axiom:II} [Eq.~(\ref{eq:H_S'_commute_H_S})], its action is given by the commutator $\mathcal{L}_{H_S'}(\cdot) = -\nI[H_S'(t), \cdot]$, where $H_S'(t)$ is a Hermitian system operator that commutes with $H_S$.
To analyse how time-translation symmetry~\ref{axiom:II} and quantum detailed balance~\ref{axiom:III} constrain the structure of the dissipator $\mathcal{D}_t$, following Ref.~\cite{dann_open_2021}, we introduce the transition operators $\{F_{S}^{jj'}\equiv\ketbra{j}{j'}_S\}_{j,j'=0}^{d_S-1}$, which form an orthonormal basis for operators acting on the system Hilbert space. By writing $\mathcal{A}_t$ in terms of its Kraus operators, i.e., $\mathcal{A}_t(\cdot)\equiv \sum_r K_r(t) (\cdot) K_r^\dagger(t)$, we have
\begin{align}
    \mathcal{A}_t(\cdot) = \sum_{j,j',m,m'=0}^{d_S-1}\xi_{jj'mm'}(t)F_S^{jj'}(\cdot) (F_S^{mm'})^\dagger,
    \label{eq:A_gen}
\end{align}
where
\begin{align}
    \xi_{jj'mm'}(t) \equiv \sum_r {\rm Tr}\left\{(F_S^{jj'})^\dagger K_r(t)\right\}{\rm Tr}\left\{F_S^{mm'} K_r^\dagger(t)\right\}.
    \label{eq:xi_def_from_A}
\end{align}
Accordingly, the dissipator $\mathcal{D}_t$ [Eq.~(\ref{eq:D_t_gen})] can be written as
\begin{align}
    \mathcal{D}_t(\cdot) &= \sum_{j,j',m,m'=0}^{d_S-1}\xi_{jj'mm'}(t)F_S^{jj'}(\cdot)(F_S^{mm'})^\dagger \nonumber\\
    &\quad - \sum_{j,j',m,m'=0}^{d_S-1}\xi_{jj'mm'}(t)\frac{1}{2}[(F_S^{mm'})^\dagger F_S^{jj'}, \cdot]_+.
    \label{eq:D_givenby_xi}
\end{align}
The coefficients $\xi_{jj'mm'}(t)$ therefore completely characterise $\mathcal{D}_t$. In particular, we have the following proposition, providing a complete characterisation of thermal Lindbladians:
\begin{proposition}\label{prop:general_form_xi}
    A generator $\mathcal{L}_t$ is a thermal Lindbladian, satisfying Axioms~\ref{axiom:I}--\ref{axiom:III}, if and only if it is of the form given in Eq.~(\ref{eq:L_gen}) with $[H_S'(t), H_S] = 0$ and its associated coefficients $\xi_{jj'mm'}(t)$ [Eq.~(\ref{eq:xi_def_from_A})] can be written as:
    \begin{align}
        \xi_{jj'mm'}(t) 
        &= \delta(\epsilon_{j'}-\epsilon_{m'} +\epsilon_m - \epsilon_j) \nE^{-\beta(\epsilon_{j}-\epsilon_{j'})/2}\nonumber \\
        &\quad \times \sum_{q,q'=0}^{d_S-1} W_{jj'qq'}(t) W_{mm'qq'}^*(t),
        \label{eq:xi_general_under_axioms}
    \end{align}
    where $\delta(x) = 1$ if $x=0$ and $0$ otherwise, and the coefficients $W_{xx'yy'}(t)$ satisfy $W_{xx'yy'}(t) = W_{x'xy'y}^*(t)$ for all $(x,x',y,y')$. 
\end{proposition}
\begin{proof}
    See App.~\ref{app:proof_of_proposition_xi}, where the explicit construction of $W_{xx'yy'}(t)$ for a given $\xi_{jj'mm'}(t)$ is provided.
\end{proof}
Eq.~(\ref{eq:xi_general_under_axioms}) provides the most general form of $\xi_{jj'mm'}(t)$ -- and thus $\mathcal{L}_t$ -- under the Axioms~\ref{axiom:I}--\ref{axiom:III}. 
It is worth noting that the coefficients $W_{xx'yy'}(t)$ in Eq.~(\ref{eq:xi_general_under_axioms}) possess a gauge freedom that leaves $\mathcal{L}_t$ invariant. Specifically, arranging the coefficients $W_{xx'yy'}(t)$ into a matrix $\bm{W}(t)\in\mathbb{C}^{d_S^2\times d_S^2}$ by treating $(x,x')$ and $(y,y')$ as row and column indices, we find that the coefficient $\xi_{jj'mm'}(t)$ remains invariant under the transformation $\bm{W}'(t)\rightarrow \bm{W}(t)\bm{U}(t)$, where $\bm{U}(t)\in\mathbb{C}^{d_S^2\times d_S^2}$ is a unitary matrix satisfying the same symmetry constraints as $\bm{W}(t)$. The significance of this gauge freedom will become apparent in Sec.~\ref{sec:constructing_ECTCM}, where we employ it for the construction of ECTCMs that yield master equations satisfying Axioms~\ref{axiom:I}–\ref{axiom:III}.

While Prop.~\ref{prop:general_form_xi} provides a complete axiomatic characterisation of Markovian quantum thermodynamics, it is, at this point, not clear if every master equation of this type can be understood as stemming from an underlying thermodynamically consistent system--environment interaction. To demonstrate that this is indeed the case, we now introduce a microscopic description of Markovian system dynamics based on explicit system--environment interactions. This framework will offer a physically transparent description of thermal processes and provide a concrete class of dynamics satisfying all three axioms.

\section{Microscopic description}\label{sec:microscopic_description}
While the thermodynamic axioms carry rich physical significance and identify the class of thermodynamically consistent generators, they do not by themselves reveal the microscopic mechanisms underlying the dynamics (given by a system--environment Hamiltonian $H_{SE}$). To be thermodynamically consistent, such an underlying model must be strictly energy conserving and the environment state initially thermal (see below for a detailed discussion), while Markovianity of the system dynamics is imposed through standard weak-coupling arguments~\cite{breuer2002theory}. It turns out, though, that these arguments break down for thermodynamically consistent $H_{SE}$ (see below), necessitating an alternative strategy to investigate the microscopic approach to Markovian thermodynamics. 

A widely employed, transparent framework  to derive and simulate Markovian open-system dynamics is provided by collision models, which simulate the noisy evolution of a system through repeated interactions with its environment, moderated by a Hamiltonian $H_{SE}$ and its induced unitary evolution operator (see Fig.~\ref{fig:collision_model})~\cite{brun2002simple,ziman2005description,Grimmer2016Open,Cattaneo2021CMcansimulate,ciccarello2022quantum,cusumano2022quantum}. 
\begin{figure}
    \centering
    \includegraphics[width=\linewidth]{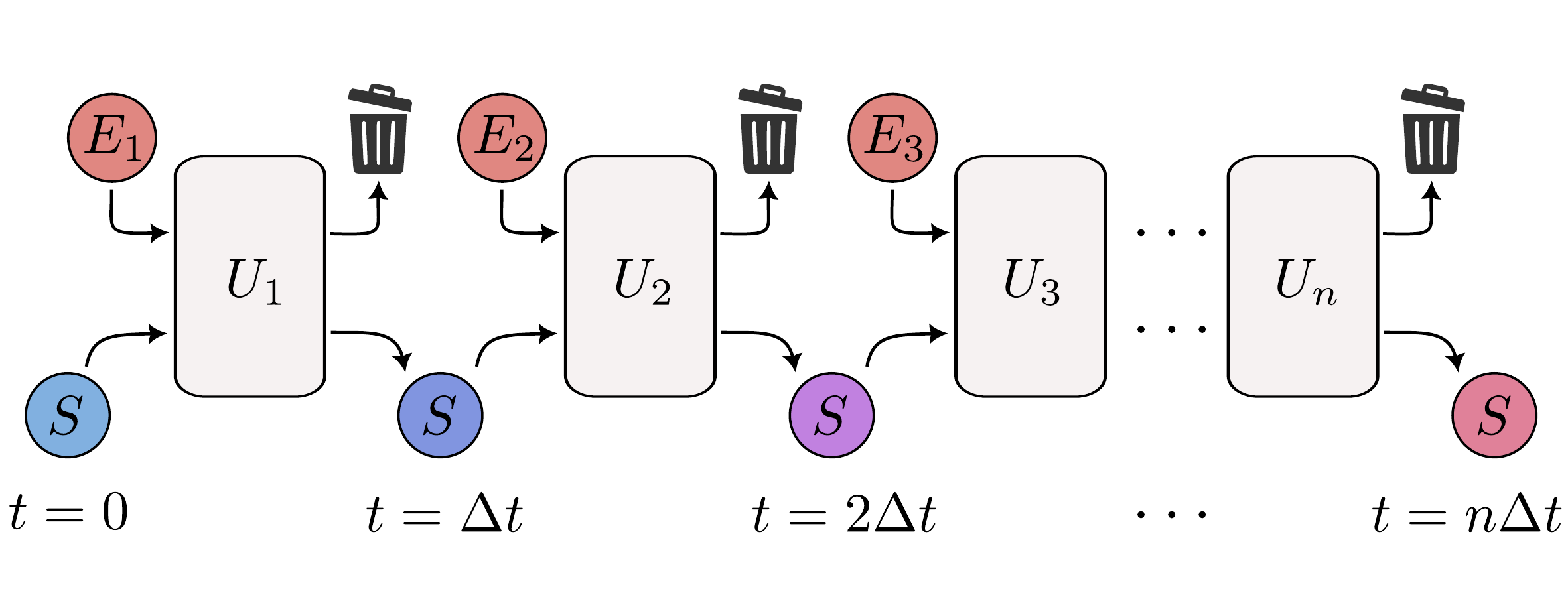}
    \caption{Quantum collision models. The system $S$ (blue discs) sequentially interacts with environmental ancillae (red discs) $\{E_i\}_{i=1}^n$ inside an interaction box representing the system–ancilla unitary $U_{i}$. After each collision occurring over a time interval $\Delta t$ the ancilla is discarded, resulting in memoryless dynamics, while the system gradually thermalises (transitions from blue to red).}
    \label{fig:collision_model}
\end{figure}
In particular, unlike conventional approaches, they offer a route to deriving Markovian master equations from strictly energy-conserving interactions. This naturally raises the question of whether a suitably constrained collision model can reproduce every thermodynamically admissible generator, thus providing a microscopic model for all thermal Lindbladians. Here, we answer this question in the affirmative by introducing thermodynamically consistent collision models and by fully characterising the resulting master equation. 

To do so, in Sec.~\ref{sec:EC}, we first introduce the energy-conservation condition, a canonical thermodynamic constraint reflecting the first law of thermodynamics. In Sec.~\ref{sec:ECTCM}, we consider collision models with thermal ancillae and energy-conserving interactions (see Fig.~\ref{fig:collision_model}), and show in Sec.~\ref{sec:ME_from_CM}  that the corresponding master equations are thermal, i.e., they satisfy Axioms~\ref{axiom:I}–\ref{axiom:III}. Finally, in Sec.~\ref{sec:constructing_ECTCM} we complete the equivalence between the axiomatic and the microscopic approaches by demonstrating that \textit{every} thermal master equation admits a corresponding thermodynamically consistent collision model.

\subsection{Energy conservation condition}\label{sec:EC}
We begin by deriving the restrictions on system--environment Hamiltonians stemming from the first law. In the framework of open quantum systems, a system $S$ is interacting with an environment $E$ with the composite Hamiltonian given by
\begin{align}
    H_{SE} = H_S\otimes\mathds{1}_E + \mathds{1}_S\otimes H_E + V,
\end{align}
where $\mathds{1}$ denotes the identity operator and $H_{S(E)}$ is the system (environment) Hamiltonian on the Hilbert space $\mathcal{H}_{S(E)}$, while $V$ is the interaction Hamiltonian on the composite space $SE$. We assume that $H_{S(E)}$ and $V$ are time-independent (explicit time-dependence will be introduced below). Given an initial state $\rho_{SE}(0)$, the state $\rho_{SE}(t)$, evolved under $H_{SE}$ until time $t$, is given by
\begin{align}
    \rho_{SE}(t) = U(t)\rho_{SE}(0)U^\dagger(t),
\end{align}
where $U(t) := \nE^{-\nI H_{SE} t}$. Thermodynamically consistent dynamics must respect the first law. That is, the energy of the local systems must be conserved under time evolution:
\begin{align}
    \langle H_S+H_E \rangle_{\rho_{SE}(t)} = \langle H_S+H_E \rangle_{\rho_{SE}(0)},
    \label{eq:H_S+H_E_conservation}
\end{align}
where $\langle X \rangle_\rho := {\rm Tr}\{X\rho\}$ is the expectation value of the operator $X$ for the state $\rho$. Here and throughout this work, we adopt the convention that local operators are implicitly tensored with the identity on the complementary subsystem.
Since Eq.~(\ref{eq:H_S+H_E_conservation}) should hold for all initial states $\rho_{SE}(0)$ and all times $t$, we have
\begin{align}
    \left[U(t), H_S+H_E\right] = 0,\quad \forall\, t.
    \label{eq:energy_conservation_U(t)}
\end{align}
Defining the thermal state $\gamma_{S(E)}$  at a finite inverse temperature $\beta < \infty$ as $\gamma_{S(E)} \equiv \nE^{-\beta H_{S(E)}}/{\rm Tr}\{\nE^{-\beta H_{S(E)}}\}$, it is easy to see that Eq.~(\ref{eq:energy_conservation_U(t)}) implies that $U(t)$ preserves thermal equilibrium, i.e., 
\begin{align}
    U(t)(\gamma_S\otimes\gamma_E)U^\dagger(t) = \gamma_S\otimes\gamma_E,\,\forall\, t.
    \label{eq:gamma_preservation_U(t)}
\end{align}
Conversely, if the eigenvalues of $H_S+H_E$ are taken, without loss of generality, to be non-negative, then Eq.~(\ref{eq:gamma_preservation_U(t)}) implies Eq.~(\ref{eq:energy_conservation_U(t)}). This follows by taking the principal logarithm of both sides of Eq.~(\ref{eq:gamma_preservation_U(t)}) (see Ref.~\citep[][Sec.~I B]{lostaglio2019introductory}). Therefore, in this open-system setting, energy conservation is equivalent to the preservation of thermal equilibrium. This relation is in fact intuitive: if the dynamics were able to drive the system out of equilibrium, external work, lower bounded by the corresponding nonequilibrium free energy~\cite{esposito_second_2011}, would necessarily be supplied during the process, thereby violating energy conservation.

We further note that Eq.~(\ref{eq:energy_conservation_U(t)}) equivalently imposes an energy conservation condition on the interaction Hamiltonian $V$:
\begin{align}
    \left[V, H_S+H_E\right] = 0.
    \label{eq:energy_conservation}
\end{align}
In order to see how energy conservation constrains the structure of $V$ (which is the relevant term for the ``collision" of the system with the environment), we consider the eigendecomposition of $H_S$ and $H_E$.
\begin{align}
    H_S &= \sum_j \epsilon_j \ketbra{j}{j}_S,\\ H_E &= \sum_k e_k \ketbra{k}{k}_E,
\end{align}
where $\{\ket{j}_S\}_j$ and $\{\ket{k}_E\}_k$ are orthonormal eigenbases of $H_S$ and $H_E$, respectively. Either Hamiltonian may generally contain energetic degeneracies.
Here, for simplicity, we do not specify the dimensions of $H_S$ and $H_E$, which can generally differ. We then have the following characterisation of energy-conserving interactions:
\begin{lemma}[\cite{son2024hierarchy}]
    Consider a general interaction Hamiltonian
    \begin{align}
        V = \sum_{j,j',k,k'} V_{jkj'k'} \ketbra{j}{j'}_S\otimes\ketbra{k}{k'}_E,
    \end{align}
    where $V_{jkj'k'} \equiv \bra{j}_S\otimes\bra{k}_E V \ket{j'}_S\otimes\ket{k'}_E$. $V$ satisfies energy conservation [Eq.~(\ref{eq:energy_conservation})] if and only if for all $(j,j',k,k')$,
    \begin{align}
    V_{jkj'k'} =  \delta(\epsilon_{j'}-\epsilon_j + e_{k'} - e_k) V_{jkj'k'}.
    \label{eq:V_under_EC}
\end{align}
\end{lemma}
\begin{proof}
The energy conservation condition [Eq.~(\ref{eq:energy_conservation})] reads
    \begin{align}
        &\quad [V, H_S+H_E] \nonumber\\
        &= \sum_{j,j',k,k'} V_{jkj'k'}[\ketbra{j}{j'}_S\otimes\ketbra{k}{k'}_E, H_S+H_E]\\
        &= \sum_{j,j',k,k'} V_{jkj'k'}\left(\epsilon_{j'}+e_{k'}-\epsilon_j - e_k\right)\ketbra{j}{j'}_S\otimes\ketbra{k}{k'}_E\\
        &= 0.
    \end{align}
    From the above, the `if' direction is straightforward. To prove the `only if' direction, we note that
    since the operators $\{\ketbra{j}{j'}_S\otimes\ketbra{k}{k'}_E\}_{j,j',k,k'}$ are linearly independent, the last equality above is equivalent to
    \begin{align}
         V_{jkj'k'}\left(\epsilon_{j'}+e_{k'}-\epsilon_j - e_k\right) = 0,\quad \forall\, (j,j',k,k').
    \end{align}
    Thus, $V_{jkj'k'} = 0$ for $(j,j',k,k')$ satisfying $\epsilon_{j'}+e_{k'}-\epsilon_j - e_k \neq 0$, which leads to $V_{jkj'k'} =  \delta(\epsilon_{j'}-\epsilon_j + e_{k'} - e_k) V_{jkj'k'}$ for all $(j,j',k,k')$.
\end{proof}
The energy-conservation condition imposes a specific block structure on the interaction Hamiltonian $V$, which in turn constrains the system dynamics resulting from the corresponding collision model. A natural question is whether  this constraint is sufficiently restrictive to ensure that the induced dynamics satisfy Axioms~\ref{axiom:I}–\ref{axiom:III}, thereby yielding a thermal Lindbladian. As we will show in Sec.~\ref{sec:ME_from_CM}, collision models consisting of thermal ancillae and energy-conserving interactions indeed always generate thermal Lindbladians. Before proceeding, however, it is instructive to note that the \textit{standard} microscopic derivation of a Markovian master equation~\cite{breuer2002theory} breaks down for energy-conserving interaction Hamiltonians and can thus not be employed for the analysis of Markovian thermodynamics from a microscopic model.

Textbook derivations of Markovian master equations are carried out in the interaction picture defined by $U_0(t)\equiv\nE^{-\nI (H_S+H_E) t}$, in which the Redfield equation is given by~\cite{breuer2002theory}
\begin{align}
    \frac{{\rm d}\rho_{S, {\rm I}}(t)}{{\rm d} t} &= -{\rm i}\mathrm{Tr}_E\{[V_{\rm I}(t), \rho_{S}(0)\otimes\rho_E]\}\nonumber\\
    &\quad -\int_0^t {\rm d} s \,\mathrm{Tr}_E \left\{\left[V_{\rm I}(t), [V_{\rm I}(s), \rho_{S,{\rm I}}(t)\otimes \rho_E]\right]\right\},
\end{align}
where $\rho_{S,{\rm I}}(t)\equiv {\rm Tr}_E\{\nE^{-\nI V t}\rho_{SE}(0)\nE^{\nI V t}\}$ and $V_{\rm I}(t)\equiv U_0^\dagger(t) V U_0(t)$. To obtain a Markovian master equation, the next step is to substitute $s$ by $t-s$ and let the upper limit $t\rightarrow \infty$~\cite{breuer2002theory}. However, 
with energy conservation requiring $[V, H_S+H_E]=0$, $V$ is stationary in the interaction picture, i.e. $V_{\rm I}(t) = V,\,\forall\, t$. Setting $t\rightarrow \infty$ will thus lead to a divergence. Therefore, it remains unclear how a Markovian master equation can be obtained from an energy-conserving interaction through the standard microscopic derivation. 

The above difficulty can be overcome using the collision-model approach, which provides a universal framework for simulating arbitrary Markovian dynamics~\cite{Cattaneo2021CMcansimulate,ciccarello2022quantum,cusumano2022quantum}. As we will show, by coupling the system sequentially to thermal ancillae through energy-conserving interactions, one can derive a well-defined Markovian master equation whose generator satisfies all the thermodynamic axioms introduced in previous section and is therefore a thermal Lindbladian.

\subsection{Energy-conserving thermal collision models}\label{sec:ECTCM}
In a collision model, the system sequentially interacts (collides) with a stream of ancillae. After colliding with the system over a time $\Delta t$, each ancilla is discarded, effectively making the system evolution memoryless. The resulting Markovian master equation is obtained in the limit $\Delta t \to 0$ (see below). After $n$ collisions, the system dynamics is described by the following map~\cite{ciccarello2022quantum}:
\begin{align}
    \mathcal{E}_n(\cdot) &:= {\rm Tr}_{E_n}\left\{U_n\dots {\rm Tr}_{E_2}\left\{U_2{\rm Tr}_{E_1}\left\{U_1(\cdot\otimes \rho_{E_1})U_1^\dagger\right\}\right.\right. \nonumber\\
    &\quad \left.\left. \otimes\rho_{E_2}U_2^\dagger\right\}\dots \otimes \rho_{E_n}U_n^\dagger\right\},
    \label{eq:collision_map_E_n}
\end{align}
where each unitary $U_i \equiv \nE^{-\nI H_{SE_i} \Delta t}$ acts on the system and ancilla $i$ for $i=1,2,\dots,n$ with 
\begin{align}
    H_{SE_i} &= H_S + H_{E_i} + gV_{SE_i} + D_{S,i} \label{eq:ECTCM_H_SE}\\
    &= H_{S,i}' + H_{E_i} + gV_{SE_i},
\end{align}
where we explicitly include a driving $D_{S,i}$ of the system such that $H_{S,i}'\equiv H_S + D_{S,i}$ and
the coefficient $g$ denotes the system--ancilla coupling strength. 
We remark that although the driving term $D_{S,i}$ and the system–ancilla interaction $V_{SE_i}$ are time-independent within each collision step, they are allowed to vary between different ancillae. This variation induces the time dependence in the master equation derived in the next subsection.
A pictorial representation of collision models is provided in Fig.~\ref{fig:collision_model}.

In accordance with the general energy-conservation considerations of the previous subsection, we impose the following conditions to ensure that the collision model is thermodynamically consistent:
\begin{microscopicbox}[boxsep=0pt,left=0.5em,right=1em]
\begin{center}
    {\fontsize{10pt}{12pt}\selectfont {\bf Energy-Conserving Thermal Collision Models (ECTCMs)}}
\end{center}

\begin{enumerate}[label=(\arabic*),leftmargin=2em]
    \item\label{CM:1} The driving $D_{S,i}$ satisfies  energy conservation:
    \begin{align}[D_{S,i}, H_S] = 0.\end{align}
    \item\label{CM:2} The system--ancilla interaction $V_{SE_i}$ satisfies energy conservation:
    \begin{align}[V_{SE_i}, H_S+H_{E_i}]=0.\end{align}
    \item\label{CM:3} All initial ancilla states are thermal:
    \begin{align}\rho_{E_i}=\gamma_{E_i} \equiv \nE^{-\beta H_{E_i}}/{\rm Tr}\{\nE^{-\beta H_{E_i}}\},\end{align}
    for a fixed $\beta$.
    \item\label{CM:4} The system--ancilla interaction $gV_{SE_i}$ does not require extra energy input to implement during any collision $i$:
    \begin{align}
        \langle gV_{SE_i} \rangle_{\rho_{SE_i}(s)} = \langle D_{S,i} \rangle_{\rho_S(0)} - \langle D_{S,i} \rangle_{\rho_{S}(s)}, \hspace{-1.2em}
    \end{align}$\forall\, \rho_{S}(0),\,\forall\, s\in[0,\Delta t]$, where, $\rho_S(0)$ is the initial state at the beginning of the $i$th collision, $\rho_{SE_i}(s)\equiv U_i(s)(\rho_S(0)\otimes\rho_{E_i})U_i^\dagger(s)$ with $U_i(s) \equiv \nE^{-\nI H_{SE_i} s}$, and $\rho_S(s)\equiv {\rm Tr}_{E_i}\{\rho_{SE_i}(s)\}$.
\end{enumerate}
\end{microscopicbox}
According to Conditions~\ref{CM:1} and \ref{CM:2}, it is straightforward to see that the unitary operator $U_i$ is energy-conserving in the sense discussed above, i.e., $[U_i, H_S + H_{E_i}] = 0$. We call the collision models satisfying Conditions~\ref{CM:1}--\ref{CM:4} {\it energy-conserving thermal collision models} (ECTCMs).

The four conditions can be interpreted as follows: Conditions~\ref{CM:1} and \ref{CM:2} ensure that neither the driving $D_{S,i}$ nor the system--ancilla interaction $V_{SE_i}$ alter the energy of the local systems individually (i.e., w.r.t. $H_S$ and $H_S+H_E$). Condition~\ref{CM:3} requires each ancilla to be prepared in a thermal state, thereby modelling a thermal environment. Finally, Condition~\ref{CM:4} imposes a stronger notion of energy conservation by requiring that any energy cost associated with implementing the interaction be supplied entirely by the driving term during the collision. We furthermore have the following equivalent formulation of Condition~\ref{CM:4}:
\begin{lemma}\label{lem:CM_4}
    Under Conditions~\ref{CM:1}, \ref{CM:2}, and \ref{CM:3}, Condition~\ref{CM:4} is equivalent to
    \begin{align}
        {\rm Tr}_{E_i}\{V_{SE_i}\gamma_{E_i}\} = 0.
        \label{eq:CM_4_vanishing_LS}
    \end{align}
\end{lemma}
\begin{proof}
    Since the total energy is conserved during any collision (generically labelled by $i$), we have
    \begin{align}
        \frac{{\rm d}}{{\rm d} s}\langle H_S + H_{E_i} + gV_{SE_i} + D_{S,i} \rangle_{\rho_{SE_i}(s)} = 0.
    \end{align}
    By Conditions~\ref{CM:1} and \ref{CM:2}, $[H_{SE_i}, H_S+H_{E_i}]=0$. Thus, $\frac{{\rm d}}{{\rm d} s}\langle H_S + H_{E_i} \rangle_{\rho_{SE_i}(s)} = 0$, which yields
    \begin{align}
        \frac{{\rm d}}{{\rm d} s}\langle gV_{SE_i} + D_{S,i} \rangle_{\rho_{SE_i}(s)} = 0.
    \end{align}
    Condition~\ref{CM:4} is thus equivalent to $\langle V_{SE_i} \rangle_{\rho_S(0)\otimes\rho_{E_i}} = 0,\,\forall\,\rho_S(0)$, where $\rho_S(0)$ is the initial state at the beginning of the $i$th collision and $\rho_{E_i}$ is the state of the fresh ancilla state. By Condition~\ref{CM:3}, it becomes $\langle V_{SE_i} \rangle_{\rho_S(0)\otimes\gamma_{E_i}} = 0,\,\forall\,\rho_S(0) \Leftrightarrow {\rm Tr}_{E_i}\{V_{SE_i}\gamma_{E_i}\} = 0$.
\end{proof}
As we will show in next subsection, Eq.~(\ref{eq:CM_4_vanishing_LS}) corresponds to a vanishing Lamb shift in the master equation obtained from the collision model. In existing collision-model derivations, this condition is introduced as a technical assumption to guarantee a well-defined continuous-time limit~\cite{ciccarello2022quantum,cusumano2022quantum}. Here we show that it admits a natural microscopic interpretation (Lem.~\ref{lem:CM_4}), allowing it to be elevated from an auxiliary assumption to a defining thermodynamic principle of ECTCMs.

\subsection{Thermal master equation from ECTCMs}\label{sec:ME_from_CM}
We now follow the standard derivation of Markovian master equations from collision models~\cite{ciccarello2022quantum} to derive the master equation corresponding to the \textit{energy-conserving} collision model introduced above. In particular, we show that the resulting master equations are thermal, i.e., satisfy Axioms~\ref{axiom:I}–\ref{axiom:III} (see Prop.~\ref{prop:L_CM_=>_L}).

During a small collision time $\Delta t$, the unitary operator $U_i$ in Eq.~(\ref{eq:collision_map_E_n}) can be approximated as follows:
\begin{align}
    U_i &\equiv \exp[-\mathrm{i}(H_{S,i}'+H_{E_i}+gV_{SE_i})\Delta t] \\
    &\simeq \mathds{1} - \mathrm{i}(H_{S,i}'+H_{E_i}+gV_{SE_i})\Delta t - \frac{1}{2}V_{SE_i}^2(g\Delta t)^2.
\end{align}
Since we seek a finite dissipative evolution in the limit $\Delta t\rightarrow 0$, the interaction strength is rescaled according to $g\propto (\Delta t)^{-1/2}$. This is the standard scaling yielding a nontrivial Markovian generator, as it keeps the second-order interaction contribution finite in the continuous-time limit~\cite{ciccarello2022quantum}. 
Consequently, the energy scales associated with $H_{S,i}’$ and $H_{E_i}$ are much smaller than that associated with the interaction term $gV_{SE_i}$, rendering second-order terms in $\Delta t$ that are not quadratic in $gV_{SE_i}$ negligible. We then have
\begin{align}
    \rho_{SE_i} &\equiv U_i \rho_{SE,i-1} U_i^\dagger \\
    &= \rho_{SE,i-1} - \mathrm{i}[H_{S,i}'+H_{E_i}+gV_{SE_i}, \rho_{SE,i-1}]\Delta t \nonumber \\
    & + \left(V_{SE_i}\rho_{SE,i-1}V_{SE_i} - \frac{1}{2}[V_{SE_i}^2, \rho_{SE,i-1}]_+\right)(g\Delta t)^2,
\end{align}
where $\rho_{SE,i-1}\equiv \rho_{S,i-1}\otimes\rho_{E_i}$ and $\rho_{SE_i}$ are the joint states before and after the $i$th collision, respectively, while $\rho_{S, i}$ is the system state after the $i$th collision, i.e., $\rho_{S,i}\equiv {\rm Tr}_{E_{i}}\{\rho_{SE_{i}}\}$. We have
\begin{align}
    \frac{\Delta \rho_{S,i}}{\Delta t} &= -\mathrm{i}[H_{S,i}'+g\mathrm{Tr}_E\{V_{SE_i}\rho_{E_i}\}, \rho_{S,i-1}] \nonumber\\
    &\quad + (g^2\Delta t)\mathrm{Tr}_{E_i}\left\{V_{SE_i}\rho_{S,i-1}\otimes\rho_{E_i} V_{SE_i}\right\} \nonumber\\
    &\quad - (g^2\Delta t){\rm Tr}_{E_i}\left\{\frac{1}{2}[V_{SE_i}^2, \rho_{S,i-1}\otimes\rho_{E_i}]_+\right\},
    \label{eq:CM_discrete_ME}
\end{align}
where $\Delta \rho_{S,i}\equiv \rho_{S,i} - \rho_{S,i-1}$.
For simplicity, we restrict our attention to identical ancillae. As demonstrated in Sec.~\ref{sec:constructing_ECTCM}, this already suffices to reproduce the full class of Markovian thermal dynamics satisfying Axioms~\ref{axiom:I}--\ref{axiom:III}. The description in terms of non-identical ancillae (e.g. to explicitly model environmental time-dependence) is equally straightforward. By Condition~\ref{CM:3}, 
\begin{align}
    \rho_{E} = \gamma_E = \sum_k p_{E}^{(k)} \ketbra{k}{k}_{E},
    \label{eq:rho_E_def}
\end{align}
with $p^{(k)}_E \equiv \bra{k}_E \rho_E \ket{k}_E = \nE^{-\beta e_k}/{\rm Tr}\{\nE^{-\beta H_E}\}$ being the thermal population on the $k$th energy level of each ancilla, the discrete-time master equation [Eq.~(\ref{eq:CM_discrete_ME})] is given by
\begin{align}
    &\frac{\Delta \rho_{S,i}}{\Delta t} = -\mathrm{i}[H_{S,i}'+gH_{S,i}^{\rm LS}, \rho_{S,i-1}] \nonumber\\
    &+ (g\Delta t^2)\sum_{k,k'}\left(L_{kk',i}\rho_{S,i-1}L_{kk',i}^\dagger - \frac{1}{2}[L_{kk',i}^\dagger L_{kk',i}, \rho_{S,i-1}]_+ \right),
\end{align}
where the Lamb shift $H_{S,i}^{\rm LS}$ is given by
\begin{align}
    H_{S,i}^{\rm LS} \equiv \mathrm{Tr}_{E_i}\{V_{SE_i}\rho_{E_i}\},
    \label{eq:H_S^LS_def}
\end{align}
and the Lindblad operator $L_{kk',i}$ is given by
\begin{align}
    L_{kk',i} \equiv \sqrt{p^{(k)}_E}\bra{k'}_{E_i} V_{SE_i}\ket{k}_{E_i}.
    \label{eq:L_kk'}
\end{align}
We now use Lem.~\ref{lem:CM_4}, which leads to $H_{S,i}^{\rm LS}=0$.
Therefore, in ECTCMs, a vanishing Lamb shift $H_{S,i}^{\rm LS}$ is equivalent to an interaction $V_{SE_{i}}$ that incurs no energetic cost. 

We then let $g^2\Delta t = 1$, $t_i\equiv i\Delta t$, and take the continuous-time limit $\Delta t\rightarrow 0$ ($t_i \rightarrow t$), which yields
\begin{align}
    \frac{{\rm d} \rho_S(t)}{{\rm d} t} &= -\nI [H_S'(t), \rho_S(t)] \nonumber\\
    &\quad + \sum_{k,k'}L_{kk'}(t)\rho_S(t)L_{kk'}^\dagger(t)\nonumber\\
    &\quad - \frac{1}{2}\sum_{k,k'}[L_{kk'}^\dagger(t) L_{kk'}(t), \rho_S(t)]_+.
    \label{eq:CMME_gen}
\end{align}
Since $V_{SE_i}$ satisfies energy conservation for all $i$ by Condition~\ref{CM:2}, the continuous-time interaction Hamiltonian $V_{SE}(t)$ takes the form given in Eq.~(\ref{eq:V_under_EC}).
The Lindbladian operator $L_{kk'}(t)$ [Eq.~(\ref{eq:L_kk'})] therefore becomes
\begin{align}
    L_{kk'}(t) = \sqrt{p^{(k)}_E} \sum_{j,j'} \delta(\epsilon_{j'}-\epsilon_j + e_{k} - e_{k'})V_{jk'j'k}(t)F_S^{jj'},
    \label{eq:L_kk'(t)}
\end{align}
where $V_{jkj'k'}(t) \equiv \bra{j}_S\otimes\bra{k}_E V_{SE}(t) \ket{j'}_S\otimes\ket{k'}_E$.
The collision model (CM) master equation [Eq.~(\ref{eq:CMME_gen})] can be written as
\begin{align}
    \frac{{\rm d} \rho_S(t)}{{\rm d} t} &= \mathcal{L}^{\rm (CM)}_t(\rho_S(t)),
    \label{eq:CMME_thermal}
\end{align}
where
\begin{align}
    \mathcal{L}^{\rm (CM)}_t(\cdot) &:= -\nI\mathcal{L}_{H_S'(t)}^{\rm (CM)}(\cdot) + \mathcal{D}^{\rm (CM)}_t(\cdot), \label{eq:L_CM_def}\\
    \mathcal{L}_{H_S'(t)}^{\rm (CM)}(\cdot) &:= [H_S'^{\rm (CM)}(t), \cdot], \label{eq:L_H_def}\\
    \mathcal{D}^{\rm (CM)}_t(\cdot) &:= \sum_{j,j',m,m'}\xi_{jj'mm'}^{\rm (CM)}(t)F_S^{jj'}(\cdot)(F_S^{mm'})^\dagger \nonumber\\
    &\quad - \frac{1}{2}\sum_{j,j',m,m'}\xi_{jj'mm'}^{\rm (CM)}(t)[(F_S^{mm'})^\dagger F_S^{jj'}, \cdot]_+,
    \label{eq:D_CM_def}
\end{align}
with
\begin{align}
    \xi_{jj'mm'}^{\rm (CM)}(t) &\equiv \delta(\epsilon_{j'}-\epsilon_j + \epsilon_{m} - \epsilon_{m'})\nonumber\\
    &\quad \times\sum_{k,k'|e_{k} - e_{k'} = \epsilon_j - \epsilon_{j'}} p^{(k)}_E V_{jk'j'k}(t)V_{mk'm'k}^*(t).
    \label{eq:xi_CM}
\end{align}
Eq.~(\ref{eq:CMME_thermal}) is the general form of Markovian master equations from ECTCMs. We note that a similar master equation was derived in Ref.~\cite{hubmann_open_2025} using ECTCMs for systems with nondegenerate energy gaps.

The coefficients $\xi_{jj'mm'}^{\rm (CM)}(t)$ hence dictate the  dissipative dynamics allowed by energy conservation. Remarkably, $\mathcal{L}_t^{\rm (CM)}$ is a thermal Lindbladian satisfying the thermodynamic Axioms~\ref{axiom:I}--\ref{axiom:III}, as stated in the following proposition:
\begin{proposition}\label{prop:L_CM_=>_L}
    ECTCMs defined by Conditions~\ref{CM:1}--\ref{CM:4} produce thermal Lindbladians characterised by Axioms~\ref{axiom:I}--\ref{axiom:III}.
\end{proposition}
\begin{proof}
See App.~\ref{app:proof_of_L_CM_=>_L}.
\end{proof}
This result shows that thermal Lindbladians emerge naturally from repeated energy-conserving interactions with a thermal environment, as modelled by ECTCMs, thereby providing a concrete microscopic foundation for Axioms~\ref{axiom:I}–\ref{axiom:III}. Moreover, as we will show in the following subsection, this correspondance is in fact {\it complete}: every thermal Lindbladian admits such a microscopic realisation.

\subsection{ECTCM construction for a given thermal Lindbladian}~\label{sec:constructing_ECTCM}
In this section, we complete the equivalence between axiomatic and microscopic descriptions of thermodynamic dynamics by proving the converse of Prop.~\ref{prop:L_CM_=>_L}.
Specifically, given a thermal Lindbladian $\mathcal{L}_t$, we explicitly construct an ECTCM that reproduces its dynamics in the continuous-time limit.
\begin{proposition}\label{prop:ME_CM_equivalence}
    Any axiomatic thermal Lindbladian satisfying Axioms~\ref{axiom:I}--\ref{axiom:III} admits an ECTCM realisation under Conditions~\ref{CM:1}--\ref{CM:4}, such that $\mathcal{L}^{\rm (CM)}_t = \mathcal{L}_t$ for all $t$, provided that the system Hamiltonian $H_S$ is non-trivial.
\end{proposition}
\begin{proof}
    The proof proceeds by explicit construction. In the remainder of this section we construct an ECTCM satisfying Conditions~\ref{CM:1}--\ref{CM:4} and show that it reproduces the target Lindbladian. The technical details are deferred to App.~\ref{app:technical_construction}.
\end{proof}
We now present the construction step by step.
To begin, consider a thermal Lindbladian $\mathcal{L}_t$ given in the following form [Eqs.~(\ref{eq:L_gen}) and (\ref{eq:A_gen})]:
\begin{align}
     \mathcal{L}_t(\cdot) &= -\nI[H_S'(t), \cdot] \nonumber\\
     &\quad + \sum_{j,j',m,m'}\xi_{jj'mm'}(t)F_S^{jj'}(\cdot)(F_S^{mm'})^\dagger \nonumber\\
    &\quad - \frac{1}{2}\sum_{j,j',m,m'}\xi_{jj'mm'}(t)[(F_S^{mm'})^\dagger F_S^{jj'}, \cdot]_+.
\end{align}
If $\mathcal{L}_t$ satisfies Axioms~\ref{axiom:I}--\ref{axiom:III}, we know from Prop.~\ref{prop:general_form_xi} that the coefficients $\xi_{jj'mm'}(t)$ can be written as
\begin{align}
    \xi_{jj'mm'}(t) &= \delta(\epsilon_{j'}-\epsilon_{m'} +\epsilon_m - \epsilon_j) \nE^{-\beta(\epsilon_{j}-\epsilon_{j'})/2}\nonumber \\
    &\quad \times \sum_{q,q'=0}^{d_S-1} W_{jj'qq'}(t) W_{mm'qq'}^*(t),
    \label{eq:xi_general_under_axioms_1}
\end{align}
where the coefficients $W_{xx'yy'}(t)$ satisfy $W_{xx'yy'}(t) = W_{x'xy'y}^*(t)$ for all $(x,x',y,y')$ (App.~\ref{app:proof_of_proposition_xi}).

On the other hand, the Lindbladian $\mathcal{L}_t^{\rm (CM)}$ obtained from the continuous-time limit of an ECTCM is given by [Eq.~(\ref{eq:L_CM_def})], 
\begin{align}
    \mathcal{L}_t^{\rm (CM)}(\cdot) &= -\nI[H_S'^{\rm (CM)}(t), \cdot] \nonumber\\
    &\quad + \sum_{j,j',m,m'}\xi_{jj'mm'}^{\rm (CM)}(t)F_S^{jj'}(\cdot)(F_S^{mm'})^\dagger \nonumber\\
    &\quad - \frac{1}{2}\sum_{j,j',m,m'}\xi_{jj'mm'}^{\rm (CM)}(t)[(F_S^{mm'})^\dagger F_S^{jj'}, \cdot]_+.
\end{align}
where $\xi_{jj'mm'}^{\rm (CM)}(t)$ is in the form of Eq.~(\ref{eq:xi_CM}):
\begin{align}
    \xi_{jj'mm'}^{\rm (CM)}(t) &\equiv \delta(\epsilon_{j'}-\epsilon_j + \epsilon_{m} - \epsilon_{m'})\nonumber\\
    &\quad \times\sum_{k,k'|e_{k} - e_{k'} = \epsilon_j - \epsilon_{j'}} p^{(k)}_E V_{jk'j'k}(t)V_{mk'm'k}^*(t),
    \label{eq:xi_CM_1}
\end{align}
with $p^{(k)}_E\equiv \bra{k}_E\gamma_E\ket{k}_E$ and $V_{jkj'k'}(t) \equiv \bra{j}_S\otimes\bra{k}_E V_{SE}(t) \ket{j'}_S\otimes\ket{k'}_E$.

\subsubsection{ECTCM Hamiltonian from $\mathcal{L}_t$}
To compare $\mathcal{L}_t^{\rm (CM)}$ with the target $\mathcal{L}_t$, we construct an ECTCM directly in the continuous-time limit. Concretely, we will define the driving term $D_S(t)$, the ancilla Hamiltonian $H_E$ (identical for all ancillae), and the system-ancilla interaction $V_{SE}(t)$. The time variable 
$t$ should be understood as the continuous-time limit of the discrete collision times $t_i\equiv i \Delta t$, where $i$ labels the ancillae and $\Delta t$ is the collision duration. For a finite $\Delta t$, the construction specifies the Hamiltonian of the $(t/\Delta t)$th ancilla and the corresponding system--ancilla interaction. We return to this point in Secs.~\ref{sec:protocol} and \ref{sec:error_analysis} when discussing the implementation protocol and the associated errors.

We define the {\bf driving term} $D_S(t)$ in the ECTCM as
\begin{align}
    D_S(t) \equiv H_S'(t) - H_S,
    \label{eq:D_S_constructed}
\end{align}
where $H_S'(t)$ is the effective system Hamiltonian from $\mathcal{L}_t$, such that $H_S'^{\rm (CM)}(t)\equiv H_S + D_S(t) = H_S'(t)$ and $[D_S(t), H_S] = 0$ since $[H_S'(t), H_S]=0$ by Axiom~\ref{axiom:II} [Eq.~(\ref{eq:H_S'_commute_H_S})]. Therefore, the ECTCM includes an energy-conserving driving term required by Condition~\ref{CM:1}, and the corresponding Lindbladian $\mathcal{L}_t^{\rm (CM)}$ has the same unitary term as that in the target $\mathcal{L}_t$.

We construct the {\bf ancilla Hamiltonian} $H_E$, taken to be identical for all ancillae in the ECTCM as follows:
\begin{align}
    H_E \equiv \sum_{\ell,\ell',q=0}^{d_S-1} \frac{\epsilon_\ell - \epsilon_{\ell'}}{2}\ketbra{(\ell,\ell'),q}{(\ell,\ell'),q}_E,
    \label{eq:H_E_constructed}
\end{align}
where $\{\epsilon_\ell\}_{\ell=0}^{d_S-1}$ is the energy spectrum of $H_S$ and the pair $(\ell, \ell')$ corresponds to a $d_S$-dimensional eigensubspace of $H_E$, in which $q$ is the index for degenerate energy levels.  Introducing the index $k\equiv \ell d_S^2 + \ell' d_S + q$, we have $H_E \equiv \sum_{k=0}^{d_S^3-1}e_k \ketbra{k}{k}_E$ with eigenvalues $\{e_k\equiv \frac{\epsilon_{\ell}-\epsilon_{\ell'}}{2}\}_{k=0}^{d_S^3-1}$.

We define the {\bf system--ancilla interaction Hamiltonian} $V_{SE}(t)$ in the ECTCM as
\begin{widetext}
\begin{align}
    V_{SE}(t) \equiv \sum_{j,j',q,q'=0}^{d_S-1}\sqrt{Z} \frac{W_{jj'qq'}(t)}{N_{(j,j')}^{1/2}}\ketbra{j}{j'}_S\otimes\left(\sum_{\ell,\ell'=0|\epsilon_{\ell'}-\epsilon_{\ell} = \epsilon_{j'}-\epsilon_j}^{d_S-1}\ketbra{(\ell',\ell),q'}{(\ell,\ell'),q}_E\right),
    \label{eq:V_constructed}
\end{align}
\end{widetext}
where $W_{jj'qq'}(t)$ is from the expression of $\xi_{jj'mm'}(t)$ [Eq.~(\ref{eq:xi_general_under_axioms_1})],
$Z\equiv \sum_{k=0}^{d_S^3-1}\nE^{-\beta e_k}$ is the partition function for the ancilla $E$ and
$N_{(j,j')}\equiv \big|\{(\ell,\ell')|\epsilon_\ell - \epsilon_{\ell'}=\epsilon_j - \epsilon_{j'}\}\big|$ is the degeneracy of the energy gap $\epsilon_j - \epsilon_{j'}$. 

In App.~\ref{app:hermiticity_EC_of_V}, we show that $V_{SE}(t)$ is an energy-conserving Hamiltonian satisfying
$V_{SE}(t) = V_{SE}^\dagger(t)$ and $[V_{SE}(t), H_S+H_E] = 0$. Substituting $V_{SE}(t)$ into the expression for $\xi_{jj'mm'}^{\rm (CM)}(t)$ in Eq.~(\ref{eq:xi_CM_1}) (see App.~\ref{app:xi_CM_=_xi} for details),
we obtain that
\begin{align}
    \xi_{jj'mm'}^{\rm (CM)}(t) = \xi_{jj'mm'}(t).
\end{align} 
This implies that, in the continuous-time limit, the ECTCM produces the same dissipator as the target Lindbladian.

Hence, the constructed ECTCM satisfies Conditions~\ref{CM:1}--\ref{CM:3} and generates the Lindbladian $\mathcal{L}_t^{\rm (CM)} = \mathcal{L}_t$. However, the derivation of $\mathcal{L}_t^{\rm (CM)}$ also relies on Condition~\ref{CM:4}, which requires the Lamb shift to vanish, as stated in Lem.~\ref{lem:CM_4}. We now verify that this condition is also satisfied by our construction.

\subsubsection{Gauge freedom and vanishing Lamb shift}
Recall that the Lamb shift is given by $H_{S}^{\rm LS}(t) \equiv \mathrm{Tr}_{E}\{V_{SE}(t)\gamma_{E}\}$ [Eq.~(\ref{eq:H_S^LS_def})]. As shown in App.~\ref{app:vanishing_H_LS}, for the constructed interaction $V_{SE}(t)$ in Eq.~(\ref{eq:V_constructed}), the condition $H_{S}^{\rm LS}(t)=0$ is equivalent to
\begin{align}
    \sum_{q=0}^{d_S-1} W_{jj'qq}(t) = 0, \quad \forall\, (j,j') \text{ with } \epsilon_j = \epsilon_{j'}.
     \label{eq:condition_on_W_for_vanishing_LS}
\end{align}
Although this condition generally does \textit{not} hold for the coefficients satisfying Axioms~\ref{axiom:I}–\ref{axiom:III} [see  Eq.~\eqref{eq:xi_general_under_axioms_1}], we recall that thermal Lindbladians possess an inherent gage freedom $\mathbf{U}(t)$, already mentioned in Sec.~\ref{sec:structure_thermal_Lindbladian}, which we can exploit to satisfy the above equation without changing the corresponding thermal Lindbladian $\mathcal{L}_t$. 

To illustrate this, we first write the coefficients $W_{xx'yy'}(t)$ as a matrix $\bm{W}(t)\in\mathbb{C}^{d_S^2\times d_S^2}$ by treating $(x,x')$ and $(y,y')$ as row and column double indices. Denoting the number of conditions in Eq.~\eqref{eq:condition_on_W_for_vanishing_LS} by $N_0$, we define the $N_0 \times d_S^2$ rectangular sub-matrix $\bm{M}(t)$ of $\bm{W}(t)$ by restricting $\bm{W}(t)$ to the rows indexed by $\mathscr{R}\equiv\{(\ell,\ell')|\epsilon_\ell = \epsilon_{\ell'}\}$, such that $\bm{M}(t):= \bm{W}(t)_{\mathscr{R},:}$. Eq.~(\ref{eq:condition_on_W_for_vanishing_LS}) can therefore be written compactly as
\begin{align}
    \bm{M}(t)\ket{v} = 0,
    \label{eq:M(t)_v}
\end{align}
where $\ket{v}\equiv\sum_{q=0}^{d_S-1}\ket{q,q}\in\mathbb{C}^{d_S^2}$.

In App.~\ref{app:proof_of_proposition_xi}, we show that $\bm{W}(t)$ admits a gauge freedom of the form $\bm{W}(t)\rightarrow \bm{W}'(t)\equiv \bm{W}(t)\bm{U}(t)$, where $\bm{U}(t)\in\mathbb{C}^{d_S^2\times d_S^2}$ is a unitary matrix satisfying certain symmetry constraints. It is straightforward to verify that $\bm{W}(t)$ and $\bm{W}'(t)$ lead to the same thermal Lindbladian $\mathcal{L}_t$, i.e., the coefficients $\xi_{jj'mm'}(t)$ [Eq.~(\ref{eq:xi_general_under_axioms_1})] remain invariant under this gauge transformation. Applying the transformation to Eq.~(\ref{eq:M(t)_v}) yields
\begin{align}
    \bm{M}(t)\bm{U}(t)\ket{v} = 0,
\end{align}
which is satisfied whenever $\bm{U}(t)\ket{v}$ lies in the kernel of $\bm{M}(t)$. In App.~\ref{app:vanishing_H_LS}, we prove that such a $\bm{U}(t)$ always exists provided that the system Hamiltonian $H_S$ is not fully degenerate. Therefore, within the constructed ECTCM, the Lamb shift $H_S^{\rm LS}(t)$ can always be eliminated by an appropriate gauge choice. The resulting ECTCM thus satisfies Condition~\ref{CM:4}.

Combining the above construction with the gauge freedom that removes the Lamb shift, we arrive at Prop.~\ref{prop:ME_CM_equivalence}, which establishes the converse direction of Prop.~\ref{prop:L_CM_=>_L}. Together, these two propositions establish the equivalence between the axiomatic and microscopic descriptions of Markovian quantum thermodynamics, therefore proving the first equivalence in Thm.~\ref{thm:equivalences}.

In addition, we emphasise that Prop.~\ref{prop:ME_CM_equivalence} is not merely an existence result, but also provides an explicit constructive protocol. In particular, the ECTCM construction specified by Eqs.~(\ref{eq:D_S_constructed}), (\ref{eq:H_E_constructed}) and (\ref{eq:V_constructed}) can be directly applied to simulate any dynamics satisfying Axioms~\ref{axiom:I}–\ref{axiom:III}. At finite times $\Delta t$, this construction provides a digital simulator of $\mathcal{L}_t$ with a simulation error that vanishes for $\Delta t \to 0$. We provide a construction protocol and its corresponding error analysis in Sec.~\ref{sec:ECTCM_simulation} and demonstrate its explicit application in Sec.~\ref{sec:applications}.

\section{Operational description}\label{sec:operational_description}
Having established the equivalence between the axiomatic and microscopic formulations of Markovian quantum thermodynamics, we now turn to the operational description within the framework of quantum resource theories~\cite{chitambar_quantum_2019,Gour_2025}.
In this resource-theoretic approach to thermodynamics, athermality -- the deviation from thermal equilibrium -- is treated as a thermodynamic resource~\cite{janzing2000thermodynamic,brandao2013resource,horodecki2013fundamental,Brando2015,lostaglio2019introductory}. Thermodynamic processes are therefore characterised as state transformations that can be implemented without access to any additional source of athermality. The resulting free operations are thermal operations, namely those that can be realised by coupling the system to a thermal ancilla through an energy-conserving unitary interaction (similar to the ECTCMs discussed above). 

To incorporate Markovianity, one further restricts attention to thermal operations that arise from continuous-time Lindbladian evolution. This leads to the notion of Markovian thermal operations (MTOs)~\cite{spaventa2022capacity,son2024hierarchy}, Markovian evolutions given by a generator $\mathcal{L}_t^{\text{MTO}}$ that result in thermal operations for \textit{every} evolution time. Unlike the axiomatic and microscopic formulations, which focus on the structure of dynamical generators and their physical realisations, respectively, the operational perspective focuses directly on the thermodynamic capabilities and limitations of quantum processes. This viewpoint is particularly useful for tasks such as designing optimal thermalisation protocols~\cite{korzekwa2022optimizing,czartowski_thermal_2023} and analysing the thermodynamic costs of quantum information processing~\cite{faist_fundamental_2018,faist_thermodynamic_2019,Hsieh2025Dynamical,luo_thermodynamic_2025}, where the underlying generators and microscopic mechanisms may not be specified explicitly.

Despite its distinct motivation and level of description, the operational approach is intimately related to the axiomatic and microscopic formulations. In particular, from the discussion above, it is immediate to see that every ECTCM engenders an MTO, since each collision itself corresponds to a thermal operation. However, the converse implication, that \textit{every} MTO generator $\mathcal{L}_t^\text{MTO}$ can indeed obtained from an ECTCM is not obvious and has remained an open question in the literature~\cite{vomEnde2023ExploringLimits}. In principle, MTO generators might exist that fundamentally cannot be implemented as the continuous-time limit of an ECTCM. In turn, this would prevent their conceptualisation in terms of an underlying thermodynamically consistent model, despite the system evolution being a thermal operation at all times. However, as we show here, this is not the case, and the sets of MTO generators and generators obtained from ECTCMs exactly coincide (under mild mathematical assumptions, see below).

The remainder of this section is organised as follows. We first formally define MTOs (Sec.~\ref{sec:MTO}). We then prove their equivalence with ECTCMs (Sec.~\ref{sec:ECTCM_=_MTO}), completing the proof of Thm.~\ref{thm:equivalences}. Finally, we use this equivalence to derive a complete axiomatic characterisation of MTO generators and discuss its implications (Sec.~\ref{sec:axiomatic_MTO}).

\subsection{Markovian thermal operations}\label{sec:MTO}
In the resource theory of athermality, thermalisation processes are treated as free operations, i.e., operations that can be implemented without consuming athermality. The most widely accepted class of free operations is given by thermal operations (TOs)~\cite{janzing2000thermodynamic,brandao2013resource,horodecki2013fundamental}. Consider a system $S$ with Hamiltonian $H_S$ interacting with an environment $E$ with Hamiltonian $H_E$. A TO
$\mathcal{T}$ is then defined by
\begin{align}
    \mathcal{T}(\cdot):= {\rm Tr}_E\left\{U(\cdot\otimes \gamma_E) U^\dagger\right\},
    \label{eq:TO_def}
\end{align}
where $\gamma_E \equiv \nE^{-\beta H_E}/{\rm Tr}\{\nE^{-\beta H_E}\}$ is the thermal state of the environment $E$ and the joint unitary $U$ is energy-conserving:
\begin{align}
    [U, H_S+H_E] = 0,
\end{align}
ensuring that the energy expectation $\langle H_S+H_E \rangle$ remains invariant under the evolution (Sec.~\ref{sec:EC}). 
Because TOs are implemented only by thermodynamically compliant resources, they necessarily preserve the Gibbs state, $\mathcal{T}(\gamma_S) = \gamma_S$ -- TOs cannot generate athermality from thermal equilibrium. 

TOs are the central state transformations considered in the resource-theoretic approach to quantum thermodynamics \cite{horodecki2013fundamental, brandao2013resource, Brando2015, Masanes2017,Muller2018Correlating,ng2019resource,shiraishi2025recovery}. However, the definition of a TO places no restriction on the memory structure of the underlying dynamics. In particular, a TO may arise from a highly non-Markovian process involving complex system--environment correlations~\cite{korzekwa2022optimizing,Lostaglio2022Continuous,czartowski_thermal_2023} and therefore does not, by itself, capture the notion of memoryless thermalisation.
In order to incorporate Markovianity into TOs, the so-called Markovian thermal operations (MTOs) have been defined ~\cite{spaventa2022capacity,son2024hierarchy}: 
A thermal operation $\mathcal{T}$ is an MTO if it satisfies the following two requirements: 
\begin{operationalbox}[boxsep=0pt,left=1em,right=1em]
\begin{center}
    {\fontsize{10pt}{12pt}\selectfont {\bf Markovian Thermal Operations (MTOs)}}
\end{center}

\begin{enumerate}[leftmargin=2em]
    \item[(i)] For some time $\tau \ge 0$ there exists a {\it time-continuous} Lindbladian operator $\mathcal{L}_t^{\rm (MTO)}$ such that
    \begin{align}
        \mathcal{T} = \overleftarrow{\rm T}\exp(\int_0^\tau \mathcal{L}_t^{\rm (MTO)}{\rm d} t ).
    \end{align} 
    \item[(ii)] For any time interval $[t_1, t_2]\subset [0,\tau]$, the map $\overleftarrow{\rm T}\exp(\int_{t_1}^{t_2} \mathcal{L}_t^{\rm (MTO)}{\rm d} t )$ is also a TO.
\end{enumerate}
\end{operationalbox}
While the above definition characterises MTOs operationally, it reveals little about their dynamical structure and their microscopic realisation. We now address both questions by relating MTOs to ECTCMs.

\subsection{Equivalence between ECTCMs and MTOs}\label{sec:ECTCM_=_MTO}
In this subsection, we establish the equivalence between ECTCMs and MTOs by proving that they define the same class of Markovian dynamics.
\begin{proposition}\label{prop:MTO_=_CM}
    The set of MTO generators $\mathcal{L}_t^{\rm (MTO)}$ and the set of Lindbladians $\mathcal{L}_t^{\rm (CM)}$ [Eq.~(\ref{eq:L_CM_def})] obtained in the continuous-time limit of ECTCMs coincide.
\end{proposition}
\begin{proof}[Proof sketch (full proof in App.~\ref{app:proof_of_MTO_=_CM})]
    It is immediate that every ECTCM induces an MTO. The continuous-time evolution generated by $\mathcal{L}_t^{\rm (CM)}$ is the limit of the $n$-collision map $\mathcal{E}_n$ with $\tau \equiv n\Delta t$ as $\Delta t\rightarrow 0$. Since each collision is a thermal operation by virtue of the defining conditions of ECTCMs, and thermal operations are closed under composition, $\mathcal{E}_n$ is itself a TO for all~$n$. Hence, its continuous-time limit is an MTO. Although this construction may require an infinite-dimensional environment in the continuous-time limit, Sec.~\ref{sec:error_analysis} shows that the dynamics can be approximated arbitrarily well using finite-dimensional environments.

    To prove the converse, we exploit the defining property of an MTO: for every sufficiently small time step $\Delta t$, the corresponding propagator is itself a TO. This provides a sequence of energy-conserving microscopic implementations, which we identify with the individual collisions of an ECTCM. Since we consider finite-dimensional systems, the corresponding Lindbladian is bounded. We then show that this boundedness forces the associated interaction Hamiltonians to scale as $H_{SE_i}(\Delta t) = O(1) + O(1/\sqrt{\Delta t})$, precisely matching the ECTCM structure and satisfying Conditions~\ref{CM:1}--\ref{CM:4}. It follows that every bounded MTO generator arises as the continuous-time limit of an ECTCM.
\end{proof}
Prop.~\ref{prop:MTO_=_CM} demonstrates that microscopic and operational formulations of Markovian quantum thermodynamics characterise the same class of dynamics. Combining this with Props.~\ref{prop:L_CM_=>_L} and \ref{prop:ME_CM_equivalence}, which establish the equivalence between axiomatic and microscopic descriptions, we therefore prove the {\it Axiomatic–Microscopic–Operational Equivalence} as stated in Thm.~\ref{thm:equivalences}.

As an immediate consequence, this triple equivalence reveals the structure of MTO generators: every $\mathcal{L}_t^{\rm (MTO)}$ is a thermal Lindbladian, and vice versa. This establishes a complete axiomatic characterisation of MTOs, thereby connecting their operational definition in terms of thermal operations with a simple generator-level description. As we discuss next, this characterisation has important implications for the long-standing problem of identifying and characterising thermodynamically admissible dynamics within the resource-theoretic framework of quantum thermodynamics.

\subsection{Axiomatic characterisation of MTO generators}\label{sec:axiomatic_MTO}
While TOs provide the standard operational framework for quantum thermodynamics in the resource theoretic setting, their characterisation remains mathematically challenging due to the highly non-unique energy-conserving unitary dilations underlying their definition with no canonical form known~\cite{vom2022bath}. Consequently, several alternative classes of free operations with more tractable mathematical descriptions have been proposed and studied, most notably Gibbs-preserving operations~(GPOs)~\cite{faist_fundamental_2018,Shiraishi2021GPO,luo_thermodynamic_2025} and their time-covariant version~(\mbox{CGPOs})~\cite{Cfiwikl2015Limitations,Shiraishi2025CovGPO}. However, the existence of strict gaps between TOs and GPOs~\cite{faist_gibbs-preserving_2015,Tajima2025Gibbs-preserving}, as well as between TOs and CGPOs~\cite{Cfiwikl2015Limitations,Ding2021Exploring}, demonstrates that this mathematical simplification comes at the cost of enlarging the set of admissible processes beyond what is physically well-motivated.

In the Markovian regime, the corresponding question asks for characterising MTOs and in particular their generators. Mirroring the situation for TOs, obtaining a tractable mathematical characterisation of MTOs has remained a highly nontrivial problem~\cite{vomEnde2023ExploringLimits}. Remarkably, Thm.~\ref{thm:equivalences} provides such a characterisation. 
By establishing the equivalence between the axiomatic and operational descriptions, it identifies MTO generators precisely as those satisfying the three thermodynamic axioms: Markovianity~\ref{axiom:I}, time-translation symmetry~\ref{axiom:II} and quantum detailed balance~\ref{axiom:III}. 
Consequently, violations of the three axioms admit a direct operational interpretation as the resources required to realise dynamics beyond the MTO framework.

In addition to this axiomatic characterisation, Ref.~\cite{vomEnde2023ExploringLimits} identified a family of MTO generators and conjectured that it exhausts the entire set of MTO generators (up to closure). In App.~\ref{app:relation_to_known_MTOs}, we show that this family coincides with the class of generators subject to Thm.~\ref{thm:equivalences} and therefore prove the conjectured completeness of the generator family under the assumptions considered in this work (namely, non-trivial system Hamiltonians and continuous time-dependent generators).

Finally, beyond the axiomatic characterisation, Thm.~\ref{thm:equivalences} is constructive: any axiomatic thermal Lindbladian generates an MTO through its microscopic ECTCM realisation. We now translate this construction into an explicit protocol for simulating arbitrary Markovian thermal dynamics using ECTCMs.

\section{ECTCM simulation for Markovian quantum thermodynamics}\label{sec:ECTCM_simulation}
Having established the equivalence between the axiomatic, microscopic, and operational formulations of Markovian quantum thermodynamics, we now turn to its constructive and practical implications. The equivalence between thermal Lindbladians and ECTCMs implies that any thermodynamically consistent Markovian dynamics can be simulated through a suitable sequence of energy-conserving collisions with thermal ancillae.
In Sec.~\ref{sec:protocol}, based on the considerations of Sec.~\ref{sec:constructing_ECTCM}, we provide such a ``digital" simulation protocol for constructing ECTCMs for arbitrary thermal Lindbladians, together with an analysis of the simulation error in Sec.~\ref{sec:error_analysis} for finite collision times $\Delta t$.

\subsection{Protocol for ECTCM simulation}\label{sec:protocol}
We summarise the ECTCM construction from Sec.~\ref{sec:constructing_ECTCM} in Protocol~\ref{tab:protocol} as a simulation protocol for arbitrary dynamics generated by thermal Lindbladians satisfying Axioms~\ref{axiom:I}–\ref{axiom:III}. The full details of the protocol are provided in the extended version, Protocol~\ref{tab:protocol_full}, presented in App.~\ref{app:full_protocol}. The protocol can be followed as a recipe: given a thermal Lindbladian $\mathcal{L}_t$ satisfying Axioms~\ref{axiom:I}–\ref{axiom:III}, an evolution time $\tau$, and collision time~$\Delta t$ as inputs, the protocol provides an ECTCM with $n=\tau/\Delta t$ collisions.
{
\renewcommand{\tablename}{Protocol}
\begin{table}[h]
\caption{\textbf{ECTCM simulation protocol}}
\label{tab:protocol}
\centering
\begin{tabular}{c p{0.93\linewidth}}
\hline\hline
\noalign{\vskip 1mm}
1. &
Choose the {\bf collision time} $\Delta t$ and the corresponding number of collisions
$n \equiv \tau/\Delta t$, where $\tau$ is the target {\bf evolution time}.
\\[2mm]
2. &
Stroboscopically sample the {\bf effective Hamiltonian} $H_S'(t)$ and the {\bf CP map} $\mathcal{A}_t$ involved in $\mathcal{L}_t$ [Eq.~(\ref{eq:L_gen})] at discrete times as
$\{H_S'(i\Delta t)\}_{i=1}^n$ and $\{\mathcal{A}_{i\Delta t}\}_{i=1}^n$, respectively.\\[2mm]
3. & Calculate the {\bf coefficients} $\xi_{jj'mm'}(i\Delta t)$ from $\mathcal{A}_{i\Delta t}$ according to Eq.~(\ref{eq:xi_def_from_A}).
\\[2mm]
4. &
Construct the {\bf coefficients}
$W_{jj'qq'}(i\Delta t)$ from
$\xi_{jj'mm'}(i\Delta t)$ [Eq.~(\ref{eq:xi_general_under_axioms})]; details of the construction, including the gauge choice, are provided in Protocol~\ref{tab:protocol_full}.
\\[2mm]
5. &
Construct the {\bf driving Hamiltonian}
$D_{S,i}$ from $H_S'(i\Delta t)$ according to Eq.~(\ref{eq:D_S_constructed}).
\\[2mm]
6. &
Construct the {\bf ancilla Hamiltonian} $H_{E_i}$
according to Eq.~(\ref{eq:H_E_constructed}).\\[2mm]
7. &
Construct the {\bf interaction Hamiltonian} $gV_{SE_i}$ according to Eq.~(\ref{eq:V_constructed}) with $g$ satisfying $g^2\Delta t = 1$.
~\\[2mm]
8. &
Implement the {\bf collision map} $\mathcal{E}_n$ [Eq.~(\ref{eq:collision_map_E_n})].\\[1mm]

\hline\hline
\end{tabular}
\end{table}
}

\subsection{Error analysis for Protocol~\ref{tab:protocol}}\label{sec:error_analysis}
Protocol~\ref{tab:protocol} implements a discrete-time map $\mathcal{E}_n$ [Eq.~(\ref{eq:collision_map_E_n})] which leads to the continuous-time Markovian dynamics $\overleftarrow{\rm T}\exp(\int_0^\tau \mathcal{L}_t^{\rm (CM)}{\rm d} t)$, with $\overleftarrow{\rm T}$ indicating chronological time-ordering, by taking the standard continuous-time limit for collision models (Sec.~\ref{sec:ME_from_CM}). By Prop.~\ref{prop:ME_CM_equivalence}, $\mathcal{L}_{t}^{\rm (CM)} = \mathcal{L}_t,\,\forall\, t$, assuming that the system Hamiltonian is not fully degenerate.
In this section, we quantify the error due to finite time steps, by analysing the distance between the discrete-time map $\mathcal{E}_n$ and the target map $\overleftarrow{\rm T}\exp(\int_0^\tau \mathcal{L}_t\,{\rm d} t)$. 
To exclude cases where $\mathcal{L}_t$ varies dramatically in time, we assume that 
$\mathcal{L}_t$ is Lipschitz continuous (i.e., its rate of change is globally bounded), thus excluding unphysical situations~\cite{heinonen2001lectures}. This means that there exists a constant $C_{\rm L}>0$ such that for any two times $r, u\ge 0$, the following bound holds:
\begin{align}
    \|\mathcal{L}_r-\mathcal{L}_u\|_{1\rightarrow 1} \le C_{\rm L} |r-u|,
\end{align}
where $\|\cdot\|_{1\rightarrow1}$ denotes the $1\rightarrow 1$ norm of linear maps~\cite{watrous2018theory}.

The relevant error (distance) is then
\begin{align}
    \varepsilon := \|\mathcal{E}_n - \overleftarrow{\rm T}\nE^{\int_0^\tau \mathcal{L}_t{\rm d} t}\|_{1\rightarrow 1}.
\end{align}
To obtain an overall error bound, we decompose the total simulation error as follows (See App.~\ref{app:error_analysis}):
\begin{align}
    \varepsilon \le n \max_{1\le i\le n}\varepsilon_{\rm s}(i) + n\max_{1\le i \le n}\left[\varepsilon_{\rm t}(i) + \varepsilon_{\rm c}(i)\right],
    \label{eq:varepsilon_bound_gen}
\end{align}
where $\varepsilon_{\rm s}(i)$ denotes the stroboscopic sampling error at the $i$th step,
\begin{align}
    \varepsilon_{\rm s}(i) := \|\nE^{\mathcal{L}_{i\Delta t}\Delta t} - \overleftarrow{\rm T}\nE^{\int_{(i-1)\Delta t}^{i\Delta t} \mathcal{L}_t{\rm d} t}\|_{1\rightarrow 1},
\end{align}
$\varepsilon_{\rm t}(i)$ denotes the Taylor truncation error,
\begin{align}
    \varepsilon_{\rm t}(i) := \|(\mathcal{I}+\mathcal{L}_i\Delta t)-\nE^{\mathcal{L}_i\Delta t}\|_{1\rightarrow 1},
\end{align}
and $\varepsilon_{\rm c}(i)$ denotes the collision error,
\begin{align}
    \varepsilon_{\rm c}(i) := \|\hat{\mathcal{E}}_i - (\mathcal{I}+\mathcal{L}_i\Delta t)\|_{1\rightarrow 1},
\end{align}
where $\hat{\mathcal{E}}_i$ is the $i$th collision channel defined as $\hat{\mathcal{E}}_i(\cdot) := {\rm Tr}_{E_i}\{U_i(\cdot\otimes\rho_{E_i})U_i^\dagger\}$.
Therefore, to bound $\varepsilon$, it suffices to bound the sampling error $\varepsilon_{\rm s}(i)$, the truncation error $\varepsilon_{\rm t}(i)$ and the collision error $\varepsilon_{\rm c}(i)$ individually. 
We defer the full analysis to App.~\ref{app:error_analysis}, where explicit bounds on these contributions are derived both for general collision-model simulations of arbitrary GKSL dynamics and, in particular, for the ECTCM simulation described in Protocol~\ref{tab:protocol}. Combining these bounds and using $\tau\equiv n\Delta t$, we prove that, for Protocol~\ref{tab:protocol} and sufficiently small $\Delta t$,
\begin{align}
    \varepsilon = O(\tau\sqrt{\Delta t}) = O\left(\frac{\tau^{3/2}}{n^{1/2}}\right).
    \label{eq:error_scaling}
\end{align}
In other words, the number of collisions for the ECTCM in Protocol~\ref{tab:protocol} to achieve an error below $\varepsilon$ scales as
\begin{align}
    n = O\left(\frac{\tau^3}{\varepsilon^2}\right).
\end{align}
In the ECTCM, the ancilla dimension is $d_S^3$ [see Eq.~(\ref{eq:H_E_constructed})]. If the ancilla can be reset to its thermal state after each collision, only a single ancilla of this size is required. Otherwise, if each collision requires an additional fresh ancilla, the total environment dimension $d_E$ scales as
\begin{align}
    d_E = (d_S^3)^{O\left({\tau^3}/{\varepsilon^2}\right)}.
    \label{eq:scaling_of_d_E}
\end{align}
This result is consistent with the standard open-system perspective that a Markovian master equation $(\varepsilon = 0)$ arises from coupling to an effectively infinite thermal reservoir~\cite{davies1974markovian,breuer2002theory}. In addition,  within the ECTCM framework, it quantifies how the dynamics would deviate from the Markovian limit, through the error $\varepsilon$, when interacting with a finite-dimensional environment $(d_E < \infty)$. 

We emphasise, however, that the requirement of an infinite-dimensional environment stems from the fact that the ECTCM simuation is \textit{universal}: it applies to every thermal Lindbladian satisfying Axioms~\ref{axiom:I}–\ref{axiom:III} and, once constructed, reproduces the corresponding continuous-time dynamics for arbitrary evolution times. This universality, however, does not imply that the construction is dimension-optimal. Indeed,
for a specific thermal Lindbladian $\mathcal{L}$ (assume that it is time-independent for simplicity), the map (i.e., MTO) $\nE^{\mathcal{L} \tau}$ at a fixed evolution time $\tau$ may admit an \textit{exact} implementation using only a finite-dimensional environment. We conjecture this behaviour to be generic and provide an explicit example in App.~\ref{app:finite_MTO_realisation}.

\section{ECTCM constructions for physical systems}\label{sec:applications}
In this section, we illustrate the practical utility of our Protocol~\ref{tab:protocol} by constructing ECTCMs to simulate two paradigmatic physical models. In Sec.~\ref{sec:TLS_model} we present an ECTCM simulation for a two-level system coupled to a bosonic environment, while in Sec.~\ref{sec:three_level_ATM}, we transform a three-level autonomous thermal machine into a three-stroke engine operating through a cycle of three sequential ECTCMs. Remarkably, we prove that the resulting engine not only faithfully reproduces the continuous dynamics for $\Delta t \to 0$, but actually retains the same efficiency as the original thermal machine, regardless of the collision time.

\subsection{Two-level atom in a bosonic environment}\label{sec:TLS_model}
A two-level system (TLS) coupled to a bosonic environment constitutes a canonical model for dissipative quantum dynamics and a paradigmatic setting for studying thermalisation in quantum systems~\cite{DeRaedt1984Thermodynamics,Leggett1987Dynamics,breuer2002theory,gardiner2004quantum,paule2018thermodynamics}. 
It therefore serves as a natural benchmark for demonstrating our ECTCM simulation protocol.

Here, following Ref.~\cite{paule2018thermodynamics}, we consider a two-level  atom characterised by the Hamiltonian $H_S = \epsilon \sigma^+\sigma^-$ with $\sigma^- \equiv \ketbra{g}{e}_S$ and $\sigma^+ = (\sigma^-)^\dagger$ being the lowering and raising operators for the atom. The atom is coupled to an infinite number of bosonic modes with the Hamiltonian
$H_E = \sum_\nu \Omega_\nu b_\nu^\dagger b_\nu$ with $b_\nu$ and $b_\nu^\dagger$ being the annihilation and creation operators for the mode $\nu$.
Within the rotating-wave approximation~\cite{burgarth2024taming}, the interaction Hamiltonian $V$ takes the multi-mode Jaynes--Cummings form
\begin{align}
    V = \sum_\nu f_\nu (\sigma^- \otimes b_\nu^\dagger + \sigma^+ \otimes b_\nu),
\end{align}
where $f_\nu$ is the coupling strength between the atom and the $\nu$th mode. Taking the limit of a continuum of modes and applying the Born--Markov approximation in the weak-coupling limit, the secular approximation selects the resonant frequency of the reservoir as the relevant damping contribution~\citep[]
{paule2018thermodynamics}. Further assuming an initial thermal state of the bosonic environment and neglecting the Lamb shift, the reduced dynamics of the atom are governed by the following GKSL master equation: $\dot{\rho}_S(t) = \mathcal{L}^{\rm (TLS)}(\rho_S(t))$, where
\begin{align}
    \mathcal{L}^{\rm (TLS)}(\cdot) := -{\rm i}[H_S, \cdot] + \Gamma \left[(n_{\rm th} + 1) \mathcal{D}_{\sigma^-}(\cdot) + n_{\rm th}\mathcal{D}_{\sigma^+}(\cdot)\right],
\end{align}
with $\mathcal{D}_X(\cdot):= X(\cdot)X^\dagger - [X^\dagger X, \cdot]_+/2$, $n_{\rm th} \equiv (\nE^{\beta \epsilon} - 1)^{-1}$ denoting the Bose--Einstein thermal occupation number and the damping coefficient $\Gamma\propto f^2$ that depends on the density of states at the atom's frequency $\epsilon$. 

Since $[V, H_S+H_E] \neq 0$ in the TLS model, $\mathcal{L}^{\rm (TLS)}$ does a priori not arise from an energy-conserving interaction. However, it is easy to check that $\mathcal{L}^{\rm (TLS)}$ satisfies Axioms~\ref{axiom:I}--\ref{axiom:III}. Consequently, Thm.~\ref{thm:equivalences} guarantees the existence of an equivalent ECTCM. Applying Protocol~\ref{tab:protocol}, we obtain an ECTCM in which each ancilla consists of three qubits with Hamiltonian
\begin{align}
    H_E^{\rm (CM)} = \frac{\epsilon}{2}(|10\rangle\!\langle 10| - |01\rangle\!\langle 01|)\otimes\mathds{1}_2,
\end{align}
identical for every collision $i$, and the  system--ancilla interaction is
\begin{align}
    V_{SE}^{\rm (CM)} = V_0(|g\rangle\!\langle e|_S\otimes |101\rangle\!\langle 010|_E + {\rm h.c.}),
\end{align}
where $V_0\equiv \sqrt{2\Gamma\left[1+\cosh(\beta\epsilon/2)\right]/\sinh(\beta\epsilon/2)}$. By construction, $[V_{SE}^{\rm (CM)}, H_S+H_E^{\rm (CM)}]=0$, while the resulting ECTCM generator reproduces the original one, i.e., $\mathcal{L}^{\rm (CM)} = \mathcal{L}^{\rm (TLS)}$.
\begin{figure}
    \centering
    \includegraphics[width=\linewidth]{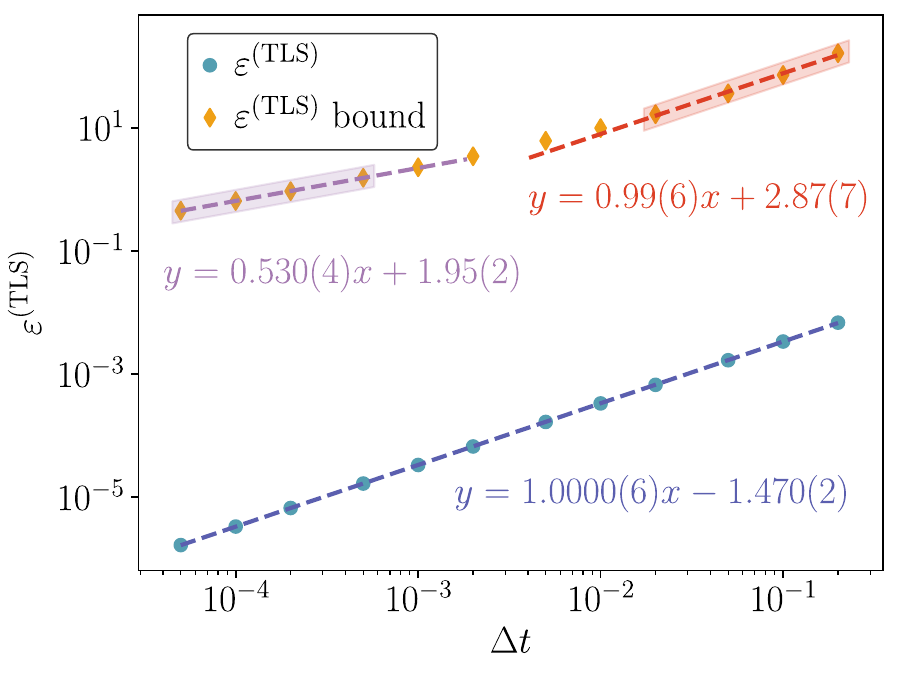}
    \caption{Error scaling of the ECTCM simulation of the TLS model. The initial state is $\rho_S(0)=\ketbra{\psi}{\psi}_S$ with $\ket{\psi}_S=(\ket{g}_S+\sqrt{3}\ket{e}_S)/2$. Similar error scaling is observed for other initial states. Other parameters are set to $\epsilon=1$, $\beta=0.1$, $\Gamma=0.01$, and $\tau=4$. The system remains out of equilibrium at time $\tau$, as quantified by $\|\rho_S(\tau)-\gamma_S\|_1 \simeq 0.63$. The error $\varepsilon^{\rm (TLS)}$ and its bound [Eq.~(\ref{eq:error_bound})] are analysed via linear fits in the log-log plot: the blue-dashed curve is fitted using all data points of $\varepsilon^{\rm (TLS)}$, while
    the purple-dashed (red-dashed) curve is obtained from the first (last) four data points of the $\varepsilon^{\rm (TLS)}$ bound, highlighted by the shaded regions. Uncertainties of fitted slopes and intercepts are given in parentheses on the last digit.}
    \label{fig:error_scaling_TLS}
\end{figure}

For an initial state $\rho_S(0)$, we write the final state at time $\tau$ in the TLS model as $\rho_S(\tau)\equiv \nE^{\mathcal{L}^{\rm (TLS)}\tau}(\rho_S(0))$, and the final state after $n\equiv\tau/\Delta t$ steps in the ECTCM with finite collision time $\Delta t$ as $\rho_{S,n}^{\rm (CM)}\equiv \mathcal{E}_n\left(\rho_S(0)\right)$, where $\mathcal{E}_n$ is the collision map [Eq.~(\ref{eq:collision_map_E_n})]. The simulation error $\varepsilon^{\rm (TLS)} \equiv \|\rho_S(\tau) - \rho_{S,n}^{\rm (CM)}\|_1$ is upper bounded by the $1\rightarrow 1$ norm $\|\nE^{\mathcal{L}^{\rm (TLS)}\tau} - \mathcal{E}_n\|_{1\rightarrow 1}$. It is therefore bounded by the error estimate established for Protocol~\ref{tab:protocol} [Eq.~(\ref{eq:error_bound})].

In Fig.~\ref{fig:error_scaling_TLS}, we plot the scalings of $\varepsilon^{\rm (TLS)}$ and its bound [Eq.~(\ref{eq:error_bound})] as functions of the collision time $\Delta t$, for a fixed initial state $\rho_S(0)$. The decay of $\varepsilon^{\rm (TLS)}$ with decreasing $\Delta t$ demonstrates the successful ECTCM simulation of $\mathcal{L}^{\rm (TLS)}$. Moreover, while the upper bound of~$\varepsilon^{\rm (TLS)}$ [Eq.~(\ref{eq:error_bound})] transitions from $O(\Delta t)$ to $O(\sqrt{\Delta t})$ in the small $\Delta t$ regime, we find that the actual error~$\varepsilon^{\rm (TLS)}$ retains the linear scaling, regardless of $\Delta t$. This is because the error terms contributing to the $O(\sqrt{\Delta t})$ behaviour in the ECTCM vanish in the particular case of the TLS model.
As discussed in Ref.~\cite{cleve2017efficient}, the linear error scaling, or equivalently the gate cost $n = \mathcal{O}(\tau^2/\varepsilon)$, is optimal for collision models. Our ECTCM construction for the TLS model is therefore optimal in this sense. In fact, in App.~\ref{app:scaling_for_d=2}, we further prove that our construction achieves the optimal error scaling for all thermal qubit dynamics  with non-trivial system Hamiltonian. 

By studying the TLS model, we demonstrate the constructive nature of Thm.~\ref{thm:equivalences}: even when a thermodynamic master equation does not appear to originate from an energy-conserving interaction, like it is the case here, our protocol yields an explicit energy-conserving digital simulation through ECTCMs, as long as the corresponding Lindbladian $\mathcal{L}_t$, i.e., satisfies Axioms~\ref{axiom:I}--\ref{axiom:III}.

Having established the equivalence for thermalisation with a \textit{single} heat bath, we now turn to a quantum system interacting with multiple baths, namely a quantum thermal engine, to explore the implications of Thm.~\ref{thm:equivalences} in this general setting.

\subsection{Three-level autonomous thermal machine}\label{sec:three_level_ATM}
One of the prototypical models of quantum thermal engines is provided by the Scovil--Schulz-DuBois three-level maser~\cite{Scovil59threelevel}. In this model, different transitions between energy levels are mediated by thermal reservoirs at different temperatures, allowing the system to operate as a minimal thermal machine in regimes such as refrigeration and heat pumping~\cite{linden2010small,kosloff2014quantum,correa2014quantum,paule2018thermodynamics}. Despite its simple structure, the model has served as a paradigmatic platform for studying general principles of quantum thermodynamics, including optimal efficiency, quantum effects in energy transport, and thermodynamic uncertainty relations~\cite{boukobza2007three,correa2014quantum,kalaee2021violating}. 

In this subsection, we present an explicit ECTCM construction of such a thermal machine to demonstrate how Protocol~\ref{tab:protocol} can be naturally extended to model interactions with multiple thermal baths, even when the resulting Lindbladian is not itself thermal [i.e., it does not satisfy Axioms~\ref{axiom:I}–\ref{axiom:III}]. In Sec.~\ref{sec:ECTCM_simulation_ATM}, we first show how the dynamics generated by a three-level autonomous thermal machine can be reproduced by a sequence of ECTCMs associated with the individual baths. In Sec.~\ref{sec:ECTCM_thermal_engine}, we adopt a complementary perspective and regard the resulting ECTCM construction as a thermal machine in its own right. This interpretation yields a finite-stroke engine whose thermodynamic behaviour can be analysed independently of the continuous-time limit, enabling direct comparisons with the autonomous machine stemming from the corresponding master equation. In Sec.~\ref{sec:heat_efficiency_ATM}, we show that the ECTCM reproduces not only the state dynamics and heat currents of the original machine for the case $\Delta t \to 0$, but also the key thermodynamic quantity -- the efficiency -- already in the regime of finite stroke length, where the simulation of the dynamics itself becomes inaccurate.

\subsubsection{ECTCM Simulation of an autonomous thermal machine}\label{sec:ECTCM_simulation_ATM}
We consider a three-level autonomous thermal machine~(ATM), where transitions between each pair of energy levels are mediated by a thermal bath~\cite{correa2014quantum,paule2018thermodynamics} (See Fig.~\ref{fig:three_level_ATM}).
Denote the three energy levels as $\{\ket{g}_S, \ket{e_a}_S, \ket{e_b}_S\}$. 
The system Hamiltonian is given by
\begin{align}
    H_S = \epsilon_1\ketbra{e_a}{e_a}_S + (\epsilon_1+\epsilon_2)\ketbra{e_b}{e_b}_S.
\end{align}
Let $\epsilon_3\equiv \epsilon_1+\epsilon_2$. When the three thermal baths, with inverse temperatures $\beta_1$, $\beta_2$ and $\beta_3$, are all bosonic and weakly coupled to the system, the system dynamics can be modelled by a GKSL master equation with Lindbladian $\mathcal{L}^{\rm (ATM)}$ as~\cite{paule2018thermodynamics}
\begin{align}
    \mathcal{L}^{\rm (ATM)}(\cdot) = -\nI[H_S, \cdot] + \Gamma\sum_{r=1}^{3}\mathcal{D}_r(\cdot),
\end{align}
where the Lamb shift has been neglected, the damping coefficient $\Gamma$ is assumed to be the same for all baths, and 
\begin{gather}
    \mathcal{D}_r(\cdot):= (n_{{\rm th}, r} + 1)\mathcal{D}_{\sigma^-_r}(\cdot) + n_{{\rm th}, r}\mathcal{D}_{\sigma^+_r}(\cdot),
\end{gather}
with  $\mathcal{D}_X(\cdot):= X(\cdot)X^\dagger - \tfrac{1}{2}[X^\dagger X, \cdot]_+$, $n_{{\rm th},r} \equiv (\nE^{\beta_r \epsilon_r} - 1)^{-1}$ and the ladder operators $\sigma_r^\pm$ (satisfying $\sigma^+_r = (\sigma^-_r)^\dagger$) defined as $\sigma_1^- \equiv \ketbra{g}{e_a}_S$, $\sigma^-_2 \equiv \ketbra{e_a}{e_b}_S$ and $\sigma^-_3 \equiv \ketbra{g}{e_b}_S$.

We denote the steady state by $\pi_{S}$, defined through $\mathcal{L}^{\rm (ATM)}(\pi_S) = 0$. Note that $\pi_S\neq \gamma_S$, i.e., the steady state is not thermal. The steady-state heat flux entering from the $r$th bath is defined as $\dot{Q}_r^{\rm (ATM)} \equiv {\rm Tr}\{H_S\mathcal{D}_r(\pi_{S})\}$.
The system operates as a refrigerator when $\beta_1 \ge \beta_3 \ge \beta_2$ and if the following condition is satisfied~\cite{correa2014quantum, paule2018thermodynamics}:
\begin{align}
    \epsilon_2 \ge \left(\frac{\beta_1 - \beta_3}{\beta_3 - \beta_2}\right)\epsilon_1.
\end{align}
In this regime, the system reaches the steady state, i.e., $\rho_S = \pi_S$, where it continuously extracts heat from the first (cold) bath ($\dot{Q}_1^{\rm (ATM)} \ge 0)$ and releases it into the third bath at intermediate temperature ($\dot{Q}_3^{\rm (ATM)} \le 0$), with the required energy supplied by the second (hot) bath ($\dot{Q}_2^{\rm (ATM)} \ge 0$), as shown in Fig.~\ref{fig:three_level_ATM}. The efficiency of this refrigerator $\eta^{\rm (ATM)}$ has been shown to be~\cite{correa2014quantum, paule2018thermodynamics}:
\begin{align}
    \eta^{\rm (ATM)} := \frac{\dot{Q}_1^{\rm (ATM)}}{\dot{Q}_2^{\rm (ATM)}} = \frac{\epsilon_1}{\epsilon_2} \le \frac{\beta_3 - \beta_2}{\beta_1 - \beta_3} =: \eta_{\rm Carnot}^{\rm (ATM)}.
    \label{eq:eta^ATM}
\end{align}
\begin{figure}
    \centering
    \includegraphics[width=\linewidth]{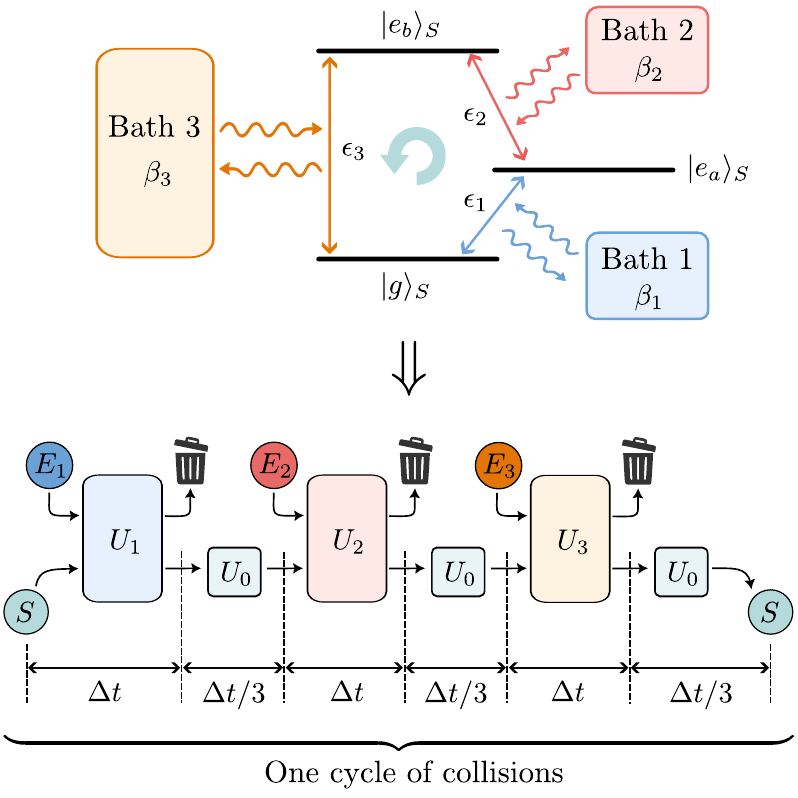}
    \caption{Three-level autonomous thermal machine and its ECTCM realisation. Top: A three-level system with energy eigenstates $\ket{g}_S$, $\ket{e_a}_S$, and $\ket{e_b}_S$, where the three transitions with energy gaps $\epsilon_1$, $\epsilon_2$, and $\epsilon_3$ are each coupled to independent thermal baths at inverse temperatures $\beta_1$, $\beta_2$, and $\beta_3$, respectively, with $\beta_1 \geq \beta_3 \geq \beta_2$. In this regime, the machine can operate as a refrigerator, absorbing heat from both the cold and hot baths and dumping it into the bath at intermediate temperature. Bottom: An energy-conserving thermal collision model (ECTCM) realisation of this thermal machine, in which the system sequentially interacts with three thermal ancillae prepared at the corresponding temperatures. Each collision consists of a unitary interaction with an ancilla lasting for $\Delta t$, followed by the system’s free unitary evolution for $\Delta t/3$. In the continuous-time limit, one collision cycle reproduces the infinitesimal-time-step dynamics of the autonomous thermal machine.}
    \label{fig:three_level_ATM}
\end{figure}
To implement this autonomous thermal machine via collision models, and in particular, ECTCMs, we notice that while $\mathcal{L}^{\rm (ATM)}$ satisfies Axioms~\ref{axiom:I} and \ref{axiom:II}, it does {\it not} satisfy Axiom~\ref{axiom:III}, quantum detailed balance with respect to a single temperature. Therefore, Protocol~\ref{tab:protocol} cannot be directly adopted to construct an ECTCM realisation of this dynamics. However, since each dissipator $\mathcal{D}_r$ satisfies $[\mathcal{D}_r, \mathcal{L}_{H_S}] = 0$ as well as detailed balance with respect to the corresponding bath inverse temperature $\beta_r$, we can regard them individually as thermal Lindbladians (without the unitary part) and use Protocol~\ref{tab:protocol} to construct an ECTCM to implement each of them. 
We thus construct three ECTCMs, one for each thermal bath, and combine them into a three-stroke cycle in which the system interacts sequentially with the three corresponding ancilla streams. 

To show that this sequential collision protocol indeed reproduces the dynamics generated by the full dissipator~$\Gamma\sum_{r=1}^3\mathcal{D}_r$, we invoke the Trotter--Suzuki formula~\cite{trotter1959product,suzuki1976generalized},
\begin{align}
    \nE^{\tau\Gamma\sum_{r=1}^3 \mathcal{D}_r} = \left(\overleftarrow{\prod_{r=1}^3}\nE^{\frac{\tau}{n} \Gamma\mathcal{D}_r}\right)^n + O(\tau^2 / n),
\end{align}
which shows that, in the limit $n \rightarrow \infty$, the collective dissipative dynamics can be approximated arbitrarily well by a sequential application of the individual dissipative dynamics.
Since each map $\nE^{\frac{\tau}{n} \Gamma\mathcal{D}_r}$ can be realised by its corresponding ECTCM, the entire dissipative dynamics admits a three-stroke ECTCM implementation.

We hence construct the collision model illustrated in Fig.~\ref{fig:three_level_ATM} to simulate this ATM, where each original full collision is decomposed into three independent collisions with different corresponding ancillae. Each of these constituent collisions is followed by the system's free unitary evolution to capture the missing unitary part in the three ECTCM constructions (see bottom of Fig.~\ref{fig:three_level_ATM}). Formally, denoting the individual collision time by $\Delta t$, the $i$th full collision map $\hat{\mathcal{E}}_i$ is given by
\begin{align}
    \hat{\mathcal{E}}_i &:= \mathcal{U}_0\circ\hat{\mathcal{E}}_{i,3}\circ\mathcal{U}_0\circ\hat{\mathcal{E}}_{i,2}\circ\mathcal{U}_0\circ\hat{\mathcal{E}}_{i,1}
    \label{eq:hat_E^(i)}
\end{align}
where $\mathcal{U}_0(\cdot):= U_0(\cdot)U_0^\dagger$ with $U_0\equiv \nE^{-\nI H_S \Delta t/3}$ and
\begin{align}
    \hat{\mathcal{E}}_{i,r}(\cdot):= {\rm Tr}_{E_{i,r}}\{U_{i,r}(\cdot\otimes\gamma_{E_{i,r}})U_{i,r}^\dagger\}
\end{align}
with $\gamma_{E_{i,r}} \equiv \nE^{-\beta_r H_{E_{i,r}}}/{\rm Tr}\{\nE^{-\beta_r H_{E_{i,r}}}\}$ and $U_{i,r} \equiv \nE^{-\nI gV_{i,r} \Delta t}$. Here, $g\equiv 1/\sqrt{\Delta t}$, and
$H_{E_{i,r}}$ and $V_{i,r}$ are constructed using Protocol~\ref{tab:protocol} by treating $\Gamma \mathcal{D}_r$ as a thermal Lindbladian. We remark that since $[V_{i,r}, H_S + H_{E_{i,r}}]=0$, and hence $[\hat{\mathcal{E}}_{i,r}, \mathcal{U}_0]=0$ for $r=1,2,3$, the timing of the system’s free evolution between collisions does not affect the overall dynamics.

After $n$ full collisions, the resulting map becomes
\begin{align}
    \mathcal{E}_n := \hat{\mathcal{E}}_n\circ \hat{\mathcal{E}}_{n-1}\circ\cdots\circ \hat{\mathcal{E}}_1.
\end{align}
For a total evolution time $\tau$, we set $n\equiv \tau/\Delta t$. 
By construction, $\mathcal{E}_n$ simulates the Markovian map $\nE^{\mathcal{L}^{\rm (ATM)} \tau}$ in the continuous-time limit ($\Delta t\rightarrow 0$). 

Since a single full collision $\hat{\mathcal{E}}_i$ requires a duration of $4\Delta t$, the collision model requires a total time $4\tau$ to reproduce dynamics over a duration $\tau$. This overhead arises naturally from recasting a continuous-time thermal machine into a sequence of finite-time strokes. The simulation time can be reduced by increasing the magnitudes of $H_S$ and $V_{i,r}$ in $U_0$ and $U_{i,r}$, respectively. 
\begin{figure*}[t]
    \centering
    \includegraphics[width=\textwidth]{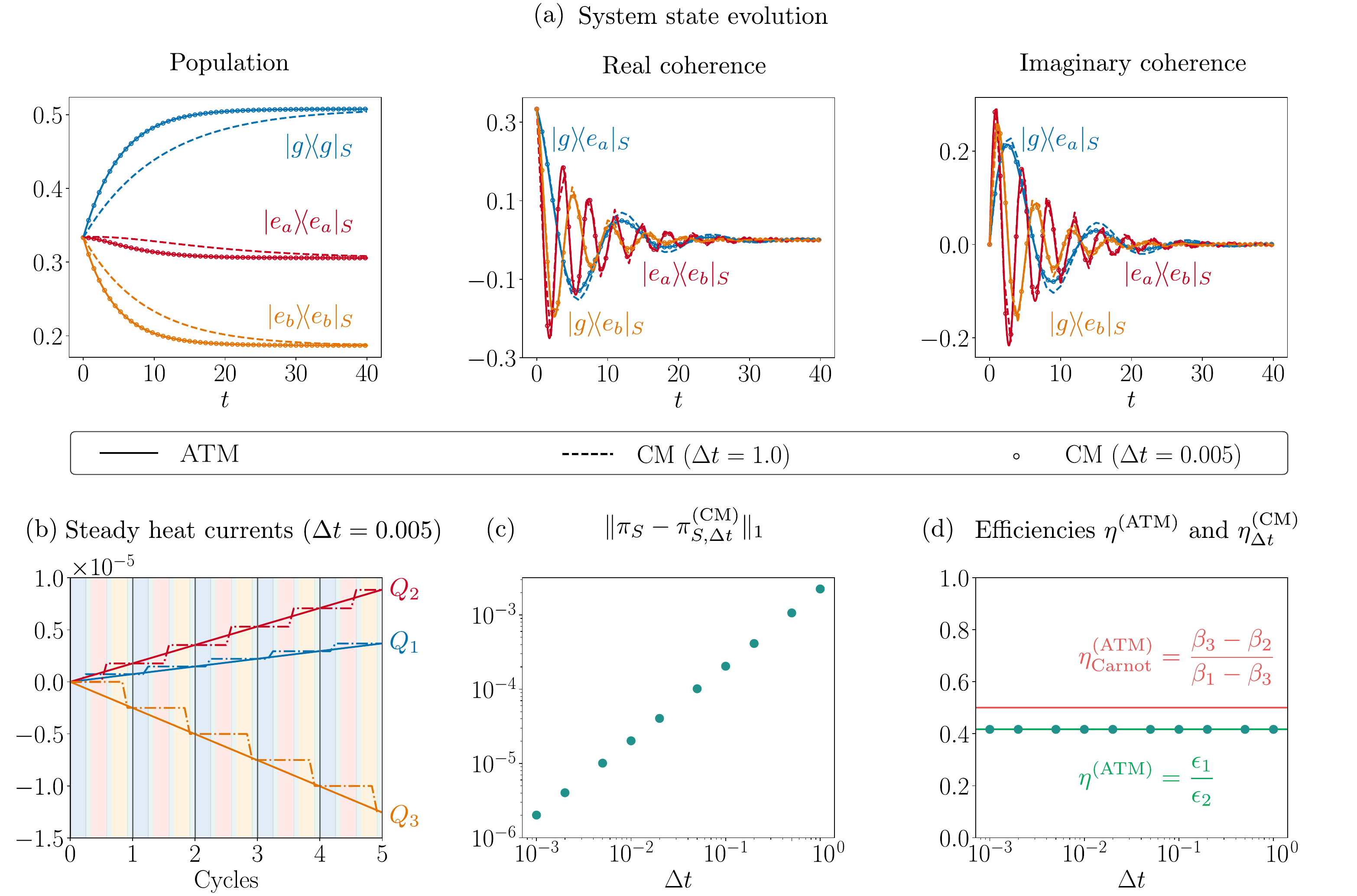}
    \caption{Numerical comparison between the three-level autonomous thermal machine (ATM) and its ECTCM realisations (CM).  Parameters: $\epsilon_1 = 0.5$, $\epsilon_2 = 1.2$, $\beta_1 = 1$, $\beta_2 = 0.4$, $\beta_3 = 0.6$ and $\Gamma=0.05$. (a) Population and coherence dynamics of the ATM (solid lines) and its ECTCM realisations with $\Delta t = 1.0$ (dashed lines) and $\Delta t = 0.005$ (circles). The initial state is $\rho_{S}(0)=\ketbra{\psi}{\psi}_S$ where $\ket{\psi}_S=(\ket{g}_S+\ket{e_a}_S+\ket{e_b}_S)/\sqrt{3}$. Similar convergence for sufficiently small $\Delta t$ is observed for other initial states. The time scale of the ECTCM is four times longer than that of the ATM, as discussed in the main text; for comparison, the time axes have been rescaled accordingly. In addition, to improve visibility, the ECTCM dynamics with $\Delta t = 0.005$ are shown only every 200 steps. (b)  Steady-state heat currents in the ATM (solid lines) and its ECTCM realisation with $\Delta t = 0.005$ (dash-dotted lines). The heat current per cycle in the ATM is given by ${Q}_r^{\rm (ATM)} = \dot{Q}_r^{\rm (ATM)} \Delta t$ for $r = 1,2,3$. The three strokes within each cycle of the ECTCM realisation are highlighted using colors consistent with Fig.~\ref{fig:three_level_ATM}. The continuously plotted dash-dotted lines represent discrete, stroke-resolved quantities that are  updated once per stroke of duration $\Delta t$. The collision model does not resolve the precise heat-transfer time within each stroke. (c) Difference between the steady states of the ATM and its ECTCM realisations for different values of $\Delta t$, quantified by the trace norm. (d) Efficiencies of the ATM (green line), its ECTCM realisations for different values of $\Delta t$ (green dots), and the Carnot efficiency (red line) (Cf. Prop.~\ref{prop::eff_ECTCM_ATM}).} 
    \label{fig:ATM_plot}
\end{figure*}

In Fig.~\ref{fig:ATM_plot}(a), we validate the analytical ECTCM realisation by showing that $\mathcal{E}_n$ reproduces the same dynamics as $\nE^{\mathcal{L}^{\rm (ATM)} \tau}$ for sufficiently small $\Delta t$. This establishes the ECTCM as a faithful microscopic realisation of the autonomous thermal machine. However, the constructed ECTCM possesses an additional interpretation: for arbitrary collision times $\Delta t$, it constitutes a finite-stroke thermal engine in its own right. In the remainder of this subsection, we investigate its thermodynamic performance and show that, surprisingly, its efficiency exactly reproduces that of the original autonomous machine for arbitrary stroke length $\Delta t$, despite the simulation of the dynamics losing faithfulness with increasing~$\Delta t$.

\subsubsection{The resulting ECTCM as a finite-stroke thermal machine}\label{sec:ECTCM_thermal_engine}
A finite-stroke thermal machine operates by sequentially coupling the working system to different thermal reservoirs for finite durations, rather than interacting with them continuously. Each cycle therefore consists of a discrete sequence of thermodynamic strokes~\cite{alicki_introduction_2018}.
To see that the ECTCM is such a thermal machine in its own right, we consider its long-time behaviour.
In the long time limit ($\tau \rightarrow \infty$), $\mathcal{E}_n$ transforms any input state into the $\Delta t$-dependent steady state $\pi_{S,\Delta t}^{\rm (CM)}$, of the single cycle of collisions, $\hat{\mathcal{E}}_i$, satisfying $\hat{\mathcal{E}}_i(\pi_{S,\Delta t}^{\rm (CM)}) = \pi_{S,\Delta t}^{\rm (CM)}$. Here, we exclude the fine-tuned values of $\Delta t$ for which the steady state $\pi_{S,\Delta t}^{\rm (CM)}$ does not exist. At this steady state, $\hat{\mathcal{E}}_i$ can be regarded as a thermal cycle consisting of three strokes associated with the three baths, respectively. We first show that the resulting engine satisfies the first and second laws of thermodynamics by charaterising its heat currents. We then analyse its efficiency and compare them with that of the original ATM.

First, the ECTCM thermal machine provides a transparent framework for tracking energy exchanges between the system and the baths. By virtue of Condition~\ref{CM:4}, in the absence of external driving of the system (i.e., $D_{i,r}=0$), the system–bath interaction neither supplies nor consumes energy during a collision. In particular, the interaction energy vanishes on average throughout each collision, i.e., 
\begin{align}
    \langle V_{i,r} \rangle_{\rho_{SE_{i,r}}(s)} \equiv {\rm Tr}\{V_{i,r}\rho_{SE_{i,r}}(s)\} = 0
    \label{eq:V_i_r_=_0}
\end{align}
$\forall\, \rho_S(0),\, s\in[0,\Delta t]$, where $\rho_S(0)$ is the initial state at the beginning of the $(i,r)$th collision and $\rho_{SE_{i,r}}(s)\equiv U_i(s)(\rho_S(0)\otimes\rho_{E_i})U_i^\dagger(s)$.
As a result, any change in the system’s energy can be attributed solely to energy exchanged with the baths.
This allows us to faithfully define the heat currents at the steady state $\pi_{S,\Delta t}^{\rm (CM)}$ solely in terms of the system Hamiltonian $H_S$ (see below). 

In a single cycle described by $\hat{\mathcal{E}}_i$ [Eq.~(\ref{eq:hat_E^(i)})],
let $\tilde{\pi}_{S,\Delta t}^{\rm (CM)}\equiv \mathcal{U}_0\circ\hat{\mathcal{E}}_{i,1}(\pi_{S,\Delta t}^{\rm (CM)})$ and $\vardbtilde{\pi}_{S,\Delta t}^{\rm (CM)}\equiv \mathcal{U}_0\circ\hat{\mathcal{E}}_{i,2}(\tilde{\pi}_{S,\Delta t}^{\rm (CM)})$ the states of the system after the first and the second stroke, respectively. 
The heat currents from each bath during a collision can be defined as
\begin{align}
    Q_{1;\Delta t}^{\rm (CM)} &\equiv {\rm Tr}\{H_S [\tilde{\pi}_{S,\Delta t}^{\rm (CM)} - \pi_{S,\Delta t}^{\rm (CM)}]\}, \label{eq:Q_1_ATM:CM}\\
    Q_{2;\Delta t}^{\rm (CM)} &\equiv {\rm Tr}\{H_S [\vardbtilde{\pi}_{S,\Delta t}^{\rm (CM)} - \tilde{\pi}_{S,\Delta t}^{\rm (CM)}]\}, \label{eq:Q_2_ATM:CM}\\
    Q_{3;\Delta t}^{\rm (CM)} &\equiv {\rm Tr}\{H_S [\pi_{S,\Delta t}^{\rm (CM)} - \vardbtilde{\pi}_{S,\Delta t}^{\rm (CM)}]\}, \label{eq:Q_3_ATM:CM}
\end{align}
where the last definition follows from the fact that $\pi_{S,\Delta t}^{\rm (CM)}$ is the fixed point of $\hat{\mathcal{E}}_i$ such that $\pi_{S,\Delta t}^{\rm (CM)} = \mathcal{U}_0\circ\hat{\mathcal{E}}_{i,3}(\vardbtilde{\pi}_{S,\Delta t}^{\rm (CM)})$. 
It follows immediately that the heat currents satisfy the {\it first law} of thermodynamics,
\begin{align}
    Q_{1;\Delta t}^{\rm (CM)} + Q_{2;\Delta t}^{\rm (CM)} + Q_{3;\Delta t}^{\rm (CM)} = 0.
\end{align}
Moreover, the free energy of a quantum state $\rho$ relative to the ambient inverse temperature $\beta$ is defined as~\cite{Vinjanampathy2016QTD}: $F_\beta(\rho) = {\rm Tr}\{H\rho\} - \beta^{-1}S(\rho)$ with $S(\rho):=-{\rm Tr}\{\rho\ln\rho\}$. Since in this ECTCM thermal machine, each collision $\hat{\mathcal{E}}_r$ is a TO which cannot increase the free energy of a state~\cite{lostaglio2019introductory}, in a single cycle, we have
\begin{align}
    &\quad \beta_1\left[F_{\beta_1}(\tilde{\pi}_{S,\Delta t}^{\rm (CM)}) - F_{\beta_1}(\pi_{S,\Delta t}^{\rm (CM)})\right] \nonumber\\
    &+ \beta_2\left[F_{\beta_2}(\vardbtilde{\pi}_{S,\Delta t}^{\rm (CM)}) - F_{\beta_2}(\tilde{\pi}_{S,\Delta t}^{\rm (CM)})\right]\nonumber\\
    &+ \beta_3\left[F_{\beta_3}(\pi_{S,\Delta t}^{\rm (CM)}) - F_{\beta_3}(\vardbtilde{\pi}_{S,\Delta t}^{\rm (CM)})\right] \le 0,
\end{align}
which leads to
\begin{align}
    \beta_1 Q_{1;\Delta t}^{\rm (CM)} + \beta_2 Q_{2;\Delta t}^{\rm (CM)} + \beta_3 Q_{3;\Delta t}^{\rm (CM)} \le 0.
\end{align}
This is the Clausius expression of the {\it second law} of thermodynamics. We hereby establish that the ECTCM thermal machine satisfies both the first and second laws of thermodynamics, and is thus thermodynamically consistent.

\subsubsection{Heat currents and efficiency of the ECTCM thermal machine}\label{sec:heat_efficiency_ATM}
With the expressions for the heat currents for the ECTCM thermal machine in hand, we can now benchmark its performance against that of the original, continuous ATM. To compare them, in Fig.~\ref{fig:ATM_plot}(b), we show that for sufficiently small $\Delta t$, the heat current per cycle $Q_{r;\Delta t}^{\rm (CM)}$ coincides with the heat currents of the original ATM calculated as $Q_r^{\rm (ATM)}\equiv\dot{Q}_r^{\rm (ATM)}\Delta t$, for $r=1,2,3$. This result demonstrates that the ECTCM realisation is not only dynamically equivalent to the original ATM but also faithfully reproduces its thermodynamic quantities as a physical model.

We finally compare the efficiencies of the ECTCM thermal machine and the original ATM. We first show that no work is required to operate the ECTCM thermal machine and then define its efficiency in terms of the heat currents exchanged with the baths, in direct analogy with the efficiency of the original ATM [Eq.~(\ref{eq:eta^ATM})].

Since enabling the system--ancilla interaction in the ECTCM incurs no energetic cost [see Eq.~(\ref{eq:V_i_r_=_0})], the only step that could, in principle, require work is the switching between different thermal ancillae and their associated interaction Hamiltonians in successive collisions. 
To see that this switching cost also vanishes,
consider a single cycle described by $\hat{\mathcal{E}}$ [Eq.~(\ref{eq:hat_E^(i)})], where we omit the cycle label $i$ for simplicity. Let $\rho_{SE_1}$ be the joint state of $S$ and ancilla $E_1$ right after their collision $\hat{\mathcal{E}}_1$ and define $\rho_{S, 1} \equiv {\rm Tr}_{E_1}\{\rho_{SE_1}\}$. The energy cost incurred by switching to collide with the second ancilla $E_2$ (i.e., switching from the interaction term $V_1$ to $V_2$ and bringing in a fresh ancilla $\gamma_{E_2}$) is then given by~\cite{ciccarello2022quantum}:
\begin{align}
    W_{1\rightarrow 2} &:= {\rm Tr}\{V_2(\rho_{S, 1}\otimes \gamma_{E_2})\} - {\rm Tr}\{V_1\rho_{SE_1}\}.
\end{align}
Following Eq.~(\ref{eq:V_i_r_=_0}), both terms in $W_{1\rightarrow 2}$ vanish. Thus, $W_{1\rightarrow 2}=0$ and similarly $W_{2\rightarrow 3} = W_{3\rightarrow 1} = 0$. 

Having shown that the ECTCM thermal machine operates without any work input, we define its efficiency solely in terms of the heat currents:
\begin{align}
    \eta^{\rm (CM)}_{\Delta t} := \frac{Q_{1;\Delta t}^{\rm (CM)}}{Q_{2;\Delta t}^{\rm (CM)}}.
\end{align}
Interestingly, we find that although, as shown in Fig.~\ref{fig:ATM_plot}(c), increasing the collision times $\Delta t$ lead to a growing deviation between the respective steady states $\pi_{S,\Delta t}^{\rm (CM)}$ and $\pi_S$, the efficiencies of the two refrigerators remain identical, as illustrated in Fig.~\ref{fig:ATM_plot}(d). In fact, the constant efficiency is a general feature of the three-stroke thermal machine we propose:
\begin{proposition}
\label{prop::eff_ECTCM_ATM}
    The efficiency of the ECTCM thermal machine in Fig.~(\ref{fig:three_level_ATM}) is a constant independent of the collision time $\Delta t$, and in particular, $\eta^{\rm (CM)}_{\Delta t} = \eta^{\rm (ATM)} = \epsilon_1/\epsilon_2,\,\forall\,\Delta t$. 
\end{proposition}
\begin{proof}
In each stroke, the operation applied on the system, $\mathcal{U}_0\circ\hat{\mathcal{E}}_r$ for $r=1,2,3$, is a thermal operation with respect to the corresponding inverse temperature $\beta_r$. Since thermal operations do not mix populations and coherence in energy eigenbasis for nondegenerate systems~\cite{lostaglio2019introductory} and the heat currents defined in Eqs.~(\ref{eq:Q_1_ATM:CM}), (\ref{eq:Q_2_ATM:CM}) and (\ref{eq:Q_3_ATM:CM}) only depend on the populations of the involved states, we can focus on the population vector of $\pi_{S,\Delta t}^{\rm (CM)}$, denoted by $\vec{p}\equiv (p_{g}, p_{e_a}, p_{e_b})$, and similarly, $\vec{\tilde{p}}$ and $\vec{\vardbtilde{p}}$ for $\tilde{\pi}_{S,\Delta t}^{\rm (CM)}$ and $\vardbtilde{\pi}_{S,\Delta t}^{\rm (CM)}$, respectively. We then have
\begin{align}
    Q_{1;\Delta t}^{\rm (CM)} &= \epsilon_1(\tilde{p}_{e_a} - p_{e_a}),\\
    Q_{2;\Delta t}^{\rm (CM)} &= \epsilon_1(\vardbtilde{p}_{e_a} - \tilde{p}_{e_a}) + (\epsilon_1+\epsilon_2)(\vardbtilde{p}_{e_b} - \tilde{p}_{e_b}).
\end{align}
Since the population in $\ket{g}_S$ is unchanged in the second stroke, we have $\tilde{p}_{e_a} + \tilde{p}_{e_b} = \vardbtilde{p}_{e_a} + \vardbtilde{p}_{e_b}$, which leads to $Q_{2;\Delta t}^{\rm (CM)} = \epsilon_2(\tilde{p}_{e_a}-\vardbtilde{p}_{e_a})$. Because the third stroke does not affect $\vardbtilde{p}_{e_a}$ and the populations return to $\vec{p}$ afterwards, as required by the definition of steady state, we have $\vardbtilde{p}_{e_a} = p_{e_a}$. Therefore, $Q_{2;\Delta t}^{\rm (CM)} = \epsilon_2(\tilde{p}_{e_a}-p_{e_a})$ and the efficiency is given by $\eta^{\rm (CM)}_{\Delta t} = {Q_{1;\Delta t}^{\rm (CM)}}/{Q_{2;\Delta t}^{\rm (CM)}} = \epsilon_1/\epsilon_2$.
\end{proof}
This demonstrates that the ECTCM realisation preserves the thermodynamic performance of the original ATM even when its state dynamics deviate from the continuous-time evolution. In particular, because the engine efficiency is reproduced exactly for arbitrary collision times $\Delta t$, one need not operate in the infinitesimal-cycle limit, where the heat exchanged per cycle becomes vanishingly small. The ECTCM construction therefore provides a physically meaningful finite-stroke counterpart of the continuous-time thermal machine, rather than merely viewing the latter as the limiting case of an increasingly rapid stroke engine.

\section{Summary and outlook}\label{sec:discussion}
In this work, we established the equivalence between three complementary formulations of Markovian quantum thermodynamics: the axiomatic description in terms of thermal Lindbladians (Prop.~\ref{prop:general_form_xi}), the microscopic description based on energy-conserving thermal collision models (ECTCMs), and the operational description in terms of Markovian thermal operations (MTOs). Specifically, we proved that every ECTCM generator is a thermal Lindbladian (Prop.~\ref{prop:L_CM_=>_L}), every thermal Lindbladian admits a microscopic realisation through an ECTCM (Prop.~\ref{prop:ME_CM_equivalence}), and the set of MTO generators coincides with the set of ECTCM generators (Prop.~\ref{prop:MTO_=_CM}). These results culminate in the \emph{Axiomatic--Microscopic--Operational Equivalence} theorem (Thm.~\ref{thm:equivalences}), establishing that the classes of generators arising from the three approaches coincide and therefore characterise exactly the same set of Markovian quantum thermodynamic processes.

This equivalence provides new insights into each formulation: 
\begin{itemize}
    \item[(i)] For the {\it axiomatic} paradigm, the equivalence theorem establishes that every thermal Lindbladian admits a microscopic realisation through an ECTCM. Consequently, the dynamics generated by thermal Lindbladians can be implemented using only thermal ancillae and energy-conserving interactions, without consuming thermodynamic resources.
    \item[(ii)] For the {\it microscopic} paradigm, the equivalence theorem reveals the precise dynamical structure generated by repeated energy-conserving interactions with thermal ancillae. In particular, the resulting dynamics are characterised exactly by the three defining axioms of thermal Lindbladians: Markovianity, time-translation covariance, and quantum detailed balance.
    \item[(iii)] For the {\it operational} paradigm, the equivalence theorem resolves the structure of MTO generators by showing that they coincide exactly with thermal Lindbladians, thereby settling the conjecture in Ref.~\cite{vomEnde2023ExploringLimits} on the complete characterisation of MTO generators. This yields a simple axiomatic characterisation of thermodynamically admissible Markovian dynamics and demonstrates that every MTO generator can be realised microscopically through an ECTCM.
\end{itemize}
Beyond establishing equivalence between the three formulations, our results endow Markovian quantum thermodynamics with a digital, thermodynamically consistent simulation paradigm via ECTCMs (Protocol~\ref{tab:protocol}). Given an arbitrary thermal Lindbladian, the protocol provides an explicit compilation into a sequence of collisions with thermal ancillae, thereby translating a continuous-time thermodynamic evolution into a discrete set of elementary thermodynamic gates. This perspective not only yields a constructive microscopic implementation of thermal dynamics but also opens the possibility of studying their simulation complexity, resource requirements, and experimental realisation using collision-model architectures. 

To illustrate the versatility of this result and its physical implications, we applied the protocol to two representative examples. First, we constructed an ECTCM simulation of a two-level system coupled to a bosonic environment, demonstrating how a paradigmatic open-system dynamics can be realised within our framework. Given that this model's original derivation starts from a Hamiltonian that is \textit{not} energy-preserving, this thermodynamically consistent simulation further highlights the non-uniqueness of implementations of Markovian thermal dynamics as well as their vastly different resource requirements. Second, we considered a three-level autonomous thermal machine coupled to multiple baths and showed that the resulting ECTCM reproduces not only the state dynamics but also the thermodynamic behaviour of the original model. In particular, the ECTCM construction naturally yields a finite-stroke thermal engine whose efficiency coincides exactly with that of the continuous-time autonomous machine for arbitrary collision times $\Delta t$. This example demonstrates that the ECTCM framework naturally extends to multi-bath settings which go beyond thermal master equations and provides a physically meaningful platform for studying finite-stroke quantum thermal machines.

The equivalence and simulation framework established in this work open up several directions for future research. 

First, our results are restricted to the Markovian regime. It remains an open question whether a similar correspondence exists between axiomatic, microscopic, and operational descriptions of non-Markovian thermodynamics. In the literature, master equations~\cite{Nazir2018,Nestmann2021How,Dann2022NonMarkovian}, collision models~\cite{ciccarello2013CMnonMarkovian,Lorenzo2017Composite,Campbell2018System} and thermal operations~\cite{Ptaszy2022NonMarkovian,czartowski_thermal_2023,Zambon2025Quantum} beyond the Markovian regime have thus far been studied independently and a unification within a single thermodynamic framework is still lacking. Our results serve as a Markovian baseline for the larger non-Markovian class of dynamics and provide a starting point for a similar consolidation of approaches to thermodynamics in the presence of memory effects.
  
Second, although the ECTCM protocol furnishes a universal microscopic realisation of thermal Lindbladians, its construction is not necessarily complexity-optimal. In particular, the requirement of strictly energy-conserving system–ancilla interactions may necessitate highly nonlocal couplings and large ancilla dimensions. Allowing controlled violations of energy conservation through the expenditure of work resources may substantially simplify the implementation. Conversely, while non-thermal Lindbladians cannot be simulated via ECTCMs, additional out-of-equilibrium resources might make the simulation possible~\cite{Cattaneo2021CMcansimulate,lacroix_making_2025}. Quantifying the required resources directly from the corresponding Lindbladian, however, is far from straightforward. The characterisation of resulting trade-offs between thermodynamic resource consumption, implementation complexity and expressivity remains an important open problem, both foundationally and experimentally.

Finally, collision-model thermal machines have been extensively studied as finite-stroke quantum heat engines~\cite{Molitor2020Stroboscopic,piccione2021power,melo2022implementation}. The close relationship between continuous-time and finite-stroke engines has also attracted attention~\cite{Uzdin2015Equivalence,lobejko2026equivalence}. Our ECTCM framework naturally bridges these two directions by providing a constructive mapping from continuous-time thermal machines to finite-stroke collision-model implementations. The three-level autonomous thermal machine presented in this work serves as an explicit example, where the resulting finite-stroke engine reproduces not only the continuous-time dynamics in the appropriate limit but also the efficiency of the original engine for arbitrary collision times. This thermodynamic equivalence opens up a systematic route for the ``digital", stroke-wise implementation of continuous thermal engines. Whether it holds in general, and what exact thermodynamic properties of continuous processes can be recovered in this way is an intriguing direction for future investigation.

\begin{acknowledgments} 
We thank S. Brattegard, T. Harris, M. Huber, M. Mitchison, S. V. Moreira, M. Moroder, A. de Oliveira Junior, P. Taranto, and J. Xuereb for valuable discussion and comments.
The research conducted in this publication was funded by Taighde Éireann -- Research Ireland under grant number IRCLA/2022/3922. YL is supported by China Scholarship Council (No.~202408060137).
\end{acknowledgments}

\newtheorem{theoremApp}{Theorem}[section]
\newtheorem{lemmaApp}{Lemma}[section]
\newtheorem{corollaryApp}{Corollary}[theoremApp]
\numberwithin{figure}{section}
\appendix
\allowdisplaybreaks

\section{Quantum detailed balance condition(s)}\label{app:QDB}
\subsection{Equivalence of quantum detailed balance of $\mathcal{D}_t$ and $\mathcal{A}_t$}
Given the inverse temperature $\beta$ and the system Hamiltonian $H_S$, let $\{\ket{j}_S\}_{j=0}^{d_S-1}$ be the orthonormal eigenbasis of $H_S$ with the corresponding eigenvalues $\{\epsilon_j\}_{j=0}^{d_S-1}$ and $d_S$ is the dimension of $H_S$. We say a Hermitian-preserving map $\mathcal{X}$ acting on the system $S$ satisfies the quantum detailed balance (QDB) condition if and only if the following holds:
\begin{align}
    &\quad \nE^{-\beta \epsilon_y}\bra{x'}_S\mathcal{X}(\ketbra{x}{y}_S)\ket{y'}_S \nonumber\\
    &= \nE^{-\beta \epsilon_{y'}}\bra{x}_S\mathcal{X}(\ketbra{x'}{y'}_S)\ket{y}_S^*,
    \label{appeq:QDB_def}
\end{align}
for all $(x,x',y,y')$. 
In this appendix, we prove the following:
\begin{lemma}
    Given a Lindbladian $\mathcal{L}$ satisfying Axioms~\ref{axiom:I} and \ref{axiom:II}, i.e., $\mathcal{L}(\cdot) = -\nI[H_S', \cdot] + \mathcal{A}(\cdot) - \frac{1}{2}[\mathcal{A}^\dagger(\mathds{1}), \cdot]_+$ with $[H_S', H_S]=0$ and $[\mathcal{A}, \mathcal{L}_{H_S'}]=0$ where $\mathcal{L}_{H_S'}(\cdot):=[H_S', \cdot]$,
    the following two statements are equivalent:
    \begin{enumerate}
        \item The dissipator $\mathcal{D}(\cdot):=\mathcal{A}(\cdot) - \frac{1}{2}[\mathcal{A}^\dagger(\mathds{1}), \cdot]_+$ satisfies the QDB condition [Eq.~(\ref{appeq:QDB_def})].
        \item $\mathcal{A}$ satisfies the QDB condition [Eq.~(\ref{appeq:QDB_def})].
    \end{enumerate}
\end{lemma}
\begin{proof}
    Taking the orthonormal basis $\{F_{S}^{jj'}\equiv\ketbra{j}{j'}_S\}_{j,j'=0}^{d_S-1}$ in the operator space, we consider the general expression of $\mathcal{A}$ [Eq.~(\ref{eq:A_gen})]:
    \begin{align}
        \mathcal{A}(\cdot) = \sum_{j,j',m,m'=0}^{d_S-1}\xi_{jj'mm'}F_S^{jj'}(\cdot) (F_S^{mm'})^\dagger.
    \end{align}
    Since $\mathcal{A}$ is CP, the coefficients $\xi_{jj'mm'}$, defined in Eq.~(\ref{eq:xi_def_from_A}), satisfy the following property [see Eq.~(\ref{eq:hermiticity_xi})]:
    \begin{align}
        \xi_{jj'mm'} = \xi_{mm'jj'}^*,\quad\forall\, (j,j',m,m').
        \label{appeq:hermiticity_xi}
    \end{align}
    By time-translation symmetry~\ref{axiom:II}, i.e., $[\mathcal{A}, \mathcal{L}_{H_S'}]=0$, we have [see Eq.~(\ref{eq:xi_energy_gaps_matching})]
    \begin{align}
        \xi_{jj'mm'} = \delta(\epsilon_{j'}-\epsilon_{m'} +\epsilon_m - \epsilon_j) \xi_{jj'mm'}.
        \label{appeq:xi_energy_gaps_matching}
    \end{align}
    The action of $\mathcal{D}$ is given by [Eq.~(\ref{eq:D_givenby_xi})]
    \begin{align}
        &\quad \bra{j}_S\mathcal{D}(\ketbra{\ell}{n}_S)\ket{m}_S \nonumber\\
        &= \xi_{j\ell m n} - \frac{1}{2}\sum_{h=0}^{d_S-1}\left(\xi_{h\ell h j}\delta_{mn} + \xi_{h m h n}\delta_{j\ell}\right)\\
        &= \xi_{j\ell m n} \nonumber\\
        &\quad - \frac{1}{2}\sum_{h=0}^{d_S-1}\left[\xi_{h\ell h j}\delta_{mn}\delta(\epsilon_{j}-\epsilon_{\ell}) + \xi_{h m h n}\delta_{j\ell}\delta(\epsilon_{m}-\epsilon_{n})\right],
        \label{appeq:D_general_action}
    \end{align}
    where the third line follows from Eq.~(\ref{appeq:xi_energy_gaps_matching}). Since $\mathcal{D}$ satisfies Eq.~(\ref{appeq:QDB_def}), we have that for all $(j,\ell,m,n)$,
    \begin{widetext}
        \begin{align}
        \nE^{-\beta(\epsilon_n - \epsilon_m)}\bra{j}_S\mathcal{D}(\ketbra{\ell}{n}_S)\ket{m}_S 
        &= \nE^{-\beta(\epsilon_n - \epsilon_m)}\xi_{j\ell m n} - \frac{1}{2}\sum_{h=0}^{d_S-1}\nE^{-\beta(\epsilon_n - \epsilon_m)}\left[\xi_{h\ell h j}\delta_{mn}\delta(\epsilon_{j}-\epsilon_{\ell}) + \xi_{h m h n}\delta_{j\ell}\delta(\epsilon_{m}-\epsilon_{n})\right]\\
        &= \nE^{-\beta(\epsilon_n - \epsilon_m)}\xi_{j\ell m n} - \frac{1}{2}\sum_{h=0}^{d_S-1}\left[\xi_{h\ell h j}\delta_{mn}\delta(\epsilon_{j}-\epsilon_{\ell}) + \xi_{h m h n}\delta_{j\ell}\delta(\epsilon_{m}-\epsilon_{n})\right] \label{appeq:QDB_D_mid_step_1}\\
        &\overset{\text{Eq.~(\ref{appeq:QDB_def})}}{=} \bra{\ell}_S\mathcal{D}(\ketbra{j}{m}_S)\ket{n}_S^*\\
        &= \xi_{\ell j n m}^* - \frac{1}{2}\sum_{h=0}^{d_S-1}\left[\xi_{hj h \ell}^*\delta_{mn}\delta(\epsilon_{j}-\epsilon_{\ell}) + \xi_{h n h m}^*\delta_{j\ell}\delta(\epsilon_{m}-\epsilon_{n})\right]\\
        &= \xi_{\ell j n m}^* - \frac{1}{2}\sum_{h=0}^{d_S-1}\left[\xi_{h\ell h j}\delta_{mn}\delta(\epsilon_{j}-\epsilon_{\ell}) + \xi_{h m h n}\delta_{j\ell}\delta(\epsilon_{m}-\epsilon_{n})\right],
        \label{appeq:QDB_D_mid_step_2}
        \end{align}
    \end{widetext}
    where in the last line we used Eq.~(\ref{appeq:hermiticity_xi}). By comparing Eqs.~(\ref{appeq:QDB_D_mid_step_1}) and (\ref{appeq:QDB_D_mid_step_2}), we obtain that, for a Lindbladian $\mathcal{L}$ satisfying Axioms~\ref{axiom:I} and \ref{axiom:II}, we have
    \begin{align*}
        \text{$\mathcal{D}$ satisfies Eq.~(\ref{appeq:QDB_def})}
        \Leftrightarrow \nE^{-\beta \epsilon_n}\xi_{j\ell m n} = \nE^{-\beta \epsilon_m}\xi^*_{\ell j n m},
        \label{appeq:QDB_xi}
    \end{align*}
    $\forall\, (j,\ell,m,n)$, which means that $\mathcal{A}$ satisfies Eq.~(\ref{appeq:QDB_def}).
\end{proof}
\subsection{Equivalence of different definitions of quantum detailed balance}
Given a parameter $s\in[0,1]$, we define the inner product of two system operators $X$ and $Y$ as~\cite{AMORIM2021389}
\begin{align}
    \langle X, Y \rangle_{\sigma,s} := {\rm Tr}\{X^\dagger \sigma^{s}Y\sigma^{1-s}\},
\end{align}
where $\sigma$ is a positive definite operator.
The QDB condition [Eq.~(\ref{appeq:QDB_def})] for a Hermitian-preserving map $\mathcal{X}$ acting on the system $S$ is equivalent to the following:
\begin{align}
    \langle A_S, \mathcal{X}^\dagger(B_S) \rangle_{\gamma_S, 0} = \langle \mathcal{X}^\dagger(A_S), B_S \rangle_{\gamma_S, 0},
    \label{appeq:QDB_def_innerproduct}
\end{align}
for all system operators $A_S$ and $B_S$. Therefore, $\mathcal{X}$ satisfies QDB $\Leftrightarrow$ $\mathcal{X}^\dagger$ is self-adjoint with respect to the inner product $\langle \cdot, \cdot \rangle_{\gamma_S, 0}$.\footnote{Note the difference between two adjoint notions appearing: Hilbert-Schmidt adjoint $\mathcal{X}^\dagger$ customarily denoted by ${}^\dagger$ and adjoint with respect to $\langle \cdot, \cdot \rangle_{\gamma_S, 0}$ inner product. As such, even though $\mathcal{X}\neq \mathcal{X}^\dagger$ in general, i.e. it is not Hilbert-Schmidt-self-adjoint, QDB of $\mathcal{X}$ requires self-adjointness of $\mathcal{X}^\dagger$ with respect to $\langle \cdot, \cdot \rangle_{\gamma_S, 0}$.} We can regard Eq.~(\ref{appeq:QDB_def_innerproduct}) as the definition of QDB. 

In the App.~D of Ref.~\cite{Alhambra2017Dynamicalmaps}, it has been shown that the QDB condition defined in terms of $\langle \cdot, \cdot \rangle_{\gamma_S, 0}$ is equivalent to the QDB condition defined via $\langle \cdot, \cdot \rangle_{\gamma_S, 1/2}$, for all Hermitian-preserving maps that satisfy time-translation symmetry, i.e., they commute with $\mathcal{L}_{H_S}$. We now generalise this equivalence to all inner products $\langle \cdot, \cdot \rangle_{\gamma_S, s}$ for $s\in[0,1]$.
\begin{lemma}
    Given a Hermitian-preserving map $\mathcal{X}$ that satisfies $[\mathcal{X}, \mathcal{L}_{H_S}]=0$, $\mathcal{X}$ satisfies the QDB condition [Eq.~(\ref{appeq:QDB_def})] if and only if $\mathcal{X}^\dagger$ is self-adjoint with respect to all inner product $\langle \cdot, \cdot \rangle_{\gamma_S, s}$ for $s\in[0,1]$.
\end{lemma}
\begin{proof}
    The `if' part is straightforward since by setting $s=0$, Eq.~(\ref{appeq:QDB_def_innerproduct}) holds, which is equivalent to Eq.~(\ref{appeq:QDB_def}). We only need to prove the `only if' part. 

    Due to the time-translation symmetry, i.e, $[\mathcal{X}, \mathcal{L}_{H_S}]=0$, $\mathcal{X}$ acts like a block matrix in the eigensubspace of $\mathcal{L}_{H_S}$. Therefore, $\bra{x'}_S\mathcal{X}(\ketbra{x}{y}_S)\ket{y'}_S = 0$ if $\epsilon_{x}-\epsilon_{y}\neq \epsilon_{x'}-\epsilon_{y'}$.
    
    Combining this with Eq.~(\ref{appeq:QDB_def}), for each $s\in[0,1]$, we have
    \begin{align}
        &\quad \nE^{-\beta \epsilon_y}\nE^{-(1-s)\beta(\epsilon_{x}-\epsilon_y)}\bra{x'}_S\mathcal{X}(\ketbra{x}{y}_S)\ket{y'}_S \nonumber\\
        &= \nE^{-\beta \epsilon_{y'}}\nE^{-(1-s)\beta(\epsilon_{x'}-\epsilon_{y'})}\bra{x}_S\mathcal{X}(\ketbra{x'}{y'}_S)\ket{y}_S^*,
    \end{align}
    for all $(x,x',y,y')$, which is equivalent to
    \begin{align}
        {\rm Tr}\{\mathcal{X}(\gamma_S^{1-s} A_S^\dagger \gamma_S^{s})B_S\} = {\rm Tr}\{A_S^\dagger\mathcal{X}(\gamma_S^{s}B_S\gamma_S^{1-s})\},
    \end{align}
    for all system operators $A_S$ and $B_S$, and therefore,
    \begin{align}
        \langle A_S, \mathcal{X}^\dagger(B_S) \rangle_{\gamma_S, s} = \langle \mathcal{X}^\dagger(A_S), B_S \rangle_{\gamma_S, s},
    \end{align}
    for all system operators $A_S$ and $B_S$, i.e., $\mathcal{X}^\dagger$ is self-adjoint with respect to $\langle \cdot, \cdot \rangle_{\gamma_S, s}$.
\end{proof}
\section{Proof of Prop.~\ref{prop:general_form_xi}}\label{app:proof_of_proposition_xi}
The coefficient $\xi_{jj'mm'}(t)$ is defined in Eq.~(\ref{eq:xi_def_from_A}):
\begin{align}
    \xi_{jj'mm'}(t) &\equiv \braket{j|_S\mathcal{A}_t(\ketbra{j'}{m'}_S)|m}_S \\
    &= \sum_r {\rm Tr}\left\{(F_S^{jj'})^\dagger K_r(t)\right\}{\rm Tr}\left\{F_S^{mm'} K_r^\dagger(t)\right\}.
    \label{appeq:xi_def_from_A}
\end{align}
where $\{F_{S}^{jj'}\equiv\ketbra{j}{j'}_S\}_{j,j'=0}^{d_S-1}$ are the transition operators defined in the eigenbasis of $H_S$ and $K_r(t)$ are the Kraus operators of the map $\mathcal{A}_t$ appearing in the generator $\mathcal{L}_t$ [see Eq.~(\ref{eq:L_gen})].

Grouping $(j,j')$ and $(m,m')$, respectively, $\xi_{jj'mm'}(t)$ can be regarded as a complex matrix $\bm{\xi}(t)\in\mathbb{C}^{d_S^2\times d_S^2}$. In particular, $\bm{\xi}(t)$ is the Choi operator of the CP map $\mathcal{A}_t$, i.e., $\bm{\xi}(t) = \sum_{x,y=0}^{d_S-1} \mathcal{A}_t(\ketbra{x}{y})\otimes \ketbra{x}{y}$. Consequently, it is easy to see that $\bm{\xi}(t)$ is Hermitian and positive semidefinite, i.e., 
\begin{align}
    \xi_{jj'mm'}(t) = \xi_{mm'jj'}^*(t),\quad\forall\, (j,j',m,m'),\, \forall\, t,
    \label{eq:hermiticity_xi}
\end{align}
and
\begin{align}
    \sum_{j,j',m,m'=0}^{d_S-1}v_{jj'}^*\xi_{jj'mm'}(t)v_{mm'}\ge 0,\,\forall\, t,
    \label{eq:PSD_xi}
\end{align}
for all vectors $\vec{v} = (v_{xy})_{x,y=0}^{d_S-1}\in\mathbb{C}^{d_S^2}$.

(a) Firstly, we prove the `only if' part in Prop.~\ref{prop:general_form_xi}, namely, the coefficients $\xi_{jj'mm'}(t)$ associated with a thermal Lindbladian $\mathcal{L}_t$ are in the form of Eq.~(\ref{eq:xi_general_under_axioms}).

By time-translation symmetry~\ref{axiom:II} [Eq.~(\ref{eq:A_commute_L_H_S})], $\mathcal{A}_t$ is block-diagonal in the eigensubspace of $\mathcal{L}_{H_S}$. Recall that the eigenoperators and eigenvalues of $\mathcal{L}_{H_S}$ are $\{\ketbra{j}{j'}_S\}_{j,j'=0}^{d_S-1}$ and $\{\epsilon_j - \epsilon_{j'}\}_{j,j'=0}^{d_S-1}$, respectively. We have
\begin{align}
    \bra{j}_S\mathcal{A}_t(\ketbra{j'}{m'}_S)\ket{m}_S = \delta(\epsilon_{j'}-\epsilon_{m'} +\epsilon_m - \epsilon_j)\xi_{jj' m m'}(t),
\end{align}
for all $(j, j', m, m')$, where $\delta(x) = 1$ if $x=0$ and $0$ otherwise. We thus have
\begin{align}
    \xi_{jj'mm'}(t) = \delta(\epsilon_{j'}-\epsilon_{m'} +\epsilon_m - \epsilon_j) \xi_{jj'mm'}(t).
    \label{eq:xi_energy_gaps_matching}
\end{align}
By quantum detailed balance~\ref{axiom:III} [Eq.~(\ref{eq:QDB_A_t})], we have
\begin{align}
    \nE^{-\beta \epsilon_{m'}}\xi_{jj' m m'}(t) = \nE^{-\beta \epsilon_{m}}\xi^*_{j' j m' m}(t),
    \label{eq:QDB_xi}
\end{align}
for all $(j,j',m,m')$ and $t$. 

Hence, the constraints on $\mathcal{A}_t$ imposed by Axioms~\ref{axiom:II} and \ref{axiom:III} are reflected in the structure of $\xi_{jj'mm'}(t)$, as given by Eqs.~(\ref{eq:xi_energy_gaps_matching}) and (\ref{eq:QDB_xi}), respectively.
We now combine these two constraints and derive an expression for $\xi_{jj'mm'}(t)$.
Define a new coefficient
\begin{align}
    \zeta_{jj' mm'}(t) \equiv \nE^{-\beta (\epsilon_{j'} - \epsilon_j)/4}\xi_{jj' mm'}(t) \nE^{-\beta (\epsilon_{m'}-\epsilon_m)/4}.
    \label{eq:zeta}
\end{align}
Grouping the indices $(j,j')$ and $(m,m')$, the matrix $\bm{\zeta}(t)\in\mathbb{C}^{d_S^2\times d_S^2}$ is Hermitian and positive semidefinite:
\begin{align}
    \zeta_{jj' m m'}(t) &\equiv \nE^{-\beta (\epsilon_{j'} - \epsilon_j)/4}\xi_{jj' mm'}(t) \nE^{-\beta (\epsilon_{m'}-\epsilon_m)/4}\\
    &= \nE^{-\beta (\epsilon_{m'}-\epsilon_m)/4}\xi_{mm'jj'}^*(t)\nE^{-\beta (\epsilon_{j'} - \epsilon_j)/4}\\
    &= \zeta_{mm' j j'}^*(t),\quad \forall\, (j,j',m,m'),\, \forall\, t,
\end{align}
where the second line follows from Eq.~(\ref{eq:hermiticity_xi}),
and 
\begin{align}
    \sum_{j,j',m,m'=0}^{d_S-1}v_{jj'}^*\zeta_{jj'mm'}(t)v_{mm'}\ge 0,\,\forall\, t,
\end{align}
for all vectors $\vec{v} = (v_{xy})_{x,y=0}^{d_S-1}\in\mathbb{C}^{d_S^2}$. These properties of $\bm{\zeta}(t)$ can also be seen by noticing that $\bm{\zeta}(t) = \bm{\Upsilon}\bm{\xi}(t)\bm{\Upsilon}^\dagger$, where $\Upsilon_{xx'yy'}=\delta_{xy}\delta_{x'y'}\nE^{-\beta(\epsilon_{x'}-\epsilon_x)/4}$, such that Hermiticity and positivity of $\bm{\zeta}(t)$ follows directly from that of $\bm{\xi}(t)$.

Furthermore, the constraints given by Eqs.~(\ref{eq:xi_energy_gaps_matching}) and (\ref{eq:QDB_xi}) imply that $\zeta_{jj'mm'}(t)$ satisfies the following symmetry:
\begin{align}
    \zeta_{jj' m m'}(t)
    &\equiv \nE^{-\beta (\epsilon_{j'} - \epsilon_j)/4}\xi_{jj' mm'}(t) \nE^{-\beta (\epsilon_{m'}-\epsilon_m)/4}\\
    &= \nE^{-\beta (\epsilon_{j'} - \epsilon_j)/4}\xi_{j' j m'm}^*(t) \nE^{3\beta(\epsilon_{m'} - \epsilon_m)/4}\\
    &= \nE^{-\beta (\epsilon_{j'} - \epsilon_j)/4}\nE^{3\beta (\epsilon_{m'}-\epsilon_m)/4}\nonumber\\
    &\quad \times\delta(\epsilon_{j}-\epsilon_{m} +\epsilon_{m'} - \epsilon_{j'})\xi_{j' j m'm}^*(t)\\
    &= \nE^{-\beta (\epsilon_j - \epsilon_{j'})/4} \xi_{j' j m'm}^*(t) \nE^{-\beta (\epsilon_{m} - \epsilon_{m'})/4}\\
    &= \zeta_{j' j m' m}^*(t),\quad \forall\, (j,j',m,m'),\, \forall\, t,
\end{align}
where in the second and the third line we used Eqs.~\eqref{eq:QDB_xi} and \eqref{eq:xi_energy_gaps_matching}, respectively.
By defining the swap matrix $\bm{S}\in\mathbb{C}^{d_S^2\times d_S^2}$ for indices $j\leftrightarrow j'$ and $m\leftrightarrow m'$, we can express the symmetry as
\begin{align}
    \bm{\zeta}(t) = \bm{S}\bm{\zeta}^*(t)\bm{S}.
    \label{eq:zeta_reshuffled_symmetry}
\end{align}
Since $\bm{\zeta}(t)$ is positive semidefinite, its positive square root $\bm{\zeta}^{1/2}(t)$ is unique. 
Noticing that $\bm{S}(\bm{\zeta}^{1/2}(t))^*\bm{S} \ge 0$ and
$(\bm{S}(\bm{\zeta}^{1/2}(t))^*\bm{S})(\bm{S}(\bm{\zeta}^{1/2}(t))^*\bm{S}) = \bm{S}(\bm{\zeta}^{1/2}(t))^*(\bm{\zeta}^{1/2}(t))^*\bm{S} = \bm{S}\bm{\zeta}^*(t)\bm{S} = \bm{\zeta}(t)$, we find that:
\begin{align}
    \bm{\zeta}^{1/2}(t) = \bm{S}(\bm{\zeta}^{1/2}(t))^*\bm{S},
    \label{eq:zeta_sqrt_reshuffled_symmetry}
\end{align}
i.e., $\bm{\zeta}^{1/2}(t)$ inherits the symmetry of $\bm{\zeta}(t)$.
We can further define
\begin{align}
    \bm{W}(t) \equiv \bm{\zeta}^{1/2}(t)\bm{U}(t),
\end{align}
where $\bm{U}(t)$ is a unitary matrix satisfying $\bm{U}(t) =\bm{S}\bm{U}^*(t)\bm{S}$. Therefore,
\begin{align}
    \bm{W}(t) = \bm{S}\bm{W}^*(t)\bm{S},
    \label{eq:W_reshuffled_symmetry}
\end{align}
and $\bm{\zeta}(t)=\bm{W}(t)\bm{W}^\dagger(t)$. The gauge freedom of $\bm{W}(t)$ introduced by $\bm{U}(t)$ is irrelevant here but will become useful in finding a collision-model realisation of $\mathcal{L}_t$ (see App.~\ref{app:vanishing_H_LS}).
We write
\begin{align}
    \zeta_{jj'mm'}(t) &= \sum_{q,q'=0}^{d_S-1} W_{jj'qq'}(t) W_{mm'qq'}^*(t),
\end{align}
and therefore,
\begin{align}
    \xi_{jj'mm'}(t) &= \nE^{-\beta(\epsilon_{j}-\epsilon_{j'})/4}\nE^{-\beta(\epsilon_{m}-\epsilon_{m'})/4} \nonumber\\
    &\quad \times \sum_{q,q'=0}^{d_S-1} W_{jj'qq'}(t) W_{mm'qq'}^*(t)\\
    &= \delta(\epsilon_{j'}-\epsilon_{m'} +\epsilon_m - \epsilon_j) \nE^{-\beta(\epsilon_{j}-\epsilon_{j'})/2}\nonumber \\
    &\quad \times \sum_{q,q'=0}^{d_S-1} W_{jj'qq'}(t) W_{mm'qq'}^*(t),
    \label{appeq:xi_general_under_axioms}
\end{align}
where in the second equality we used Eq.~(\ref{eq:xi_energy_gaps_matching}). $\xi_{jj'mm'}(t)$ is therefore of the desired form given in Eq.~(\ref{eq:xi_general_under_axioms}).

(b) We then prove the `if' part of Prop.~\ref{prop:general_form_xi}, namely, the coefficient $\xi_{jj'mm'}(t)$ given in Eq.~(\ref{eq:xi_general_under_axioms}), leads to a map $\mathcal{A}_t$ satisfying time-translation symmetry~\ref{axiom:II} [Eq.~(\ref{eq:A_commute_L_H_S})] and quantum detailed balance~\ref{axiom:III} [Eq.~(\ref{eq:QDB_A_t})]. Equivalently, it suffices to show that $\xi_{jj'mm'}(t)$ satisfies Eqs.~(\ref{eq:hermiticity_xi}), (\ref{eq:PSD_xi}), (\ref{eq:xi_energy_gaps_matching}) and (\ref{eq:QDB_xi}). The first three conditions follow straightforwardly. For the last condition, we obtain
\begin{widetext}
    \begin{align}
    \nE^{-\beta(\epsilon_{m'}-\epsilon_m)}\xi_{jj'mm'}(t) &= \delta(\epsilon_{j'}-\epsilon_{m'} +\epsilon_m - \epsilon_j) \nE^{-\beta(\epsilon_{j'}-\epsilon_{j})/2} \sum_{q,q'=0}^{d_S-1} W_{jj'qq'}(t) W_{mm'qq'}^*(t)\\
    &= \delta(\epsilon_{j'}-\epsilon_{m'} +\epsilon_m - \epsilon_j) \nE^{-\beta(\epsilon_{j'}-\epsilon_{j})/2} \sum_{q,q'=0}^{d_S-1} W_{j'jqq'}^*(t) W_{m'mqq'}(t)\\
    &= \xi_{j'jm'm}^*(t),
\end{align}
\end{widetext}
where in the second line we used the property that $W_{xx'yy'}(t) = W_{x'xy'y}^*(t)$ for all $(x,x',y,y')$ and swapped the indices $q\leftrightarrow q'$. This completes the proof of Prop.~\ref{prop:general_form_xi}.

\section{Proof of Prop.~\ref{prop:L_CM_=>_L}}\label{app:proof_of_L_CM_=>_L}
In this appendix,  we prove that the Lindbladian $\mathcal{L}_t^{\rm (CM)}$ [Eq.~(\ref{eq:L_CM_def})] arising from the continuous-time limit of an ECTCM is a thermal Lindbladian satisfying Axioms~\ref{axiom:I}--\ref{axiom:III}.

Recall the definition of the coefficient $\xi_{jj'mm'}^{\rm (CM)}$ [Eq.~(\ref{eq:xi_CM})]:
\begin{align}
    \xi_{jj'mm'}^{\rm (CM)}(t) &\equiv \delta(\epsilon_{j'}-\epsilon_j + \epsilon_{m} - \epsilon_{m'})\nonumber\\
    &\quad \times\sum_{k,k'|e_{k} - e_{k'} = \epsilon_j - \epsilon_{j'}} p^{(k)}_E V_{jk'j'k}(t)V_{mk'm'k}^*(t).
\end{align}
We define the map $\mathcal{A}^{\rm (CM)}_t$ as
\begin{align}
    \mathcal{A}^{\rm (CM)}_t(\cdot) = \sum_{j,j',m,m'}\xi_{jj'mm'}^{\rm (CM)}(t)F_S^{jj'}(\cdot) (F_S^{mm'})^\dagger.
\end{align}
Note that $\mathcal{A}^{\rm (CM)}_t$ is CP since $\mathcal{A}^{\rm (CM)}_t(\cdot) \equiv \sum_{k,k'}L_{kk'}(t)(\cdot)L_{kk'}^\dagger(t)$ is constructed from the Lindblad operators $L_{kk'}(t)$ in Eq.~(\ref{eq:L_kk'(t)}).
The ECTCM Lindbladian $\mathcal{L}^{\rm (CM)}_t$ can be written as 
\begin{align}
    \mathcal{L}^{\rm (CM)}_t(\cdot) &= -\nI[H_S'^{\rm (CM)}(t), \cdot] \nonumber\\
    &\quad + \mathcal{A}^{\rm (CM)}_t(\cdot) - \frac{1}{2}[(\mathcal{A}^{\rm (CM)}_t)^\dagger(\mathds{1}), \cdot]_+,
    \label{appeq:L_CM_with_A}
\end{align}
coinciding with the Lindbladian form required by Axiom~\ref{axiom:I}. We now show that $\mathcal{L}^{\rm (CM)}_t$ satisfies time-translation symmetry~\ref{axiom:II} and quantum detailed balance~\ref{axiom:III}.

Since $[H_S'^{\rm (CM)}(t), H_S]=0,\,\forall\, t$ in ECTCMs [Condition~\ref{CM:1}],
for the time-translation symmetry~\ref{axiom:II}, we only need to show that $[\mathcal{A}^{\rm (CM)}_t, \mathcal{L}_{H_S}] = 0$, i.e., $\mathcal{A}^{\rm (CM)}$ acts like a block matrix in the eigensubspace of $\mathcal{L}_{H_S}$. Recall that the eigenoperators and eigenvalues of $\mathcal{L}_{H_S}$ are $\{\ketbra{j}{j'}_S\}_{j,j'}$ and $\{\epsilon_j - \epsilon_{j'}\}_{j,j'}$, respectively. We have
~\\
\begin{align}
    &\quad \mathcal{A}^{\rm (CM)}_t(\ketbra{\ell}{n}_S) \nonumber\\
    &=  \sum_{j,j',m,m'}\xi_{jj'mm'}^{\rm (CM)}(t)F_S^{jj'}(\ketbra{\ell}{n}_S) (F_S^{mm'})^\dagger\\
    &= \sum_{j,m}\xi_{j\ell mn}^{\rm (CM)}(t)\ketbra{j}{m}_S\\
    &= \sum_{j,m}\delta(\epsilon_{\ell}-\epsilon_j + \epsilon_{m} - \epsilon_{n})\xi_{j\ell mn}^{\rm (CM)}(t)\ketbra{j}{m}_S,
\end{align}
where in the third line we used the expression of $\xi_{jj'mm'}^{\rm (CM)}(t)$ in Eq.~(\ref{eq:xi_CM}). Therefore, $\mathcal{A}^{\rm (CM)}_t$ cannot send an eigenoperator of $\mathcal{L}_{H_S}$ outside of the corresponding eigen-subspace. We conclude that $[\mathcal{A}^{\rm (CM)}_t, \mathcal{L}_{H_S}] = 0$ and thus $\mathcal{L}^{\rm (CM)}_t$ satisfies time-translation symmetry~\ref{axiom:II}.
 
Showing that $\mathcal{L}^{\rm (CM)}_t$ satisfies quantum detailed balance~\ref{axiom:III}, is equivalent to showing that $\mathcal{A}^{\rm (CM)}_t$ satisfies Eq.~(\ref{eq:QDB_A_t}), namely,
\begin{align}
    \nE^{-\beta \epsilon_n}\xi_{j\ell m n}^{\rm (CM)}(t) = \nE^{-\beta \epsilon_m}\left(\xi_{\ell j n m}^{\rm (CM)}(t)\right)^*,
    \label{eq:QDB_xi_CM}
\end{align}
for all $(j,\ell,m,n)$.
We firstly notice that for all tuples $(j,\ell,m,n)$ for which $\epsilon_\ell - \epsilon_j \neq \epsilon_n - \epsilon_m$, we have $\xi_{j\ell mn}^{\rm (CM)}(t) = \xi_{\ell j n m}^{\rm (CM)}(t) = 0$ by Eq.~(\ref{eq:xi_CM}), so Eq.~(\ref{eq:QDB_xi_CM}) trivially holds. When $\epsilon_\ell - \epsilon_j = \epsilon_n - \epsilon_m$, we consider
\begin{align}
    &\quad \nE^{-\beta (\epsilon_n - \epsilon_m)}\xi_{j\ell mn}^{\rm (CM)}(t) \nonumber \\
    &= \sum_{k,k'|e_{k} - e_{k'} = \epsilon_j - \epsilon_{\ell}} \nE^{-\beta (\epsilon_n - \epsilon_m)}p^{(k)}_E V_{jk'\ell k}(t)V_{mk'nk}^*(t) \\
    &= \sum_{k,k'|e_{k'} - e_{k} = \epsilon_j - \epsilon_{\ell}} \nE^{-\beta (\epsilon_n - \epsilon_m)}p_E^{(k')} V_{jk\ell k'}(t)V_{mknk'}^*(t)\\
    &= \sum_{k,k'|e_{k} - e_{k'} = \epsilon_{\ell} - \epsilon_j} p^{(k)}_E V_{jk\ell k'}(t)V_{mknk'}^*(t)\\
    &= \sum_{k,k'|e_{k} - e_{k'} = \epsilon_{\ell} - \epsilon_j} p^{(k)}_E V_{\ell k'jk}^*(t)V_{nk'mk}(t)\\
    &= \left(\xi_{\ell j n m}^{\rm (CM)}(t)\right)^*,
\end{align}
where in the third line we relabelled $k\leftrightarrow k'$, the fourth line follows from that $\nE^{-\beta (\epsilon_n - \epsilon_m)}p_E^{(k')} = \nE^{-\beta (\epsilon_n - \epsilon_m + e_{k'})}/Z = \nE^{-\beta e_{k}}/Z = p_E^{(k)}$ where $Z\equiv {\rm Tr}\{\nE^{-\beta H_E}\}$, and the fifth line follows from the Hermiticity of $V(t)$, i.e., $V_{xyzw}(t)=V_{zwxy}^*(t),\,\forall\, (x,y,z,w)$.
Therefore, we have proven Eq.~(\ref{eq:QDB_xi_CM}), which implies that $\mathcal{L}^{\rm (CM)}_t$ satisfies quantum detailed balance~\ref{axiom:III}. This completes the proof of Prop.~\ref{prop:L_CM_=>_L}.

\section{Technical details for the ECTCM construction in Sec.~\ref{sec:constructing_ECTCM}}\label{app:technical_construction}
In this appendix, we present the explicit calculations underlying Sec.~\ref{sec:constructing_ECTCM}. That is, for a given thermal Lindbladian $\mathcal{L}_t$ satsifying Axioms~\ref{axiom:I}--\ref{axiom:III}, we provide an explicit construction for an ECTCM that realises it. \\ 
In this collision model, the ancilla Hamiltonian $H_E$ is defined as follows [Eq.~(\ref{eq:H_E_constructed})]:
\begin{align}
    H_E \equiv \sum_{\ell,\ell',q=0}^{d_S-1} \frac{\epsilon_\ell - \epsilon_{\ell'}}{2}\ketbra{(\ell,\ell'),q}{(\ell,\ell'),q}_E,
    \label{appeq:H_E_constructed}
\end{align}
where $\{\epsilon_\ell\}_{\ell=0}^{d_S-1}$ is the energy spectrum of $H_S$ and the pair $(\ell, \ell')$ corresponds to a $d_S$-dimensional eigensubspace of $H_E$, in which $q$ is the index for degenerate energy levels.  Introducing the index $k\equiv \ell d_S^2 + \ell' d_S + q$, we have $H_E \equiv \sum_{k=0}^{d_S^3-1}e_k \ketbra{k}{k}_E$ with $\{e_k\equiv \frac{\epsilon_{\ell}-\epsilon_{\ell'}}{2}\}_{k=0}^{d_S^3-1}$ being the eigenvalues.

The system--ancilla interaction Hamiltonian $V_{SE}(t)$ in the ECTCM is given by Eq.~(\ref{eq:V_constructed}):
\begin{widetext}
\begin{align}
    V_{SE}(t) \equiv \sum_{j,j',q,q'=0}^{d_S-1}\sqrt{Z} \frac{W_{jj'qq'}(t)}{N_{(j,j')}^{1/2}}\ketbra{j}{j'}_S\otimes\left(\sum_{\ell,\ell'=0|\epsilon_{\ell'}-\epsilon_{\ell} = \epsilon_{j'}-\epsilon_j}^{d_S-1}\ketbra{(\ell',\ell),q'}{(\ell,\ell'),q}_E\right),
    \label{appeq:V_constructed}
\end{align}
\end{widetext}
where $W_{jj'qq'}(t)$ is from the expression of $\xi_{jj'mm'}(t)$ [Eq.~(\ref{eq:xi_general_under_axioms_1})],
$Z\equiv \sum_{k=0}^{d_S^3-1}\nE^{-\beta e_k}$ is the partition function for the ancilla $E$ and
$N_{(j,j')}\equiv \big|\{(\ell,\ell')|\epsilon_\ell - \epsilon_{\ell'}=\epsilon_j - \epsilon_{j'}\}\big|$ is the degeneracy of the energy gap $\epsilon_j - \epsilon_{j'}$.

\begin{widetext}
\subsection{Hermiticity and energy conservation of $V_{SE}(t)$}\label{app:hermiticity_EC_of_V}
It is straightforward to check that $V_{SE}(t) = V_{SE}^\dagger(t)$:

\begin{align}
    V_{SE}^\dagger(t) &= \sum_{j,j',q,q'=0}^{d_S-1} \sqrt{Z} \frac{W_{jj'qq'}^*(t)}{N_{(j,j')}^{1/2}}\ketbra{j'}{j}_S\otimes\left(\sum_{\ell,\ell'=0|\epsilon_{\ell'}-\epsilon_{\ell} = \epsilon_{j'}-\epsilon_j}^{d_S-1}\ketbra{(\ell,\ell'),q}{(\ell',\ell),q'}_E\right) \\
    &= \sum_{j,j',q,q'=0}^{d_S-1} \sqrt{Z} \frac{W_{j'jq'q}(t)}{N_{(j',j)}^{1/2}}\ketbra{j'}{j}_S\otimes\left(\sum_{\ell,\ell'=0|\epsilon_{\ell'}-\epsilon_{\ell} = \epsilon_{j}-\epsilon_{j'}}^{d_S-1}\ketbra{(\ell',\ell),q}{(\ell,\ell'),q'}_E\right)\\
    &= \sum_{j,j',q,q'=0}^{d_S-1} \sqrt{Z} \frac{W_{jj'qq'}(t)}{N_{(j,j')}^{1/2}}\ketbra{j}{j'}_S\otimes\left(\sum_{\ell,\ell'=0|\epsilon_{\ell'}-\epsilon_{\ell} = \epsilon_{j'}-\epsilon_{j}}^{d_S-1}\ketbra{(\ell',\ell),q'}{(\ell,\ell'),q}_E\right)\\
    &= V_{SE}(t),
\end{align}\end{widetext}
where in the second line we used the symmetry: $W_{xx'yy'}(t) = W_{x'xy'y}^*(t)$ for all $(x,x',y,y')$ [Eq.~(\ref{eq:W_reshuffled_symmetry})], $N_{(j,j')}=N_{(j',j)}$ and swapped the indices $\ell\leftrightarrow \ell'$, and in the third line we swapped the indices $j\leftrightarrow j'$ and $q\leftrightarrow q'$.\\
\noindent In addition, to fulfill the Condition~\ref{CM:2} for ECTCMs, we show that $V_{SE}(t)$ is energy conserving, i.e., it satisfies $[V_{SE}(t), H_S+H_E] = 0$.
\clearpage
\begin{widetext}
    \begin{align}
    &\quad [V_{SE}(t), H_S + H_E] \nonumber\\
    &= \sum_{j,j',q,q'=0}^{d_S-1} \sum_{\ell,\ell'=0|\epsilon_{\ell'}-\epsilon_{\ell} = \epsilon_{j'}-\epsilon_{j}}^{d_S-1} \sqrt{Z} \frac{W_{jj'qq'}(t)}{N_{(j,j')}^{1/2}} \left[\left(\epsilon_{j'}+\frac{\epsilon_\ell-\epsilon_{\ell'}}{2}\right)-\left(\frac{\epsilon_{\ell'}-\epsilon_\ell}{2} + \epsilon_j\right)\right]\ketbra{j}{j'}_S\otimes\ketbra{(\ell',\ell),q'}{(\ell,\ell'),q}_E\\
    &= \sum_{j,j',q,q'=0}^{d_S-1} \sum_{\ell,\ell'=0|\epsilon_{\ell'}-\epsilon_{\ell} = \epsilon_{j'}-\epsilon_{j}}^{d_S-1} \sqrt{Z} \frac{W_{jj'qq'}(t)}{N_{(j,j')}^{1/2}} \left[\left(\epsilon_{j'}+\frac{\epsilon_j-\epsilon_{j'}}{2}\right)-\left(\frac{\epsilon_{j'}-\epsilon_j}{2} + \epsilon_j\right)\right]\ketbra{j}{j'}_S\otimes\ketbra{(\ell',\ell),q'}{(\ell,\ell'),q}_E\\
    &= 0\add{,}
\end{align}
where we have used that the second sum only runs over pairs $(\ell, \ell')$ for which $\epsilon_{\ell'}-\epsilon_{\ell} = \epsilon_{j'}-\epsilon_{j}$ holds.
\end{widetext}

\subsection{Proof of $\xi_{jj'mm'}^{\rm (CM)}(t) = \xi_{jj'mm'}(t)$}\label{app:xi_CM_=_xi}
The coefficient $\xi_{jj'mm'}(t)$ of the thermal Lindbladian is given by Eq.~(\ref{eq:xi_general_under_axioms_1}):
~\\
\begin{align}
    \xi_{jj'mm'}(t) &= \delta(\epsilon_{j'}-\epsilon_{m'} +\epsilon_m - \epsilon_j) \nE^{-\beta(\epsilon_{j}-\epsilon_{j'})/2}\nonumber \\
    &\quad \times \sum_{q,q'=0}^{d_S-1} W_{jj'qq'}(t) W_{mm'qq'}^*(t),
    \label{appeq:xi_general_under_axioms_1}
\end{align}
where $W_{xx'yy'}(t)$ is constructed from $\xi_{jj'mm'}(t)$ as described in App.~\ref{app:proof_of_proposition_xi}.

The coefficient $\xi_{jj'mm'}^{\rm (CM)}(t)$ of the ECTCM is in the form of Eq.~(\ref{eq:xi_CM_1}):
\begin{align}
    \xi_{jj'mm'}^{\rm (CM)}(t) &\equiv \delta(\epsilon_{j'}-\epsilon_j + \epsilon_{m} - \epsilon_{m'})\nonumber\\
    &\quad \times\sum_{k,k'|e_{k} - e_{k'} = \epsilon_j - \epsilon_{j'}} p^{(k)}_E V_{jk'j'k}(t)V_{mk'm'k}^*(t).
    \label{appeq:xi_CM_1}
\end{align}
where $p^{(k)}_E\equiv \bra{k}_E\gamma_E\ket{k}_E$ and $V_{jkj'k'}(t) \equiv \bra{j}_S\otimes\bra{k}_E V_{SE}(t) \ket{j'}_S\otimes\ket{k'}_E$.

Before proceeding with the calculation showing that $\xi_{jj'mm'}^{\rm (CM)}(t) = \xi_{jj'mm'}(t)$, we briefly explain the rationale behind the construction of the ancilla Hamiltonian $H_E$ [Eq.~(\ref{appeq:H_E_constructed})] and the system--ancilla interaction Hamiltonian $V_{SE}(t)$ [Eq.~(\ref{appeq:V_constructed})]. By comparing Eqs.~(\ref{appeq:xi_general_under_axioms_1}) and~(\ref{appeq:xi_CM_1}), one immediately identifies the correspondences:
\begin{align*}
    p^{(k)}_E &\leftrightarrow \nE^{-\beta(\epsilon_{j}-\epsilon_{j'})/2},\\
    V_{jkj'k'}(t) &\leftrightarrow W_{jj'qq'}(t).
\end{align*}
The first correspondence suggests that the ancilla energy levels labelled by $k$ should be chosen to reproduce the Boltzmann factors $\nE^{-\beta(\epsilon_{j}-\epsilon_{j'})/2}$. This naturally leads to an ancilla spectrum composed of the energy differences $\{\frac{\epsilon_{\ell}-\epsilon_{\ell'}}{2}\}_{\ell,\ell'=0}^{d_S-1}$. The second correspondence, however, imposes an additional requirement. Since the interaction Hamiltonian $V_{SE}(t)$ must encode the coefficients $W_{jj'qq'}(t)$ , the ancilla index $k$ must carry not only information about the energy difference $(\ell,\ell')$ but also the additional index $q$. This observation motivates enlarging the ancilla Hilbert space so that a single ancilla label $k$ can simultaneously accommodate both types of information.
Consequently, the ancilla must have dimension $d_S^3$, with the tuple index $(\ell,\ell',q)$ encoded as $k\equiv \ell d_S^2 + \ell' d_S + q$.

Now, to prove $\xi_{jj'mm'}^{\rm (CM)}(t) = \xi_{jj'mm'}(t)$ explicitly, we first compute the Lindblad operator $L_{kk'}(t)$ of the collision model, assuming a thermal initial environment state, given by Eq.~(\ref{eq:L_kk'}): 
\begin{widetext}
\begin{align}
    L_{kk'}(t) &\equiv \sqrt{p^{(k)}_E}\bra{k'}_E V_{SE}(t)\ket{k}_E\\
    &= \sum_{j,j',q,q'=0}^{d_S-1}\sum_{\ell,\ell'=0|\epsilon_{\ell'}-\epsilon_{\ell} = \epsilon_{j'}-\epsilon_{j}}^{d_S-1}\delta(k-\ell d_S^2-\ell'd_S-q)\delta(k'-\ell'd_S^2-\ell d_S-q')\nE^{-\beta (\epsilon_j - \epsilon_{j'})/4} \frac{W_{jj'qq'}(t)}{N_{(j,j')}^{1/2}} F_S^{jj'}.
\end{align}
We then have
\begin{align}
    \xi_{jj'mm'}^{\rm (CM)}(t) &\equiv \sum_{k,k'=0}^{d_S^3-1} {\rm Tr}\left\{(F_S^{jj'})^\dagger L_{kk'}(t)\right\}{\rm Tr}\left\{F_S^{mm'} L_{kk'}^\dagger(t)\right\} \label{eq:xi_CM_from_L_kk'}\\
    &= \sum_{k,k'=0}^{d_S^3-1}\hspace{1.5mm}\sum_{q,q'=0}^{d_S-1}\hspace{1.5mm}\sum_{p,p'=0}^{d_S-1}\hspace{1.5mm}\sum_{\ell,\ell'=0|\epsilon_{\ell'}-\epsilon_{\ell} = \epsilon_{j'}-\epsilon_{j}}^{d_S-1}\hspace{1.5mm}\sum_{n,n'=0|\epsilon_{n'}-\epsilon_{n} = \epsilon_{m'}-\epsilon_{m}}^{d_S-1}\nonumber\\
    &\quad \times \delta(k-\ell d_S^2-\ell'd_S-q)\delta(k'-\ell'd_S^2-\ell d_S-q')\delta(k-nd_S^2-n'd_S-p)\delta(k'-n'd_S^2-nd_S-p')\nonumber\\
    &\quad \times\nE^{-\beta (\epsilon_j - \epsilon_{j'})/4}\nE^{-\beta (\epsilon_m - \epsilon_{m'})/4} \frac{W_{jj'qq'}(t)}{N_{(j,j')}^{1/2}}\frac{W_{mm'pp'}^*(t)}{N_{(m,m')}^{1/2}}\\
    &= \sum_{q,q'=0}^{d_S-1}\hspace{1.5mm}\sum_{p,p'=0}^{d_S-1}\hspace{1.5mm}\sum_{\ell,\ell'=0|\epsilon_{\ell'}-\epsilon_{\ell} = \epsilon_{j'}-\epsilon_{j}}^{d_S-1}\hspace{1.5mm}\sum_{n,n'=0|\epsilon_{n'}-\epsilon_{n} = \epsilon_{m'}-\epsilon_{m}}^{d_S-1}\nonumber \\
    &\quad \times \delta(\ell d_S^2+\ell'd_S+q - nd_S^2-n'd_S-p)\delta(\ell'd_S^2+\ell d_S+q'-n'd_S^2-nd_S-p') \nonumber\\
    &\quad \times\nE^{-\beta (\epsilon_j - \epsilon_{j'})/4}\nE^{-\beta (\epsilon_m - \epsilon_{m'})/4} \frac{W_{jj'qq'}(t)}{N_{(j,j')}^{1/2}}\frac{W_{mm'pp'}^*(t)}{N_{(m,m')}^{1/2}}\label{eq:xi_CM_constraucted_deltas}\\
    &= \sum_{q,q'=0}^{d_S-1}\hspace{1.5mm}\sum_{p,p'=0}^{d_S-1}\hspace{1.5mm}\sum_{\ell,\ell'=0|\epsilon_{\ell'}-\epsilon_{\ell} = \epsilon_{j'}-\epsilon_{j}}^{d_S-1}\hspace{1.5mm}\sum_{n,n'=0|\epsilon_{n'}-\epsilon_{n} = \epsilon_{m'}-\epsilon_{m}}^{d_S-1} \delta_{\ell n}\delta_{\ell' n'}\delta_{qp}\delta_{q'p'}\nonumber\\
    &\quad \times \nE^{-\beta (\epsilon_j - \epsilon_{j'})/4}\nE^{-\beta (\epsilon_m - \epsilon_{m'})/4} \frac{W_{jj'qq'}(t)}{N_{(j,j')}^{1/2}}\frac{W_{mm'pp'}^*(t)}{N_{(m,m')}^{1/2}}\\
    &= \delta(\epsilon_{j'}-\epsilon_{m'}+\epsilon_m - \epsilon_j)\sum_{q,q'=0}^{d_S-1}\left(\sum_{\ell,\ell'=0|\epsilon_{\ell'}-\epsilon_{\ell} = \epsilon_{j'}-\epsilon_{j}}^{d_S-1} \frac{1}{N_{(j,j')}}\right)\nE^{-\beta (\epsilon_j - \epsilon_{j'})/2}W_{jj'qq'}(t)W_{mm'qq'}^*(t)\\
    &= \delta(\epsilon_{j'}-\epsilon_{m'}+\epsilon_m - \epsilon_j)\sum_{q,q'=0}^{d_S-1}\nE^{-\beta (\epsilon_j - \epsilon_{j'})/2}W_{jj'qq'}(t)W_{mm'qq'}^*(t)
    \label{eq:xi_CM_constructed}\\
    &= \xi_{jj'mm'}(t),
\end{align}
where in Eq.~(\ref{eq:xi_CM_constraucted_deltas}), to evaluate the Kronecker deltas, we used the fact that the base-$d_S$ expansion of an integer is unique, so equality of two integers implies equality of the corresponding coefficients in their base-$d_S$ representations.
\end{widetext}
Consequently, we have shown that, for \textit{any} thermal Lindbladian $\mathcal{L}_t$, we can construct a driving term $\mathcal{D}_S(t)$ [see Eq.~(\ref{eq:D_S_constructed})], an ancilla Hamiltonian $H_E$ and interaction Hamiltonian $V_{SE}(t)$ satisfying Conditions~\ref{CM:1}--\ref{CM:2}, such that the resulting collision model Lindblaidan $\mathcal{L}^{\rm (CM)}_t$ -- obtained from assuming an initial thermal state on the environment, hence satisfying Condition~\ref{CM:3} -- coincides with $\mathcal{L}_t$. In the next subsection, we show that $V_{SE}$ can be chosen such that the corresponding Lamb shift vanishes, i.e., the collision model also satisfies Condition~\ref{CM:4}, and is thus an ECTCM.

Before concluding this subsection, we remark that the above construction can be summarised in the following lemma:
\begin{lemma}[Thermal Hermitian dilation]\label{lem:thermal_hermitian_dilation}
    Given a CP map $\mathcal{A}$ that satisfies the two conditions:
    \begin{itemize}
        \item (Time-translation symmetry) $[\mathcal{A}, \mathcal{L}_{H_S}] = 0$.
        \item (Quantum detailed balance) For all $(x,x',y,y')$, $\nE^{-\beta \epsilon_y}\bra{x'}_S\mathcal{A}(\ketbra{x}{y}_S)\ket{y'}_S = \nE^{-\beta \epsilon_{y'}}\bra{x}_S\mathcal{A}(\ketbra{x'}{y'}_S)\ket{y}_S^*$,
    \end{itemize}
    where $\{\ket{j}_S\}_{j=0}^{d_S-1}$ is the orthonormal eigenbasis of $H_S$ with the corresponding eigenvalues $\{\epsilon_j\}_{j=0}^{d_S-1}$, there always exists a thermal Hermitian dilation:
    \begin{align}
        \mathcal{A}(\cdot) = {\rm Tr}_E\{V_{SE}(\cdot\otimes \gamma_E)V_{SE}\},
    \end{align}
    where $H_E$ and $V_{SE}$ are constructed in Eqs.~(\ref{appeq:H_E_constructed}) and (\ref{appeq:V_constructed}), respectively, satisfying  $[V_{SE}, H_S+H_E] = 0$, and $\gamma_E \equiv \nE^{-\beta H_E}/{\rm Tr}\{\nE^{-\beta H_E}\}$ is the thermal state of $E$. We note that
    \begin{align}
        \mathcal{A}\text{ is TP } \Leftrightarrow {\rm Tr}_E\{V_{SE}^2 \gamma_E\} = \mathds{1}_S.
    \end{align}
\end{lemma}
\begin{proof}
    Since the CP map $\mathcal{A}$ satisfies the pertinent parts of Axioms~\ref{axiom:II} and \ref{axiom:III}, it can be written as $\mathcal{A}(\cdot) = \sum_{j,j',m,m'=0}^{d_S-1}\xi_{jj'mm'}F_S^{jj'}(\cdot) (F_S^{mm'})^\dagger$ where $F_S^{xy}\equiv \ketbra{x}{y}_S$ and the coefficient $\xi_{jj'mm'}$ is of the form of Eq.~(\ref{appeq:xi_general_under_axioms_1}) without the variable of time. The thermal Hermitian dilation ${\rm Tr}_E\{V_{SE}(\cdot\otimes \gamma_E)V_{SE}\}$ with $H_E$ and $V_{SE}$ defined in Eqs.~(\ref{appeq:H_E_constructed}) and (\ref{appeq:V_constructed}) can be written as ${\rm Tr}_E\{V_{SE}(\cdot\otimes \gamma_E)V_{SE}\} = \sum_{j,j',m,m'=0}^{d_S-1}\xi_{jj'mm'}^{\rm (CM)}F_S^{jj'}(\cdot) (F_S^{mm'})^\dagger$ where $\xi_{jj'mm'}^{\rm (CM)}$ is given by Eq.~(\ref{eq:xi_CM_from_L_kk'}). Following the derivation of Eq.~(\ref{eq:xi_CM_constructed}), we have $\xi_{jj'mm'} = \xi_{jj'mm'}^{\rm (CM)}$, and therefore, $\mathcal{A}(\cdot) = {\rm Tr}_E\{V_{SE}(\cdot\otimes \gamma_E)V_{SE}\}$.
\end{proof}
To understand this lemma, we note that the Hermitian operator $V_{SE}$ appears in the first order Taylor expansion of the unitary $U_{SE}\equiv \nE^{-\nI V_{SE}\Delta t}$ with respect to $\Delta t$. The lemma therefore shows that every CP map satisfying thermodynamic Axioms~\ref{axiom:II} and~\ref{axiom:III} can be interpreted as an effective partial description of a short-time energy-conserving collision with a thermal ancilla. The remaining contribution required to complete the Lindbladian description is precisely the anticommutator term involving $\mathcal{A}^\dagger(\mathds{1})$, as appears in $\mathcal{L}^{\rm (CM)}$ [Eq.~(\ref{appeq:L_CM_with_A})].

Recall that Prop.~\ref{prop:general_form_xi} can be regarded as a characterisation of CP maps satisfying thermodynamic Axioms~\ref{axiom:II} and~\ref{axiom:III}. 
Lem.~\ref{lem:thermal_hermitian_dilation} thus offers an alternative characterisation of the same class of maps from a collision-model perspective. 
We will use this lemma in App.~\ref{app:error_analysis} to bound the norm of $\mathcal{A}$ in terms of the norm of $V_{SE}$.

\medskip

\subsection{Vanishing Lamb shift by a gauge freedom}\label{app:vanishing_H_LS}
In this subsection, we conclude the proof that all thermal Linbladians can be realised by ECTCMs by showing that the Lamb shift induced by the constructed 
$V_{SE}(t)$ [Eq.~(\ref{eq:V_constructed})] can be eliminated by exploiting a gauge freedom.
Recall that the Lamb shift is given by $H_{S}^{\rm LS}(t) \equiv \mathrm{Tr}_{E}\{V_{SE}(t)\gamma_{E}\}$ [Eq.~(\ref{eq:H_S^LS_def})]. We see that a vanishing Lamb shift, i.e., $H_{S}^{\rm LS}(t) = 0$ implies
\begin{widetext}
\begin{align}
    {\rm Tr}_E\{V_{SE}(t) \gamma_E\} &= \sum_{k=0}^{d_S^3-1}\sum_{j,j',q,q'=0}^{d_S-1}\sum_{\ell,\ell'=0|\epsilon_{\ell'}-\epsilon_{\ell} = \epsilon_{j'}-\epsilon_j}^{d_S-1} \delta(k-\ell'd_S^2-\ell d_S - q')\delta(k - \ell d_S^2 - \ell' d_S - q) \nE^{-\beta e_k}\frac{W_{jj'qq'}(t)}{(ZN_{(j,j')})^{1/2}} F_S^{jj'}\\
    &= \sum_{j,j'|\epsilon_j=\epsilon_{j'}}^{d_S-1}\sum_{q=0}^{d_S-1}\sum_{\ell = 0}^{d_S-1}\frac{W_{jj'qq}(t)}{(ZN_{(j,j')})^{1/2}} F_S^{jj'}\\
    &= \sum_{j,j'|\epsilon_j=\epsilon_{j'}}^{d_S-1}\sum_{q=0}^{d_S-1}\frac{d_S W_{jj'qq}(t)}{(ZN_0)^{1/2}} F_S^{jj'} = 0\\
    &\Leftrightarrow \sum_{q=0}^{d_S-1} W_{jj'qq}(t) = 0, \quad \forall\, (j,j') \text{ with } \epsilon_j = \epsilon_{j'},
    \label{eq:condition_on_W_for_LS}
\end{align}
where 
the second equality follows from the uniqueness of the base-$d_S$ expansion of integers and
$N_0 \equiv \big|\{(\ell, \ell')|\epsilon_\ell = \epsilon_{\ell'}\}\big|$, and the final equivalence follows from the linear independence of the $F_S^{jj'}$ basis elements. 
\end{widetext}
To enforce Eq.~(\ref{eq:condition_on_W_for_LS}) while keeping the terms $\xi_{jj'mm'}^{\rm (CM)}(t)$ -- and thus the resulting Lindbladian $\mathcal{L}^{(\rm{CM})}_t$ -- invariant, we ignore the time dependence for simplicity and recall that $W_{jj'qq'}$ comes from the square root of $\zeta_{jj'mm'}$ in Eq.~(\ref{eq:zeta}): $\bm{W} \equiv \bm{\zeta}^{1/2}\bm{U}$, where the bold letter denotes the matrices in $\mathbb{C}^{d_S^2\times d_S^2}$ treating $(j,j')$ and $(m,m')$ as row and column indices, respectively, and $\bm{U}$ is a gauge unitary matrix satisfies the symmetry:
\begin{align}
    \bm{U} = \bm{S}\bm{U}^*\bm{S},
    \label{eq:U_reshuffled_symmetry}
\end{align}
with $\bm{S}$ being the swap operator acting on $(j,j')$ and $(m,m')$, respectively.
Since the matrix $\bm{U}$ does not influence the resulting coefficient $\xi_{jj'mm'}^{\rm (CM)}$ in Eq.~(\ref{eq:xi_CM_constructed}),
we now show how the condition of vanishing Lamb shift [Eq.~(\ref{eq:condition_on_W_for_LS})] on $\bm{W}$ can be satisfied by using the gauge freedom of $\bm{W}$, i.e., choosing an appropriate $\bm{U}$.

To this end, define the rectangular matrix $\bm{M}\in \mathbb{C}^{N_0\times d_S^2}$ with $N_0$ rows and $d_S^2$ columns, where $N_0$ denotes the number of equations representing the condition in Eq.~(\ref{eq:condition_on_W_for_LS}), i.e.,
\begin{align}
    \bm{M}_{\alpha,:} := (\bm{\zeta}^{1/2})_{r_\alpha,:},\quad \alpha=0,1,\dots,N_0-1,
    \label{eq:M_def}
\end{align}
where each $r_\alpha$ is a pair $(j,j')$ and $\alpha$ enumerates the elements of $\mathscr{R}\equiv\{(\ell,\ell')|\epsilon_\ell = \epsilon_{\ell'}\}$ as $\mathscr{R}=\{r_0,r_1,\dots,r_{N_0-1}\}$.
Define the vector $\ket{v}\equiv \sum_{q=0}^{d_S-1}\ket{q,q}\in\mathbb{C}^{d_S^2}$. Condition Eq.~(\ref{eq:condition_on_W_for_LS}) is equivalent to
\begin{align}
    \exists\, \bm{U},\text{ such that } \bm{M}\bm{U}\!\ket{v} = 0,
    \label{eq:M_U_v_=_0}
\end{align}
Therefore, to satisfy Eq.~(\ref{eq:M_U_v_=_0}), we require that $\bm{U}\!\ket{v}\in \mathfrak{ker}(\bm{M})$ where $\mathfrak{ker}(\bm{X})$ denotes the kernel of the matrix $\bm{X}$. Define $\mathfrak{reach}(\mathscr{X};\ket{\psi})$ as the set of reachable states via matrices in the set $\mathscr{X}$ acting on $\ket{\psi}$, i.e.,
\begin{align}
    \mathfrak{reach}(\mathscr{X};\ket{\psi}) := \{\bm{X}\!\ket{\psi} | \bm{X}\in\mathscr{X}\},
\end{align}
and denote the set of $\bm{U}$'s satisfying the symmetry [Eq.~(\ref{eq:U_reshuffled_symmetry})] as $\mathscr{U}$.
Eq.~(\ref{eq:M_U_v_=_0}) can be expressed as
\begin{align}
    \mathfrak{ker}(\bm{M})\cap \mathfrak{reach}(\mathscr{U};\ket{v}) \neq \varnothing,
    \label{eq:ker_reach_non_empty}
\end{align}
i.e., the kernel of $\bm{M}$ intersects with the reachable set $\mathfrak{reach}(\mathscr{U};\ket{v})$. We have the following characterisation of $\mathfrak{reach}(\mathscr{U};\ket{v})$:
\begin{lemma}\label{lem:reach_U_v}
    Let ${\rm Herm}^{d_S\times d_S}$ denote the set of Hermitian matrices in $\mathbb{C}^{d_S\times d_S}$ and $\|\!\ket{\psi}\!\|_2$ be the 2-norm of $\ket{\psi}$.
    \begin{align}
        \mathfrak{reach}(\mathscr{U};\ket{v}) &= \{\ket{\phi}|\ket{\phi} = \sum_{\ell,\ell'=0}^{d_S-1} H_{\ell\ell'}\ket{\ell,\ell'},\nonumber\\
        &\quad \|\!\ket{\phi}\!\|_2=\sqrt{d_S},\, \bm{H}\in{\rm Herm}^{d_S\times d_S}\} \\
        &=: \overline{\mathscr{V}}_{\rm Herm}.
    \end{align}
\end{lemma}
\begin{proof}
    For a general vector $\ket{\psi} \equiv \sum_{\ell,\ell'=0}^{d_S-1}\psi_{\ell\ell'}\ket{\ell,\ell'}$, we note that $\bm{S}\!\ket{\psi}^* = \sum_{\ell,\ell'=0}^{d_S-1}\psi_{\ell'\ell}^*\ket{\ell,\ell'}$. It is straightforward that
    \begin{align}
        \overline{\mathscr{V}}_{\rm Herm} = \{\ket{\phi} | \bm{S}\!\ket{\phi}^* = \ket{\phi},\, \|\!\ket{\phi}\!\|_2=\sqrt{d_S}\}.
        \label{eq:V_Herm_representation}
    \end{align}
    For a vector $\ket{u}\equiv \bm{U}\!\ket{v}\in \mathfrak{reach}(\mathscr{U};\ket{v})$, we have
    \begin{align}
        \bm{S}\!\ket{u}^* = \bm{S}\bm{U}^*\!\ket{v} = \bm{S}\bm{U}^*\bm{S}\!\ket{v} = \bm{U}\!\ket{v} = \ket{u},
    \end{align}
    where the equality $\bm{S}\!\ket{v} = \ket{v}$ is used. Besides, $\|\!\ket{u}\!\|_2 = \|\!\ket{v}\!\|_2 = \sqrt{d_S}$.
    Therefore, $\ket{u}\in\overline{\mathscr{V}}_{\rm Herm}$, meaning that $\mathfrak{reach}(\mathscr{U};\ket{v}) \subseteq \overline{\mathscr{V}}_{\rm Herm}$.

    On the other hand, consider a vector $\ket{\phi}\in\overline{\mathscr{V}}_{\rm Herm}$. If $\ket{\phi} = \ket{v}$, it is straightforward that $\ket{\phi}\in \mathfrak{reach}(\mathscr{U};\ket{v})$, so we only need to consider the case when $\ket{\phi} \neq \ket{v}$. Normalise $\ket{\phi}$ and $\ket{v}$ as $\ket{\tilde{\phi}}\equiv \ket{\phi}/\sqrt{d_S}$ and $\ket{\tilde{v}}\equiv \ket{v}/\sqrt{d_S}$, respectively.
    We consider the Householder reflection unitary $\bm{U}'\in\mathbb{C}^{d_S^2\times d_S^2}$ \cite{Householder1958}:
    \begin{align}
        \bm{U}' \equiv \mathds{1} - \frac{2(\ket{\tilde{\phi}}-\ket{\tilde{v}})(\bra{\tilde{\phi}} - \bra{\tilde{v}})}{\|\!\ket{\tilde{\phi}}-\ket{\tilde{v}}\!\|_2^2}.
        \label{eq:U_householder_reflection}
    \end{align}
    Since $\braket{\phi|v}$ is always real, we have $\bm{U'}\!\ket{v} = \ket{\phi}$. Moreover, $\bm{U}' = \bm{S}(\bm{U}')^*\bm{S}$, which follows from the fact that 
    $\bm{S}\!\ket{v}^* = \ket{v}$ and $\bm{S}\!\ket{\phi}^* = \ket{\phi}$.
    Therefore, any $\ket{\phi}\in \overline{\mathscr{V}}_{{\rm Herm}}$ is reachable via a symmetry-constrained $\bm{U'}$ acting on $\ket{v}$, i.e., $\overline{\mathscr{V}}_{\rm Herm}\subseteq \mathfrak{reach}(\mathscr{U};\ket{v})$. We finally conclude that $\mathfrak{reach}(\mathscr{U};\ket{v})=\overline{\mathscr{V}}_{\rm Herm}$.
\end{proof}
By Lem.~\ref{lem:reach_U_v}, we know that the reachable set $\mathfrak{reach}(\mathscr{U};\ket{v})$ is equal to $\overline{\mathscr{V}}_{\rm Herm}$, the set of vectorised Hermitian matrices in ${\rm Herm}^{d_S\times d_S}$ with a norm constraint. Let $\mathscr{V}_{\rm Herm}$ be the space of vectorised Hermitian matrices without the norm constraint.
The condition Eq.~(\ref{eq:ker_reach_non_empty}) is equivalent to
\begin{align}
    {\rm dim}(\mathfrak{ker}(\bm{M})\cap \mathscr{V}_{\rm Herm}) > 0,
    \label{eq:dim(ker(M)_cap_V_Herm)_>_0}
\end{align}
since once there exists a non-trivial intersection (i.e., non-zero vector) $\ket{\phi}\in \mathfrak{ker}(\bm{M})\cap \mathscr{V}_{\rm Herm}$, its normalised counterpart $\ket{\phi'} \equiv \sqrt{d_S}\ket{\phi}/\|\!\ket{\phi}\!\|_2 \in \mathfrak{ker}(\bm{M})\cap\mathfrak{reach}(\mathscr{U};\ket{v})$, and if Eq.~(\ref{eq:ker_reach_non_empty}) holds, the common element of $\mathfrak{ker}(\bm{M})$ and $\mathfrak{reach}(\mathscr{U};\ket{v})$ cannot be the zero vector due to the norm constraint in $\mathfrak{reach}(\mathscr{U};\ket{v})$, meaning that Eq.~(\ref{eq:ker_reach_non_empty}) holds if and only if ${\rm dim}(\mathfrak{ker}(\bm{M})\cap \mathscr{V}_{\rm Herm}) > 0$. 

To show that this is indeed the case, we now consider the structure of $\bm{M}$. Recall that $\bm{\zeta}^{1/2}$ satisfies the same symmetry as $\bm{\zeta}$ [Eq.~(\ref{eq:zeta_sqrt_reshuffled_symmetry})], i.e., $\zeta_{jj'mm'} = \zeta_{j'jm'm}^*$ for all $(j,j',m,m')$. Since $\bm{M}$ consists of $N_0$ rows picked from $\bm{\zeta}^{1/2}$ [Eq.~(\ref{eq:M_def})], we have
\begin{align}
    M_{jj'mm'} + M_{j'jmm'} &= M_{jj'm'm}^* + M_{j'jm'm}^*,
    \label{eq:M_jj'+M_j'j}\\
    \nI(M_{jj'mm'} - M_{j'jmm'}) &= -\nI(M_{jj'm'm}^* - M_{j'jm'm}^*),
    \label{eq:i(M_jj'-M_j'j)}
\end{align}
for all $(j,j')\in \{(\ell,\ell')|\epsilon_\ell = \epsilon_{\ell'}\}$ and $(m,m')$. For $(j,j')\in \{(\ell,\ell')|\epsilon_\ell = \epsilon_{\ell'}\}$, define
\begin{align}
    \ket{\bm{H}_{jj'}^+} &\equiv \sum_{m,m'=0}^{d_S-1} \frac{M_{jj'mm'}^* + M_{j'jmm'}^*}{2} \ket{m,m'},\\
    \ket{\bm{H}_{jj'}^-} &\equiv \sum_{m,m'=0}^{d_S-1} \frac{-\nI(M_{jj'mm'}^* - M_{j'jmm'}^*)}{2} \ket{m,m'}.
\end{align}
By Eqs.~(\ref{eq:M_jj'+M_j'j}) and (\ref{eq:i(M_jj'-M_j'j)}), $\ket{\bm{H}_{jj'}^+}$ and $\ket{\bm{H}_{jj'}^-}$ are vectorised Hermitian matrices, i.e., $\ket{\bm{H}_{jj'}^{\pm}}\in\mathscr{V}_{\rm Herm}$. The matrix $\bm{M}$ can thus be written as
\begin{align}
    \bm{M} &= \sum_{j,j'=0|\epsilon_j=\epsilon_{j'}}^{d_S-1}\sum_{m,m'=0}^{d_S-1}M_{jj'mm'}\ketbra{j,j'}{m,m'}\\
    &= \sum_{j=0}^{d_S-1}\ketbra{j,j}{\bm{H}_{jj}^+} \nonumber\\
    &\quad + \sum_{j,j'=0|\epsilon_j=\epsilon_{j'},\, j>j'} (\ket{j,j'} + \ket{j',j})\!\bra{\bm{H}_{jj'}^+} \nonumber\\
    &\quad - \sum_{j,j'=0|\epsilon_j=\epsilon_{j'},\, j>j'} \nI(\ket{j,j'} - \ket{j',j})\!\bra{\bm{H}_{jj'}^-}.
\end{align}
Since 
\begin{align}
    &{\rm dim}\left[{\rm span}\left(\{\ket{\bm{H}_{jj}}\}_{j=0}^{d_S-1}\cup\{\ket{\bm{H}_{jj'}^+}, \ket{\bm{H}_{jj'}^-}\}_{j,j'=0\big|\substack{\epsilon_j=\epsilon_{j'}\\j>j'}}^{d_S-1}\right)\right] \nonumber\\
    &\le |\{(j,j')|\epsilon_j = \epsilon_{j'}\}| \equiv N_0,
    \label{eq:dim_span_H_in_M}
\end{align}
as long as $N_0 < d_S^2 \equiv {\rm dim}(\mathscr{V}_{\rm Herm})$, there always exists a non-trivial vector $\ket{\phi}\in \mathscr{V}_{\rm Herm}$ which is orthogonal to all $\ket{\bm{H}_{jj'}^\pm}$'s in $\bm{M}$, such that $\bm{M}\!\ket{\phi} = 0$, i.e., $\ket{\phi}\in \mathfrak{ker}(\bm{M})\cap\mathscr{V}_{\rm Herm}$, and thus, Eq.~(\ref{eq:dim(ker(M)_cap_V_Herm)_>_0}), ${\rm dim}(\mathfrak{ker}(\bm{M})\cap \mathscr{V}_{\rm Herm}) > 0$, is satisfied.
In conclusion, we have
\begin{align*}
    N_0 < d_S^2 &\Rightarrow \text{Eq.~(\ref{eq:dim(ker(M)_cap_V_Herm)_>_0})} \Leftrightarrow \text{Eq.~(\ref{eq:ker_reach_non_empty})} \\
    &\Leftrightarrow \text{Eq.~(\ref{eq:M_U_v_=_0})} \Leftrightarrow \text{Eq.~(\ref{eq:condition_on_W_for_LS})}\\
    &\Leftrightarrow {\rm Tr}_E\{V_{SE}\gamma_E\} = 0.
\end{align*}
We note that when $N_0 = d_S^2$ such that $\bm{M} = \bm{\zeta}^{1/2}$, the dimension of the span in Eq.~(\ref{eq:dim_span_H_in_M}) can still be strictly less than $N_0$, and thus, $d_S^2$, if $\bm{\zeta}$, and therefore $\bm{M}$, is not full rank. 
Therefore, the Lamb shift in the constructed collision model can always vanish under an appropriate gauge, except the special case where the system has fully degenerate energies ($N_0=d_S^2$) and the coefficient matrix $\bm{\zeta}$ [Eq.~(\ref{eq:zeta})] is full-rank, for which its ECTCM realisation is unknown. 

\section{Proof of Prop.~\ref{prop:MTO_=_CM}}\label{app:proof_of_MTO_=_CM}
In this appendix, we prove that the set of MTO generators $\mathcal{L}_t^{\rm (MTO)}$ coincides with the set of ECTCM generators $\mathcal{L}_t^{\rm (CM)}$. Specifically, we show that every ECTCM generator generates an MTO and, conversely, that every {\it bounded} MTO generator necessarily arises as the continuous-time limit of an ECTCM.

Before proceeding with the proof of Prop.~\ref{prop:MTO_=_CM}, we establish two auxillary technical lemmas. The first characterises when a GKSL dissipator vanishes identically. The second applies this characterisation to the second-order term in the expansion of an infinitesimal thermal operation and shows that, if this term vanishes, all other terms in the expansion except the zeroth-order vanish, and hence rendering the induced system dynamics trivial.
\begin{lemma}\label{lem:zero_dissipator}
    A GKSL dissipator $\mathcal{D}(\cdot) := \mathcal{A}(\cdot) - \frac{1}{2}[\mathcal{A}^\dagger(\mathds{1}),\cdot]_+$ where $\mathcal{A}$ is a CP map, is the zero map, i.e., $\mathcal{D}(\rho) = 0$ for all states $\rho$, if and only if $\mathcal{A}\propto \mathcal{I}$.
\end{lemma}
\begin{proof}
    The `if' part is immediate. We thereby focus on the `only if' direction. The condition that $\mathcal{D}$ is the zero map is equivalent to an equality
    \begin{align}
        \mathcal{A}(\cdot) = \frac{1}{2}[\mathcal{A}^\dagger(\mathds{1}),\cdot]_+.
    \end{align}
    Let us denote $A\equiv \mathcal{A}^\dagger(\mathds{1})$. Since $\mathcal{A}$ is CP, so is the map $\Phi(\cdot):=[A,\cdot]_+$. This requires that the Choi operator of $\Phi$, denoted by $J_\Phi$, in the eigenbasis of $A$ with the corresponding eigenvalues $\{a_i\}_i$ is positive semidefinite, i.e.,
    \begin{align}
        J_\Phi \equiv \sum_{i,j} \Phi(\ketbra{i}{j})\otimes\ketbra{i}{j} = \sum_{i,j}(a_i+a_j)\ketbra{ii}{jj} \ge 0.
    \end{align}
    Now we restrict the Choi operator to the subspace spanned by two vectors $\{\ket{ii}, \ket{jj}\}$. The positivity of $J_{\Phi}$ implies that
    \begin{align}
        \begin{pmatrix}
            2a_i & a_i + a_j \\
            a_i + a_j & 2a_j
        \end{pmatrix} \ge 0.
    \end{align}
    Since the determinant of this matrix is $-(a_i-a_j)^2$, this matrix is positive semidefinite if and only if $a_i = a_j$. Therefore, we have 
    \begin{align*}
        J_\Phi \ge 0 \Leftrightarrow a_i = a_j,\,\forall\, i,j \Leftrightarrow A \propto \mathds{1}.
    \end{align*}
    Hence, $\mathcal{A}(\cdot) = \frac{1}{2}\Phi(\cdot) \propto [\mathds{1},\cdot]_+ = (\cdot)$, namely, $\mathcal{A}\propto\mathcal{I}$.
\end{proof}
\begin{lemma}\label{lem:zero_second_order}
    Consider an infinitesimal thermal operation given in the following form:
    \begin{align}
        \mathcal{T}(\cdot):={\rm Tr}_E\{\nE^{-\nI H_{SE}\Delta t}(\cdot\otimes\gamma_E)\nE^{\nI H_{SE}\Delta t}\},
    \end{align}
    where $[H_{SE}, H_S+H_E] = 0$ and $\gamma_E$ is the thermal state of $E$. Suppose that $\mathcal{T}$ admits the following expansion around $\Delta t=0$:
    \begin{align}
        \mathcal{T} = \sum_{m=0}^\infty \frac{(-\nI \Delta t)^m}{m!}\mathcal{T}^{(m)},
    \end{align}
    where $\mathcal{T}^{(m)}(\cdot):= {\rm Tr}_E\{\mathcal{L}_{H_{SE}}^m(\cdot\otimes\gamma_E)\}$ with $\mathcal{L}_{H_{SE}}(\cdot) := [H_{SE}, \cdot]$. If $\mathcal{T}^{(2)}$ is the zero map, then $\mathcal{T} = \mathcal{I}$.
\end{lemma}
\begin{proof}
    By direct calculation, we have $\mathcal{T}^{(2)} = -2\mathcal{D}$ where $\mathcal{D}(\cdot) := \mathcal{A}(\cdot) - \frac{1}{2}[\mathcal{A}^\dagger(\mathds{1}),\cdot]_+$ and
    \begin{align}
        \mathcal{A}(\cdot) := {\rm Tr}_E\{H_{SE}(\cdot\otimes\gamma_E)H_{SE}\}.
    \end{align}
    If $\mathcal{T}^{(2)}$ is the zero map, so is $\mathcal{D}$. By Lem.~\ref{lem:zero_dissipator}, this implies that $\mathcal{A}\propto \mathcal{I}$. To see how this condition restricts $H_{SE}$, we write $\gamma_{E} = \sum_k p_{E}^{(k)}\ketbra{k}{k}_{E}$. The Kraus operators of $\mathcal{A}$ are 
    \begin{align}
        L_{kk'}\equiv \sqrt{p_{E}^{(k)}}\bra{k'}_{E_i}H_{SE}\ket{k}_{E}
    \end{align}
    for all $k$, $k'$. Since $\mathcal{A}\propto \mathcal{I}$, $L_{kk'}\propto \mathds{1}_S,\,\forall\,k,k'$. Given that $\gamma_{E}$ is full-rank, we have 
    \begin{align}
        \bra{k'}_{E}H_{SE}\ket{k}_{E} \propto \mathds{1}_S,\,\forall\,k,k',
    \end{align}
    which implies that $H_{SE} = \mathds{1}_S\otimes B_{E}$ where $B_E$ is an operator on $E$. Since $[H_{SE}, H_S+H_E] = 0$, we further have $[H_{SE}, H_E] = [B_E, H_E] = 0$, and therefore, 
    \begin{align}
        \mathcal{L}_{H_{SE}}(\rho_S\otimes\gamma_E) = \rho_S\otimes[B_E, \gamma_E] = 0,\,\forall\,\rho_S,
    \end{align}
    i.e., $\mathcal{L}_{H_{SE}}(\cdot\otimes\gamma_E)$ is the zero map. This yields that $\mathcal{T}^{(m)}(\cdot) := {\rm Tr}_E\{\mathcal{L}_{H_{SE}}^m(\cdot\otimes\gamma_E)\}$ vanishes identically for all positive integer $m$, which implies reduction of the whole map to identity, $\mathcal{T} = \mathcal{I}$.
\end{proof}
    We now turn to the proof of Prop.~\ref{prop:MTO_=_CM}. As we shall see, Lem.~\ref{lem:zero_second_order} helps single out ECTCMs as the only possible microscopic realisation of bounded MTO generators.

    Firstly, it is straightforward that $\mathcal{L}^{\rm (CM)}_t$ generates an MTO. The map $\overleftarrow{\rm T}\exp(\int_0^\tau \mathcal{L}_t^{\rm (CM)}{\rm d} t )$ arises as the continuous-time limit of the $n$-collision map $\mathcal{E}_n$ with $\tau \equiv n\Delta t$ as $\Delta t\rightarrow 0$. Moreover, $\mathcal{E}_n$ is a concatenation of $n$ collisions, each of which is a TO by virtue of the defining conditions of ECTCMs. Since TOs are closed under composition, $\mathcal{E}_n$ is itself a TO for every $n$. Consequently, its continuous-time limit is an MTO.

    To show that every bounded $\mathcal{L}_t^{\rm (MTO)}$ -- generating an MTO dynamics -- necessarily requires an ECTCM realisation, we note that for a generic MTO $\mathcal{T} = \overleftarrow{\rm T}\exp(\int_0^\tau \mathcal{L}_t^{\rm (MTO)}{\rm d} t )$, since $\mathcal{L}_t^{\rm (MTO)}$ is continuous, we can write it as a product integral~\citep[][p.7]{Dollard_Friedman_1984}. That is, we have
    \begin{align}
        \mathcal{T} = \overleftarrow{\rm T}\exp\left(\int_0^\tau \mathcal{L}_t^{\rm (MTO)}{\rm d} t \right) = \lim_{\Delta u_i\rightarrow 0,\,\forall\,i} \overleftarrow{\prod_{i=1}^n}\nE^{\mathcal{L}_{u_i}^{\rm (MTO)}\Delta u_i},
    \end{align}
    where $\{u_0,u_1,u_2,\cdots,u_n\}$ is a partition of $[0,\tau]$ with $\Delta u_i \equiv u_i - u_{i-1}$, and $\overleftarrow{\prod}_{m=1}^M \mathcal{X}_m := \mathcal{X}_M\mathcal{X}_{M-1}\cdots\mathcal{X}_1$ denotes the time ordered product (for simplicity, we omit the concatenation symbol $\circ$). The limit is independent of the specific choice of the partition. Choosing $\Delta u_i = \Delta t,\,\forall\, i$ and $\tau \equiv n \Delta t$, we have
    \begin{align}
        \mathcal{T} = \lim_{n\rightarrow \infty} \overleftarrow{\prod_{i=1}^n}\nE^{\mathcal{L}_{i\Delta t}^{\rm (MTO)}\Delta t}.
        \label{appeq:T_MTO_expression}
    \end{align}
    We will now show that each $\mathcal{L}_{i\Delta t}^{\rm (MTO)}$ can {\it only} arise from a collision in an ECTCM. Since for every $i$, the map $\mathcal{T}_i\equiv \nE^{\mathcal{L}_{i\Delta t}^{\rm (MTO)}\Delta t}$ is a TO by definition, we have
    \begin{align}
        \mathcal{T}_i(\cdot) = {\rm Tr}_{E_i}\{U_{SE_i}(\cdot\otimes\gamma_{E_i})U_{SE_i}^\dagger\}.
    \end{align}
    Without loss of generality, we write 
    \begin{align}
        U_{SE_i} = \nE^{-\nI \Delta t\,H_{SE_i}(\Delta t)}.
    \end{align}
    Since $[U_{SE_i}, H_S+H_{E_i}]=0$, $H_{SE_i}(\Delta t)$ can always be chosen to satisfy $[H_{SE_i}(\Delta t), H_S+H_{E_i}] = 0$, in which the environment Hamiltonian $H_{E_i}$ can be obtained from $\gamma_{E_i}$ and is bounded and independent of the discretisation parameter $\Delta t$.
    
    To cast the map $\mathcal{T}_i$ into a collision, 
    we further interpret the unitary $U_{SE_i}$ as the endpoint of a unitary trajectory $U_{SE_i}(s):= \nE^{-\nI H_{SE_i}(\Delta t)s}$ for $s\in[0,\Delta t]$, i.e., 
    \begin{align}
        \mathcal{T}_i(\cdot) &= \mathcal{T}_i(s)(\cdot)|_{s=\Delta t} \nonumber\\
        &:= {\rm Tr}_{E_i}\{\nE^{-\nI H_{SE_i}(\Delta t)s}(\cdot\otimes\gamma_{E_i}) \nE^{\nI H_{SE_i}(\Delta t)s}\}|_{s=\Delta t}.
    \end{align}
    Expanding $U_{SE_i}(s)$ around $s = 0$, we have
    \begin{align}
        \mathcal{T}_i(s) &= \sum_{m=0}^\infty \frac{(-\nI s)^m}{m!}\mathcal{T}_i^{(m)},
    \end{align}
    where $\mathcal{T}_i^{(m)}(\cdot):= {\rm Tr}_{E_i}\{\mathcal{L}_{H_{SE_i}(\Delta t)}^m(\cdot\otimes\gamma_{E_i})\}$. By setting $s=\Delta t$, we therefore obtain an expansion of $\mathcal{T}_i$:
    \begin{align}
        \mathcal{T}_i &= \sum_{m=0}^\infty \frac{(-\nI \Delta t)^m}{m!}\mathcal{T}_i^{(m)}.
        \label{appeq:T_i_expansion_w_T}
    \end{align}
    In particular,
    \begin{widetext}
        \begin{align}
        \mathcal{T}_i^{(0)}(\cdot) &= (\cdot)\\
        \mathcal{T}_i^{(1)}(\cdot) &= [{\rm Tr}_{E_i}\{H_{SE_i}(\Delta t)\gamma_{E_i}\}, \cdot]\\
        \mathcal{T}_i^{(2)}(\cdot) &= -2\left(\mathrm{Tr}_{E_i}\left\{H_{SE_i}(\Delta t)(\cdot\otimes\gamma_{E_i}) H_{SE_i}(\Delta t)\right\} -  {\rm Tr}_{E_i}\left\{\frac{1}{2}[(H_{SE_i}(\Delta t))^2, \cdot\otimes\gamma_{E_i}]_+\right\}\right).
    \end{align}
    \end{widetext}
    We hence write
    \begin{align}
        \mathcal{T}_i = \mathcal{I} + \Delta t\mathcal{T}_i^{(1)} + \frac{(\Delta t)^2}{2}\mathcal{T}_i^{(2)} + \tilde{O}(\Delta t),
        \label{appeq:T_i_expansion_w_H}
    \end{align}
    where $\tilde{O}(\Delta t)$ denotes the scaling of the remainder terms in the expansion, which depends on the scaling of $\|H_{SE_i}\|_\infty$ and will be discussed case by case below. On the other hand, recalling that $\mathcal{T}_i\equiv \nE^{\mathcal{L}_{i\Delta t}^{\rm (MTO)}\Delta t}$, we can alternatively write
    \begin{align}
        \mathcal{T}_i = \mathcal{I} + \Delta t \mathcal{L}_{i\Delta t}^{\rm (MTO)} + O\qty((\Delta t)^2).
        \label{appeq:T_i_expansion_w_L}
    \end{align}
    One may be concerned about the apparent $\Delta t$-dependence of $\mathcal{L}_{i\Delta t}^{\rm (MTO)}$. However, it is important to note that $i\Delta t$ should be regarded as $\lfloor \frac{t}{\Delta t} \rfloor \Delta t$, which is independent of the discretisation parameter $\Delta t$. 
    
    Comparing Eqs.~(\ref{appeq:T_i_expansion_w_L}) and (\ref{appeq:T_i_expansion_w_H}), we see that $\mathcal{L}_{i\Delta t}^{\rm (MTO)}$ collects all contributions that are the first order in $\Delta t$. We thereby consider four cases:
    
    (1) If $\|H_{SE_i}(\Delta t)\|_\infty$ remains bounded as $\Delta t \rightarrow 0$, 
    $\tilde{O}(\Delta t) = O((\Delta t)^3)$ and
    $\mathcal{L}_{i\Delta t}^{\rm (MTO)}$ contains only $\mathcal{T}_i^{(1)}$, namely, the Hamiltonian term generated by ${\rm Tr}_{E_i}\{H_{SE_i}(\Delta t)\gamma_{E_i}\}$, resulting in purely unitary dynamics.

    (2) If $\|H_{SE_i}(\Delta t)\|_\infty$ diverges as $\Delta t \rightarrow 0$ with the scaling $O(1/\sqrt{\Delta t})$, the second-order map $\mathcal{T}_i^{(2)}$ in Eq.~(\ref{appeq:T_i_expansion_w_H}) contributes at first order in $\Delta t$ and is therefore incorporated into $\mathcal{L}_{i\Delta t}^{\rm (MTO)}$, while the remainder terms scale as $\tilde{O}(\Delta t) = O((\Delta t)^{3/2})$ do not contribute to $\mathcal{L}_{i\Delta t}^{\rm (MTO)}$.
    Moreover, if $H_{SE_i}(\Delta t)$ admits a nonvanishing ${\rm Tr}_{E_i}\{H_{SE_i}(\Delta t)\gamma_{E_i}\}$, then it will lead to an unbounded first-order contribution $\mathcal{T}_i^{(1)}$ to the generator. Therefore, for $\mathcal{L}_{i\Delta t}^{\rm (MTO)}$ to remain bounded, ${\rm Tr}_{E_i}\{H_{SE_i}(\Delta t)\gamma_{E_i}\}$ must vanish.

    (3) If $\|H_{SE_i}(\Delta t)\|_\infty$ diverges more slowly than $O(1/\sqrt{\Delta t})$ as \mbox{$\Delta t \rightarrow 0$}, as discussed above, ${\rm Tr}_{E_i}\{H_{SE_i}(\Delta t)\gamma_{E_i}\}$ should still vanish in order for a bounded generator $\mathcal{L}_{i\Delta t}^{\rm (MTO)}$. Moreover, the second-order map $\mathcal{T}_i^{(2)}$ and the remainder terms $\tilde{O}(\Delta t)$ in Eq.~(\ref{appeq:T_i_expansion_w_H}) decay faster than $O(\Delta t)$, and therefore do not contribute to $\mathcal{L}_{i\Delta t}^{\rm (MTO)}$, implying that the generator is the zero map.

    (4) If $\|H_{SE_i}(\Delta t)\|_\infty$ diverges faster than $O(1/\sqrt{\Delta t})$ as $\Delta t \rightarrow 0$, as discussed above, ${\rm Tr}_{E_i}\{H_{SE_i}(\Delta t)\gamma_{E_i}\}$ should still vanish in order for a bounded generator $\mathcal{L}_{i\Delta t}^{\rm (MTO)}$. Moreover, the second-order map $\mathcal{T}_i^{(2)}$ in Eq.~(\ref{appeq:T_i_expansion_w_H}) decays slower than $O(\Delta t)$ and would therefore lead to a divergent contribution to $\mathcal{L}_{i\Delta t}^{\rm (MTO)}$. Consequently, $\mathcal{T}_i^{(2)}$ must vanish as well if $\mathcal{L}_{i\Delta t}^{\rm (MTO)}$ is to remain bounded. By Lem.~\ref{lem:zero_second_order}, this implies that all maps $\mathcal{T}_i^{(m)}$ with $m\ge 1$ in Eq.~(\ref{appeq:T_i_expansion_w_T}) vanish identically and the generator $\mathcal{L}_{i\Delta t}^{\rm (MTO)}$ is the zero map.

    Through the above discussion, we note that for a bounded, nonvanishing $\mathcal{L}_{i\Delta t}^{\rm (MTO)}$ to exist, $\|H_{SE_i}(\Delta t)\|_\infty$ can only scale as either $O(1)$ or $O(1/\sqrt{\Delta t})$ as $\Delta t \rightarrow 0$, yielding the following general form of $H_{SE_i}(\Delta t)$:
    \begin{align}
        H_{SE_i}(\Delta t) = \frac{1}{\sqrt{\Delta t}}V_{SE_i}(\Delta t) + D_{SE_i}(\Delta t),
    \end{align}
    where ${\rm Tr}_{E_i}\{V_{SE_i}(\Delta t)\gamma_{E_i}\} = 0$,
    $D_{SE_i}(\Delta t)$ and $V_{SE_i}(\Delta t)$ scale as $O(1)$ as $\Delta t \rightarrow 0$.  Since for a bounded $D_{SE_i}(\Delta t)$, only its thermal mean, i.e., $D_{S_i}(\Delta t) \equiv {\rm Tr}_{E_i}\{D_{SE_i}(\Delta t)\gamma_{E_i}\}$, contributes to the system dynamics, we can effectively replace $D_{SE_i}(\Delta t)$ by $D_{S_i}(\Delta t)$ without altering the resulting dynamics. Besides, without loss of generality, we can include the system Hamiltonian $H_S$ and the environmental Hamiltonian $H_{E_i}$ into $H_{SE_i}(\Delta t)$. We therefore obtain
    \begin{align}
        H_{SE_i}(\Delta t) = H_S + H_{E_i} + \frac{1}{\sqrt{\Delta t}}V_{SE_i}(\Delta t) + D_{S,i}(\Delta t).
    \end{align}
    Since the two terms $D_{S,i}(\Delta t)$ and $\frac{1}{\sqrt{\Delta t}}V_{SE_i}(\Delta t)$ have different scalings, $[H_{SE_i}(\Delta t), H_S+H_{E_i}] = 0$ implies that
    $[D_{S,i}(\Delta t), H_S] = 0$ and $[V_{SE_i}(\Delta t), H_S+H_{E_i}] = 0$.
    Furthermore, since the limit $\Delta t \rightarrow 0$ will be taken in the expression of the MTO $\mathcal{T}$ in consideration [Eq.~(\ref{appeq:T_MTO_expression})], we can replace $D_{S,i}(\Delta t)$ and $V_{SE_i}(\Delta t)$ by their respective limits: $D_{S,i}\equiv \lim_{\Delta t\rightarrow 0} D_{S,i}(\Delta t)$ and $V_{SE_i}\equiv \lim_{\Delta t\rightarrow 0}V_{SE_i}(\Delta t)$, which leads to 
    \begin{align}
        H_{SE_i}(\Delta t) =  H_S + H_{E_i} + \frac{1}{\sqrt{\Delta t}}V_{SE_i} + D_{S,i}.
    \end{align}
    This is a standard Hamiltonian in ECTCMs [see Eq.~(\ref{eq:ECTCM_H_SE})]. Following the above discussion, we remark that this system--environment model indeed satisfies the defining conditions of ECTCMs [Conditions~\ref{CM:1}--\ref{CM:4}]:
    \begin{align*}
        &\text{Condition~\ref{CM:1}: } [D_{S,i}, H_S]=0;\\
        &\text{Condition~\ref{CM:2}: } [V_{SE_i}, H_S+ H_{E_i}]=0;\\
        &\text{Condition~\ref{CM:3}: } \gamma_{E_i} \text{ is the thermal state of $E_i$};\\
        &\text{Condition~\ref{CM:4}: } {\rm Tr}_{E_i}\{V_{SE_i}\gamma_{E_i}\} = 0\text{ (Lem.~\ref{lem:CM_4})}.
    \end{align*}
    Therefore, each $\mathcal{L}_{i\Delta t}^{\rm (MTO)}$ arises in the continuous-time limit of a single collision in this ECTCM. Combining with the fact that every ECTCM generator $\mathcal{L}_t^{\rm (CM)}$ generates an MTO, we hence conclude that the set of MTO generators coincides exactly with the set of ECTCM generators.

\section{Relation to the known family of MTO generators in Ref.~\cite{vomEnde2023ExploringLimits}}\label{app:relation_to_known_MTOs}
In this appendix, we show that the class of generators characterised in Thm.~\ref{thm:equivalences} coincides with the set of MTO generators -- conjectured to provide a complete characterisation of \textit{all} MTO generators -- identified in Ref.~\cite{vomEnde2023ExploringLimits}.
\begin{lemma}[{\citep[][Thm.~2]{vomEnde2023ExploringLimits}}]\label{lem:MTO_generator}
    Let ${\rm TO}(H_S, \beta)$ denote the set of thermal operations on a system $S$ with Hamiltonian $H_S$ at inverse temperature $\beta$. If a Lindbladian $\mathcal{L}^{\rm (MTO)}$ is of the following form:
    \begin{align}
        \mathcal{L}^{\rm (MTO)} = -\nI\mathcal{L}_{H_S'} + \mathcal{D}^{\rm (MTO)},
        \label{eq:L_MTO}
    \end{align}
    where $H_S'$ is a system operator satisfying $[H_S', H_S]=0$ and 
    \begin{align}
        \mathcal{D}^{\rm (MTO)}(\cdot) &:= \sum_{k,k'}J_{kk'}(\cdot)J_{kk'}^\dagger - \frac{1}{2}\sum_{k,k'}[J_{kk'}^\dagger J_{kk'}, \cdot]_+,\\
        J_{kk'} &:= \nE^{-\beta e_k/2}\bra{k'}_E V_{SE} \ket{k}_E,
        \label{eq:J_kk'}
    \end{align}
    with $[V_{SE}, H_S+H_E]=0$ and $H_E\equiv \sum_k e_k\ketbra{k}{k}_E$, then $(\nE^{\mathcal{L}^{\rm (MTO)}t})_{t\ge 0}$ is a continuous one-parameter semigroup in the closure of ${\rm TO}(H_S, \beta)$.
\end{lemma}
 We can define a time-dependent Lindbladian $\mathcal{L}_t^{\rm (MTO)}$ such that $\mathcal{L}_t^{\rm (MTO)}$ is an MTO generator defined in Eq.~(\ref{eq:L_MTO}) for every time $t$. By Lem.~\ref{lem:MTO_generator}, the map $\mathcal{T}_\tau := \overleftarrow{\rm T}\exp\left(\int_0^\tau \mathcal{L}_t^{\rm (MTO)}{\rm d} t \right)$ is a TO for all times $\tau$ (potentially requiring an infinite environment for its realisation). However, since Lem.~\ref{lem:MTO_generator} provides only a partial characterisation of MTO generators, it is a priori unclear whether \textit{all} MTO dynamics are generated by (time-dependent) generators $\mathcal{L}^{(\rm{MTO})}_t$ satisfying Eqs.~(\ref{eq:L_MTO})-(\ref{eq:J_kk'}). Here, we answer this question affirmatively by demonstrating that the set of these generators $\mathcal{L}^{(\rm{MTO})}_t$ coincides exactly with the set of thermal generators $\mathcal{L}_t$ and $\mathcal{L}^{(\rm{CM})}_t$ arising from the axiomatic and microscopic approaches, respectively: 
\begin{proposition}
\label{prop:equiv_MTO_gen}
    For systems with non-trivial Hamiltonians, the following sets of generators of Markovian thermal dynamics coincide:
    \begin{enumerate}
        \item The set of axiomatic Lindbladians $\mathcal{L}_t$ defined by Axioms~\ref{axiom:I}--\ref{axiom:III}.
        \item The set of Lindbladians $\mathcal{L}_t^{\rm (CM)}$ obtained in the continuous-time limit of energy-conserving collision models defined by Conditions~\ref{CM:1}--\ref{CM:4}.
        \item The set of MTO generators $\mathcal{L}_t^{\rm (MTO)}$ defined in Lem.~\ref{lem:MTO_generator}.
    \end{enumerate}
\end{proposition}
\begin{proof}
    ($1\Leftrightarrow 2$) This equivalence has been established by Props.~\ref{prop:L_CM_=>_L} and ~\ref{prop:ME_CM_equivalence}.

    ($2\Rightarrow 3$) Comparing $\mathcal{L}_t^{\rm (MTO)}$ [Eq.~(\ref{eq:L_MTO})] with the master equation obtained from collision models [Eq.~(\ref{eq:CMME_gen})] or $\mathcal{L}_t^{\rm (CM)}$ [Eq.~(\ref{eq:L_CM_def})], we find that $\mathcal{L}_t^{\rm (MTO)}$ has the same structure as $\mathcal{L}_t^{\rm (CM)}$ except that the $J_{kk'} = \nE^{-\beta e_k/2}\bra{k'}_E V_{SE} \ket{k}_E$ in the definition [Eq.~\eqref{eq:J_kk'}] of $\mathcal{L}_t^{\rm (MTO)}$ may not allow for an interaction Hamiltonian $V_{SE}(t)$ that leads to a vanishing Lamb shift in the corresponding ECTCM. Therefore, $\mathcal{L}_t^{\rm (CM)}$ is always a generator of MTOs with additional constraints.
    
    ($3\Rightarrow 1$) We can follow the same lines as in the proof of Prop.~\ref{prop:L_CM_=>_L} (App.~\ref{app:proof_of_L_CM_=>_L}), when we show that $\mathcal{L}_t^{\rm (CM)}$ satisfies Axioms~\ref{axiom:I}--\ref{axiom:III}. A priori, $\mathcal{L}_t^{\rm (MTO)}$ is less constrained than $\mathcal{L}_t^{\rm (CM)}$ due to the absence of the vanishing Lamb shift condition. However, since this condition was not used in the corresponding part of the proof of Prop.~\ref{prop:L_CM_=>_L}, we can directly employ the same reasoning to show that $\mathcal{L}_t^{(\text{MTO})}$ is a thermal Lindbladian $\mathcal{L}_t$.
\end{proof}
By Prop.~\ref{prop:MTO_=_CM}, the class of $\mathcal{L}_t^{\rm (CM)}$ coincides with the class of MTO generators, Prop.~\ref{prop:equiv_MTO_gen} therefore implies that the set of MTO generators in Lem.~\ref{lem:MTO_generator} exhausts the entire set of MTO generators. Consequently, we establish the conjecture in Ref.~\cite{vomEnde2023ExploringLimits} that this identified family is complete (up to closure), at least under the assumptions considered here (namely, non-trivial system Hamiltonians and continuous time-dependent generators).

\section{Full version of Protocol~\ref{tab:protocol}}\label{app:full_protocol}
In this appendix, we present the full version of the ECTCM simulation protocol~\ref{tab:protocol} in Protocol~\ref{tab:protocol_full}.
{
\renewcommand{\tablename}{Protocol}
\renewcommand{\thetable}{\thesection.\Roman{table}}
\setcounter{table}{0}
\begin{table*}[t]
\caption{\bf Explicit construction of the energy-conserving thermal collision model (ECTCM).}
\centering
\begin{tabular}{c p{0.95\linewidth}}
\hline
\hline
\noalign{\vskip 2mm}
{\bf 1.} & {\bf Extract coefficients from the given thermal master equation under Axioms~\ref{axiom:I}--\ref{axiom:III}.} \\[2mm]
(a) & Calculate the coefficient $\xi_{jj'mm'}(t)$ from the CP map $\mathcal{A}_t$ in the target Lindbladian $\mathcal{L}_t$ [Eq.~(\ref{eq:xi_def_from_A})]. \\[2mm]
(b) & Construct the coefficient $\zeta_{jj'mm'}(t)$ from $\xi_{jj'mm'}(t)$ [Eq.~(\ref{eq:zeta})]. \\[2mm]
(c) & Rewrite $\zeta_{jj'mm'}(t)$ as a matrix $\bm{\zeta}(t) \in\mathbb{C}^{d_S^2\times d_S^2}$ by treating $(j,j')$ ($(m,m')$) as row (column) indices, and obtain the square root $\bm{\zeta}^{1/2}(t)\in\mathbb{C}^{d_S^2\times d_S^2}$. \\
&~\\
{\bf 2.} & {\bf Construct the coefficient matrix $\bm{W}$ by fixing the gauge.} \\[2mm]
(a) & Calculate the degeneracy coefficient $N_0 \equiv |\{(j,j')|\epsilon_j = \epsilon_{j'}\}|$ and pick the corresponding $N_0$ rows from $\bm{\zeta}^{1/2}(t)$ to form $\bm{M}(t)\in\mathbb{C}^{N_0\times d_S^2}$ [Eq.~(\ref{eq:M_def})]. \\[2mm]
(b) & Find the eigenbasis of the kernel of $\bm{M}(t)$, denoted as $\{\ket{e_\alpha(t)}\}_{\alpha=0}^{{\rm dim}(\mathfrak{ker}(\bm{M}(t)))-1}$. \\[2mm]
(c) & Iterate over $\alpha$ to find $\alpha'$ such that $\ket{e_{\alpha'}(t)} + \bm{S}\ket{e_{\alpha'}(t)}^* \neq 0$ where $\bm{S}\in\mathbb{C}^{d_S^2\times d_S^2}$ is the swap matrix, and denote $\ket{\phi(t)} \equiv \ket{e_{\alpha'}(t)} + \bm{S}\ket{e_{\alpha'}(t)}^*$.\\[2mm]
(d) & Define $\ket{v} \equiv \sum_{j=0}^{d_S-1}\ket{j,j}$ and construct the Householder reflection unitary $\bm{U}(t)\in\mathbb{C}^{d_S^2\times d_S^2}$ between (normalised) $\ket{\phi(t)}$ and $\ket{v}$ [Eq.~(\ref{eq:U_householder_reflection})]. \\[2mm]
(e) & Construct the matrix $\bm{W}(t) \equiv \bm{\zeta}^{1/2}(t)\bm{U}(t) \in \mathbb{C}^{d_S^2\times d_S^2}$. \\
&~\\
{\bf 3.} & {\bf Construct the collision model.} \\[2mm]
(a) & Set the collision time $\Delta t$ and the number of collision $n\equiv \tau/\Delta t$ where $\tau$ is the evolution time. \\[2mm]
(b) & Sample the effective Hamiltonian $H_S'(t)$ as $\{H_S'(i\Delta t)\}_{i=1}^n$, the CP map $\mathcal{A}_t$ as $\{\mathcal{A}_{i\Delta t}\}_{i=1}^n,$ and follow Steps 1 and 2 to construct the corresponding $\bm{W}(i\Delta t)$ for $i=1,2,\dots, n$. \\[2mm]
(c) & Construct the driving $D_{S,i}\equiv H_S'(i\Delta t) - H_S$. \\[2mm]
(d) & Construct the same Hamiltonian $H_E$ for all ancillae [Eq.~(\ref{eq:H_E_constructed})]. \\[2mm]
(e) & Construct the interaction Hamiltonian $gV_{SE_i}$ from $\bm{W}(i\Delta t)$ [Eq.~(\ref{eq:V_constructed})] with constant $g$ satisfying $g^2\Delta t = 1$. \\[2mm]
(f) & Implement the collision map $\mathcal{E}_n$ [Eq.~(\ref{eq:collision_map_E_n})].\\[2mm]
\hline
\hline
\end{tabular}
\label{tab:protocol_full}
\end{table*}
}
We remark that in Step~2(c) of Protocol~\ref{tab:protocol_full}, the vector $\ket{\phi(t)}$ is indeed in the intersection $\mathfrak{ker}(\bm{M}(t))\cap \mathscr{V}_{\rm Herm}$. In particular, $\ket{\phi(t)}\in \mathscr{V}_{\rm Herm}$ is guaranteed by definition.
We show that $\ket{\phi(t)}\in \mathfrak{ker}(\bm{M}(t))$ as follows (ignoring the time dependence)
\begin{align}
    {\bm M}\ket{\phi} &= {\bm M}{\bm S}\ket{e_{\alpha'}}^*\\
    &= \sum_{j,j'=0|\epsilon_j=\epsilon_{j'}}^{d_S-1}\sum_{m,m'=0}^{d_S-1} M_{jj'mm'}(e_{\alpha'})_{m'm}^*\ket{j,j'}\\
    &= \sum_{j,j'=0|\epsilon_j=\epsilon_{j'}}^{d_S-1}\sum_{m,m'=0}^{d_S-1} \left[M_{j'jm'm}(e_{\alpha'})_{m'm}\right]^*\ket{j,j'}\\
    &= 0,
\end{align}
where in the third line, we used the symmetry: $M_{jj'mm'}=M_{j'jm'm}^*$ and the last line follows from the fact that $\sum_{m,m'=0}^{d_S-1}M_{j'jm'm}(e_{\alpha'})_{m'm} = 0$ for all $(j',j)\in \{(\ell,\ell')|\epsilon_\ell = \epsilon_{\ell'}\}$ since ${\bm M}\ket{e_{\alpha'}} = 0$.

\section{Technical details of the error analysis for Protocol~\ref{tab:protocol}}\label{app:error_analysis}
In the main text, the overall error $\varepsilon$ of Protocol~\ref{tab:protocol} is given by [Eq.~(\ref{eq:varepsilon_bound_gen})]
\begin{align}
    \varepsilon &\le n \max_{1\le i\le n}\varepsilon_{\rm s}(i) + n\max_{1\le i \le n}\left[\varepsilon_{\rm t}(i) + \varepsilon_{\rm c}(i)\right].
\end{align}
In this appendix, we first derive this decomposition and then provide an explicit bound on $\varepsilon$ by bounding the sampling error $\varepsilon_{\rm s}(i)$, the truncation error $\varepsilon_{\rm t}(i)$ and the collision error $\varepsilon_{\rm c}(i)$, respectively.
\subsection{Useful inequalities}\label{app:useful_ineqs}
First, we present several useful inequalities and specify the assumptions. Following the error analysis of general collision models in Ref.~\cite{Cattaneo2021CMcansimulate},
we consider the trace norm and operator norm for an operator $X$~\citep[][p.~33]{watrous2018theory}:
\begin{align}
    \|X\|_1 &:= {\rm Tr}\{\sqrt{X^\dagger X}\}, \label{eq:trace_norm}\\
    \|X\|_\infty &:= \sup_{\|Y\|_1= 1}|{\rm Tr}\{Y^\dagger X\}|. \label{eq:operator_norm}
\end{align}
The $1\rightarrow 1$ norm for a channel $\Lambda$ is defined as~\citep[][p.~167]{watrous2018theory}
\begin{align}
    \|\Lambda\|_{1\rightarrow 1} := \sup_{\|X\|_1= 1} \|\Lambda(X)\|_1.
    \label{eq:1to1_norm}
\end{align}
In the following, we will use H\"older's inequality: 
\begin{align}
    \|XY\|_1 \le \|X\|_1\|Y\|_\infty,
    \label{eq:Holder_ineq}
\end{align}
for any operators $X$ and $Y$. In addition, the following two lemmas will be useful~\citep[][p.~167]{watrous2018theory}:
\begin{lemmaApp}\label{applem:triangle_1to1}
    The $1\rightarrow 1$ norm satisfies the triangle inequality, i.e., $\|\Lambda - \Omega\|_{1\rightarrow 1} \le \|\Lambda - \Phi\|_{1\rightarrow 1} + \|\Phi - \Omega\|_{1\rightarrow 1}$, for any linear maps $\Lambda$, $\Omega$, and $\Phi$. 
\end{lemmaApp}
\begin{proof}
    \begin{align}
        \|\Lambda - \Omega\|_{1\rightarrow 1} &= \sup_{\|X\|_1= 1} \|\Lambda(X) - \Omega(X)\|_1 \\
        &\le \sup_{\|X\|_1= 1} \left(\|\Lambda(X) - \Phi(X)\|_1 \right. \nonumber\\
        &\quad \left. + \|\Phi(X) - \Omega(X)\|_1\right)\\
        &\le \sup_{\|X\|_1= 1}\|\Lambda(X) - \Phi(X)\|_1 \nonumber\\
        &\quad + \sup_{\|X\|_1= 1}\|\Phi(X) - \Omega(X)\|_1\\
        &= \|\Lambda - \Phi\|_{1\rightarrow 1} + \|\Phi - \Omega\|_{1\rightarrow 1},
    \end{align}
    where the first inequality follows from the triangle inequality for trace norm.
\end{proof}
\begin{lemmaApp}[\cite{Kliesch2011DissipativeQCTtheorem}]\label{applem:subadditive_1to1}
    For any linear maps $\Lambda_1$, $\Lambda_2$, $\Omega_1$, and $\Omega_2$, the inequality holds:
    $\|\Lambda_2 \circ \Lambda_1 - \Omega_2 \circ \Omega_1\|_{1\rightarrow 1} \le \|\Lambda_2\|_{1\rightarrow 1}\|\Lambda_1 - \Omega_1\|_{1 \rightarrow 1} + \|\Lambda_2 - \Omega_2\|_{1 \rightarrow 1}\|\Omega_1\|_{1\rightarrow 1}$.
\end{lemmaApp}
\begin{proof}
    \begin{align}
        &\quad \|\Lambda_2 \circ \Lambda_1 - \Omega_2 \circ \Omega_1\|_{1\rightarrow 1} \nonumber \\
        &\le \|\Lambda_2 \circ \Lambda_1 - \Lambda_2\circ\Omega_1\|_{1\rightarrow 1} + \|\Lambda_2 \circ \Omega_1 - \Omega_2 \circ \Omega_1\|_{1\rightarrow 1}\\
        &= \sup_{\|X\|_1= 1}\|\Lambda_2\circ\Lambda_1(X) - \Lambda_2\circ\Omega_1(X)\|_{1}\nonumber\\
        &\quad + \sup_{\|X\|_1= 1}\|\Lambda_2 \circ \Omega_1(X) - \Omega_2 \circ \Omega_1(X)\|_{1}\\
        &= \sup_{\|X\|_1= 1}\Big\|\Lambda_2\left(\frac{\Lambda_1(X)-\Lambda_2(X)}{\|\Lambda_1-\Lambda_2\|_{1\rightarrow 1}}\right)\Big\|_1\cdot \|\Lambda_1-\Lambda_2\|_{1\rightarrow 1}\nonumber\\
        &\quad + \sup_{\|X\|_1= 1}\Big\|(\Lambda_2-\Omega_2)\left(\frac{\Omega_1(X)}{\|\Omega_1\|_{1\rightarrow 1}}\right)\Big\|_1\cdot \|\Omega_1\|_{1\rightarrow 1}\\
        &\le \sup_{\|X\|_1= 1}\|\Lambda_2(X)\|_1\cdot \|\Lambda_1-\Lambda_2\|_{1\rightarrow 1}\nonumber\\
        &\quad + \sup_{\|X\|_1= 1}\|(\Lambda_2-\Omega_2)(X)\|_1\cdot \|\Omega_1\|_{1\rightarrow 1}\\
        &= \|\Lambda_2\|_{1\rightarrow 1}\|\Lambda_1 - \Omega_1\|_{1 \rightarrow 1} + \|\Lambda_2 - \Omega_2\|_{1 \rightarrow 1}\|\Omega_1\|_{1\rightarrow 1},
    \end{align}
    where the first inequality follows from Lem.~\ref{applem:triangle_1to1} and the second inequality is obtained by noticing that $\sup_{\|X\|_1\le 1}\|\Phi(X)\|_{1} = \sup_{\|X\|_1= 1}\|\Phi(X)\|_{1}$ for any linear map $\Phi$, since if the optimal $X^*$ for the first optimisation has $\|X^*\|_1 < 1$, the operator $\tilde{X}^*\equiv X^*/\|X^*\|_1$ will lead to a greater value, i.e., $\|\Phi(\tilde{X}^*)\|_{1} > \|\Phi(X^*)\|_{1}$, contradicting to the optimality of $X^*$.
\end{proof}
To exclude (unphysical) cases where $\mathcal{L}_t$ varies too dramatically in time, we assume that 
$\mathcal{L}_t$ is Lipschitz continuous~\cite{royden1988real} throughout. That is, there exists a constant $C_{\rm L}>0$ such that for any two times $r, u\ge 0$, the following bound holds:
\begin{align}
    \|\mathcal{L}_r-\mathcal{L}_u\|_{1\rightarrow 1} \le C_{\rm L} |r-u|.
    \label{eq:Lipschitz_continuity}
\end{align}
We note that Lipschitz continuity implies continuity of the superoperator-valued function $\mathcal{L}: t \mapsto \mathcal{L}_t$.

\subsection{Derivation of Eq.~(\ref{eq:varepsilon_bound_gen})}
The relevant error (distance) is
\begin{align}
    \varepsilon := \|\mathcal{E}_n - \overleftarrow{\rm T}\nE^{\int_0^\tau \mathcal{L}_t{\rm d} t}\|_{1\rightarrow 1}.
\end{align}
Intuitively, this error can be decomposed into two contributions:
\begin{itemize}
    \item[(i)] the error arising from the stroboscopic discretisation of the continuous-time evolution, and
    \item[(ii)] the error arising from approximating each infinitesimal evolution by a finite-duration collision.
\end{itemize}
To quantify the first contribution (i), we use the shorthand notation $\mathcal{L}_i := \mathcal{L}_{i\Delta t}$ where $i=1,2,\dots,n$ and define the product integral approximation error~\cite{Dollard1977ProductIntegrals}:
\begin{align}
    \varepsilon_{\rm pia} := \|\nE^{\mathcal{L}_n\Delta t}\circ\cdots\circ \nE^{\mathcal{L}_1\Delta t} - \overleftarrow{\rm T}\nE^{\int_0^\tau \mathcal{L}_t{\rm d} t}\|_{1\rightarrow 1},
\end{align}
together with its corresponding sampling error at $i$th step:
\begin{align}
    \varepsilon_{\rm s}(i) := \|\nE^{\mathcal{L}_i\Delta t} - \overleftarrow{\rm T}\nE^{\int_{(i-1)\Delta t}^{i\Delta t} \mathcal{L}_t{\rm d} t}\|_{1\rightarrow 1}.
\end{align}

The second contribution (ii) is quantified by the discrete-step error:
\begin{align}
    \varepsilon_{\rm ds} := \|\hat{\mathcal{E}}_n \circ \cdots \circ \hat{\mathcal{E}}_1 - \nE^{\mathcal{L}_n\Delta t}\circ\cdots\circ \nE^{\mathcal{L}_1\Delta t}\|_{1\rightarrow 1},
\end{align}
where $\hat{\mathcal{E}}_i$ is the $i$th collision channel defined as $\hat{\mathcal{E}}_i(\cdot) := {\rm Tr}_{E_i}\{U_i(\cdot\otimes\rho_{E_i})U_i^\dagger\}$, such that $\mathcal{E}_n = \hat{\mathcal{E}}_n \circ \cdots \circ \hat{\mathcal{E}}_1$.
To analyse $\varepsilon_{\rm ds}$, we further introduce the single-step error at the $i$th~step:
\begin{align}
    \varepsilon_{\rm ss}(i) := \|\hat{\mathcal{E}}_{i} - \nE^{\mathcal{L}_i\Delta t}\|_{1\rightarrow 1},
\end{align}
which is consists of the Taylor truncation error at $i$th step:
\begin{align}
    \varepsilon_{\rm t}(i) := \|(\mathcal{I}+\mathcal{L}_i\Delta t)-\nE^{\mathcal{L}_i\Delta t}\|_{1\rightarrow 1},
\end{align}
and the collision error at $i$th step:
\begin{align}
    \varepsilon_{\rm c}(i) := \|\hat{\mathcal{E}}_i - (\mathcal{I}+\mathcal{L}_i\Delta t)\|_{1\rightarrow 1}.
\end{align}
Combining the above definitions, the total simulation error admits the following decomposition:
\begin{align}
    \varepsilon &\le \varepsilon_{\rm pia} + \varepsilon_{\rm ds}\\
    &\le n \max_{1\le i\le n}\varepsilon_{\rm s}(i) + n \max_{1\le i \le n}\varepsilon_{\rm ss}(i)\\
    &\le n \max_{1\le i\le n}\varepsilon_{\rm s}(i) + n\max_{1\le i \le n}\left[\varepsilon_{\rm t}(i) + \varepsilon_{\rm c}(i)\right],
\end{align}
where the first and third inequalities follow from Lem.~(\ref{applem:triangle_1to1}), while the second inequality comes from Lem.~(\ref{applem:subadditive_1to1}), with the fact that 
the $1\rightarrow 1$ norm of completely positive and trace-preserving (CPTP) maps is 1 and
the time-ordered exponential $\overleftarrow{\rm T}\nE^{\int_{t'}^t \mathcal{L}_s{\rm d} s}$ is multiplicative, i.e., $\overleftarrow{\rm T}\nE^{\int_{t_1}^{t_3} \mathcal{L}_s{\rm d} s} = \overleftarrow{\rm T}\nE^{\int_{t_2}^{t_3} \mathcal{L}_s{\rm d} s}\circ \overleftarrow{\rm T}\nE^{\int_{t_1}^{t_2} \mathcal{L}_s{\rm d} s},\,\forall\, t_1\le t_2\le t_3$~\citep[][p.~11]{Dollard_Friedman_1984}. 

To obtain a bound on $\varepsilon$, we now bound the sampling error $\varepsilon_{\rm s}(i)$, the truncation error $\varepsilon_{\rm t}(i)$ and the collision error $\varepsilon_{\rm c}(i)$ individually.

\subsection{Sampling error}
In order to bound the $i$th step sampling error $\varepsilon_{\rm s}(i) := \|\nE^{\mathcal{L}_i\Delta t} - \overleftarrow{\rm T}\nE^{\int_{(i-1)\Delta t}^{i\Delta t} \mathcal{L}_t{\rm d} t}\|_{1\rightarrow 1}$ where $\mathcal{L}_i := \mathcal{L}_{i\Delta t}$, we define the average Lindbladian $\bar{\mathcal{L}}_i$ as follows:
\begin{align}
    \bar{\mathcal{L}}_i := \frac{1}{\Delta t}\int_{(i-1)\Delta t}^{i\Delta t}\mathcal{L}_t\, {\rm d}t.
\end{align}
By Lem.~\ref{applem:triangle_1to1}, the $i$th step sampling error $\varepsilon_{\rm s}(i)$ can be bounded as follows:
\begin{align}
    \varepsilon_{\rm s}(i) &\le \|\nE^{\mathcal{L}_i\Delta t} - \nE^{\bar{\mathcal{L}}_i\Delta t}\|_{1\rightarrow 1} \nonumber\\
    &\quad + \|\nE^{\bar{\mathcal{L}}_i\Delta t} - \overleftarrow{\rm T}\nE^{\int_{(i-1)\Delta t}^{i\Delta t} \mathcal{L}_t{\rm d} t}\|_{1\rightarrow 1}.
    \label{eq:decomp_varepsilon_s(k)}
\end{align}
To further bound both terms in Eq.~(\ref{eq:decomp_varepsilon_s(k)}), we introduce additional notation. First, since $\mathcal{L}_t$ is continuous, we rewrite $\overleftarrow{\rm T}\nE^{\int_{s}^{t} \mathcal{L}_u{\rm d} u}$ as a product integral~\citep[][p.7]{Dollard_Friedman_1984}, such that
\begin{align}
\label{eqn::propagator}
    \overleftarrow{\rm T}\nE^{\int_{s}^{t} \mathcal{L}_u{\rm d} u} = \lim_{\Delta u_k\rightarrow 0,\,\forall\,k} \overleftarrow{\prod_{k=1}^n}\nE^{\mathcal{L}_{u_k}\Delta u_k} =: \mathcal{P}_\mathcal{L}(t,s),
\end{align}
where $\{u_0,u_1,u_2,\cdots,u_n\}$ is a partition of $[s,t]$ with $\Delta u_k \equiv u_k - u_{k-1}$, and $\overleftarrow{\prod}_{m=1}^M \mathcal{X}_m := \mathcal{X}_M\mathcal{X}_{M-1}\cdots\mathcal{X}_1$ denotes the time ordered product (for simplicity, we ignore the concatenation symbol $\circ$). We further denote the product integral as a propagator $\mathcal{P}_\mathcal{L}(t,s)$ from time $s$ to time $t$ generated by $\mathcal{L}_u$. The derivatives of $\mathcal{P}_\mathcal{L}(t,s)$ are given by~\citep[][p.12]{Dollard_Friedman_1984}:
\begin{align}
    \frac{\partial}{\partial t}\mathcal{P}_\mathcal{L}(t,s) &= \mathcal{L}_t\mathcal{P}_\mathcal{L}(t,s),\\
    \frac{\partial}{\partial s}\mathcal{P}_\mathcal{L}(t,s) &= -\mathcal{P}_\mathcal{L}(t,s)\mathcal{L}_s.
\end{align}
As a propagator, $\mathcal{P}_\mathcal{L}(t,s)$ has a unique inverse given by~\citep[][suppl. mat.]{Kliesch2011DissipativeQCTtheorem}
\begin{align}
    \mathcal{P}^{-1}_\mathcal{L}(t,s) &:= \lim_{\Delta u_k\rightarrow 0,\,\forall\,k}\overleftarrow{\prod_{k=n}^1}\nE^{-\mathcal{L}_{u_k}\Delta u_k}\\
    &= \left(\mathcal{P}_\mathcal{L}(t,s)\right)^{-1},
\end{align}
such that $\mathcal{P}^{-1}_\mathcal{L}(t,s)\mathcal{P}_\mathcal{L}(t,s) = \mathcal{P}_\mathcal{L}(t,s)\mathcal{P}^{-1}_\mathcal{L}(t,s) = \mathcal{I}$ where $\mathcal{I}$ is the identity map. Accordingly, the derivatives of $\mathcal{P}^{-1}_\mathcal{L}(t,s)$ are given by
\begin{align}
    \frac{\partial}{\partial t}\mathcal{P}^{-1}_\mathcal{L}(t,s) &= -\mathcal{P}^{-1}_\mathcal{L}(t,s)\mathcal{L}_t,\\
    \frac{\partial}{\partial s}\mathcal{P}^{-1}_\mathcal{L}(t,s) &= \mathcal{L}_t\mathcal{P}^{-1}_\mathcal{L}(t,s).
\end{align}
Now, we are ready to derive error bounds for the sampling error. In this derivation, the following identity will be useful:
\begin{align}
    &\quad \mathcal{P}_{\mathcal{L}'}(t,s) - \mathcal{P}_{\mathcal{L}}(t,s) \nonumber\\
    &= \mathcal{P}_\mathcal{L}(t,s)\left[\mathcal{P}_{\mathcal{L}}^{-1}(t,s)\mathcal{P}_{\mathcal{L}'}(t,s)-\mathcal{I}\right]\\
    &= \mathcal{P}_\mathcal{L}(t,s)\int_s^t \frac{\partial}{\partial u}\left[\mathcal{P}_{\mathcal{L}}^{-1}(u,s)\mathcal{P}_{\mathcal{L}'}(u,s)-\mathcal{I}\right] {\rm d}u\\
    &= \mathcal{P}_\mathcal{L}(t,s)\int_s^t \mathcal{P}_{\mathcal{L}}^{-1}(u,s)(\mathcal{L}'_u-\mathcal{L}_u)\mathcal{P}_{\mathcal{L}'}(u,s) {\rm d}u\\
    &= \int_s^t \mathcal{P}_{\mathcal{L}}(t,u)(\mathcal{L}'_u-\mathcal{L}_u)\mathcal{P}_{\mathcal{L}'}(u,s) {\rm d}u,
    \label{eq:propagator_diff}
\end{align}
Using Eq.~(\ref{eq:propagator_diff}), the first term in Eq.~(\ref{eq:decomp_varepsilon_s(k)}) can be bounded as
\begin{align}
    &\quad \|\nE^{\mathcal{L}_i\Delta t} - \nE^{\bar{\mathcal{L}}_i\Delta t}\|_{1\rightarrow 1} \nonumber\\
    &= \Big\|\int_{0}^{\Delta t} \nE^{\bar{\mathcal{L}}_i(\Delta t - s)}(\mathcal{L}_i-\bar{\mathcal{L}}_i)\nE^{\mathcal{L}_i s}{\rm d}s\Big\|_{1\rightarrow 1}\\
    &\le \int_{0}^{\Delta t}\|\nE^{\bar{\mathcal{L}}_i(\Delta t - s)}(\mathcal{L}_i-\bar{\mathcal{L}}_i)\nE^{\mathcal{L}_i s}\|_{1\rightarrow 1}\, {\rm d} s\\[-0.5ex]
    &\le \int_{0}^{\Delta t}\|\nE^{\bar{\mathcal{L}}_i(\Delta t - s)}\|_{1\rightarrow 1}\|\mathcal{L}_i-\bar{\mathcal{L}}_i\|_{1\rightarrow 1}\|\nE^{\mathcal{L}_i s}\|_{1\rightarrow 1}\, {\rm d} s\\
    &= \Delta t\|\mathcal{L}_i-\bar{\mathcal{L}}_i\|_{1\rightarrow 1},
    \label{eq:1st_term_varepsilon_s(k)}
\end{align}
where the first and second inequalities follow from Lems.~\ref{applem:triangle_1to1} and \ref{applem:subadditive_1to1}, respectively, and the last line uses the fact that $\|\nE^{\bar{\mathcal{L}}_i(\Delta t - s)}\|_{1\rightarrow 1} = \|\nE^{\mathcal{L}_i s}\|_{1\rightarrow 1} = 1$ since both $\nE^{\bar{\mathcal{L}}_i(\Delta t - s)}$ and $\nE^{\mathcal{L}_i s}$ are CPTP maps. Under the assumption of the Lipschitz continuity of $\mathcal{L}_t$ [Eq.~(\ref{eq:Lipschitz_continuity})], the above can be further bounded as
\begin{align}
    \|\nE^{\mathcal{L}_i\Delta t} - \nE^{\bar{\mathcal{L}}_i\Delta t}\|_{1\rightarrow 1} &\le
    \Delta t \|\mathcal{L}_i-\bar{\mathcal{L}}_i\|_{1\rightarrow 1}\\
    &= \Big\|\int_{(i-1)\Delta t}^{i(\Delta t)} (\mathcal{L}_i - \mathcal{L}_t) {\rm d}t \Big\|_{1\rightarrow 1}\\
    &\le \int_{(i-1)\Delta t}^{i(\Delta t)} \|\mathcal{L}_i - \mathcal{L}_t\|_{1\rightarrow 1} {\rm d}t\\
    &\le \int_{(i-1)\Delta t}^{i(\Delta t)} C_{\rm L}|i\Delta t - t|{\rm d}t\\
    &= \frac{C_{\rm L}}{2}(\Delta t)^2,
\end{align}
where we recall the shorthand notation: $\mathcal{L}_i\equiv \mathcal{L}_{i\Delta t}$, the first inequality follows Lem.~\ref{applem:triangle_1to1} and the second inequality uses Eq.~(\ref{eq:Lipschitz_continuity}).

For the second term in Eq.~(\ref{eq:decomp_varepsilon_s(k)}), we follow the similar lines as the proof of Thm.~7 in the Supplemental Material of Ref.~\cite{Kliesch2011DissipativeQCTtheorem} to derive an upper bound. Define $t_{i}\equiv i\Delta t$.
Using Eq.~(\ref{eq:propagator_diff}), we have
\begin{widetext}
    \begin{align}
    &\quad\|\nE^{\bar{\mathcal{L}}_i \Delta t} - \overleftarrow{\rm T}\nE^{\int_{t_{i-1}}^{t_i} \mathcal{L}_t{\rm d} t}\|_{1\rightarrow 1} \nonumber\\
    &= \Big\|\int_{t_{i-1}}^{t_i} \mathcal{P}_{\mathcal{L}}(t_i,u)(\bar{\mathcal{L}}_i-\mathcal{L}_u)\nE^{\bar{\mathcal{L}}_i (u-t_{i-1})} {\rm d}u \Big\|_{1\rightarrow 1}\\
    &= \frac{1}{\Delta t}\Big\|\int_{t_{i-1}}^{t_i}\int_{t_{i-1}}^{t_i}\mathcal{P}_{\mathcal{L}}(t_i,u)(\mathcal{L}_r-\mathcal{L}_u)\nE^{\bar{\mathcal{L}}_i (u-t_{i-1})} {\rm d}r {\rm d}u\Big\|_{1\rightarrow 1}\\
    &= \frac{1}{\Delta t}\Big\|\int_{t_{i-1}}^{t_i}\int_{t_{i-1}}^{t_i}\left[\mathcal{P}_{\mathcal{L}}(t_i,r)\mathcal{L}_u\nE^{\bar{\mathcal{L}}_i (r-t_{i-1})} - \mathcal{P}_{\mathcal{L}}(t_i,u)\mathcal{L}_u\nE^{\bar{\mathcal{L}}_i (u-t_{i-1})}\right]{\rm d}r {\rm d}u\Big\|_{1\rightarrow 1}\\
    &\le \frac{1}{\Delta t}\int_{t_{i-1}}^{t_i}\int_{t_{i-1}}^{t_i}\|\mathcal{L}_u\|_{1\rightarrow 1}\left(
    \|\mathcal{P}_{\mathcal{L}}(t_i,r)\|_{1\rightarrow 1}\|\nE^{\bar{\mathcal{L}}_i (r-t_{i-1})}-\nE^{\bar{\mathcal{L}}_i (u-t_{i-1})}\|_{1\rightarrow 1} \right.\nonumber\\
    &\left. \quad +\|\mathcal{P}_{\mathcal{L}}(t_i,r) - \mathcal{P}_{\mathcal{L}}(t_i,u)\|_{1\rightarrow 1}\|\nE^{\bar{\mathcal{L}}_i (u-t_{i-1})}\|_{1\rightarrow 1}\right) {\rm d}r {\rm d}u\\
    &= \frac{1}{\Delta t}\int_{t_{i-1}}^{t_i}\int_{t_{i-1}}^{t_i}\|\mathcal{L}_u\|_{1\rightarrow 1}\left(\|\nE^{\bar{\mathcal{L}}_i (r-t_{i-1})}-\nE^{\bar{\mathcal{L}}_i (u-t_{i-1})}\|_{1\rightarrow 1} + \|\mathcal{P}_{\mathcal{L}}(t_i,r) - \mathcal{P}_{\mathcal{L}}(t_i,u)\|_{1\rightarrow 1}\right){\rm d}r {\rm d}u,
\end{align}
where the inequality follows from Lems.~\ref{applem:triangle_1to1} and \ref{applem:subadditive_1to1}. Using these lemmas again, we have
\begin{align}
    \|\mathcal{P}_\mathcal{L}(t, r) - \mathcal{P}_{\mathcal{L}}(t, u)\|_{1\rightarrow 1} &= \Big\|\int_u^r \frac{\partial}{\partial s}\mathcal{P}_{\mathcal{L}}(t, s) {\rm d}s \Big\|_{1\rightarrow 1}\\
    &= \Big\|\int_u^r \mathcal{P}_{\mathcal{L}}(t, s)\mathcal{L}_s {\rm d}s \Big\|_{1\rightarrow 1}\\
    &\le \Big|\int_u^r \|\mathcal{P}_{\mathcal{L}}(t, s)\|_{1\rightarrow 1} \|\mathcal{L}_s\|_{1 \rightarrow 1} {\rm d} s \Big|\\
    &= \Big|\int_u^r \|\mathcal{L}_s\|_{1 \rightarrow 1} {\rm d} s \Big|,
\end{align}
and similarly,
\begin{align}
    \|\nE^{\bar{\mathcal{L}}_i (r-t_{i-1})}-\nE^{\bar{\mathcal{L}}_i (u-t_{i-1})}\|_{1\rightarrow 1} \le \Big|\int_u^r \|\bar{\mathcal{L}}_i\|_{1 \rightarrow 1} {\rm d} s \Big|,
\end{align}
where the absolute values are taken in both cases since $u$ can be greater or smaller than $r$.
We therefore bound the second term in Eq.~(\ref{eq:decomp_varepsilon_s(k)}) as
\begin{align}
    \|\nE^{\bar{\mathcal{L}}_i \Delta t} - \overleftarrow{\rm T}\nE^{\int_{t_{i-1}}^{t_i} \mathcal{L}_t{\rm d} t}\|_{1\rightarrow 1}  &\le \frac{1}{\Delta t}\int_{t_{i-1}}^{t_i}\int_{t_{i-1}}^{t_i}\|\mathcal{L}_u\|_{1\rightarrow 1}\left(\Big|\int_u^r \|\mathcal{L}_s\|_{1 \rightarrow 1} {\rm d} s \Big| + \Big|\int_u^r \|\bar{\mathcal{L}}_i\|_{1 \rightarrow 1} {\rm d} s \Big|\right){\rm d}r {\rm d}u\\
    &\le \max_{t\in[t_{i-1},t_{i}]}\|\mathcal{L}_t\|_{1 \rightarrow 1}^2 \frac{2}{\Delta t}\int_{t_{i-1}}^{t_i}\int_{t_{i-1}}^{t_i}|r-u| {\rm d}r {\rm d}u\\
    &= \frac{2(\Delta t)^2}{3} \max_{t\in[t_{i-1},t_{i}]}\|\mathcal{L}_t\|_{1 \rightarrow 1}^2.
    \label{eq:2nd_term_varepsilon_s(k)}
\end{align}
\end{widetext}
Hence, combining Eqs.~(\ref{eq:decomp_varepsilon_s(k)}), (\ref{eq:1st_term_varepsilon_s(k)}) and (\ref{eq:2nd_term_varepsilon_s(k)}), we obtain an upper bound on $\varepsilon_{\rm s}(i)$:
\begin{align}
    \varepsilon_{\rm s}(i) &\le \Delta t\|\mathcal{L}_i - \bar{\mathcal{L}}_i\|_{1\rightarrow 1} + \frac{2(\Delta t)^2}{3} \max_{t\in[(i-1)\Delta t, i\Delta t]}\|\mathcal{L}_t\|_{1 \rightarrow 1}^2
    \label{eq:varepsilon_s(k)_bound_w_L_no_Lip}\\
    &\le \left(\frac{C_{\rm L}}{2} + \frac{2}{3}\max_{s\in[(i-1)\Delta t, i\Delta t]}\|\mathcal{L}_s\|_{1 \rightarrow 1}^2\right) (\Delta t)^2,
    \label{eq:varepsilon_s(k)_bound_w_L_w_Lip}
\end{align}
where the second inequality holds when $\mathcal{L}_t$ is Lipschitz continuous [Eq.~(\ref{eq:Lipschitz_continuity})] with the Lipschitz constant $C_{\rm L}$.

Recall that $\mathcal{L}_t$ satisfies Axiom~\ref{axiom:I}, i.e., $\mathcal{L}_t$ is a Lindbladian in the GKSL form [Eq.~(\ref{eq:L_gen})]. We thus have
\begin{align}
    &\quad \|\mathcal{L}_t\|_{1\rightarrow 1} \nonumber\\
    &= \sup_{\|X\|_1= 1} \|-\nI[H_S'(t), X] + \mathcal{A}_t(X) - \frac{1}{2}[\mathcal{A}_t^\dagger(\mathds{1}), X]_+\|_1\\
    &\le \sup_{\|X\|_1= 1}\left(2\|H_S'(t) X\|_1 + \|\mathcal{A}_t(X)\|_1 + \|\mathcal{A}_t^\dagger(\mathds{1})X\|_1\right)\\
    &\le 2\|H_S'(t)\|_\infty + \|\mathcal{A}_t\|_{1\rightarrow 1} + \|\mathcal{A}^\dagger_t(\mathds{1})\|_\infty,
\end{align}
where the first inequality is from the triangle inequality of the trace norm and the second inequality follows from H\"older's inequality [Eq.~(\ref{eq:Holder_ineq})]. Recall the definition of the operator norm [Eq.~(\ref{eq:operator_norm})]. We have
\begin{align}
    \|\mathcal{A}^\dagger_t(\mathds{1})\|_\infty &= \sup_{\|Y\|_1= 1} |{\rm Tr}\{Y^\dagger \mathcal{A}_t^\dagger(\mathds{1})\}| \\
    &= \sup_{\|Y^\dagger\|_1= 1}|{\rm Tr}\{\mathcal{A}_t(Y)\}|\\
    &= \sup_{\|Y\|_1= 1}|{\rm Tr}\{\mathcal{A}_t(Y)\}|\\
    &\le \sup_{\|Y\|_1= 1} {\rm Tr}\{\sqrt{(\mathcal{A}_t(Y))^\dagger\mathcal{A}_t(Y)}\}\\
    &= \|\mathcal{A}_t\|_{1\rightarrow 1},
\end{align}
where in the last line we used the definitions of the trace norm [Eq.~(\ref{eq:trace_norm})] and the $1\rightarrow 1$ norm [Eq.~(\ref{eq:1to1_norm})]. Therefore, $\|\mathcal{L}_t\|_{1\rightarrow 1}$ can be bounded as
\begin{align}
    \|\mathcal{L}_t\|_{1\rightarrow 1} \le 2(\|H_S'(t)\|_\infty + \|\mathcal{A}_t\|_{1\rightarrow 1}).
\end{align}
The bound [Eq.~(\ref{eq:varepsilon_s(k)_bound_w_L_w_Lip})] therefore implies
\begin{align}
    &\quad \varepsilon_{\rm s}(i) \nonumber\\
    &\le \left[\frac{C_{\rm L}}{2} + \frac{4}{3}\max_{t\in[(i-1)\Delta t, i\Delta t]}\left(\|H_S'(t)\|_\infty + \|\mathcal{A}_t\|_{1\rightarrow 1}\right)^2\right] (\Delta t)^2.
\end{align}
We remark that this bound is general for all Lipschitz-continuous Lindbladian $\mathcal{L}_t$.

In addition, as $\mathcal{L}_t$ is a thermal Lindbladian, $\mathcal{A}_t$ satisfies the two axioms, time-translation symmetry~\ref{axiom:II} and quantum detailed balance~\ref{axiom:III}. Then, by Lem.~\ref{lem:thermal_hermitian_dilation}, $\mathcal{A}_t$ admits the following dilation:
\begin{align}
    \mathcal{A}_t(\cdot) = {\rm Tr}_{E}\{V_{SE}(t)(\cdot\otimes\rho_{E})V_{SE}(t)\}.
\end{align}
Therefore, we have
\begin{align}
    \|\mathcal{A}_t\|_{1\rightarrow 1} &= \sup_{\|X_S\|_1= 1}\|{\rm Tr}_{E}\{V_{SE}(t)(X_S\otimes\rho_{E})V_{SE}(t)\}\|_1\\
    &\le \sup_{\|X_S\|_1= 1}\|V_{SE}(t)(X_S\otimes\rho_{E})V_{SE}(t)\|_1\\
    &\le \sup_{\|X_S\|_1= 1}\|V_{SE}(t)\|_\infty^2\|X_S\otimes\rho_{E}\|_1\\
    &= \|V_{SE}(t)\|_\infty^2,
    \label{eq:bound_on_A_t_1to1_norm}
\end{align}
where the first inequality follows from the fact that the trace norm is contractive under partial trace and the second inequality uses H\"older's inequality [Eq.~(\ref{eq:Holder_ineq})] and the fact that $\|X_S\otimes \rho_{E_i}\|_1 = \|X_S\|_1\|\rho_{E_i}\|_1 = \|X_S\|_1$. This leads to
\begin{align}
    \|\mathcal{L}_t\|_{1\rightarrow 1} \le 2(\|H_S'(t)\|_\infty + \|V_{SE}(t)\|_\infty^2).
    \label{eq:bound_on_L_t_1to1_norm}
\end{align}
Consequently, the $i$th step sampling error $\varepsilon_{\rm s}(i)$ can be further bounded as
\begin{align}
    &\quad \varepsilon_{\rm s}(i) \nonumber\\
    &\le \left[\frac{C_{\rm L}}{2} + \frac{4}{3}\max_{t\in[(i-1)\Delta t, i\Delta t]}\left(\|H_S'(t)\|_\infty + \|V_{SE}(t)\|_\infty^2\right)^2\right] \nonumber\\
    &\quad \times(\Delta t)^2,
    \label{eq:varepsilon_s(k)_bound_w_Lip}
\end{align}
which corresponds to the first part of the error bound provided in Eq.~\eqref{eq:error_bound} in the main text.
\subsection{Truncation error}\label{app:truncation_error}
We now derive bounds on the truncation error $\varepsilon_{\rm t}(i)$, stemming from first order approximations made in the derivation of the ECTCM master equation. Following a similar procedure as in Ref.~\cite{Cattaneo2021CMcansimulate}, we consider
\begin{align}
    \varepsilon_{\rm t}(i) &= \|(\mathcal{I}+\mathcal{L}_i\Delta t)-\nE^{\mathcal{L}_i\Delta t}\|_{1\rightarrow 1}\\
    &= \|\sum_{m=2}^\infty \frac{(\mathcal{L}_i\Delta t)^m}{m!}\|_{1\rightarrow 1}\\
    &\le \sum_{m=2}^\infty \frac{(\|\mathcal{L}_i\|_{1\rightarrow 1}\Delta t)^m}{m!}\\
    &\le \frac{(\|\mathcal{L}_i\|_{1\rightarrow 1}\Delta t)^2}{2!}\sum_{m=0}^\infty \frac{(\|\mathcal{L}_i\|_{1\rightarrow 1}\Delta t)^m}{m!},
\end{align}
where the first inequality comes from Lem.~\ref{applem:triangle_1to1} and the second inequality comes from the fact that $2!m! \le (m+2)!$ for all nonnegative integers $m$. For $\Delta t$ sufficiently small, we have
\begin{align}
    \|\mathcal{L}_i\|_{1\rightarrow 1}\Delta t \le 1,
    \label{eq:L_k_Delta_t_<_1}
\end{align}
in which case
\begin{align}
    \varepsilon_{\rm t}(i) &\le \frac{(\|\mathcal{L}_i\|_{1\rightarrow 1}\Delta t)^2}{2!}\sum_{m=0}^\infty \frac{(\|\mathcal{L}_i\|_{1\rightarrow 1}\Delta t)^m}{m!}\\
    &\le \frac{\nE}{2}(\|\mathcal{L}_i\|_{1\rightarrow 1}\Delta t)^2.
    \label{eq:varepsilon_t(k)_gen_bound}
\end{align}
According to Eq.~(\ref{eq:bound_on_L_t_1to1_norm}), $\|\mathcal{L}_i\|_{1\rightarrow 1}$ can be bounded as
\begin{align}
    \|\mathcal{L}_i\|_{1\rightarrow 1} \le 2(\|H_{S,i}'\|_\infty + \|\mathcal{A}_i\|_{1\rightarrow 1}),
\end{align}
where we used the shorthand notations $H_{S,i}':= H_S'(i\Delta t)$ and $\mathcal{A}_i:=\mathcal{A}(i\Delta t)$.
This leads to the bound on the truncation error:
\begin{align}
    \varepsilon_{\rm t}(i) \le 2\nE [(\|H_{S,i}'\|_\infty + \|\mathcal{A}_i\|_{1\rightarrow 1})\Delta t]^2.
\end{align}
Since $\mathcal{A}_i$ satisfies the two axioms, time-translation symmetry~\ref{axiom:II} and quantum detailed balance~\ref{axiom:III}, its $1\rightarrow 1$ norm can be bounded [Eq.~(\ref{eq:bound_on_A_t_1to1_norm})]: $\|\mathcal{A}_i\|_{1\rightarrow 1} \le \|V_{SE_i}\|_\infty^2$. Hence, we obtain
\begin{align}
    \varepsilon_{\rm t}(i) \le 2\nE [(\|H_{S,i}'\|_\infty + \|V_{SE_i}\|_\infty^2)\Delta t]^2,
    \label{eq:varepsilon_t(k)_bound}
\end{align}
yielding the second line in Eq.~\eqref{eq:error_bound}. Note that the above holds when Eq.~(\ref{eq:L_k_Delta_t_<_1}), or the following stronger condition,
\begin{align}
    2(\|H_{S,i}'\|_\infty + \|V_{SE_i}\|_\infty^2)\Delta t \le 1,
    \label{eq:(H_S+V^2)_Delta_t_<_1}
\end{align}
is satisfied. 

\subsection{Collision error}
We finally derive the upper bound of the collision error $\varepsilon_{\rm c}(i) = \|\hat{\mathcal{E}}_i - (\mathcal{I}+\mathcal{L}_i\Delta t)\|_{1\rightarrow 1}$. Recall that $\hat{\mathcal{E}}_i(\cdot) = {\rm Tr}_{E_i}\{U_i(\cdot\otimes \rho_{E_i}) U_i^\dagger\}$, where $U_i \equiv \nE^{-\nI H_{SE_i}\Delta t}$
with $H_{SE_i}\equiv H_{S,i}'+H_{E_i}+gV_{SE_i}$. We expand $U_i$ in terms of $\Delta t$:
\begin{align}
    U_i &= \mathds{1} - \nI (H_{i}+ g V_{SE_i}) \Delta t - \frac{1}{2}(H_{i} + g V_{SE_i})^2 (\Delta t)^2\nonumber\\
    &\quad + O((\Delta t)^3),
\end{align}
where $H_i \equiv H_{S,i}' + H_{E_i}$.
The map $\hat{\mathcal{E}}_i(\cdot)$ can be written as
\begin{widetext}
    \begin{align}
    \hat{\mathcal{E}}_i(\cdot) &:= {\rm Tr}_{E_i}\{U_i(\cdot\otimes \rho_{E_i}) U_i^\dagger\}\\
    &= (\cdot) -\mathrm{i} \Delta t  [H_{S,i}', \cdot] + (g\Delta t)^2\mathrm{Tr}_{E_i}\left\{V_{SE_i}(\cdot\otimes\rho_{E_i}) V_{SE_i}\right\} - (g\Delta t)^2{\rm Tr}_{E_i}\left\{\frac{1}{2}[V_{SE_i}^2, \cdot\otimes\rho_{E_i}]_+\right\} \nonumber\\
    &\quad + O(g (\Delta t)^2) + O((\Delta t)^2) \\
    &= (\cdot) -\mathrm{i} \Delta t  [H_{S,i}', \cdot] + \Delta t\,\mathrm{Tr}_{E_i}\left\{V_{SE_i}(\cdot\otimes\rho_{E_i}) V_{SE_i}\right\} - \Delta t\,{\rm Tr}_{E_i}\left\{\frac{1}{2}[V_{SE_i}^2, \cdot\otimes\rho_{E_i}]_+\right\} + O((\Delta t)^{3/2})\\
    &= \mathcal{I}(\cdot) + \mathcal{L}^{\rm (CM)}_{i\Delta t}(\cdot)\Delta t + {\rm R_c}({\rm Tr}_{E_i}\{U_i(\cdot\otimes \rho_{E_i}) U_i^\dagger\}),
    \label{eq:R_c_for_E_k}
    \end{align}
\end{widetext}
where in the first line we used the fact that ${\rm Tr}_{E_i}\{V_{SE_i}\rho_{E_i}\} = 0$ as required for a vanishing Lamb shift, in the third line we used the fact that $g^2\Delta t = 1$, $\mathcal{L}_t^{\rm (CM)}$ is the continuous-time Lindbladian at time $t$, defined in Eq.~(\ref{eq:L_CM_def}) and ${\rm R_c}({\rm Tr}_{E_i}\{U_i(\cdot\otimes \rho_{E_i}) U_i^\dagger\})$ is the remainder of the expansion of $\hat{\mathcal{E}}_i$ with respect to $\Delta t$:
\begin{widetext}
    \begin{align}
        R_{\rm c}({\rm Tr}_{E_i}\{U_i(\cdot\otimes \rho_{E_i}) U_i^\dagger\}) &\equiv 
        \sum_{m=3}^\infty \frac{(-\nI \Delta t)^m}{m!}{\rm Tr}_{E_i}\{\mathcal{L}_{H_{SE_i}}^m(\cdot\otimes\gamma_{E_i})\}
        + (\Delta t)^2\,{\rm Tr}_{E_i}\left\{H_i(\cdot\otimes\rho_{E_i}) H_i\right\} \nonumber\\
        &\quad + (\Delta t)^{3/2}\,{\rm Tr}_{E_i}\left\{H_i(\cdot\otimes\rho_{E_i})V_{SE_i}\right\} + (\Delta t)^{3/2}\,{\rm Tr}_{E_i}\left\{V_{SE_i}(\cdot\otimes\rho_{E_i})H_i\right\},
        \label{appeq:R_c_explicit}
    \end{align}
    with $\mathcal{L}_{X}(\cdot):=[X, \cdot]$.
\end{widetext}
We note that there are three kinds of terms in ${\rm R_c}$ with different scalings of $\Delta t$~\cite{Cattaneo2021CMcansimulate}: one containing only $V_{SE_i}$ with the scaling $O((\Delta t)^{3/2})$, one containing only $H_{i}$ with the scaling $O((\Delta t)^2)$ and one containing both $V_{SE_i}$ and $H_i$ with the scaling $O((\Delta t)^{3/2})$.

In Prop.~\ref{prop:ME_CM_equivalence}, we have shown that $\mathcal{L}_t^{\rm (CM)} = \mathcal{L}_t$ for all $t$. Thus, we have $\mathcal{L}^{\rm (CM)}_{i\Delta t} = \mathcal{L}_{i\Delta t} \equiv \mathcal{L}_{i}$. The collision error is hence given by
\begin{align}
    \varepsilon_{\rm c}(i) &= \|\hat{\mathcal{E}}_i - (\mathcal{I}+\mathcal{L}_i\Delta t)\|_{1\rightarrow 1}\\
    &= \|{\rm R_c}({\rm Tr}_{E_i}\{U_i(\cdot\otimes \rho_{E_i}) U_i^\dagger\})\|_{1\rightarrow 1}.
\end{align}
Using the fact that trace norm is contractive under partial trace, we can remove the partial trace ${\rm Tr}_{E_i}$ and obtain an upper bound on $\varepsilon_{\rm c}(i)$:
\begin{align}
    \varepsilon_{\rm c}(i) &= \|{\rm R_c}({\rm Tr}_{E_i}\{U_i(\cdot\otimes \rho_{E_i}) U_i^\dagger\})\|_{1\rightarrow 1}\\
    &= \|{\rm Tr}_{E_i}\{{\rm R_c}[U_i(\cdot\otimes \rho_{E_i}) U_i^\dagger]\}\|_{1\rightarrow 1}\\
    &=\sup_{\|X_S\|_1= 1}\|{\rm Tr}_{E_i}\{{\rm R_c}[U_i(X_S \otimes \rho_{E_i})U_i^\dagger]\}\|_{1}\\
    &\le \sup_{\|X_S\|_1= 1}\|{\rm R_c}[U_i(X_S \otimes \rho_{E_i})U_i^\dagger]\|_{1}.
\end{align}
Here, from the second line onward, $R_{\rm c}$ denotes the collection of all contributions beyond second order in the expansion of the unitary operator $U_i$. Specifically, it comprises the sum over $m\ge 3$ in Eq.~(\ref{appeq:R_c_explicit}) together with the three additional terms that originate from the second-order expansion of $U_i$, prior to taking the partial trace over the environment.
We now write the remainder explicitly:
\begin{widetext}
    \begin{align}
        &\quad \sup_{\|X_S\|_1= 1}\|{\rm R_c}[U_i(X_S \otimes \rho_{E_i})U_i^\dagger]\|_{1} \nonumber\\
        &=\sup_{\|X_S\|_1= 1}\|{\rm R_c}[\nE^{-\nI \Delta t (H_i + gV_{SE_i})}(X_S \otimes \rho_{E_i})\nE^{\nI \Delta t (H_i + gV_{SE_i})}]\|_{1}\\
        &= 
        \sup_{\|X_S\|_1= 1}\Big\|\tilde{\sum}_{j,m =1}^\infty \frac{[-\nI \Delta t (H_i + gV_{SE_i})]^j}{j!}(X_S\otimes\rho_{E_i})\frac{[\nI \Delta t (H_i + gV_{SE_i})]^m}{m!}\Big\|_1\\
        &\le \sup_{\|X_S\|_1= 1}\tilde{\sum}_{j,m =1}^\infty\Big\|\frac{[-\nI \Delta t (H_i + gV_{SE_i})]^j}{j!}(X_S\otimes\rho_{E_i})\frac{[\nI \Delta t (H_i + gV_{SE_i})]^m}{m!}\Big\|_1\\
        &\le \sup_{\|X_S\|_1= 1}\tilde{\sum}_{j,m=1}^\infty \frac{(\Delta t \|H_i + gV_{SE_i}\|_\infty)^j}{j!}\|X_S\otimes\rho_{E_i}\|_1\frac{(\Delta t \|H_i + gV_{SE_i}\|_\infty)^m}{m!}\\
        &\le \sup_{\|X_S\|_1= 1}\tilde{\sum}_{j,m=1}^\infty \frac{[\Delta t (\|H_i\|_\infty + g\|V_{SE_i}\|_\infty)]^j}{j!}\|X_S\otimes\rho_{E_i}\|_1\frac{[\Delta t (\|H_i\|_\infty + g\|V_{SE_i}\|_\infty)]^m}{m!}\\
        &= {\rm R_c}\left(\nE^{2\Delta t \left(\|H_i\|_\infty + g\|V_{SE_i}\|_\infty\right)}\right).
    \end{align}
\end{widetext}
    Here, in the third line we introduced the notation $\tilde{\sum}$ to express the sum over all possible combinations in the remainder ${\rm R_c}$ [cf.~Eq.~(\ref{appeq:R_c_explicit})], which therefore is not the full summation $\sum$. The first and third inequalities come from the triangle inequalities of the trace norm and the operator norm, respectively, while the second inequality comes from H\"older's inequality [Eq.~(\ref{eq:Holder_ineq})] and the submultiplicativity of the operator norm. The remainder ${\rm R_c}$ in the last line can be written explicitly by decompose it into three different terms ${\rm R}_{\rm c}^{(V)}$, ${\rm R}_{\rm c}^{(H)}$ and ${\rm R}_{\rm c}^{(VH)}$, where the superscript denotes what terms are contained [see the discussion below Eq.~(\ref{eq:R_c_for_E_k})].
\begin{align}
    {\rm R_c}\left(\nE^{2\Delta t \left(\|H_i\|_\infty + g\|V_{SE_i}\|_\infty\right)}\right) &= {\rm R}_{\rm c}^{(V)} + {\rm R}_{\rm c}^{(H)} + {\rm R}_{\rm c}^{(VH)},
    \label{appeq:R_c_decomposition}
\end{align}
where
\begin{widetext}
    \begin{align}
    {\rm R}_{\rm c}^{(V)} &:= \sum_{m=3}^\infty \frac{(2\Delta t g \|V_{SE_i}\|_\infty)^m}{m!} 
    \le \frac{(2\Delta t g \|V_{SE_i}\|_\infty)^3}{3!}\nE^{2\Delta t g \|V_{SE_i}\|_\infty} = \frac{4\|V_{SE_i}\|_\infty^3 (\Delta t)^{3/2}}{3}\nE^{2\sqrt{\Delta t} \|V_{SE_i}\|_\infty}, \\
    {\rm R}_{\rm c}^{(H)} &:= \sum_{m=2}^\infty \frac{(2\Delta t \|H_i\|_\infty)^m}{m!}
    \le \frac{(2\Delta t \|H_i\|_\infty)^2}{2!}\nE^{2\Delta t \|H_i\|_\infty} =
    2\|H_i\|_\infty^2(\Delta t)^2\nE^{2\Delta t \|H_i\|_\infty}, \\
    {\rm R}_{\rm c}^{(VH)} &:= \sum_{m=2}^\infty \frac{(2\Delta t)^m \sum_{m'=1}^{m-1}\begin{pmatrix}\begin{smallmatrix}
         m\\
         m'
    \end{smallmatrix}
    \end{pmatrix}\|H_i\|_\infty^{m'}(g\|V_{SE_i}\|_\infty)^{m-m'}}{m!} \le 
    \sum_{m=2}^\infty \frac{(4\Delta t)^m\|H_i\|_\infty(g\|V_{SE_i}\|_\infty)^{m-1}}{m!}\\
    &\le \frac{(4\Delta t)^2\|H_i\|_\infty g\|V_{SE_i}\|_\infty}{2!}\nE^{4\Delta tg\|V_{SE_i}\|_\infty}
    = 8\|H_i\|_\infty \|V_{SE_i}\|_\infty(\Delta t)^{3/2}\nE^{4\sqrt{\Delta t}\|V_{SE_i}\|_\infty},
    \end{align}
\end{widetext}
where in ${\rm R}_{\rm c}^{(VH)}$ we used the condition that
\begin{align}
    g\|V_{SE_i}\|_\infty \ge \|H_i\|_\infty,
    \label{eq:gV_>_H_k}
\end{align}
which holds for small $\Delta t$ since $g^2\Delta t = 1$,
and the fact that $\sum_{m'=1}^{m-1}\begin{pmatrix}\begin{smallmatrix}
         m\\
         m'
    \end{smallmatrix}
    \end{pmatrix} \le \sum_{m'=0}^{m-1}\begin{pmatrix}\begin{smallmatrix}
         m\\
         m'
    \end{smallmatrix}
    \end{pmatrix} = 2^m$.
The final expression of all three remainders are obtained under the condition $g^2\Delta t = 1$.
Putting these remainders together, we obtain an upper bound on $\varepsilon_{\rm c}(i)$:
\begin{align}
    \varepsilon_{\rm c}(i) &\le \frac{4\|V_{SE_i}\|_\infty^3 (\Delta t)^{3/2}}{3}\nE^{2\sqrt{\Delta t} \|V_{SE_i}\|_\infty} \nonumber\\
    &\quad + 2\|H_i\|_\infty^2(\Delta t)^2\nE^{2\Delta t \|H_i\|_\infty} \nonumber\\
    &\quad + 8\|H_i\|_\infty \|V_{SE_i}\|_\infty(\Delta t)^{3/2}\nE^{4\sqrt{\Delta t}\|V_{SE_i}\|_\infty}.
    \label{eq:varepsilon_c(k)_gen_bound}
\end{align}
Now we recall that in our construction of ECTCMs, $H_{E_i} = H_E$ [Eq.~(\ref{eq:H_E_constructed})] is independent of $i$ and $\|H_{E}\|_\infty = \max_{\ell,\ell'=0,1,\dots,d_S-1}(\frac{\epsilon_\ell-\epsilon_{\ell'}}{2})$ where $\{\epsilon_\ell\}_{\ell=0}^{d_S-1}$ are the eigenvalues of $H_S$. Thus, we have 
\begin{align}
    \|H_i\|_\infty &= \|H_{S,i}' + H_E\|_\infty \\
    &\le \|H_{S,i}'\|_\infty + \|H_E\|_\infty \\
    &\le \|H_{S,i}'\|_\infty + \|H_S\|_\infty \\
    &\le 2\|\tilde{H}_{S,i}\|_\infty,
\end{align}
where $\tilde{H}_{S,i} := \arg\max_{X\in\{H_{S,i}', H_S\}}\|X\|_\infty$,
and the collision error can be bounded as
\begin{align}
    \varepsilon_{\rm c}(i) &\le \frac{4\|V_{SE_i}\|_\infty^3 (\Delta t)^{3/2}}{3}\nE^{2\sqrt{\Delta t} \|V_{SE_i}\|_\infty} \nonumber\\
    &\quad + 8\|\tilde{H}_{S,i}\|_\infty^2(\Delta t)^2\nE^{4\Delta t \|\tilde{H}_{S,i}\|_\infty} \nonumber\\
    &\quad + 16\|\tilde{H}_{S,i}\|_\infty \|V_{SE_i}\|_\infty(\Delta t)^{3/2}\nE^{4\sqrt{\Delta t}\|V_{SE_i}\|_\infty}.
    \label{eq:varepsilon_c(k)_bound}
\end{align}
The bound holds when Eq.~(\ref{eq:gV_>_H_k}) is satisfied, i.e., for $\Delta t$ sufficiently small. Alternatively, the bound also holds under the stronger condition that $g^2\Delta t = 1$ and:
\begin{align}
    \|V_{SE_i}\|_\infty \ge 2\|\tilde{H}_{S,i}\|_\infty\sqrt{\Delta t}.
    \label{eq:V_>_2H_S_sqrt_Delta_t}
\end{align}
\subsection{Total error}\label{app:total_error}
The total error $\varepsilon$ is bounded as follows [Eq.~(\ref{eq:varepsilon_bound_gen})]
\begin{align}
    \varepsilon \le n\max_{1\le i \le n}\varepsilon_{\rm s}(i) + n\max_{1\le i \le n}[\varepsilon_{\rm t}(i) + \varepsilon_{\rm c}(i)].
    \label{appeq:varepsilon_bound_gen}
\end{align}
We summarise in Table~\ref{tab:error_bounds} the error bounds derived above for $\varepsilon_{\rm s}(i)$, $\varepsilon_{\rm t}(i)$ and $\varepsilon_{\rm c}(i)$. The table presents both the general bounds for collision-model simulations of arbitrary GKSL dynamics and those specialised to the explicit construction introduced in Protocol~\ref{tab:protocol}. Using the above equation and referring to Table~\ref{tab:error_bounds}, a general bound on the total error $\varepsilon$ can be directly obtained. In the following, we focus on bounding the error $\varepsilon$ for our specific construction. 
{\renewcommand{\arraystretch}{1.3}
    \setcounter{table}{0}
    \begin{table*}
        \begin{ruledtabular}
        \begin{tabular}{ccc}
        Error types & General collision-model simulation & Protocol~\ref{tab:protocol} \\
        \hline
        $\varepsilon_{\rm s}(i)$  &  Eq.~(\ref{eq:varepsilon_s(k)_bound_w_L_w_Lip}) [Eq.~(\ref{eq:Lipschitz_continuity})] & Eq.~(\ref{eq:varepsilon_s(k)_bound_w_Lip}) [Eq.~(\ref{eq:Lipschitz_continuity})]\\
        $\varepsilon_{\rm t}(i)$  & Eq.~(\ref{eq:varepsilon_t(k)_gen_bound}) [Eq.~(\ref{eq:L_k_Delta_t_<_1})] & Eq.~(\ref{eq:varepsilon_t(k)_bound}) [Eq.~(\ref{eq:(H_S+V^2)_Delta_t_<_1})]\\
        $\varepsilon_{\rm c}(i)$  & Eq.~(\ref{eq:varepsilon_c(k)_gen_bound}) [Eq.~(\ref{eq:gV_>_H_k}) with $g^2\Delta t =1$] & Eq.~(\ref{eq:varepsilon_c(k)_bound}) [Eq.~(\ref{eq:V_>_2H_S_sqrt_Delta_t})]
        \end{tabular}
        \end{ruledtabular}
        \caption{Error bounds for general collision-model simulation and the explicit construction provided in Protocol~\ref{tab:protocol}. The equations in squared brackets specify the conditions under which the corresponding bounds hold.}
        \label{tab:error_bounds}
    \end{table*}
}

The upper bound of $\varepsilon_{\rm s}(i)$ is given in Eq.~(\ref{eq:varepsilon_s(k)_bound_w_Lip}),
\begin{align}
    &\quad \varepsilon_{\rm s}(i) \nonumber\\
    &\le \left[\frac{C_{\rm L}}{2} + \frac{4}{3}\max_{t\in[(i-1)\Delta t, i\Delta t]}\left(\|H_S'(t)\|_\infty + \|V_{SE}(t)\|_\infty^2\right)^2\right] \nonumber\\
    &\quad \times(\Delta t)^2,
\end{align}
with the assumption of the Lipschitz continuity of $\mathcal{L}_t$ [Eq.~(\ref{eq:Lipschitz_continuity})].

Using the upper bounds for $\varepsilon_{\rm t}(i)$ [Eq.~(\ref{eq:varepsilon_t(k)_bound})] and $\varepsilon_{\rm c}(i)$ [Eq.~(\ref{eq:varepsilon_c(k)_bound})], we have
\begin{align}
    \varepsilon_{\rm t}(i) + \varepsilon_{\rm c}(i) &\le 2\nE(\|H_{S,i}'\|_\infty + \|V_{SE_i}\|_\infty^2)^2 (\Delta t)^2\nonumber\\
    &\quad + \frac{4\|V_{SE_i}\|_\infty^3 (\Delta t)^{3/2}}{3}\nE^{2\sqrt{\Delta t} \|V_{SE_i}\|_\infty} \nonumber\\
    &\quad + 8\|\tilde{H}_{S,i}\|_\infty^2(\Delta t)^2\nE^{4\Delta t \|\tilde{H}_{S,i}\|_\infty} \nonumber\\
    &\quad + 16\|\tilde{H}_{S,i}\|_\infty \|V_{SE_i}\|_\infty(\Delta t)^{3/2}\nE^{4\sqrt{\Delta t}\|V_{SE_i}\|_\infty},
\end{align}
where $\tilde{H}_{S,i} := \arg\max_{X\in\{H_{S,i}', H_S\}}\|X\|_\infty$. The conditions for this bound to hold are [Eqs.~(\ref{eq:(H_S+V^2)_Delta_t_<_1}) and (\ref{eq:V_>_2H_S_sqrt_Delta_t})]
\begin{align}
    2(\|H_{S,i}'\|_\infty + \|V_{SE_i}\|_\infty^2)\Delta t &\le 1, \\
    2\|\tilde{H}_{S,i}\|_\infty\sqrt{\Delta t}&\le \|V_{SE_i}\|_\infty,
\end{align}
which is satisfied for $\Delta t$ sufficiently small. We note that the operator $V_{SE}(t)$ is the Hermitian operator constructed from $\mathcal{A}_t$ according to Lem.~\ref{lem:thermal_hermitian_dilation} and by construction, we have $V_{SE_i} = V_{SE}(i\Delta t)$ for the system--ancilla interaction $V_{SE_i}$. Denote
\begin{align}
    \|V_{\max}\|_{\infty} &:= \max_{t\in [0,\infty)}\|V_{SE}(t)\|_{\infty},\\
    \|\tilde{H}_{\max}\|_\infty &:= \max\left\{\max_{t\in [0,\infty)}\|H_{S}'(t)\|_\infty, \|H_S\|_\infty\right\}.
\end{align}
We have
\begin{align}
        \varepsilon &\le n\max_{1\le i \le n}\varepsilon_{\rm s}(i) + n\max_{1\le i \le n}[\varepsilon_{\rm t}(i) + \varepsilon_{\rm c}(i)] \\
        &\le n \left[\left(\frac{C_{\rm L}}{2}+(\frac{4}{3}+2\nE)(\|\tilde{H}_{\max}\|_\infty + \|V_{\max}\|_\infty^2)^2\right) (\Delta t)^2 \right. \nonumber\\
        &\quad + \frac{4\|V_{\max}\|_\infty^3 (\Delta t)^{3/2}}{3}\nE^{2\sqrt{\Delta t} \|V_{\max}\|_\infty} \nonumber\\
        &\quad + 8\|\tilde{H}_{\max}\|_\infty^2(\Delta t)^2\nE^{4\Delta t \|\tilde{H}_{\max}\|_\infty} \nonumber\\
        &\left. \quad + 16\|\tilde{H}_{\max}\|_\infty \|V_{\max}\|_\infty(\Delta t)^{3/2}\nE^{4\sqrt{\Delta t}\|V_{\max}\|_\infty}\right],
        \label{eq:error_bound}
\end{align}
where the second bound holds when
\begin{align}
    2(\|\tilde{H}_{\max}\|_\infty + \|V_{\max}\|_\infty^2)\Delta t &\le 1, \\
    2\|\tilde{H}_{\max}\|_\infty\sqrt{\Delta t}&\le \|V_{\max}\|_\infty.
\end{align}
Recalling that in the ECTCM, $n\equiv \tau/\Delta t$, we thus obtain Eq.~(\ref{eq:error_scaling}) in the main text, namely, $\varepsilon = O(\sqrt{\Delta t})$ for fixed $\tau$ and sufficiently small $\Delta t$.

\setcounter{section}{9}
\section{Example of MTO realisation with finite-sized environment}\label{app:finite_MTO_realisation}
In this appendix, we present an explicit example showing that an axiomatic thermal Lindbladian generates a Markovian thermal operation that admits a realisation with a finite-size environment.

Consider the following Lindbladian for a qubit system with the Hamiltonian $H_S = \epsilon_0 F_S^{00} + \epsilon_1 F_S^{11}$ where
$F_S^{jj'} \equiv \ketbra{j}{j'}_S$ is the transition operator in the system's energy eigenbasis and the Lindbladian of the dynamics is given by:
\begin{align}
\label{eqn::dissipator_ex}
    \mathcal{L}(\cdot) &= -\nI[H_S, \cdot] \nonumber + \mathcal{D}(\cdot),\\
    \mathcal{D}(\cdot) &= \frac{1}{2}\ln\left(\frac{2}{1+\Delta_\gamma^2}\right)\left(F_S^{00}(\cdot)F_S^{00} - \frac{1}{2}[F_S^{00}, \cdot]_+\right) \nonumber\\
    &\quad + \frac{1}{2}\ln\left(\frac{2}{1+\Delta_\gamma^2}\right)\left(F_S^{11}(\cdot)F_S^{11} - \frac{1}{2}[F_S^{11}, \cdot]_+\right) \nonumber\\
    &\quad + \gamma_0\ln 2\left(F_S^{01}(\cdot)F_S^{10} - \frac{1}{2}[F_S^{11}, \cdot]_+\right) \nonumber\\
    &\quad + \gamma_1\ln 2\left(F_S^{10}(\cdot)F_S^{01} - \frac{1}{2}[F_S^{00}, \cdot]_+\right).
\end{align}
where $\gamma_{0(1)} \equiv \nE^{-\beta \epsilon_{0(1)}}/(\nE^{-\beta \epsilon_{0}} + \nE^{-\beta \epsilon_{1}})$ are thermal populations at the inverse temperature $\beta$ and $\Delta_\gamma\equiv \gamma_0-\gamma_1$. 
Recalling the expression of the dissipator in terms of the coefficient $\xi_{jj'mm'}$ [Eq.~(\ref{eq:D_givenby_xi})], we can identify the non-vanishing coefficients:
\begin{align}
    &\xi_{0101} = \gamma_0 \ln 2,\\
    &\xi_{1010} = \gamma_1 \ln 2,\\
    &\xi_{0000} = \xi_{1111} = \frac{1}{2}\ln\left(\frac{2}{1+\Delta_\gamma^2}\right),
\end{align}
It is straightforward to check that $\mathcal{L}$ satisfies the axioms~\ref{axiom:I}--\ref{axiom:III}. The resulting map at $\tau=1$ is 
\begin{align}
    \mathcal{E}_{\tau=1} \equiv \nE^{\mathcal{L}} = \mathcal{U}_{\tau=1}\circ\nE^{\mathcal{D}},
\end{align}
where we used the time-translation symmetry of $\mathcal{D}$ directly following from $[\mathcal{L}, \mathcal{L}_{H_S}]=0$ and $\mathcal{U}_\tau(\cdot) := \nE^{-\nI H_S \tau}(\cdot)\nE^{\nI H_S \tau}$. Since $\mathcal{U}_\tau$ is a thermal operation, then if $\mathcal{T} = \nE^{\mathcal{D}}$ is a TO, so is $\mathcal{E}_{t=\tau} = \mathcal{U}_{t=\tau}\circ\mathcal{T}$. We thus only need to show that $\mathcal{T}$ is a TO that can be realised with a finite-dimensional environment for $\nE^{\mathcal{D}}$.

Vectorising the system state $\rho_S$ in the basis $\{\ket{00}_S, \ket{01}_S, \ket{10}_S, \ket{11}_S\}$, we can express the superoperator $\mathcal{D}$ as a matrix:
\begin{align}
    \bm{D} = \begin{pmatrix}
        -\gamma_1\ln 2 & 0 & 0 & \gamma_0\ln 2\\
        0 & \ln\Big(\frac{\sqrt{1+\Delta_\gamma^2}}{2}\Big) & 0 & 0\\
        0 & 0 & \ln\Big(\frac{\sqrt{1+\Delta_\gamma^2}}{2}\Big) & 0\\
        \gamma_1\ln 2 & 0 & 0 & -\gamma_0\ln 2
    \end{pmatrix}.
\end{align}
Thus, the map $\nE^{\mathcal{D}}$ can be written as
\begin{align}
    \nE^{\bm{D}} = \begin{pmatrix}
        1-\frac{\gamma_1}{2} & 0 & 0 & \frac{\gamma_0}{2}\\
        0 & \frac{\sqrt{1+\Delta_\gamma^2}}{2} & 0 & 0\\
        0 & 0 & \frac{\sqrt{1+\Delta_\gamma^2}}{2} & 0\\
        \frac{\gamma_1}{2} & 0 & 0 & 1-\frac{\gamma_0}{2}
    \end{pmatrix}.
\end{align}
This can be realised as a thermal operation with a copy of $S$ (denoted as $S'$ with $H_{S'}=H_S$) being the environment, i.e.,
\begin{align}
    \mathcal{T}(\cdot) = {\rm Tr}_{S'}\{U(\cdot\otimes \gamma_{S'})U^\dagger\},
\end{align}
where $U\equiv (\nE^{\nI H_S \arctan(\Delta_\gamma)/(\epsilon_0-\epsilon_1)}\otimes\mathds{1}_{S'})\nE^{-\nI\pi F_{\rm swap}/4}$ with $F_{\rm swap}\equiv \sum_{j,j'\in\{0,1\}}\ketbra{j}{j'}\otimes\ketbra{j'}{j}$ being the swap operator. It is easy to check that $[U, H_S\otimes\mathds{1}_{S'}+\mathds{1}_S\otimes H_{S'}] = 0$.

It is straightforward to check that the superoperator $\mathbf{T}$ of $\mathcal{T}$ expressed in the basis $\{\ket{00}_S, \ket{01}_S, \ket{10}_S, \ket{11}_S\}$ satisfies $\mathbf{T} = e^\mathbf{D}$, i.e., it corresponds to the dynamics generated by the dissipator $\mathcal{D}$ in Eq.~\eqref{eqn::dissipator_ex}.
Therefore, in this example, the dimension of the required environment is just 2, equal to the system dimension. 

\section{Optimal error scaling for $d_S = 2$}\label{app:scaling_for_d=2}
In this appendix, we prove that for two-level systems with nontrivial Hamiltonians, the simulation error of Protocol~\ref{tab:protocol} scales as $O(\Delta t)$ instead of $O(\sqrt{\Delta t})$, achieving the optimal scaling for collision models~\cite{correa2014quantum}.

We consider a general two-level system $S$ with the Hamiltonian $H_S = \epsilon_0\ketbra{0}{0}_S + \epsilon_1\ketbra{1}{1}_S$ where $\epsilon_0\neq \epsilon_1$. According to Protocol~\ref{tab:protocol}, the Hamiltonian of an ancilla $E$ that consists of three qubits is constructed as follows [Eq.~(\ref{eq:H_E_constructed})]:
\begin{align}
    H_E = \left(\frac{\epsilon_0-\epsilon_1}{2}\ketbra{01}{01} + \frac{\epsilon_1-\epsilon_0}{2}\ketbra{10}{10}\right)\otimes\mathds{1}_2,
    \label{appeq:H_E_2d}
\end{align}
where $\mathds{1}_2 \equiv \ketbra{0}{0}+\ketbra{1}{1}$, and
the system--ancilla interaction Hamiltonian $V_{SE}$ [Eq.~(\ref{eq:V_constructed})] is constructed as follows with the time dependence omitted for simplicity:
\begin{widetext}
    \begin{align}
        V_{SE} &= \sqrt{\frac{Z}{2}}\ketbra{0}{0}_S\otimes(\ketbra{00}{00} + \ketbra{11}{11})\otimes\underbrace{\left(\sum_{q,q'=0}^{1} {W_{00qq'}}\ketbra{q'}{q}\right)}_{\displaystyle =: Q_{00}} + \sqrt{Z}\ketbra{0}{1}_S\otimes\ketbra{10}{01}\otimes\underbrace{\left(\sum_{q,q'=0}^{1} {W_{01qq'}}\ketbra{q'}{q}\right)}_{\displaystyle =: Q_{01}} \nonumber\\
        &\quad + \sqrt{\frac{Z}{2}}\ketbra{1}{1}_S\otimes(\ketbra{00}{00} + \ketbra{11}{11})\otimes\underbrace{\left(\sum_{q,q'=0}^{1} {W_{11qq'}}\ketbra{q'}{q}\right)}_{\displaystyle =: Q_{11}} + \sqrt{Z}\ketbra{1}{0}_S\otimes\ketbra{01}{10}\otimes\underbrace{\left(\sum_{q,q'=0}^{1} {W_{10qq'}}\ketbra{q'}{q}\right)}_{\displaystyle =: Q_{10}},
        \label{appeq:V_SE_2d}
    \end{align}
\end{widetext}
where the first qubit corresponds to the system $S$, $Z \equiv {\rm Tr}\{\nE^{-\beta H_E}\}$ for some inverse temperature $\beta$, and the coefficients $W_{xx'yy'}$ come from a thermal Lindbladian on $S$ whose details are irrelevant here.

According to the explicit block structure of $V_{SE}$, we have the following lemma:
\begin{lemma}\label{applem:vanishing_odd_moments}
    The interaction Hamiltonian $V_{SE}$ in the form of Eq.~(\ref{appeq:V_SE_2d}) satisfies ${\rm Tr}_{E}\{(V_{SE})^n\gamma_E\} = 0$ for all odd $n>1$, if ${\rm Tr}_E\{V_{SE}\gamma_E\} = 0$ is satisfied, where $\gamma_E\equiv \nE^{-\beta H_E} / Z$.
\end{lemma}
\begin{proof}
    Firstly, we note that $V_{SE}$ can be written as the direct sum of two block matrices $V_{SE}^{(1)}$ and $V_{SE}^{(2)}$ in the bases $\{\ket{000}, \ket{011} \ket{100}, \ket{111}\}$ and $\{\ket{010}, \ket{101}\}$, respectively, i.e., $V_{SE} = \sqrt{Z/2}V_{SE}^{(1)}\oplus \sqrt{Z}V_{SE}^{(2)}$, where
    \begin{align}
        V_{SE}^{(1)} &= \begin{pNiceMatrix}[first-row, first-col]
                  & \ket{000} & \ket{011} & \ket{100} & \ket{111}\\
        \ket{000} & Q_{00} & 0 & 0 & 0\\
        \ket{011} & 0 & Q_{00} & 0 & 0\\
        \ket{100} & 0 & 0 & Q_{11} & 0\\
        \ket{111} & 0 & 0 & 0 & Q_{11}
        \end{pNiceMatrix}, \label{appeq:V_SE^(1)}\\
        V_{SE}^{(2)} &= \begin{pNiceMatrix}[first-row, first-col]
                  & \ket{010} & \ket{101}\\
        \ket{010} & 0 & Q_{01} \\
        \ket{101} & Q_{10} & 0
        \end{pNiceMatrix},
        \label{appeq:V_SE^(2)}
    \end{align}
    with the four matrices $Q_{00}$, $Q_{01}$, $Q_{10}$ and $Q_{11}$ defined in Eq.~(\ref{appeq:V_SE_2d}). Note that in the bases used in Eqs.~(\ref{appeq:V_SE^(1)}) and (\ref{appeq:V_SE^(2)}), the first qubit corresponds to the system. We now consider 
    \begin{align}
        {\rm Tr}_{E}\{(V_{SE})^n\gamma_E\} &= \left(\sqrt{Z}/2\right)^{n/2}{\rm Tr}_{E}\{(V_{SE}^{(1)})^n\gamma_E\}\nonumber\\
        &\quad \oplus Z^{n/2} {\rm Tr}_{E}\{(V_{SE}^{(2)})^n\gamma_E\},
    \end{align}
    for odd $n\ge 1$. From the structure of $V_{SE}^{(2)}$ [Eq.~(\ref{appeq:V_SE^(2)})], it is straightforward to verify that all odd powers $(V_{SE}^{(2)})^n$ have vanishing diagonal entries. However, since $\gamma_E$ is diagonal in the energy eigenbasis of $E$, it follows that ${\rm Tr}_{E}\{(V_{SE}^{(2)})^n\gamma_E\} = 0$ for all odd $n\ge 1$. Hence, in order to show ${\rm Tr}_{E}\{(V_{SE})^n\gamma_E\} = 0$, we only need to show that ${\rm Tr}_{E}\{(V_{SE}^{(1)})^n\gamma_E\} = 0$.

    According to the structure of $H_E$ [Eq.~(\ref{appeq:H_E_2d})], $\gamma_E$ is fully degenerate on the relevant support of $V_{SE}^{(1)}$. Since $(V_{SE}^{(1)})^n$ remains block-diagonal for all $n$, we have
    ~\\
    \begin{align}
        {\rm Tr}_{E}\{(V_{SE}^{(1)})^n\gamma_E\} &\propto\, \begin{pNiceMatrix}[first-row, first-col]
                  & \ket{0}_S & \ket{1}_S \\
        \ket{0}_S & {\rm Tr}\{(Q_{00})^n\} & 0\\
        \ket{1}_S & 0 & {\rm Tr}\{(Q_{11})^n\}
        \end{pNiceMatrix}.
    \end{align}
    When ${\rm Tr}_E\{V_{SE}\gamma_E\} = 0$ is satisfied, we have ${\rm Tr}_{E}\{V_{SE}^{(1)}\gamma_E\}=0$, and thus, ${\rm Tr}\{Q_{00}\} = {\rm Tr}\{Q_{11}\} = 0$.
    Using the Cayley-Hamilton Theorem~\citep[][p.~312]{axler2024linear} for two-dimensional matrices, we have
    \begin{align}
        Q_{00}^2 &= {\rm Tr}\{Q_{00}\}Q_{00} - {\rm det}(Q_{00})\mathds{1}_2 = - {\rm det}(Q_{00})\mathds{1}_2,\\
        Q_{11}^2 &= {\rm Tr}\{Q_{11}\}Q_{11} - {\rm det}(Q_{11})\mathds{1}_2 = - {\rm det}(Q_{11})\mathds{1}_2.
    \end{align}
    Therefore, all odd powers of $Q_{00}$ and $Q_{11}$ are traceless. This leads to ${\rm Tr}_{E}\{(V_{SE}^{(1)})^n\gamma_E\} = 0$, and therefore, ${\rm Tr}_{E}\{(V_{SE})^n\gamma_E\} = 0$ for all odd $n\ge 1$.
\end{proof}
Since in Protocol~\ref{tab:protocol}, ${\rm Tr}_E\{V_{SE}\gamma_E\}$ appears as the Lamb shift, which vanishes in ECTCMs, by Lem.~\ref{applem:vanishing_odd_moments}, it follows that ${\rm Tr}_{E}\{(V_{SE})^n\gamma_E\} = 0$ for all odd $n\ge 1$.
We now show how the vanishing odd moments ${\rm Tr}_{E}\{(V_{SE})^n\gamma_E\}$ lead to the optimal error scaling. 

As discussed around Eq.~(\ref{eq:R_c_for_E_k}), the single-collision error $\varepsilon_{\rm c}(i) = \|\hat{\mathcal{E}}_i - (\mathcal{I}+\mathcal{L}_i\Delta t)\|_{1\rightarrow 1}$ can be decomposed into three parts: terms involving only $H_i\equiv H_{S,i}' + H_{E_i}$, whose lowest-order contribution is of order $H_i^2$, terms involving only $V_{SE_i}$ with lowest-order contribution arising at $(V_{SE_i})^3$, and mixed terms containing both $H_i$ and $V_{SE_i}$ with the lowest-order term as $H_i V_{SE_i}$. Since $H_{S,i}$ is accompanied by a factor of $\Delta t$, while $V_{SE_i}$ is multiplied by $\sqrt{\Delta t}$, these three contributions therefore scale as $O((\Delta t)^2)$, $O((\Delta t)^{3/2})$ and $O((\Delta t)^{3/2})$, respectively, in the limit of sufficiently small $\Delta t$. Now, when all odd moments ${\rm Tr}_{E_i}\{(V_{SE_i})^n\gamma_{E_i}\}$ vanish, the lowest-order contributions in the second and the third terms vanish, rendering all three contributions of order $O((\Delta t)^2)$. Therefore, substituting this scaling of $\varepsilon_{\rm c}(i)$ into the total error $\varepsilon$ [Eq.~(\ref{appeq:varepsilon_bound_gen})] and noting that $n\equiv \tau/\Delta t$, we find that after $n$ collisions, the error $\varepsilon$ scales as $n O((\Delta t)^2) = O(\Delta t)$, achieving the desired scaling.

More explicitly, following the similar lines as in Ref.~\citep[][Suppl. Mat.]{Cattaneo2021CMcansimulate}, we modify the expansions of the remainders $R_{\rm c}^{(V)}$ and $R_{\rm c}^{(VH)}$ in Eq.~(\ref{appeq:R_c_decomposition}) by removing the vanishing terms with odd order of $V_{SE_i}$:
\begin{widetext}
    \begin{align}
        {\rm R}_{\rm c}^{(V)} &= \sum_{m=2}^\infty \frac{(2\Delta t g \|V_{SE_i}\|_\infty)^{(2m)}}{(2m)!} \le \frac{(2\Delta t g \|V_{SE_i}\|_\infty)^4}{4!}\cosh(2\Delta t g \|V_{SE_i}\|_\infty) \nonumber\\
        & = \frac{2}{3}\|V_{SE_i}\|_\infty^4 (\Delta t)^2 \cosh(2\sqrt{\Delta t} \|V_{SE_i}\|_\infty),\\
        {\rm R}_{\rm c}^{(VH)} &= \sum_{m=1}^\infty \frac{(2\Delta t)^{2m+1} \sum_{m'=1}^{m}\begin{pmatrix}\begin{smallmatrix}
         2m+1\\
         2m'
    \end{smallmatrix}
    \end{pmatrix}\|H_i\|_\infty^{2m-2m'+1}(g\|V_{SE_i}\|_\infty)^{2m'}}{(2m+1)!} \nonumber \\
    &\quad + \sum_{m=2}^\infty \frac{(2\Delta t)^{2m} \sum_{m'=1}^{m-1}\begin{pmatrix}\begin{smallmatrix}
         2m\\
         2m'
    \end{smallmatrix}
    \end{pmatrix}\|H_i\|_\infty^{2m-2m'}(g\|V_{SE_i}\|_\infty)^{2m'}}{(2m)!}\\
    &\le \|H_i\|_\infty\sum_{m=1}^\infty\frac{(4\Delta t)^{2m+1}(g\|V_{SE_i}\|_\infty)^{2m}}{(2m+1)!} + \|H_i\|_\infty^2 \sum_{m=2}^\infty \frac{(4\Delta t)^{2m}(g\|V_{SE_i}\|_\infty)^{2m-2}}{(2m)!}\\
    &\le 8\|H_i\|_\infty\|V_{SE_i}\|_\infty (\Delta t)^{3/2}\sinh(4\sqrt{\Delta t}\|V_{SE_i}\|_\infty) + \frac{32}{3}\|H_i\|_\infty^2\|V_{SE_i}\|_\infty^2(\Delta t)^3\cosh(4\sqrt{\Delta t}\|V_{SE_i}\|_\infty),
    \end{align}
\end{widetext}
where in the first inequality in $R_{\rm c}^{(VH)}$, we used the condition $g\|V_{SE_i}\|_\infty \ge \|H_i\|_\infty$ and the fact that $\sum_{m'=1}^{m}\begin{pmatrix}
    \begin{smallmatrix}
         2m+1\\
         2m'
    \end{smallmatrix}
    \end{pmatrix}$ and $\sum_{m'=1}^{m-1}\begin{pmatrix}\begin{smallmatrix}
         2m\\
         2m'
    \end{smallmatrix}
    \end{pmatrix}$ are upper bounded by $2^{2m+1}$ and $2^{2m}$, respectively. Since for small $\Delta t$, $\sinh(4\sqrt{\Delta t}\|V_{SE_i}\|_\infty)$ scales as $O(\sqrt{\Delta t})$, $R_{\rm c}^{(VH)}$ in total scales as $O((\Delta t)^2)$. The same holds for $R_{\rm c}^{(V)}$. One can therefore use the modified bounds on $R_{\rm c}^{(V)}$ and $R_{\rm c}^{(VH)}$ to bound $\varepsilon_{\rm c}(i)$, and hence $\varepsilon$, which consequently scales an $O(\Delta t)$ scaling upon taking $n\equiv \tau/\Delta t$.

\begin{figure}[]
    \centering
    \includegraphics[width=\linewidth]{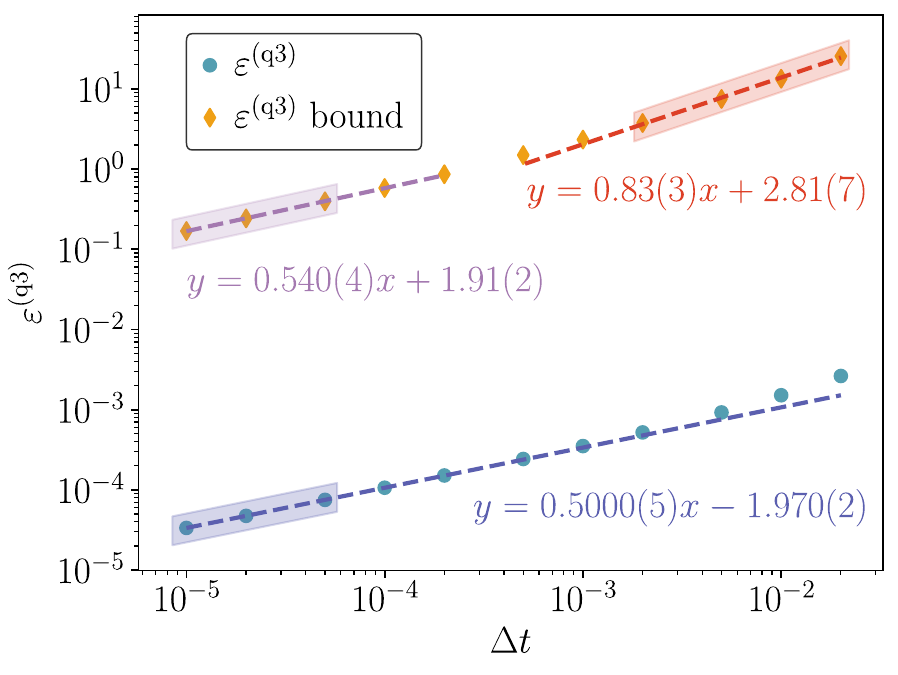}
    \caption{Error scaling of the ECTCM simulation of the qutrit model. The initial state is $\rho_S(0)=\ketbra{\psi}{\psi}_S$ with $\ket{\psi}_S=\sum_{j=0}^2 \ket{j}_S/\sqrt{3}$. Similar error scaling is observed for other initial states.
    The evolution time is $\tau=2$. The system remains out of equilibrium at time $\tau$, as quantified by $\|\rho_S^{\rm (q3)}(\tau)-\gamma_S\|_1 \simeq 0.32$. The error $\varepsilon^{\rm (q3)}$ and its bound [Eq.~(\ref{eq:error_bound})] are analysed via linear fits in the log-log plot: the blue-dashed and purple-dashed curved are fitted using the first four data points of $\varepsilon^{\rm (q3)}$ and its bound, respectively, while
    the red-dashed curve is obtained from the last four data points of the $\varepsilon^{\rm (q3)}$ bound, highlighted by the shaded regions. Uncertainties of fitted slopes and intercepts are given in parentheses on the last digit.}
    \label{fig:error_scaling_qutrit}
\end{figure}
    
    Since the proof of the vanishing odd moments ${\rm Tr}_{E_i}\{(V_{SE_i})^n\gamma_{E_i}\}$ (Lem.~\ref{applem:vanishing_odd_moments}) highly depends on the condition $d_S = 2$, in higher dimensional systems, we do not expect this favourable scaling of simulation errors to hold. Here, we present a concrete example of a three-dimensional system and show that the simulation error of Protocol~\ref{tab:protocol} scales as $O(\sqrt{\Delta t})$ for sufficiently small $\Delta t$.

    Let $S$ be a three-level system with two degenerate ground states, whose Hamiltonian is given by $H_S = \ln \sqrt{2}\ketbra{2}{2}_S$. A thermal Lindbladian $\mathcal{L}^{\rm (q3)}$ on $S$ satisfying Axioms~\ref{axiom:I}--\ref{axiom:III} is given by $\mathcal{L}^{\rm (q3)}(\cdot) = -\nI[H_S, \cdot] + \mathcal{D}^{\rm (q3)}(\cdot)$ where the background temperature $\beta = 2$ and the dissipator
    \begin{align}
        \mathcal{D}^{\rm (q3)}(\cdot) &= \sum_{j,j',m,m'=0}^{2}\xi_{jj'mm'}^{\rm (q3)}F_S^{jj'}(\cdot)(F_S^{mm'})^\dagger \nonumber\\
        &\quad - \sum_{j,j',m,m'=0}^{2}\xi_{jj'mm'}^{\rm (q3)}\frac{1}{2}[(F_S^{mm'})^\dagger F_S^{jj'}, (\cdot)]_+,
\end{align}
is charaterised by the coefficient $\xi_{jj'mm'}^{\rm (q3)}$ with the transition operators defined as $\{F_{S}^{jj'}\equiv\ketbra{j}{j'}_S\}_{j,j'=0}^{2}$. Grouping $(j,j')$ and $(m,m')$, respectively, $\xi_{jj'mm'}^{\rm (q3)}$ can be written as a complex matrix $\bm{\xi}\in\mathbb{C}^{9\times 9}$.
Concretely, we consider the following specific $\bm{\xi}^{\rm (q3)}$:
\begin{widetext}
    \begin{align}
        \bm{\xi}^{\rm (q3)} = \begin{pNiceMatrix}[first-row, first-col]
              & (0,0) & (0,1) & (0,2) & (1,0) & (1,1) & (1,2) & (2,0) & (2,1) & (2,2) \\
        (0,0) & 5 & -4 & 0 & -4 & -3 & 0 & 0 & 0 & -3 \\
        (0,1) & -4 & 5 & 0 & 3 & 2 & 0 & 0 & 0 & 2\\
        (0,2) & 0 & 0 & 3 & 0 & 0 & 2 & 0 & 0 & 0\\
        (1,0) & -4 & 3 & 0 & 5 & 2 & 0 & 0 & 0 & 2\\
        (1,1) & -3 & 2 & 0 & 2 & 5 & 0 & 0 & 0 & 1\\
        (1,2) & 0 & 0 & 2 & 0 & 0 & 3 & 0 & 0 & 0\\
        (2,0) & 0 & 0 & 0 & 0 & 0 & 0 & 1.5 & 1 & 0\\
        (2,1) & 0 & 0 & 0 & 0 & 0 & 0 & 1 & 1.5 & 0\\
        (2,2) & -3 & 2 & 0 & 2 & 1 & 0 & 0 & 0 & 5
        \end{pNiceMatrix} \times 0.1,
    \end{align}
    which satisfies the conditions [Eqs.~(\ref{eq:hermiticity_xi}), (\ref{eq:PSD_xi}), (\ref{eq:xi_energy_gaps_matching}) and (\ref{eq:QDB_xi})] we require for a thermal Lindbladian.
\end{widetext}

Following Protocol~\ref{tab:protocol}, we can construct an ECTCM to simulate this qutrit's dynamics.
For an initial state $\rho_S(0)$, we write the final state at time $\tau$ as $\rho_S(\tau)\equiv \nE^{\mathcal{L}^{\rm (q3)}\tau}(\rho_S(0))$, and the final state after $n=\tau/\Delta t$ steps in the ECTCM with finite collision time $\Delta t$ as $\rho_{S,n}^{\rm (CM)}\equiv \mathcal{E}_n\left(\rho_S(0)\right)$, where $\mathcal{E}_n$ is the collision map [Eq.~(\ref{eq:collision_map_E_n})]. The simulation error $\varepsilon^{\rm (q3)} \equiv \|\rho_S(\tau) - \rho_{S,n}^{\rm (CM)}\|_1$ is upper-bounded by the $1\rightarrow 1$ norm $\|\nE^{\mathcal{L}^{\rm (q3)}\tau} - \mathcal{E}_n\|_{1\rightarrow 1}$, and is therefore bounded by Eq.~(\ref{eq:error_bound}).

In Fig.~\ref{fig:error_scaling_qutrit}, we plot the scalings of $\varepsilon^{\rm (q3)}$ and its bound [Eq.~(\ref{eq:error_bound})] as functions of the collision time $\Delta t$, for a fixed initial state $\rho_S^{\rm (q3)}(0)$. The decay of $\varepsilon^{\rm (q3)}$ with decreasing $\Delta t$ demonstrates the faithful ECTCM simulation of $\mathcal{L}^{\rm (q3)}$. Besides, $\varepsilon^{\rm (q3)}$ follows a trend consistent with the prediction of the upper bound [Eq.~(\ref{eq:error_bound})], exhibiting a crossover from a faster decay regime to an $O(\sqrt{\Delta t})$ scaling in the small $\Delta t$ limit. 

\vfill
\FloatBarrier
\bibliography{references}

@article{Lostaglio2017coherence,
  title = {Markovian evolution of quantum coherence under symmetric dynamics},
  author = {Lostaglio, Matteo and Korzekwa, Kamil and Milne, Antony},
  journal = {Phys. Rev. A},
  volume = {96},
  issue = {3},
  pages = {032109},
  numpages = {20},
  year = {2017},
  month = {Sep},
  publisher = {American Physical Society},
  doi = {10.1103/PhysRevA.96.032109},
  url = {https://link.aps.org/doi/10.1103/PhysRevA.96.032109}
}

@article{horodecki2013fundamental,
  title={Fundamental limitations for quantum and nanoscale thermodynamics},
  author={Horodecki, Micha{\l} and Oppenheim, Jonathan},
  journal={Nat. Commun.},
  volume={4},
  number={1},
  pages={2059},
  year={2013},
  publisher={Nature Publishing Group UK London},
  doi={10.1038/ncomms3059}
}

@article{lostaglio2019introductory,
  title={An introductory review of the resource theory approach to thermodynamics},
  author={Lostaglio, Matteo},
  journal={Rep. Prog. Phys.},
  volume={82},
  number={11},
  pages={114001},
  year={2019},
  publisher={IOP Publishing},
  doi={10.1088/1361-6633/ab46e5}
}

@article{ciccarello2013CMnonMarkovian,
  title = {Collision-model-based approach to non-Markovian quantum dynamics},
  author = {Ciccarello, F. and Palma, G. M. and Giovannetti, V.},
  journal = {Phys. Rev. A},
  volume = {87},
  issue = {4},
  pages = {040103},
  numpages = {5},
  year = {2013},
  month = {Apr},
  publisher = {American Physical Society},
  doi = {10.1103/PhysRevA.87.040103},
  url = {https://link.aps.org/doi/10.1103/PhysRevA.87.040103}
}

@article{ciccarello2022quantum,
  title={Quantum collision models: Open system dynamics from repeated interactions},
  author={Ciccarello, Francesco and Lorenzo, Salvatore and Giovannetti, Vittorio and Palma, G Massimo},
  journal={Phys. Rep.},
  volume={954},
  pages={1--70},
  year={2022},
  doi={10.1016/j.physrep.2022.01.001},
  publisher={Elsevier}
}

@book{breuer2002theory,
  title={The theory of open quantum systems},
  author={Breuer, Heinz-Peter and Petruccione, Francesco},
  year={2002},
  doi={10.1093/acprof:oso/9780199213900.001.0001},
  publisher={OUP Oxford}
}

@article{son2024hierarchy,
  title={A hierarchy of thermal processes collapses under catalysis},
  author={Son, Jeongrak and Ng, Nelly HY},
  journal={Quantum Sci. Technol.},
  volume={10},
  number={1},
  pages={015011},
  year={2024},
  doi={10.1088/2058-9565/ad7ef5},
  publisher={IOP Publishing}
}

@article{spaventa2022capacity,
  title     = {Capacity of non-{M}arkovianity to boost the efficiency of molecular switches},
  author    = {Spaventa, Giovanni and Huelga, Susana F and Plenio, Martin B},
  journal   = {Phys. Rev. A},
  volume    = {105},
  number    = {1},
  pages     = {012420},
  year      = {2022},
  publisher = {APS},
  url       = {https://arxiv.org/abs/2103.14534}
}

@article{HOLEVO1993211,
title = {A note on covariant dynamical semigroups},
journal = {Rep. Math. Phys.},
volume = {32},
number = {2},
pages = {211-216},
year = {1993},
issn = {0034-4877},
doi = {https://doi.org/10.1016/0034-4877(93)90014-6},
author = {A.S. Holevo},
}

@article{ROGA2010311,
title = {Davies maps for qubits and qutrits},
journal = {Rep. Math. Phys.},
volume = {66},
number = {3},
pages = {311-329},
year = {2010},
issn = {0034-4877},
doi = {https://doi.org/10.1016/S0034-4877(11)00003-6},
url = {https://www.sciencedirect.com/science/article/pii/S0034487711000036},
author = {Wojciech Roga and Mark Fannes and Karol Życzkowski}
}

@article{Lostaglio2022Continuous,
  title = {Continuous thermomajorization and a complete set of laws for {M}arkovian thermal processes},
  author = {Lostaglio, Matteo and Korzekwa, Kamil},
  journal = {Phys. Rev. A},
  volume = {106},
  issue = {1},
  pages = {012426},
  numpages = {18},
  year = {2022},
  month = {Jul},
  publisher = {American Physical Society},
  doi = {10.1103/PhysRevA.106.012426},
  url = {https://link.aps.org/doi/10.1103/PhysRevA.106.012426}
}

@article{Stinespring_1955,
    author = {Stinespring, W. Forrest},
    title = {{Positive functions on $C^{\ast}$-algebras}},
    journal = {Proc. Am. Math. Soc.},
    volume = {6},
    pages = {211},
    year = {1955},
    doi = {10.1090/S0002-9939-1955-0069403-4},
    url = {https://doi.org/10.1090/S0002-9939-1955-0069403-4}
}

@article{Cattaneo2021CMcansimulate,
  title = {{Collision Models Can Efficiently Simulate Any Multipartite Markovian Quantum Dynamics}},
  author = {Cattaneo, Marco and De Chiara, Gabriele and Maniscalco, Sabrina and Zambrini, Roberta and Giorgi, Gian Luca},
  journal = {Phys. Rev. Lett.},
  volume = {126},
  issue = {13},
  pages = {130403},
  numpages = {8},
  year = {2021},
  month = {Apr},
  publisher = {American Physical Society},
  doi = {10.1103/PhysRevLett.126.130403},
  url = {https://link.aps.org/doi/10.1103/PhysRevLett.126.130403}
}

@article{Kliesch2011DissipativeQCTtheorem,
  title = {{Dissipative Quantum Church-Turing Theorem}},
  author = {Kliesch, M. and 
  Barthel, T. and Gogolin, C. and Kastoryano, M. and Eisert, J.},
  journal = {Phys. Rev. Lett.},
  volume = {107},
  issue = {12},
  pages = {120501},
  numpages = {5},
  year = {2011},
  month = {Sep},
  publisher = {American Physical Society},
  doi = {10.1103/PhysRevLett.107.120501},
  url = {https://link.aps.org/doi/10.1103/PhysRevLett.107.120501}
}

@article{Dollard1977ProductIntegrals,
    author = {Dollard, John D. and Friedman, Charles N.},
    title = {Product integrals and the {S}chrödinger equation},
    journal = {J. Math. Phys.},
    volume = {18},
    number = {8},
    pages = {1598-1607},
    year = {1977},
    month = {08},
    issn = {0022-2488},
    doi = {10.1063/1.523446},
    url = {https://doi.org/10.1063/1.523446},
}

@book{Dollard_Friedman_1984, 
    place={Cambridge}, 
    series={Encyclopedia of Mathematics and its Applications}, 
    title={Product Integration with Application to Differential Equations}, 
    publisher={Cambridge University Press}, 
    author={Dollard, John Day and Friedman, Charles N.}, 
    year={1984}, 
    doi={10.1017/CBO9781107340701},
    collection={Encyclopedia of Mathematics and its Applications}
}

@book{watrous2018theory,
  title={The theory of quantum information},
  author={Watrous, John},
  year={2018},
  doi={10.1017/9781316848142},
  publisher={Cambridge University Press}
}

@article{vomEnde2023ExploringLimits,
author = {vom Ende, Frederik and Malvetti, Emanuel and Dirr, Gunther and Schulte-Herbr\"{u}ggen, Thomas},
title = {{Exploring the Limits of Controlled Markovian Quantum Dynamics with Thermal Resources}},
journal = {Open Syst. Inf. Dyn},
volume = {30},
number = {01},
pages = {2350005},
year = {2023},
doi = {10.1142/S1230161223500051}
}

@article{AMORIM2021389,
title = {Complete positivity and self-adjointness},
journal = {Linear Algebra Appl.},
volume = {611},
pages = {389-439},
year = {2021},
issn = {0024-3795},
doi = {https://doi.org/10.1016/j.laa.2020.10.038},
author = {Érik Amorim and Eric A. Carlen}
}

@article{Alhambra2017Dynamicalmaps,
  title = {Dynamical maps, quantum detailed balance, and the {P}etz recovery map},
  author = {Alhambra, \'Alvaro M. and Woods, Mischa P.},
  journal = {Phys. Rev. A},
  volume = {96},
  issue = {2},
  pages = {022118},
  numpages = {12},
  year = {2017},
  month = {Aug},
  publisher = {American Physical Society},
  doi = {10.1103/PhysRevA.96.022118},
  url = {https://link.aps.org/doi/10.1103/PhysRevA.96.022118}
}

@book{alicki2007quantum,
  title={Quantum dynamical semigroups and applications},
  author={Alicki, Robert},
  year={2007},
  doi={10.1007/3-540-70861-8},
  publisher={Springer}
}

@article{brandao2013resource,
  title={Resource theory of quantum states out of thermal equilibrium},
  author={Brand{\~a}o, Fernando GSL and Horodecki, Micha{\l} and Oppenheim, Jonathan and Renes, Joseph M and Spekkens, Robert W},
  journal={Phys. Rev. Lett.},
  volume={111},
  number={25},
  pages={250404},
  year={2013},
  doi={10.1103/PhysRevLett.111.250404},
  publisher={APS}
}

@article{davies1974markovian,
  title={Markovian master equations},
  author={Davies, E Brian},
  journal={Commun. Math. Phys.},
  volume={39},
  number={2},
  pages={91--110},
  year={1974},
  doi={10.1007/BF01608389},
  publisher={Springer}
}

@article{moroder_thermodynamics_2024,
    title = {Thermodynamics of the {Quantum} {Mpemba} {Effect}},
    volume = {133},
    issn = {0031-9007, 1079-7114},
    url = {https://link.aps.org/doi/10.1103/PhysRevLett.133.140404},
    doi = {10.1103/PhysRevLett.133.140404},
    number = {14},
    urldate = {2025-09-05},
    journal = {Phys. Rev. Lett.},
    author = {Moroder, Mattia and Culhane, Oisín and Zawadzki, Krissia and Goold, John},
    month = {oct},
    year = {2024},
    pages = {140404},
}

@article{lacroix_making_2025,
    title = {Making quantum collision models exact},
    volume = {8},
    issn = {2399-3650},
    url = {https://www.nature.com/articles/s42005-025-02201-2},
    doi = {10.1038/s42005-025-02201-2},
    number = {1},
    urldate = {2025-11-19},
    journal = {Commun. Phys.},
    author = {Lacroix, Thibaut and Cilluffo, Dario and Huelga, Susana F. and Plenio, Martin B.},
    month = jul,
    year = {2025},
    pages = {268},
}

@article{shiraishi2025recovery,
  title={Recovery of the second law in fully quantum thermodynamics},
  author={Shiraishi, Naoto and Takagi, Ryuji},
  journal={\href{ 	
https://doi.org/10.48550/arXiv.2510.05642}{arXiv:2510.05642 (2025)}}
}

@article{gorini1976completely,
  title={Completely positive dynamical semigroups of {$N$}-level systems},
  author={Gorini, Vittorio and Kossakowski, Andrzej and Sudarshan, Ennackal Chandy George},
  journal={J. Math. Phys.},
  volume={17},
  number={5},
  pages={821--825},
  year={1976},
  doi={10.1063/1.522979},
  publisher={American Institute of Physics}
}

@article{lindblad1976generators,
  title={On the generators of quantum dynamical semigroups},
  author={Lindblad, Goran},
  journal={Commun. Math. Phys.},
  volume={48},
  number={2},
  pages={119--130},
  year={1976},
  doi={10.1007/BF01608499},
  publisher={Springer}
}

@article{dann_open_2021,
    title = {Open system dynamics from thermodynamic compatibility},
    volume = {3},
    issn = {2643-1564},
    url = {https://link.aps.org/doi/10.1103/PhysRevResearch.3.023006},
    doi = {10.1103/PhysRevResearch.3.023006},
    number = {2},
    urldate = {2025-02-28},
    journal = {Phys. Rev. Res.},
    author = {Dann, Roie and Kosloff, Ronnie},
    month = apr,
    year = {2021},
    pages = {023006},
}

@article{hewgill2021quantum,
  title={Quantum thermodynamically consistent local master equations},
  author={Hewgill, Adam and De Chiara, Gabriele and Imparato, Alberto},
  journal={Phys. Rev. Res.},
  volume={3},
  number={1},
  pages={013165},
  year={2021},
  doi={10.1103/PhysRevResearch.3.013165},
  publisher={APS}
}

@article{dann2021quantum,
  title={Quantum thermo-dynamical construction for driven open quantum systems},
  author={Dann, Roie and Kosloff, Ronnie},
  journal={Quantum},
  volume={5},
  pages={590},
  year={2021},
  doi={10.22331/q-2021-11-25-590},
  publisher={Verein zur F{\"o}rderung des Open Access Publizierens in den Quantenwissenschaften}
}

@article{chen2025thermodynamically,
  title={{Thermodynamically Consistent Lindbladians for Quantum Stochastic Thermodynamics}},
  author={Chen, Jin-Fu},
  journal={\href{ 	
https://doi.org/10.48550/arXiv.2502.20118
}{arXiv:2502.20118 (2025)}}
}

@book{heinonen2001lectures,
  title={Lectures on analysis on metric spaces},
  author={Heinonen, Juha},
  year={2001},
  doi={10.1007/978-1-4613-0131-8},
  publisher={Springer Science \& Business Media}
}

@article{Bardet2023RapidThermal,
  title = {{Rapid Thermalization of Spin Chain Commuting Hamiltonians}},
  author = {Bardet, Ivan and Capel, \'Angela and Gao, Li and Lucia, Angelo and P\'erez-Garc\'{\i}a, David and Rouz\'e, Cambyse},
  journal = {Phys. Rev. Lett.},
  volume = {130},
  issue = {6},
  pages = {060401},
  numpages = {7},
  year = {2023},
  month = {Feb},
  publisher = {American Physical Society},
  doi = {10.1103/PhysRevLett.130.060401},
  url = {https://link.aps.org/doi/10.1103/PhysRevLett.130.060401}
}

@article{Householder1958,
  title = {Unitary Triangularization of a Nonsymmetric Matrix},
  volume = {5},
  ISSN = {1557-735X},
  url = {http://dx.doi.org/10.1145/320941.320947},
  DOI = {10.1145/320941.320947},
  number = {4},
  journal = {J. ACM},
  publisher = {Association for Computing Machinery (ACM)},
  author = {Householder,  Alston S.},
  year = {1958},
  month = oct,
  pages = {339–342}
}

@book{paule2018thermodynamics,
  title={Thermodynamics and synchronization in open quantum systems},
  author={Paule, Gonzalo Manzano},
  year={2018},
  doi={10.1007/978-3-319-93964-3},
  publisher={Springer}
}

@inproceedings{cleve2017efficient,
  title={{Efficient Quantum Algorithms for Simulating Lindblad Evolution}},
  author={Cleve, Richard and Wang, Chunhao},
  booktitle={ICALP 2017},
  series = {LIPIcs},
  volume={80},
  pages={17:1--17:14},
  doi={10.4230/LIPIcs.ICALP.2017.17},
  publisher={Schloss Dagstuhl – Leibniz-Zentrum für Informatik},
  year={2017}
}

@article{correa2014quantum,
  title={Quantum-enhanced absorption refrigerators},
  author={Correa, Luis A and Palao, Jos{\'e} P and Alonso, Daniel and Adesso, Gerardo},
  journal={Sci. Rep.},
  volume={4},
  number={1},
  pages={3949},
  year={2014},
  doi={10.1038/srep03949},
  publisher={Nature Publishing Group UK London}
}

@book{axler2024linear,
  title={Linear algebra done right},
  author={Axler, Sheldon},
  year={2024},
  doi={10.1007/978-3-031-41026-0},
  publisher={Springer}
}

@book{royden1988real,
  title={Real Analysis},
  author={Royden, H.L.},
  isbn={9780024041517},
  lccn={86033216},
  series={Lecture Notes in Mathematics},
  url={https://books.google.ie/books?id=J_1CCoeDXsgC},
  year={1988},
  publisher={Macmillan}
}

@article{popovic2018entropy,
  title={Entropy production and correlations in a controlled non-{M}arkovian setting},
  author={Popovic, Maria and Vacchini, Bassano and Campbell, Steve},
  journal={Phys. Rev. A},
  volume={98},
  number={1},
  pages={012130},
  year={2018},
  doi={10.1103/PhysRevA.98.012130},
  publisher={APS}
}

@article{strasberg2019non,
  title={Non-{M}arkovianity and negative entropy production rates},
  author={Strasberg, Philipp and Esposito, Massimiliano},
  journal={Phys. Rev. E},
  volume={99},
  number={1},
  pages={012120},
  year={2019},
  doi={10.1103/PhysRevE.99.012120},
  publisher={APS}
}

@article{marvian2012symmetry,
  title={Symmetry, asymmetry and quantum information},
  author={Marvian, Iman},
  year={2012},
  publisher={University of Waterloo},
  journal={PhD Thesis},
  url={http://hdl.handle.net/10012/7088}
}

@article{Marvian2016QuantifyCoherence,
  title = {How to quantify coherence: Distinguishing speakable and unspeakable notions},
  author = {Marvian, Iman and Spekkens, Robert W.},
  journal = {Phys. Rev. A},
  volume = {94},
  issue = {5},
  pages = {052324},
  numpages = {23},
  year = {2016},
  month = {Nov},
  publisher = {American Physical Society},
  doi = {10.1103/PhysRevA.94.052324},
  url = {https://link.aps.org/doi/10.1103/PhysRevA.94.052324}
}

@book{kelly1979reversibility,
  title={Reversibility and stochastic networks},
  author={Kelly, Frank P},
  year={1979},
  publisher={J. Wiley},
  doi={10.1002/net.3230130110}
}

@article{esposito_second_2011,
    title = {Second law and {Landauer} principle far from equilibrium},
    volume = {95},
    issn = {0295-5075},
    number = {4},
    urldate = {2025-02-16},
    journal = {Europhys. Lett.},
    author = {Esposito, M. and Van den Broeck, C.},
    year = {2011},
    pages = {40004},
    doi={10.1209/0295-5075/95/40004}
}

@article{spohn1978entropy,
  title={Entropy production for quantum dynamical semigroups},
  author={Spohn, Herbert},
  journal={J. Math. Phys.},
  volume={19},
  number={5},
  pages={1227--1230},
  year={1978},
  publisher={American Institute of Physics},
  doi={10.1063/1.523789}
}

@article{horowitz2014equivalent,
  title={Equivalent definitions of the quantum nonadiabatic entropy production},
  author={Horowitz, Jordan M and Sagawa, Takahiro},
  journal={J. Stat. Phys.},
  volume={156},
  number={1},
  pages={55--65},
  year={2014},
  doi={s10955-014-0991-1},
  publisher={Springer}
}

@article{chitambar_quantum_2019,
    title = {Quantum resource theories},
    volume = {91},
    issn = {0034-6861, 1539-0756},
    url = {https://link.aps.org/doi/10.1103/RevModPhys.91.025001},
    doi = {10.1103/RevModPhys.91.025001},
    number = {2},
    urldate = {2023-11-13},
    journal = {Rev. Mod. Phys.},
    author = {Chitambar, Eric and Gour, Gilad},
    month = apr,
    year = {2019},
    pages = {025001},
}

@book{Gour_2025, place={Cambridge}, title={Quantum Resource Theories}, publisher={Cambridge University Press}, author={Gour, Gilad}, year={2025}, doi={10.1017/9781009560870}}

@article{cusumano2022quantum,
  title={Quantum collision models: {A} beginner guide},
  author={Cusumano, Stefano},
  journal={Entropy},
  volume={24},
  number={9},
  pages={1258},
  year={2022},
  doi={doi.org/10.3390/e24091258},
  publisher={MDPI}
}

@article{agarwal1973open,
  title={Open quantum {M}arkovian systems and the microreversibility},
  author={Agarwal, GS475498},
  journal={Z. Phys. A: Hadrons Nucl.},
  volume={258},
  number={5},
  pages={409--422},
  year={1973},
  doi={10.1007/BF01391504},
  publisher={Springer}
}

@article{alicki1976detailed,
  title={On the detailed balance condition for non-{H}amiltonian systems},
  author={Alicki, Robert},
  journal={Rep. Math. Phys},
  volume={10},
  number={2},
  pages={249--258},
  year={1976},
  doi={10.1016/0034-4877(76)90046-X},
  publisher={Elsevier}
}

@article{fagnola2009two,
  title={On two quantum versions of the detailed balance condition},
  author={Fagnola, Franco and Umanit{\`a}, Veronica},
  journal={Noncommut. Harmon. Anal. Appl. Probab. II, Banach Cent. Publ.},
  volume={89},
  pages={105--119},
  doi={10.4064/bc89-0-5},
  year={2010}
}

@article{scandi_thermalization_2026,
    title = {Thermalization in {Open} {Many}-{Body} {Systems} and {KMS} {Detailed} {Balance}},
    volume = {16},
    issn = {2160-3308},
    url = {https://link.aps.org/doi/10.1103/sfp3-3sqf},
    doi = {10.1103/sfp3-3sqf},
    number = {1},
    urldate = {2026-03-27},
    journal = {Phys. Rev. X},
    author = {Scandi, Matteo and Alhambra, \'Alvaro M.},
    year = {2026},
    pages = {011040},
}

@article{temme2010chi,
  title={The $\chi^2$-divergence and mixing times of quantum {M}arkov processes},
  author={Temme, Kristan and Kastoryano, Michael James and Ruskai, Mary Beth and Wolf, Michael Marc and Verstraete, Frank},
  journal={J. Math. Phys.},
  volume={51},
  pages={122201},
  year={2010},
  doi={10.1063/1.3511335},
  publisher={AIP Publishing}
}

@article{chen2025efficient,
  title={Efficient quantum thermal simulation},
  author={Chen, Chi-Fang and Kastoryano, Michael and Brand{\~a}o, Fernando GSL and Gily{\'e}n, Andr{\'a}s},
  journal={Nature},
  volume={646},
  number={8085},
  pages={561--566},
  year={2025},
  doi={10.1038/s41586-025-09583-x},
  publisher={Nature Publishing Group UK London}
}

@article{davies1976markovian,
  title={Markovian master equations. {II}},
  author={Davies, E Brian},
  journal={Math. Ann.},
  volume={219},
  number={2},
  pages={147--158},
  year={1976},
  doi={10.1007/BF01351898},
  publisher={Springer}
}

@article{alicki2009thermalization,
  title={On thermalization in {K}itaev's {2D} model},
  author={Alicki, Robert and Fannes, Mark and Horodecki, Michal},
  journal={J. Phys. A: Math. Theor.},
  volume={42},
  number={6},
  pages={065303},
  doi={10.1088/1751-8113/42/6/065303},
  year={2009}
}

@article{petz1986sufficient,
  title={Sufficient subalgebras and the relative entropy of states of a von {N}eumann algebra},
  author={Petz, D{\'e}nes},
  journal={Commun. Math. Phys.},
  volume={105},
  number={1},
  pages={123--131},
  year={1986},
  doi={10.1007/BF01212345},
  publisher={Springer}
}

@article{petz1988sufficiency,
  title={Sufficiency of channels over von {N}eumann algebras},
  author={Petz, D{\'e}nes},
  journal={Q. J. Math.},
  volume={39},
  number={1},
  pages={97--108},
  year={1988},
  doi={10.1093/qmath/39.1.97},
  publisher={Oxford University Press}
}

@article{parzygnat2023axioms,
  title={Axioms for retrodiction: achieving time-reversal symmetry with a prior},
  author={Parzygnat, Arthur J and Buscemi, Francesco},
  journal={Quantum},
  volume={7},
  pages={1013},
  doi = {10.22331/q-2023-05-23-1013},
  year={2023},
  publisher={Verein zur F{\"o}rderung des Open Access Publizierens in den Quantenwissenschaften}
}

@article{wilde2015recoverability,
  title={Recoverability in quantum information theory},
  author={Wilde, Mark M},
  journal={Proc. R. Soc. A},
  volume={471},
  pages={2182},
  year={2015},
  doi={10.1098/rspa.2015.0338},
  publisher={The Royal Society}
}

@article{parzygnat2023time,
  title={From time-reversal symmetry to quantum {B}ayes’ rules},
  author={Parzygnat, Arthur J and Fullwood, James},
  journal={PRX Quantum},
  volume={4},
  number={2},
  pages={020334},
  year={2023},
  doi={10.1103/PRXQuantum.4.020334},
  publisher={APS}
}

@mastersthesis{hubmann_open_2025,
    address = {Wien},
    title = {Open quantum evolution from thermodynamic collision models},
    url = {https://doi.org/10.25365/thesis.78592},
    author = {Hubmann, Felix},
    year = {2025},
}

@article{korzekwa2022optimizing,
  title={Optimizing thermalization},
  author={Korzekwa, Kamil and Lostaglio, Matteo},
  journal={Phys. Rev. Lett.},
  volume={129},
  number={4},
  pages={040602},
  year={2022},
  doi={10.1103/PhysRevLett.129.040602},
  publisher={APS}
}

@article{Kosloff2013QuantumThermo,
AUTHOR = {Kosloff, Ronnie},
TITLE = {Quantum Thermodynamics: A Dynamical Viewpoint},
JOURNAL = {Entropy},
VOLUME = {15},
YEAR = {2013},
NUMBER = {6},
PAGES = {2100--2128},
URL = {https://www.mdpi.com/1099-4300/15/6/2100},
ISSN = {1099-4300},
DOI = {10.3390/e15062100}
}

@article{burgarth_control_2023,
	title = {Control of {Quantum} {Noise}: {On} the {Role} of {Dilations}},
	volume = {24},
	issn = {1424-0661},
	shorttitle = {Control of {Quantum} {Noise}},
	url = {https://doi.org/10.1007/s00023-022-01211-y},
	doi = {10.1007/s00023-022-01211-y},
	number = {1},
	journal = {Ann. Henri Poincar{\'e}},
	author = {Burgarth, Daniel and Facchi, Paolo and Hillier, Robin},
	month = jan,
	year = {2023},
	pages = {325--347},
}

@article{burgarth2024taming,
  title={Taming the rotating wave approximation},
  author={Burgarth, Daniel and Facchi, Paolo and Hillier, Robin and Ligab{\`o}, Marilena},
  journal={Quantum},
  volume={8},
  pages={1262},
  year={2024},
  doi={10.22331/q-2024-02-21-1262},
  publisher={Verein zur F{\"o}rderung des Open Access Publizierens in den Quantenwissenschaften}
}

@article{trotter1959product,
  title={On the product of semi-groups of operators},
  author={Trotter, Hale F},
  journal={Proc. Am. Math. Soc.},
  volume={10},
  number={4},
  pages={545--551},
  year={1959},
  doi={10.2307/2033649},
  publisher={JSTOR}
}

@article{suzuki1976generalized,
  title={Generalized {T}rotter's formula and systematic approximants of exponential operators and inner derivations with applications to many-body problems},
  author={Suzuki, Masuo},
  journal={Commun. Math. Phys.},
  volume={51},
  number={2},
  pages={183--190},
  year={1976},
  doi={10.1007/BF01609348},
  publisher={Springer}
}

@article{Scovil59threelevel,
  title = {{Three-Level Masers as Heat Engines}},
  author = {Scovil, H. E. D. and Schulz-DuBois, E. O.},
  journal = {Phys. Rev. Lett.},
  volume = {2},
  issue = {6},
  pages = {262--263},
  numpages = {0},
  year = {1959},
  month = {Mar},
  publisher = {American Physical Society},
  doi = {10.1103/PhysRevLett.2.262},
  url = {https://link.aps.org/doi/10.1103/PhysRevLett.2.262}
}

@article{lostaglio_elementary_2018,
    title = {Elementary {Thermal} {Operations}},
    volume = {2},
    url = {https://quantum-journal.org/papers/q-2018-02-08-52/},
    doi = {10.22331/q-2018-02-08-52},
    urldate = {2023-10-09},
    journal = {Quantum},
    publisher = {Verein zur Förderung des Open Access Publizierens in den Quantenwissenschaften},
    author = {Lostaglio, Matteo and Alhambra, \'Alvaro M. and Perry, Christopher},
    year = {2018},
    pages = {52},
}

@article{linden2010small,
  title={How small can thermal machines be? {T}he smallest possible refrigerator},
  author={Linden, Noah and Popescu, Sandu and Skrzypczyk, Paul},
  journal={Phys. Rev. Lett.},
  volume={105},
  number={13},
  pages={130401},
  year={2010},
  doi={10.1103/PhysRevLett.105.130401},
  publisher={APS}
}

@article{kalaee2021violating,
  title={Violating the thermodynamic uncertainty relation in the three-level maser},
  author={Kalaee, Alex Aras
  
  h Sand and Wacker, Andreas and Potts, Patrick P},
  journal={Phys. Rev. E},
  volume={104},
  number={1},
  pages={L012103},
  year={2021},
  doi={10.1103/PhysRevE.104.L012103},
  publisher={APS}
}

@article{kosloff2014quantum,
  title={Quantum heat engines and refrigerators: {C}ontinuous devices},
  author={Kosloff, Ronnie and Levy, Amikam},
  journal={Annu. Rev. Phys. Chem.	},
  volume={65},
  number={1},
  pages={365--393},
  year={2014},
  doi={10.1146/annurev-physchem-040513-103724},
  publisher={Annual Reviews}
}

@article{boukobza2007three,
  title={Three-level systems as amplifiers and attenuators: {A} thermodynamic analysis},
  author={Boukobza, E and Tannor, DJ},
  journal={Phys. Rev. Lett.},
  volume={98},
  number={24},
  pages={240601},
  year={2007},
  doi={10.1103/PhysRevLett.98.240601},
  publisher={APS}
}

@article{faist_fundamental_2018,
    title = {Fundamental {Work} {Cost} of {Quantum} {Processes}},
    volume = {8},
    issn = {2160-3308},
    url = {https://link.aps.org/doi/10.1103/PhysRevX.8.021011},
    doi = {10.1103/PhysRevX.8.021011},
    number = {2},
    urldate = {2023-11-20},
    journal = {Phys. Rev. X},
    author = {Faist, Philippe and Renner, Renato},
    month = apr,
    year = {2018},
    pages = {021011},
}

@article{faist_thermodynamic_2019,
    title = {Thermodynamic {Capacity} of {Quantum} {Processes}},
    volume = {122},
    url = {https://link.aps.org/doi/10.1103/PhysRevLett.122.200601},
    doi = {10.1103/PhysRevLett.122.200601},
    number = {20},
    urldate = {2023-11-20},
    journal = {Phys. Rev. Lett.},
    publisher = {American Physical Society},
    author = {Faist, Philippe and Berta, Mario and Brand{\~a}o, Fernando},
    month = may,
    year = {2019},
    pages = {200601},
}

@article{czartowski_thermal_2023,
    title = {Thermal {Recall}: {Memory}-{Assisted} {Markovian} {Thermal} {Processes}},
    volume = {4},
    shorttitle = {Thermal {Recall}},
    url = {https://link.aps.org/doi/10.1103/PRXQuantum.4.040304},
    doi = {10.1103/PRXQuantum.4.040304},
    number = {4},
    urldate = {2023-10-09},
    journal = {PRX Quantum},
    publisher = {American Physical Society},
    author = {Czartowski, Jakub and de Oliveira Junior, A. and Korzekwa, Kamil},
    month = oct,
    year = {2023},
    pages = {040304},
}

@article{Hsieh2025Dynamical,
  title = {{Dynamical Landauer Principle: Quantifying Information Transmission by Thermodynamics}},
  author = {Hsieh, Chung-Yun},
  journal = {Phys. Rev. Lett.},
  volume = {134},
  issue = {5},
  pages = {050404},
  numpages = {9},
  year = {2025},
  month = {Feb},
  publisher = {American Physical Society},
  doi = {10.1103/PhysRevLett.134.050404},
  url = {https://link.aps.org/doi/10.1103/PhysRevLett.134.050404}
}

@article{luo_thermodynamic_2025,
  title = {Thermodynamic criteria for signaling in quantum channels},
  author = {Luo, Yutong and Milz, Simon and Binder, Felix C.},
  journal = {Phys. Rev. Res.},
  volume = {7},
  issue = {4},
  pages = {043327},
  numpages = {26},
  year = {2025},
  month = {Dec},
  publisher = {American Physical Society},
  doi = {10.1103/vjfk-vsmw},
  url = {https://link.aps.org/doi/10.1103/vjfk-vsmw}
}

@article{Vinjanampathy2016QTD,
author = {Sai Vinjanampathy and Janet Anders},
title = {Quantum thermodynamics},
journal = {Contemp. Phys.},
volume = {57},
number = {4},
pages = {545--579},
year = {2016},
publisher = {Taylor \& Francis},
doi = {10.1080/00107514.2016.1201896}
}

@article{Brando2015,
  title = {The second laws of quantum thermodynamics},
  volume = {112},
  ISSN = {1091-6490},
  url = {http://dx.doi.org/10.1073/pnas.1411728112},
  DOI = {10.1073/pnas.1411728112},
  number = {11},
  journal = {PNAS},
  publisher = {Proceedings of the National Academy of Sciences},
  author = {Brand{\~a}o,  Fernando and Horodecki,  Michał and Ng,  Nelly and Oppenheim,  Jonathan and Wehner,  Stephanie},
  year = {2015},
  month = Feb,
  pages = {3275–3279}
}

@article{Masanes2017,
  title = {A general derivation and quantification of the third law of thermodynamics},
  volume = {8},
  ISSN = {2041-1723},
  pages = {14538},
  doi = {10.1038/ncomms14538},
  journal = {Nat. Commun.},
  publisher = {Springer Science and Business Media LLC},
  author = {Masanes,  Lluís and Oppenheim,  Jonathan},
  year = {2017},
  month = Mar 
}

@article{janzing2000thermodynamic,
  title={{Thermodynamic Cost of Reliability and Low Temperatures: Tightening Landauer's Principle and the Second Law}},
  author={Janzing, D. and Wocjan, P. and Zeier, R. and Geiss, R. and Beth, {Th.}},
  journal={Int. J. Theor. Phys.},
  volume={39},
  number={12},
  pages={2717--2753},
  year={2000},
  doi={10.1023/A:1026422630734},
  publisher={Springer}
}

@article{Muller2018Correlating,
  title = {{Correlating Thermal Machines and the Second Law at the Nanoscale}},
  author = {M\"uller, Markus P.},
  journal = {Phys. Rev. X},
  volume = {8},
  issue = {4},
  pages = {041051},
  numpages = {23},
  year = {2018},
  month = {Dec},
  publisher = {American Physical Society},
  doi = {10.1103/PhysRevX.8.041051},
  url = {https://link.aps.org/doi/10.1103/PhysRevX.8.041051}
}

@article{vom2022bath,
  title={Which bath Hamiltonians matter for thermal operations?},
  author={Vom Ende, Frederik},
  journal={J. Math. Phys.},
  volume={63},
  pages={112202},
  year={2022},
  doi={10.1063/5.0117534},
  publisher={AIP Publishing}
}

@article{Shiraishi2021GPO,
  title = {Quantum Thermodynamics of Correlated-Catalytic State Conversion at Small Scale},
  author = {Shiraishi, Naoto and Sagawa, Takahiro},
  journal = {Phys. Rev. Lett.},
  volume = {126},
  issue = {15},
  pages = {150502},
  numpages = {6},
  year = {2021},
  month = {Apr},
  publisher = {American Physical Society},
  doi = {10.1103/PhysRevLett.126.150502},
  url = {https://link.aps.org/doi/10.1103/PhysRevLett.126.150502}
}

@article{Shiraishi2025CovGPO,
  title = {{Quantum Thermodynamics with Coherence: Covariant Gibbs-Preserving Operation Is Characterized by the Free Energy}},
  author = {Shiraishi, Naoto},
  journal = {Phys. Rev. Lett.},
  volume = {134},
  issue = {16},
  pages = {160402},
  numpages = {6},
  year = {2025},
  month = {Apr},
  publisher = {American Physical Society},
  doi = {10.1103/PhysRevLett.134.160402},
  url = {https://link.aps.org/doi/10.1103/PhysRevLett.134.160402}
}

@article{Cfiwikl2015Limitations,
  title = {Limitations on the Evolution of Quantum Coherences: Towards Fully Quantum Second Laws of Thermodynamics},
  author = {\ifmmode \acute{C}\else \'{C}\fi{}wikli\ifmmode \acute{n}\else \'{n}\fi{}ski, Piotr and Studzi\ifmmode \acute{n}\else \'{n}\fi{}ski, Micha\l{} and Horodecki, Micha\l{} and Oppenheim, Jonathan},
  journal = {Phys. Rev. Lett.},
  volume = {115},
  issue = {21},
  pages = {210403},
  numpages = {5},
  year = {2015},
  month = {Nov},
  publisher = {American Physical Society},
  doi = {10.1103/PhysRevLett.115.210403},
  url = {https://link.aps.org/doi/10.1103/PhysRevLett.115.210403}
}

@article{faist_gibbs-preserving_2015,
    title = {Gibbs-preserving maps outperform thermal operations in the quantum regime},
    volume = {17},
    issn = {1367-2630},
    url = {https://dx.doi.org/10.1088/1367-2630/17/4/043003},
    doi = {10.1088/1367-2630/17/4/043003},
    number = {4},
    urldate = {2023-09-22},
    journal = {New J. Phys.},
    publisher = {IOP Publishing},
    author = {Faist, Philippe and Oppenheim, Jonathan and Renner, Renato},
    year = {2015},
    pages = {043003},
}

@article{Tajima2025Gibbs-preserving,
  title = {Gibbs-Preserving Operations Requiring Infinite Amount of Quantum Coherence},
  author = {Tajima, Hiroyasu and Takagi, Ryuji},
  journal = {Phys. Rev. Lett.},
  volume = {134},
  issue = {17},
  pages = {170201},
  numpages = {7},
  year = {2025},
  month = {Apr},
  publisher = {American Physical Society},
  doi = {10.1103/PhysRevLett.134.170201},
  url = {https://link.aps.org/doi/10.1103/PhysRevLett.134.170201}
}

@article{Ding2021Exploring,
  title = {Exploring the gap between thermal operations and enhanced thermal operations},
  author = {Ding, Yuqiang and Ding, Feng and Hu, Xueyuan},
  journal = {Phys. Rev. A},
  volume = {103},
  issue = {5},
  pages = {052214},
  numpages = {9},
  year = {2021},
  month = {May},
  publisher = {American Physical Society},
  doi = {10.1103/PhysRevA.103.052214},
  url = {https://link.aps.org/doi/10.1103/PhysRevA.103.052214}
}

@article{Molitor2020Stroboscopic,
  title = {Stroboscopic two-stroke quantum heat engines},
  author = {Molitor, Otavio A. D. and Landi, Gabriel T.},
  journal = {Phys. Rev. A},
  volume = {102},
  issue = {4},
  pages = {042217},
  numpages = {10},
  year = {2020},
  month = {Oct},
  publisher = {American Physical Society},
  doi = {10.1103/PhysRevA.102.042217},
  url = {https://link.aps.org/doi/10.1103/PhysRevA.102.042217}
}

@article{melo2022implementation,
  title={Implementation of a two-stroke quantum heat engine with a collisional model},
  author={Melo, Filipe V and S{\'a}, Nahum and Roditi, Itzhak and Souza, Alexandre M and Oliveira, Ivan S and Sarthour, Roberto S and Landi, Gabriel T},
  journal={Phys. Rev. A},
  volume={106},
  number={3},
  pages={032410},
  year={2022},
  doi={10.1103/PhysRevA.106.032410},
  publisher={APS}
}

@article{piccione2021power,
  title={Power maximization of two-stroke quantum thermal machines},
  author={Piccione, Nicol{\`o} and De Chiara, Gabriele and Bellomo, Bruno},
  journal={Phys. Rev. A},
  volume={103},
  number={3},
  pages={032211},
  year={2021},
  doi={10.1103/PhysRevA.103.032211},
  publisher={APS}
}

@article{Uzdin2015Equivalence,
  title = {{Equivalence of Quantum Heat Machines, and Quantum-Thermodynamic Signatures}},
  author = {Uzdin, Raam and Levy, Amikam and Kosloff, Ronnie},
  journal = {Phys. Rev. X},
  volume = {5},
  issue = {3},
  pages = {031044},
  numpages = {21},
  year = {2015},
  month = {Sep},
  publisher = {American Physical Society},
  doi = {10.1103/PhysRevX.5.031044},
  url = {https://link.aps.org/doi/10.1103/PhysRevX.5.031044}
}

@article{lobejko2026equivalence,
  title={{Equivalence of Discrete and Continuous Otto-like Engines Assisted by Catalysts: Mapping Catalytic Advantages from the Discrete to the Continuous Framework}},
  author={{\L}obejko, Marcin and Biswas, Tanmoy and Horodecki, Micha{\l}},
  journal={Phys. Rev. Lett.},
  volume={136},
  number={18},
  pages={180402},
  year={2026},
  doi={10.1103/p2bv-hv5s},
  publisher={APS}
}

@article{oDonovan2025quantum,
  title={Quantum master equation from the eigenstate thermalization hypothesis},
  author={O'Donovan, Peter and Strasberg, Philipp and Modi, Kavan and Goold, John and Mitchison, Mark T},
  journal={Phys. Rev. B},
  volume={112},
  number={1},
  pages={014312},
  year={2025},
  doi={10.1103/kh96-ct3y},
  publisher={APS}
}

@article{brun2002simple,
  title={A simple model of quantum trajectories},
  author={Brun, Todd A},
  journal={Am. J. Phys.},
  volume={70},
  number={7},
  pages={719--737},
  year={2002},
  doi={10.1119/1.1475328},
  publisher={American Association of Physics Teachers}
}

@article{ziman2005description,
  title={Description of quantum dynamics of open systems based on collision-like models},
  author={Ziman, M{\"a}rio and {\v{S}}telmachovi{\v{c}}, Peter and Bu{\v{z}}ek, Vladim{\'\i}r},
  journal={Open Syst. Inf. Dyn.},
  volume={12},
  number={1},
  pages={81--91},
  year={2005},
  doi={10.1007/s11080-005-0488-0},
  publisher={Springer}
}

@article{Grimmer2016Open,
  title = {Open dynamics under rapid repeated interaction},
  author = {Grimmer, Daniel and Layden, David and Mann, Robert B. and Mart\'{\i}n-Mart\'{\i}nez, Eduardo},
  journal = {Phys. Rev. A},
  volume = {94},
  issue = {3},
  pages = {032126},
  numpages = {28},
  year = {2016},
  month = {Sep},
  publisher = {American Physical Society},
  doi = {10.1103/PhysRevA.94.032126}
}

@Inbook{ng2019resource,
author="Ng, Nelly Huei Ying
and Woods, Mischa Prebin",
editor="Binder, Felix
and Correa, Luis A.
and Gogolin, Christian
and Anders, Janet
and Adesso, Gerardo",
title={{Resource Theory of Quantum Thermodynamics: Thermal Operations and Second Laws}},
bookTitle="Thermodynamics in the Quantum Regime: Fundamental Aspects and New Directions",
year="2018",
publisher="Springer International Publishing",
address="Cham",
pages="625--650",
isbn="978-3-319-99046-0",
doi="10.1007/978-3-319-99046-0_26",
url="https://doi.org/10.1007/978-3-319-99046-0_26"
}

@incollection{alicki_introduction_2018,
    address = {Cham},
    title = {Introduction to {Quantum} {Thermodynamics}: {History} and {Prospects}},
    isbn = {978-3-319-99046-0},
    url = {https://doi.org/10.1007/978-3-319-99046-0_1},
    doi = {10.1007/978-3-319-99046-0_1},
    booktitle = {Thermodynamics in the {Quantum} {Regime}: {Fundamental} {Aspects} and {New} {Directions}},
    publisher = {Springer International Publishing},
    author = {Alicki, Robert and Kosloff, Ronnie},
    editor={Binder, Felix
    and Correa, Luis A.
    and Gogolin, Christian
    and Anders, Janet
    and Adesso, Gerardo},
    year = {2018},
    pages = {1--33},
}

@article{Campbell2018System,
  title = {System-environment correlations and {M}arkovian embedding of quantum non-{M}arkovian dynamics},
  author = {Campbell, Steve and Ciccarello, Francesco and Palma, G. Massimo and Vacchini, Bassano},
  journal = {Phys. Rev. A},
  volume = {98},
  issue = {1},
  pages = {012142},
  numpages = {11},
  year = {2018},
  month = {Jul},
  publisher = {American Physical Society},
  doi = {10.1103/PhysRevA.98.012142},
  url = {https://link.aps.org/doi/10.1103/PhysRevA.98.012142}
}

@article{Ptaszy2022NonMarkovian,
  title = {Non-Markovian thermal operations boosting the performance of quantum heat engines},
  author = {Ptaszy\ifmmode \acute{n}\else \'{n}\fi{}ski, Krzysztof},
  journal = {Phys. Rev. E},
  volume = {106},
  issue = {1},
  pages = {014114},
  numpages = {11},
  year = {2022},
  month = {Jul},
  publisher = {American Physical Society},
  doi = {10.1103/PhysRevE.106.014114},
  url = {https://link.aps.org/doi/10.1103/PhysRevE.106.014114}
}

@article{Dann2022NonMarkovian,
  title = {Non-Markovian dynamics under time-translation symmetry},
  author = {Dann, Roie and Megier, Nina and Kosloff, Ronnie},
  journal = {Phys. Rev. Res.},
  volume = {4},
  issue = {4},
  pages = {043075},
  numpages = {29},
  year = {2022},
  month = {Nov},
  publisher = {American Physical Society},
  doi = {10.1103/PhysRevResearch.4.043075},
  url = {https://link.aps.org/doi/10.1103/PhysRevResearch.4.043075}
}

@article{Lorenzo2017Composite,
  title = {Composite quantum collision models},
  author = {Lorenzo, Salvatore and Ciccarello, Francesco and Palma, G. Massimo},
  journal = {Phys. Rev. A},
  volume = {96},
  issue = {3},
  pages = {032107},
  numpages = {11},
  year = {2017},
  month = {Sep},
  publisher = {American Physical Society},
  doi = {10.1103/PhysRevA.96.032107},
  url = {https://link.aps.org/doi/10.1103/PhysRevA.96.032107}
}

@article{Zambon2025Quantum,
  title = {{Quantum Processes as Thermodynamic Resources: The Role of Non-Markovianity}},
  author = {Zambon, Guilherme and Adesso, Gerardo},
  journal = {Phys. Rev. Lett.},
  volume = {134},
  issue = {20},
  pages = {200401},
  numpages = {9},
  year = {2025},
  month = {May},
  publisher = {American Physical Society},
  doi = {10.1103/PhysRevLett.134.200401},
  url = {https://link.aps.org/doi/10.1103/PhysRevLett.134.200401}
}

@Inbook{Nazir2018,
author="Nazir, Ahsan
and Schaller, Gernot",
editor="Binder, Felix
and Correa, Luis A.
and Gogolin, Christian
and Anders, Janet
and Adesso, Gerardo",
title={{The Reaction Coordinate Mapping in Quantum Thermodynamics}},
bookTitle="Thermodynamics in the Quantum Regime: Fundamental Aspects and New Directions",
year="2018",
publisher="Springer International Publishing",
address="Cham",
pages="551--577",
isbn="978-3-319-99046-0",
doi="10.1007/978-3-319-99046-0_23",
url="https://doi.org/10.1007/978-3-319-99046-0_23"
}

@article{Nestmann2021How,
  title = {{How Quantum Evolution with Memory is Generated in a Time-Local Way}},
  author = {Nestmann, K. and Bruch, V. and Wegewijs, M. R.},
  journal = {Phys. Rev. X},
  volume = {11},
  issue = {2},
  pages = {021041},
  numpages = {22},
  year = {2021},
  month = {May},
  publisher = {American Physical Society},
  doi = {10.1103/PhysRevX.11.021041},
  url = {https://link.aps.org/doi/10.1103/PhysRevX.11.021041}
}

@book{gardiner2004quantum,
  title={{Quantum noise: a handbook of Markovian and non-Markovian quantum stochastic methods with applications to quantum optics}},
  author={Gardiner, Crispin and Zoller, Peter},
  year={2004},
  isbn={978-3-540-22301-6},
  publisher={Springer Science \& Business Media}
}

@article{hofer2017markovian,
  title={Markovian master equations for quantum thermal machines: local versus global approach},
  author={Hofer, Patrick P and Perarnau-Llobet, Mart{\'\i} and Miranda, L David M and Haack, G{\'e}raldine and Silva, Ralph and Brask, Jonatan Bohr and Brunner, Nicolas},
  journal={New J. Phys.},
  volume={19},
  number={12},
  pages={123037},
  year={2017},
  doi={10.1088/1367-2630/aa964f},
  publisher={IOP Publishing}
}

@article{fagnola2015entropy,
  title={Entropy production for quantum {M}arkov semigroups},
  author={Fagnola, Franco and Rebolledo, Rolando},
  journal={Commun. Math. Phys.},
  volume={335},
  number={2},
  pages={547--570},
  year={2015},
  doi={10.1007/s00220-015-2320-1},
  publisher={Springer}
}

@article{fagnola2015entropy_DB,
  title={Entropy production and detailed balance for a class of quantum {M}arkov semigroups},
  author={Fagnola, Franco and Rebolledo, Rolando},
  journal={Open Syst. Inf. Dyn.},
  volume={22},
  number={03},
  pages={1550013},
  year={2015},
  doi={10.1142/S1230161215500134},
  publisher={World Scientific}
}

@article{Potts_2021,
doi = {10.1088/1367-2630/ac3b2f},
url = {https://doi.org/10.1088/1367-2630/ac3b2f},
year = {2021},
month = {dec},
publisher = {IOP Publishing},
volume = {23},
number = {12},
pages = {123013},
author = {Potts, Patrick P and Kalaee, Alex Arash Sand and Wacker, Andreas},
title = {A thermodynamically consistent Markovian master equation beyond the secular approximation},
journal = {New J. Phys.}
}

@article{Leggett1987Dynamics,
  title = {Dynamics of the dissipative two-state system},
  author = {Leggett, A. J. and Chakravarty, S. and Dorsey, A. T. and Fisher, Matthew P. A. and Garg, Anupam and Zwerger, W.},
  journal = {Rev. Mod. Phys.},
  volume = {59},
  issue = {1},
  pages = {1--85},
  numpages = {0},
  year = {1987},
  month = {Jan},
  publisher = {American Physical Society},
  doi = {10.1103/RevModPhys.59.1},
  url = {https://link.aps.org/doi/10.1103/RevModPhys.59.1}
}

@article{DeRaedt1984Thermodynamics,
  title = {Thermodynamics of a two-level system coupled to bosons},
  author = {De Raedt, Bart and De Raedt, Hans},
  journal = {Phys. Rev. B},
  volume = {29},
  issue = {10},
  pages = {5325--5336},
  numpages = {0},
  year = {1984},
  month = {May},
  publisher = {American Physical Society},
  doi = {10.1103/PhysRevB.29.5325},
  url = {https://link.aps.org/doi/10.1103/PhysRevB.29.5325}
}

@article{Mehl2013Noise,
  title = {Noise analysis of qubits implemented in triple quantum dot systems in a {D}avies master equation approach},
  author = {Mehl, Sebastian and DiVincenzo, David P.},
  journal = {Phys. Rev. B},
  volume = {87},
  issue = {19},
  pages = {195309},
  numpages = {26},
  year = {2013},
  month = {May},
  publisher = {American Physical Society},
  doi = {10.1103/PhysRevB.87.195309},
  url = {https://link.aps.org/doi/10.1103/PhysRevB.87.195309}
}

@article{naeij2025open,
  title={Open quantum system approaches to superconducting qubits},
  author={Naeij, Hamid Reza},
  journal={Quantum Inf. Process.},
  volume={24},
  number={7},
  pages={220},
  year={2025},
  doi={10.1007/s11128-025-04838-y},
  publisher={Springer}
}

@article{cleri2024quantum,
  title={Quantum computers, quantum computing, and quantum thermodynamics},
  author={Cleri, Fabrizio},
  journal={Front. Quantum Sci. Technol.},
  volume={3},
  pages={1422257},
  year={2024},
  doi={10.3389/frqst.2024.1422257},
  publisher={Frontiers Media SA}
}
\end{document}